\documentclass{amsart}
\usepackage{graphicx}
\graphicspath{{images/}}

\usepackage{a4wide}
\usepackage[hidelinks]{hyperref}
\usepackage{xcolor}
\usepackage{verbatim}
\usepackage{stackrel}

\theoremstyle{plain}
\newtheorem{thm}{Theorem}
\newtheorem{prop}[thm]{Proposition}
\newtheorem{lemma}[thm]{Lemma}
\newtheorem{cor}[thm]{Corollary}
\theoremstyle{definition}
\newtheorem{definition}[thm]{Definition}

\newtheorem{remark}[thm]{Remark}

\usepackage{amssymb,amsmath,amsthm,mathrsfs}
\usepackage{graphicx}        
\usepackage{upgreek} 

\renewcommand{\max}{{\mathrm{max}}}
\newcommand{\roId}{{\mathrm{Id}}}

\newcommand{\rotr}{\mathrm{tr}}

\newcommand{\grad}{\mathrm{grad}}
\newcommand{\roloc}{\mathrm{loc}}

\newcommand{\tn}[1]{\ensuremath{\mathbb{T}^{#1}}}

\newcommand{\rn}[1]{\ensuremath{\mathbb{R}^{#1}}}

\newcommand{\sn}[1]{\ensuremath{\mathbb{S}^{#1}}}
\newcommand{\nn}[1]{\ensuremath{\mathbb{N}^{#1}}}
\newcommand{\bM}{\bar{M}}

\newcommand{\bS}{\bar{S}}
\newcommand{\bL}{\bar{L}}

\newcommand{\bx}{\bar{x}}

\newcommand{\bkappa}{\bar{\kappa}}

\newcommand{\g}{\gamma}

\newcommand{\dka}{\dot{\kappa}}
\newcommand{\dvka}{\dot{\varkappa}}
\newcommand{\Ndot}{\dot{N}}
\newcommand{\G}{\Gamma}
\newcommand{\de}{\delta}
\newcommand{\bge}{\bar{g}}

\renewcommand{\d}{\partial}

\newcommand{\mfl}{\mathfrak{l}}
\newcommand{\mfC}{\mathfrak{C}}
\newcommand{\mfS}{\mathfrak{S}}
\newcommand{\chmfC}{\check{\mathfrak{C}}}
\newcommand{\mfD}{\mathfrak{D}}
\newcommand{\mfL}{\mathfrak{L}}

\newcommand{\mff}{\mathfrak{f}}
\newcommand{\mfu}{\mathfrak{u}}
\newcommand{\mfv}{\mathfrak{v}}

\newcommand{\mfb}{\mathfrak{b}}

\newcommand{\dmfk}{\dot{\mathfrak{k}}}
\newcommand{\mfg}{\mathfrak{g}}

\newcommand{\bmfh}{\bar{\mathfrak{h}}}

\newcommand{\mfN}{\mathfrak{N}}
\newcommand{\mfn}{\mathfrak{n}}
\newcommand{\mfR}{\mathfrak{R}}
\newcommand{\mfK}{\mathfrak{K}}
\newcommand{\hmfK}{\hat{\mathfrak{K}}}

\newcommand{\chmfE}{\check{\mathfrak{E}}}
\newcommand{\chmfF}{\check{\mathfrak{F}}}
\newcommand{\tr}{\mathrm{tr}}

\newcommand{\rodiv}{\mathrm{div}}

\newcommand{\robg}{\mathrm{bg}}
\newcommand{\rofin}{\mathrm{fin}}
\newcommand{\roper}{\mathrm{per}}
\newcommand{\ro}{\mathbb{R}}
\newcommand{\bbE}{\mathbb{E}}
\newcommand{\bbP}{\mathbb{P}}

\newcommand{\bbU}{\mathbb{U}}
\newcommand{\bbV}{\mathbb{V}}
\newcommand{\rodim}{\mathrm{dim}}
\newcommand{\rotot}{\mathrm{tot}}
\newcommand{\rorem}{\mathrm{rem}}

\newcommand{\roinj}{\mathrm{inj}}

\newcommand{\ioih}{\iota_{\mathrm{ih}}}

\newcommand{\roScal}{\mathrm{Scal}}

\newcommand{\refer}{\mathrm{ref}}

\newcommand{\msA}{\mathscr{A}}
\newcommand{\msE}{\mathscr{E}}
\newcommand{\msR}{\mathscr{R}}

\newcommand{\msC}{\mathscr{C}}
\newcommand{\bmsC}{\overline{\mathscr{C}}}
\newcommand{\msD}{\mathscr{D}}
\newcommand{\msH}{\mathscr{H}}

\newcommand{\msO}{\mathscr{O}}
\newcommand{\ldr}[1]{\langle #1\rangle}
\newcommand{\sfh}{\mathsf{h}}
\newcommand{\sfg}{\mathsf{g}}
\newcommand{\sfx}{\mathsf{x}}
\newcommand{\sfy}{\mathsf{y}}

\newcommand{\bna}{\overline{\nabla}}
\newcommand{\mT}{\mathcal{T}}
\newcommand{\mK}{\mathcal{K}}
\newcommand{\mB}{\mathcal{B}}
\newcommand{\mD}{\mathcal{D}}
\newcommand{\mI}{\mathcal{I}}
\newcommand{\chmI}{\check{\mathcal{I}}}

\newcommand{\mJ}{\mathcal{J}}

\newcommand{\mC}{\mathcal{C}}
\newcommand{\mN}{\mathcal{N}}
\newcommand{\mM}{\mathcal{M}}
\newcommand{\mL}{\mathcal{L}}

\newcommand{\hOm}{\hat{\Omega}}

\newcommand{\mH}{\mathcal{H}}

\newcommand{\ml}{\mathcal{L}}

\renewcommand{\a}{\alpha}
\renewcommand{\b}{\beta}

\newcommand{\bfI}{\mathbf{I}}
\newcommand{\bfJ}{\mathbf{J}}
\newcommand{\bfK}{\mathbf{K}}

\newcommand{\chk}{\check{k}}
\newcommand{\chkappa}{\check{\kappa}}
\newcommand{\chN}{\check{N}}

\newcommand{\chxi}{\check{\xi}}
\newcommand{\chh}{\check{h}}

\newcommand{\che}{\check{e}}
\newcommand{\chE}{\check{E}}
\newcommand{\chF}{\check{F}}

\newcommand{\hE}{\hat{E}}

\newcommand{\hC}{\hat{C}}

\newcommand{\hu}{\hat{u}}
\newcommand{\hv}{\hat{v}}
\newcommand{\hf}{\hat{f}}
\newcommand{\chom}{\check{\omega}}
\newcommand{\chphi}{\check{\phi}}

\newcommand{\me}{\mathcal{E}}
\newcommand{\hme}{\hat{\me}}
\newcommand{\chme}{\check{\me}}

\newcommand{\e}{\epsilon}

\newcommand{\abs}[1]{|#1|}
\newcommand{\norm}[1]{\left\|#1\right\|}
\newcommand{\Vi}{\mathcal{V}}
\newcommand{\Ui}{\mathcal{U}}

\newcommand{\supp}{\mathrm{supp}}

\newcommand{\BaE}{\mathrm{E}}
\newcommand{\BaF}{\mathrm{F}}
\newcommand{\BaG}{\mathrm{G}}

\newcommand{\Vol}{{\mathrm{Vol}}}

\newcommand{\U}{{\mathcal U}}
\newcommand{\V}{{\mathcal V}}
\newcommand{\W}{{\mathcal W}}

\newcommand{\bnabla}{\overline{\nabla}}
\newcommand{\bR}{\bar{R}}

\renewcommand{\S}{\Sigma}
\newcommand{\s}{\sigma}

\newcommand{\n}{\nabla}

\newcommand{\h}{h}
\renewcommand{\k}{k}

\newcommand{\bh}{\bar \h}
\newcommand{\bk}{\bar \k}

\newcommand{\bga}{{\bar \gamma}}

\newcommand{\bGa}{{\bar \Gamma}}

\usepackage{subfiles}

\begin{document}

\author{Hans Ringstr\"{o}m}

\title[Local existence theory for a class of CMC gauges]{Local existence theory for a class of CMC gauges
  for the Einstein-non-linear scalar field equations}
\begin{abstract}
  The purpose of this article is to develop a local existence theory for a class of CMC gauges for the Einstein-non-linear
  scalar field equations. We do so in the context of closed and parallelisable initial manifolds. The assumption that the
  initial manifold is parallelisable is not a topological restriction in the case of closed oriented $3$-manifolds, but it
  is a restriction in higher dimensions. The results include local existence, uniqueness, a continuation criterion and Cauchy
  stability. Previous results concerning related gauges have been restricted
  to the case of $\tn{n}$ spatial topology. Moreover, they have not covered the Einstein-non-linear scalar field setting.
  The main motivation for developing this theory is that it forms the basis for a family of past global non-linear stability
  results we derive in a separate article. 
\end{abstract}

\maketitle

\section{Introduction}

In a sequence of works, \cite{rasql,rasq,rsh,specks3}, Rodnianski and Speck use what they refer to as CMC transported spatial coordinates
to prove stable big bang formation. They consider the Einstein-scalar field, stiff fluid and higher dimensional vacuum equations. With the
exception of \cite{specks3}, treating the case of $\sn{3}$-spatial topology, they assume $\tn{n}$-spatial topology. The articles \cite{rasq,specks3}
are concerned with with the stability of spatially homogeneous solutions, in the $3+1$-dimensional setting, that are close to isotropy. However,
\cite{rsh} concerns the stability of Kasner solutions in higher dimensions; i.e., spatially homogeneous vacuum solutions with spatial topology
$\tn{n}$ that are, by necessity, anisotropic. Finally, \cite{GIJ}, written in collaboration with Fournodavlos, demonstrates stability of spatially
homogeneous solutions with $\tn{n}$-spatial topology for the maximal range in which stability can hold. The matter models considered in \cite{GIJ}
are vacuum and scalar fields. However, in \cite{GIJ}, the authors consider a different system of equations, including a Fermi Walker transported frame. 
A related result is to be found in \cite{fau}. In this article, the authors prove stable big bang formation for the FLRW solutions to the Einstein
scalar field equations with compact hyperbolic spatial topology. Moreover, in this setting, the authors are able to prove future global non-linear
stability. In other words, they obtain a global non-linear stability result for the solutions they consider. It is important to note that all the
results mentioned so far are based on CMC gauges. One issue associated with this choice is that it is non-local, even though the causal structure
close to the big bang is local. The CMC gauges thus introduce a connection between different regions of spacetime that should, from the point of
view of physics, be independent. For these reasons, it is of interest to mention \cite{bao,baotwo,boz}, in which the authors do prove stable big bang
formation using a gauge that can be localised.

It is natural to try to extend the results of \cite{GIJ} to more general spatial topologies. Moreover, it is of interest to find a condition on
initial data that ensures big bang formation, in particular conclusions concerning the singularity similar to those of \cite{GIJ}, but which 
is independent of any background solution. There is also a natural notion of initial data on the singularity introduced in \cite{RinQC}. 
If one wishes to prove stable big bang formation, the most natural class of solutions to prove stability of is those with asymptotics corresponding
to data on the singularity. Finally, it is natural to wish to prove stable big bang formation generally for spatially locally homogeneous solutions
with quiescent asymptotics. As always, the Einstein-scalar field equations is a natural system to consider. However, given the interest in the physics
literature in both inflation and different models of dark energy, it is natural to consider the Einstein-non-linear scalar field equations.
In this case, it is reasonable to expect a large class of solutions with quiescent singularities as well as future asymptotics that can be demonstrated
to be stable. In other words, it is natural to expect a large class of solutions for which one can demonstrate both future and past global non-linear
stability.

We address all the questions raised in the previous paragraph in separate articles; one concerning the asymptotics of spatially homogeneous solutions
to the Einstein-non-linear scalar field equations, see \cite{RinSH}, and one concerning finding stable regimes etc., see \cite{GPR}. However, the
results require a local existence theory for the relevant gauge fixed equations. The purpose of this article is to provide such a theory.

The local existence theory we develop here concerns a class of gauge fixed equations which is similar to the equations used in \cite{GIJ}. However,
we here allow additional degrees of freedom: an antisymmetric matrix depending on the spacetime coordinates and the unknowns can be specified
freely etc. Moreover, we consider more general manifolds. In particular, we only require the spatial manifold to be closed and parallelisable.
Finally, we address the Einstein-non-linear scalar field setting. The main issues in proving local existence in this setting are similar to those
already addressed in \cite[Section~14.3, p.~4412-4420]{rasq}. However, in \cite{rasq}, the authors refer to \cite{CaK} as the main source of the
relevant ideas. The main complication is due to the appearance of too many derivatives of the metric components (in our case, the frame coefficients)
in the second order equation for the components of the second fundamental form. However, this issue can be addressed by integrating by parts in both
space and time; this idea seems to be common to all the local existence proofs and seems to originate with \cite{CaK}.

The actual equations that appear in the study of the Einstein-non-linear scalar field system are quite lengthy. For this reason, it is natural to
develop a local existence theory for an abstract model system and then to demonstrate that the Einstein-non-linear scalar field system can be
considered to be a special case of the general theory. Next, we describe the abstract model systems we consider and state the relevant well
posedness results. 

\subsection{Model systems}\label{ssection:modelsystem}
Let $\bM$ be a closed, connected and oriented $n$-dimensional manifold, $\mI\subset\rn{}$ be an open interval, $\bh_{\refer}$ be a Riemannian metric
on $\bM$ and assume that there is a global orthonormal frame $\{E_{i}\}$ on $\bM$ with respect to $\bh_{\refer}$. Let $M:=\bM\times \mI$ and
$n_{i}\in\nn{}$, $i=1,2$; here and in what follows we use the notation $\nn{}_0$ for the set of non-negative integers (i.e., $0,1,2,...$) and
$\nn{}$ for the set of positive integers (i.e., $1,2,3,\dots$). Let $f_{i}$, $i=1,2,3,4$, be smooth functions of the following types:
\begin{align*}
  & f_{1}:M\times\rn{n_{1}+n_{2}+n+1}\rightarrow\rn{n_{1}},\\
  & f_{2}:M\times\rn{(n+1)n_{1}+(n+2)n_{2}+(n+2)(n+1)}\rightarrow\rn{n_{2}},\\
  & f_{3}:M\times\rn{(n+1)n_{1}+n_{2}}\rightarrow\rn{},\\
  & f_{4}:M\times\rn{(n+1)n_{1}+(n+2)n_{2}+(n+2)(n+1)/2}\rightarrow\rn{}.  
\end{align*}
Next, for $i,j=1,\dots,n$, and $l=1,2$, let $\bh_{l}^{ij}$ be smooth functions of the following types:
\[
\bh_{1}^{ij}:M\times\rn{n_{1}+n+1}\rightarrow \rn{},\ \ \
\bh_{2}^{ij}:M\times\rn{n_{1}}\rightarrow \rn{}.
\]
Assume also that $\bh_{l}^{ij}=\bh_{l}^{ji}$ for $i,j=1,\dots,n$ and $l=1,2$. Assume, moreover, that there is a $1<\lambda\in\rn{}$ such that if
$\xi\in\rn{n}$, then
\begin{equation}\label{eq:lambdabd}
  \lambda^{-1}|\xi|^{2}\leq \bh_{1}^{ij}(p)\xi_{i}\xi_{j}\leq \lambda |\xi|^{2},\ \ \
  \lambda^{-1}|\xi|^{2}\leq \bh_{2}^{ij}(q)\xi_{i}\xi_{j}\leq \lambda |\xi|^{2}
\end{equation}
for all $p\in M\times\rn{n_{1}+n+1}$ and $q\in M\times \rn{n_{1}}$, where we use the Einstein summation convention. Next, for $i,j=1,\dots,n$,
let $\g^{ij}$ be smooth functions of the following type:
\[
\g^{ij}:M\times\rn{n_{1}+n_{2}+n+1}\rightarrow M_{n_{2}\times n_{1}}(\rn{});
\]
here $M_{k\times l}(\mathbb{K})$ denotes the $k\times l$-matrices with entries in $\mathbb{K}$. Finally, let $\zeta$ be a smooth function of the
following type:
\[
\zeta:M\times\rn{(n+1)n_{1}+n_{2}}\rightarrow\rn{},
\]
where we, in addition, assume that there is a $1<\lambda_{\zeta}\in\rn{}$ such that
\begin{equation}\label{eq:lambdazetalowerandupperbd}
  \lambda_{\zeta}^{-1}\leq \zeta(p)\leq \lambda_{\zeta},
\end{equation}
for all $p\in M\times\rn{(n+1)n_{1}+n_{2}}$. In addition to the above, assume $f_{i}$, $\bh_{l}^{ij}$, $\g^{ij}$ and $\zeta$ and all their
partial derivatives to be globally bounded (where we by partial derivatives in the $\bM$ direction mean applications of $E_{\bfI}$, where
$E_{\bfI}\psi:=E_{i_{1}}\cdots E_{i_{k}}\psi$ for $\bfI=(i_{1},\dots,i_{k})$ and $1\leq i_j\leq n$ for $1\leq j\leq k$). 

Given the above functions, consider the following system of equations (the meaning of the notation $f_{i}[u,v,N_{1}]$ etc. is clarified in connection with
(\ref{eq:fonemeaning}) below):
\begin{subequations}\label{seq:themodel}
  \begin{align}
    \d_{t}u & = f_{1}[u,v,N_{1}],\label{eq:themodelu}\\
    \d_{t}^{2}v & = \bh^{ij}_{1}[u,N_{1}]E_{i}E_{j}v+\g^{ij}[u,v,N_{1}]E_{i}E_{j}u+f_{2}[u,v,N_{1},N_{2}],\label{eq:themodelv}\\
    \Delta_{\bh_{2}[u]}N_{1} = & \zeta[u,v]N_{1}+f_{3}[u,v],\label{eq:themodelN}\\
    \Delta_{\bh_{2}[u]}N_{2} = & \zeta[u,v]N_{2}+f_{4}[u,v,N_{1}],\label{eq:themodelNdot}    
  \end{align}
\end{subequations}
where $\bh_{2}[u]$ is the Riemannian metric such that $\bh_{2}^{ij}[u]$ are the components of the inverse of the matrix with components
$\bh_{2}[u](E_{i},E_{j})$. We are interested in solving the initial value problem for these equations. In other words, given $t_{0}\in \mI$
and initial data $u_{0}\in C^{\infty}(\bM,\rn{n_{1}})$ and $v_{0,i}\in C^{\infty}(\bM,\rn{n_{2}})$, $i=0,1$, we wish to find an open interval
$\mI_{0}\subseteq\mI$ containing $t_{0}$ and a solution $(u,v,N_{1},N_{2})$ to (\ref{seq:themodel}) on $M_{0}:=\bM\times \mI_{0}$, with
\begin{equation}\label{eq:domandrangesol}
  u\in C^{\infty}(M_{0},\rn{n_{1}}),\ \ \
  v\in C^{\infty}(M_{0},\rn{n_{2}}),\ \ \
  N_{1},N_{2}\in C^{\infty}(M_{0},\rn{}),  
\end{equation}
such that $u(\cdot,t_{0})=u_{0}$, $v(\cdot,t_{0})=v_{0,0}$ and $(\d_{t}v)(\cdot,t_{0})=v_{0,1}$ (in practice, we are also interested in finding solutions
in other regularity classes). In (\ref{seq:themodel}), $f_{1}[u,v,N_{1}]$ should, given $u,v,N_{1}$, be thought of as the function that takes $p\in M_{0}$ to
\begin{equation}\label{eq:fonemeaning}
  f_{1}(p,u(p),v(p),N_{1}(p),(E_{1}N_{1})(p),\dots,(E_{n}N_{1})(p)).
\end{equation}
The interpretation of $f_{2}[u,v,N_{1},N_{2}]$ is similar, only that it depends on $u$, $E_{i}u$, $v$, $E_{i}v$, $\d_{t}v$, $N_{1}$, $N_{2}$, $E_{i}N_{1}$,
$E_{i}N_{2}$, $E_{i}E_{j}N_{1}$ and $E_{i}E_{j}N_{2}$ (where we assume $i\leq j$ in the last two expressions, so that the last two expressions only give rise
to dependence on $n(n+1)$ variables). Similarly, $f_{3}$ only depends on $u$, $E_{i}u$ and $v$. However, $f_{4}$ only depends on $u$, $v$, $E_{i}u$, $E_{i}v$,
$\d_{t}v$, $N_{1}$, $E_{i}N_{1}$ and $E_{i}E_{j}N_{1}$ (where we assume $i\leq j$ in the last expression). Next, $\bh_{1}^{ij}$ only depends on $u$, $N_{1}$ and
$E_{i}N_{1}$ and $\bh_{2}^{ij}$ only depends on $u$. The $\g^{ij}$'s only depend on $u$, $v$, $N_{1}$ and $E_{i}N_{1}$. Finally, $\zeta$ has the same dependence
as $f_{3}$.

\begin{definition}\label{def:modelsystem}
  An equation of the form (\ref{seq:themodel}), where $\bM$; $\mI$; $\bh_{\refer}$; $\{ E_{i}\}$; $f_{i}$, $i=1,2,3,4$; $\bh_{l}^{ij}$, $i,j=1,\dots,n$ and
  $l=1,2$; $\g^{ij}$, $i,j=1,\dots,n$; and $\zeta$ satisfy the above conditions, is said to be a \textit{model system}. 
\end{definition}
\begin{remark}
  Given a model system, we below use all the notation introduced prior to the statement of Definition~\ref{def:modelsystem} without further comment. 
\end{remark}
\begin{remark}
  We are mainly interested in the case that $N_{2}=\d_{t}N_{1}$, so that (\ref{eq:themodelNdot}) can, essentially, be derived from (\ref{eq:themodelN}).
  However, it is convenient to assume that $f_{4}$ and all its derivatives are bounded, and the equation obtained by time differentiating
  (\ref{eq:themodelN}) does not have that property. In fact, one would have to modify the resulting equation in a neighbourhood of the initial data in
  order to obtain a system of the desired form. Here we therefore prefer to consider (\ref{seq:themodel}) and discuss the necessary modifications when
  applying the general theory to a particular system. 
\end{remark}

We state the basic existence and uniqueness theorem in finite degree of regularity in Proposition~\ref{prop:localexistence} below. In
Lemma~\ref{lemma:contcriterionfdreg}, we then state the associated continuation criterion. Combining these results yields the following corollary for
smooth solutions.

\begin{cor}\label{cor:localexmodsys}
  Fix a model system and
  \begin{equation}\label{eq:modsysid}
    \Ui_{0}\in C^{\infty}(\bM,\rn{n_{1}}),\ \ \ \Vi_{0,0}\in C^{\infty}(\bM,\rn{n_{2}}),\ \ \
    \Vi_{0,1}\in C^{\infty}(\bM,\rn{n_{2}}).
  \end{equation}
  Let $t_{0}\in \mI=:(t_{-},t_{+})$. Then there are $t_{a},t_{b}$ with $t_0\in (t_{a},t_{b})\subseteq\mI$, and a unique smooth solution $u,v,N_{1},N_{2}$
  to (\ref{seq:themodel}) on $\bM\times (t_{a},t_{b})$ satisfying
  \begin{equation}\label{eq:uvdtvindata}
    u(\cdot,t_{0})=\Ui_{0},\ \ \
    v(\cdot,t_{0})=\Vi_{0,0},\ \ \
    (\d_{t}v)(\cdot,t_{0})=\Vi_{0,1}.
  \end{equation}
  Moreover, either $t_{b}=t_{+}$ or 
  \begin{equation}\label{eq:Tbcharacterization}
    \lim_{t\uparrow t_{b}}\sup_{t_{0}\leq\tau\leq t}\left(\|u(\cdot,\tau)\|_{C^{2}}+\|v(\cdot,\tau)\|_{C^{2}}
      +\|\d_{t}v(\cdot,\tau)\|_{C^{1}}\right)=\infty.
  \end{equation}
  The statement concerning $t_{a}$ is similar. 
\end{cor}
\begin{proof}
  The proof is to be found at the end of Subsection~\ref{ssection:cont crit smooth sol}.  
\end{proof}
\begin{definition}\label{def:maxintex}
  Fix a model system, a $t_{0}\in\mI$ and initial data as in (\ref{eq:modsysid}). Then the interval $(t_{a},t_{b})$ obtained by appealing to Corollary~\ref{cor:localexmodsys}
  is called the \textit{maximal interval of existence} corresponding to the initial data.
\end{definition}

Next, we state the relevant Cauchy stability result.

\begin{prop}\label{prop:CauchyStability}
  Fix a model system, a $t_{0}\in\mI$ and initial data as in (\ref{eq:modsysid}). Let $u$, $v$, $N_{1}$ and $N_{2}$ denote the corresponding solution, obtained by appealing
  to Corollary~\ref{cor:localexmodsys}, and $\mI_{\robg}$ be the corresponding maximal interval of existence; see Definition~\ref{def:maxintex}. Fix an integer $m>n/2+1$.
  Finally, fix a $t_{\rofin}\in\mI_{\robg}$. Then, given $\e>0$, there is a $\delta>0$ such that if
  \begin{equation}\label{eq:roperid}
    \Ui_{0,\roper}\in C^{\infty}(\bM,\rn{n_{1}}),\ \ \ \Vi_{0,0,\roper}\in C^{\infty}(\bM,\rn{n_{2}}),\ \ \
    \Vi_{0,1,\roper}\in C^{\infty}(\bM,\rn{n_{2}})
  \end{equation}
  are such that
  \[
    \|\Ui_{0,\roper}-\Ui_{0}\|_{H^{m+1}(\bM)}
    +\|\Vi_{0,0,\roper}-\Vi_{0,0}\|_{H^{m+1}(\bM)}
    +\|\Vi_{0,1,\roper}-\Vi_{0,1}\|_{H^{m}(\bM)}<\delta,
  \]
  and if $u_{\roper}$, $v_{\roper}$, $N_{1,\roper}$ and $N_{2,\roper}$ is the smooth solution to (\ref{seq:themodel}) satisfying
  \begin{equation}\label{eq:roperidcond}
    u_{\roper}(\cdot,t_{0})=\Ui_{0,\roper},\ \ \
    v_{\roper}(\cdot,t_{0})=\Vi_{0,0,\roper},\ \ \
    (\d_{t}v_{\roper})(\cdot,t_{0})=\Vi_{0,1,\roper},
  \end{equation}
  cf. Corollary~\ref{cor:localexmodsys}, with corresponding maximal existence interval $\mI_{\roper}$, cf. Definition~\ref{def:maxintex}, then
  $t_{\rofin}\in\mI_{\roper}$ and 
  \[
  \|u-u_{\roper}\|_{H^{m+1}(\bM_{t_{\rofin}})}+\|v-v_{\roper}\|_{H^{m+1}(\bM_{t_{\rofin}})}+\|\d_{t}v-\d_{t}v_{\roper}\|_{H^{m+1}(\bM_{t_{\rofin}})}<\e.
  \]
\end{prop}
\begin{proof}
  The proof is to be found in Subsection~\ref{ssection:stabilitymodelsystem} below. 
\end{proof}

\subsection{Einstein's equations}
Next, we apply the above results to the \textit{Einstein non-linear scalar field equations}, which, given a \textit{potential} $V\in C^{\infty}(\ro)$,
can be written
\begin{subequations}\label{seq:ENLSFE}
  \begin{align}
    G = & T,\\
    \Box_{g}\phi = & V'\circ\phi.
  \end{align}
\end{subequations}
A solution to these equations consists of a Lorentz manifold $(M,g)$ and a function $\phi\in C^{\infty}(M)$ satisfying (\ref{seq:ENLSFE}), where
$\Box_{g}$ is the wave operator associated with $g$;
\begin{align*}
  G := & \mathrm{Ric}-\tfrac{1}{2}Sg;\\
  T := & d\phi\otimes d\phi-\left[\tfrac{1}{2}|d\phi|_g^2+V\circ\phi\right]g;
\end{align*}
and $\mathrm{Ric}$ and $S$ denote the Ricci tensor and scalar curvature of $g$ respectively. In what follows, we also refer to $\phi$ as the
\textit{scalar field}; $G$ as the \textit{Einstein tensor}; and $T$ as the \textit{stress energy tensor}. We are interested in solving
(\ref{seq:ENLSFE}) given initial data, defined as follows. 
\begin{definition}\label{def:id}
  \textit{Initial data for the Einstein non-linear scalar field equations}, with potential $V\in C^{\infty}(\rn{})$, consist of the following: a connected
  smooth manifold $\bM$; a smooth Riemannian metric $h$ on $\bM$; a smooth symmetric covariant $2$-tensor field $\kappa$ on $\bM$; and two functions
  $\phi_{0}$, $\phi_{1}\in C^\infty(\bM)$ satisfying the \textit{constraint equations}
  \begin{subequations}\label{seq:constraintsBA}
    \begin{align}
      \roScal_h-|\kappa|_{h}^{2}+(\rotr_{h}\kappa)^{2} = & \phi_{1}^{2}+|d\phi_0|_{h}^2+2V\circ\phi_{0},\label{eq:ham con original}\\
      \rodiv_{h}\kappa-d\rotr_{h}\kappa = & \phi_1d\phi_0.\label{eq:mom con original}
    \end{align}
  \end{subequations}  
  Here $\roScal_h$ denotes the scalar curvature associated with $h$. In case $\rotr_{h}\kappa$ is constant, the initial data are said to have
  \textit{constant mean curvature} or to be \textit{CMC}. 
\end{definition}
\begin{remark}
  The equations (\ref{eq:ham con original}) and (\ref{eq:mom con original}) are referred to as the \textit{Hamiltonian} and
  \textit{momentum constraint equations} respectively.
\end{remark}
\begin{remark}
  Needless to say, one can consider initial data with a lower degree of regularity. 
\end{remark}
In \cite{GPR}, the system (\ref{seq:ENLSFE}) is considered in the context of the formation of big bang singularities. However, in order to obtain
a system of equations with a well posed initial value problem and which can be used to deduce geometric conclusions concerning the asymptotics in
the direction of the singularity, it is necessary to fix the gauge. This can be done in many ways, but in \cite{GPR} the equations used are a
special case of the system (\ref{seq:thesystem}) we discuss next. In fact, the conclusions of this article are used in the proof of the past global
existence results in \cite{GPR} and in the proof of the results concerning
the existence of developments given data on the singularity in \cite{andres}. However, the step from Corollary~\ref{cor:localexmodsys} and
Proposition~\ref{prop:CauchyStability} to the type of existence and stability results needed in, e.g., \cite{GPR} is substantial. Our next goal is
therefore to formulate local existence and Cauchy stability results for solutions to a slight generalisation of the equations in
\cite{GPR}. In order to formulate the equations, assume first $\bM$ to be a closed $n$-dimensional manifold
and assume that there is a global frame $\{E_{i}\}$ of $T\bM$ with dual frame
$\{\eta^i\}$. We wish to solve the equations on $M:=\bM\times\mI$, where $\mI\subset (0,\infty)$ is an open interval. Moreover, we, by standard abuse
of notation, let $t$ denote both the second coordinate of a point in $\bM\times \mI$ as well as the function that maps $(\bx,s)\in M$ to $s$. In
particular, $t\in\mI$. In this subsection we are interested in the following system of equations: 
\begin{subequations}\label{seq:thesystem}
  \begin{align}
    \d_{t}e_{I}^{i} = & f_{I}^{J}e_{J}^{i},\label{eq:eIievo}\\
    \d_{t}\omega_{i}^{I} = & -f_{J}^{I}\omega_{i}^{J},\label{eq:omIievo}\\
    \d_{t}\bk_{IJ} = & \bna_{I}\bna_{J}N+2N\bk_{IL}\bk_{LJ}-Nt^{-1}\bk_{IJ}+Ne_{I}(\phi)e_{J}(\phi)\label{eq:dtbkIJ}\\
    &+\tfrac{2N}{n-1}(V\circ\phi)\de_{IJ}-N\bR_{IJ}+f_{I}^{L}\bk_{LJ}+\bk_{IL}f_{J}^{L},\nonumber\\
    \Delta_{\bge}N = & t^{-2}(N-1)+N\big[\bk_{IJ}\bk_{IJ}+[e_{0}(\phi)]^{2}-\tfrac{2}{n-1}V\circ\phi-\tfrac{1}{t^{2}}\big],\label{eq:lapseeq}\\
    \Box_{g}\phi = & -\vartheta e_{0}(\phi)+V'\circ\phi,\label{eq:scalarfieldfe}\\
    (\bna_{e_{L}}\bk)(e_{L},e_{I}) = & e_{I}(\theta)+e_{I}(\phi)e_{0}(\phi),\label{eq:momcon}\\
    t^{-2} = & \bk_{IJ}\bk_{IJ}+[e_{0}(\phi)]^{2}+e_I(\phi)e_I(\phi)+2V\circ\phi-\bS,\label{eq:Hamcon}
  \end{align}
\end{subequations}
where $V\in C^\infty(\ro)$, $i,I,J\in\{1,\dots,n\}$ and we sum over repeated indices, even if both are downstairs indices. The variables in
(\ref{seq:thesystem}) are $e_{I}^{i}$, $\omega^{I}_{i}$, $\bk_{IJ}$, $N$ and $\phi$; we refer to them as the \textit{basic variables}. A solution
is tacitly understood to consist of a manifold $M$ as above, on which the basic variables are defined. Moreover, it is understood that $N>0$ on
$M$ and that, initially, $\omega^I_ie^i_J=\delta^I_J$. In order to explain how the
constituents of (\ref{seq:thesystem}) are to be interpreted, define, given $e_I^i$ and $\omega^I_i$, $e_{I}:=e_{I}^{i}E_{i}$ and
$\omega^{I}:=\omega^{I}_{i}\eta^{i}$. Then, due to (\ref{eq:eIievo}) and (\ref{eq:omIievo}),
\begin{equation}\label{eq:omega IJ minus delta IJ ev eq}
  \d_t(\omega^J(e_I)-\de_I^J)=-f_K^J(\omega^K(e_I)-\de^K_I)+f_I^K(\omega^J(e_K)-\de^J_K).
\end{equation}
Since $\omega^J(e_I)=\de_I^J$ initially, (\ref{eq:omega IJ minus delta IJ ev eq}) yields the conclusion that $\omega^J(e_I)=\de_I^J$ whenever the
solution is defined. In particular, $\{e_I\}$ is thus a frame and $\{\omega^I\}$ is its dual frame on the tangent spaces of the hypersurfaces
$\bM_{t}:=\bM\times\{t\}$. We define a Lorentz metric on $M$ by
\begin{equation}\label{eq:gdefintro}
  g:=-N^{2}dt\otimes dt+\textstyle{\sum}_{I}\omega^{I}\otimes\omega^{I}.
\end{equation}
Letting $e_{0}:=N^{-1}\d_{t}$, $\{e_{\a}\}$ is then an orthonormal frame with respect to $g$. The induced metric on the hypersurfaces
$\bM_{t}$ is 
\begin{equation}\label{eq:bgeitoomegaIintro}
  \bge:=\textstyle{\sum}_{I}\omega^{I}\otimes\omega^{I}.
\end{equation}
In (\ref{seq:thesystem}), $\bna$, $\bR$ and $\bS$ denote the Levi-Civita connection, the Ricci curvature and the scalar curvature of
$\bge$ respectively. Moreover, the upper case Latin indices refer to components with respect to the frame $\{e_I\}$. Next, $\bk_{IJ}$ is required
to be symmetric. Moreover,
\[
\bk:=\bk_{IJ}\omega^{I}\otimes\omega^{J};
\]
$\theta:=\bk_{II}$; and $\vartheta:=\theta-t^{-1}$. Turning to (\ref{eq:lapseeq}), we assume 
\begin{equation}\label{eq:zetadefintro}
  \zeta:=\bk_{IJ}\bk_{IJ}+(e_0\phi)^{2}-\tfrac{2}{n-1}V\circ\phi
\end{equation}
to satisfy $\zeta>0$ on $M$. Combining this assumption with (\ref{eq:lapseeq}) and the maximum principle, it follows that $N>0$; see
(\ref{eq:Ntzlowerbd}) and the adjacent text for more details. Concerning $f_{I}^{J}$, we assume
\begin{equation}\label{eq:fIJsym}
  f_{I}^{J}+f_{J}^{I}=-2N\bk_{IJ}
\end{equation}
(recall that $\bk_{IJ}$ is symmetric). Due to this requirement, $\bk_{IJ}$ are the components of the second fundamental
form of the $\bM_t$-hypersurfaces with respect to the frame $\{e_I\}$; see Lemma~\ref{lemma:frameorth}. Let $a_I^J:=(f_I^J-f_J^I)/2$. Then we, finally,
assume $a_I^J$ to only depend on the spacetime coordinates, $e^i_I$, $\omega^I_i$, $N$, $\bk_{IJ}$, $\phi$, $N^{-1}\d_{t}\phi$, $E_{i}\phi$ and $E_{i}N$.
That the condition $a_I^J=0$ yields the equations considered in \cite{GPR} is verified in Remark~\ref{remark:GPR special case}.

Next, let us state the relevant well posedness results for (\ref{seq:thesystem}). We use the following notation:
\begin{subequations}\label{seq:leCknorms}
  \begin{align}
    \|e(\cdot,t)\|_{C^{k}} := & \big(\textstyle{\sum}_{i,I}\sum_{|\bfI|\leq k}\sup_{\bx\in\bM}|(E_{\bfI}e_{I}^{i})(\bx,t)|^{2}\big)^{1/2},\\
    \|\omega(\cdot,t)\|_{C^{k}} := & \big(\textstyle{\sum}_{i,I}\sum_{|\bfI|\leq k}\sup_{\bx\in\bM}|(E_{\bfI}\omega_{i}^{I})(\bx,t)|^{2}\big)^{1/2},\\
    \|\bk(\cdot,t)\|_{C^{k}} := & \big(\textstyle{\sum}_{I,J}\sum_{|\bfI|\leq k}\sup_{\bx\in\bM}|(E_{\bfI}\bk_{IJ})(\bx,t)|^{2}\big)^{1/2},\\
    \|\d_{t}\bk(\cdot,t)\|_{C^{k}} := & \big(\textstyle{\sum}_{I,J}\sum_{|\bfI|\leq k}\sup_{\bx\in\bM}|(E_{\bfI}\d_{t}\bk_{IJ})(\bx,t)|^{2}\big)^{1/2},\\    
    \|\phi(\cdot,t)\|_{C^{k}} := & \big(\textstyle{\sum}_{|\bfI|\leq k}\sup_{\bx\in\bM}|(E_{\bfI}\phi)(\bx,t)|^{2}\big)^{1/2}.
  \end{align}
\end{subequations}
For an explanation of the meaning of $E_{\bfI}$ etc., we refer the reader to Subsection~\ref{ssection:conventionsframe}. Next, given two real numbers
$t_{a}$ and $t_{b}$, define $I_{t_{a},t_{b}}$ to be the closed and non-empty interval with $t_{a}$ and $t_{b}$ as endpoints. In order to formulate the
continuation criterion, it is convenient to introduce the following expression
\begin{equation}\label{eq:mCdef}
  \begin{split}
    \mC(t) := & \textstyle{\sup}_{\tau\in I_{t_{0},t}}[\|e(\cdot,\tau)\|_{C^{2}}+\|\omega(\cdot,\tau)\|_{C^{2}}+\|\bk(\cdot,\tau)\|_{C^{2}}
      +\|\d_{t}\bk(\cdot,\tau)\|_{C^{1}}+\|\phi(\cdot,\tau)\|_{C^{3}}\\
      & \phantom{\sup \ \ +xy} +\|\d_{t}\phi(\cdot,\tau)\|_{C^{2}}+\|e_{0}\phi(\cdot,\tau)\|_{C^{2}}
      +\|\d_{t}e_{0}\phi(\cdot,\tau)\|_{C^{1}}+\|1/\zeta(\cdot,\tau)\|_{C^{0}}]
  \end{split}
\end{equation}
where $\zeta$ is defined by (\ref{eq:zetadefintro}). 

Next, we state local existence and uniqueness of solutions to (\ref{seq:thesystem}), including a continuation criterion, given initial data.
\begin{thm}\label{thm:fosyslocalexistence}
  Fix a function $V\in C^{\infty}(\rn{})$ and let $(\bM,h,\kappa,\phi_{0},\phi_{1})$ be smooth CMC initial data for the Einstein-non-linear scalar
  field equations with potential $V$, see Definition~\ref{def:id}, where $\bM$ is a closed, connected, oriented and parallelisable $n$-dimensional
  manifold, $\tr_{h}\kappa>0$ and
  \begin{equation}\label{eq:zetapos}
    |\kappa|_{h}^{2}+\phi_{1}^{2}-\tfrac{2}{n-1}V\circ\phi_{0}>0
  \end{equation}
  on $\bM$. Let $t_{0}:=1/\tr_{h}\kappa$. Fix a frame $\{E_{i}\}$ of $\bM$ with dual frame $\{\eta^{i}\}$. Next, let $\{\xi_{I}\}$ be an orthonormal
  frame of $(\bM,h)$ with dual frame $\{\rho^{I}\}$. Let $a_{I}^{J}$, where $I,J\in\{1,\dots,n\}$, be such that $a_{I}^{J}=-a_{J}^{I}$ and such that
  $a_{I}^{J}$ only depends on the spacetime coordinates, $e^i_I$, $\omega^I_i$, $\bk_{IJ}$, $\phi$, $N^{-1}\d_{t}\phi$, $E_{i}\phi$, $N$ and $E_{i}N$,
  where $e^i_I$, $\omega^I_i$, $\bk_{IJ}$, $\phi$ and $N$ are the basic variables of the system (\ref{seq:thesystem}). Define $f_{I}^{J}$, where
  $I,J\in\{1,\dots,n\}$, by
  \begin{equation}\label{eq:fIJthmdef}
    f_{I}^{J}:=-N\bk_{IJ}+a_{I}^{J}.
  \end{equation}
  Let $e_{I}^{i}|_{t_{0}}:=\xi_{I}^{i}$ and $\omega_{i}^{I}|_{t_{0}}:=\rho_{i}^{I}$, where $\xi_{I}^{i}$ and $\rho_{i}^{I}$ are defined by
  $\xi_{I}=\xi_{I}^{i}E_{i}$ and $\rho^{I}=\rho_{i}^{I}\eta^{i}$ respectively. Next, let $\bk_{IJ}|_{t_{0}}:=\kappa(\xi_{I},\xi_{J})$ and $N|_{t_{0}}$
  be defined as the unique solution to the equation that results by replacing $t$, $\phi$ and $e_{0}\phi$ in (\ref{eq:lapseeq}) by
  $t_{0}$, $\phi_{0}$ and $\phi_{1}$ respectively. Define, finally, $\phi|_{t_{0}}:=\phi_{0}$ and $(\d_{t}\phi)|_{t_{0}}:=N|_{t_{0}}\phi_{1}$.
  Then there is an open interval $\mI\subseteq (0,\infty)$, containing $t_{0}$, and a unique solution to (\ref{seq:thesystem}) on
  $M:=\bM\times\mI$, corresponding to these initial data, such that $N>0$ and $\zeta>0$; such that $\bk_{II}=1/t$; and such that
  $\bk_{IJ}=\bk_{JI}$. Moreover, if $\bM_{t}:=\bM\times\{t\}$ for $t\in\mI$ and 
  \begin{equation}\label{eq:gthmdef}
    g:=-N^{2}dt\otimes dt+\textstyle{\sum}\omega^{I}\otimes\omega^{I},
  \end{equation}
  where $\omega^{I}:=\omega^{I}_{i}\eta^{i}$, then $(M,g,\phi)$ is a solution to the Einstein-non-linear scalar field equations which induces the
  initial data $(h,\kappa,\phi_{0},\phi_{1})$ on $\bM_{t_{0}}$; and is such that the mean curvature of $\bM_{t}$ is $1/t$ for
  $t\in\mI$. Finally, if $\mI=(t_{-},t_{+})$, then either $t_{-}=0$ or $\mC(t)\rightarrow\infty$ as $t\downarrow t_{-}$. Similarly, either
  $t_{+}=\infty$ or $\mC(t)\rightarrow\infty$ as $t\uparrow t_{+}$.  
\end{thm}
\begin{remark}
  Orientable $3$-dimensional closed manifolds are parallelisable. The requirement that $\bM$ be parallelisable is thus only a non-trivial
  condition in case $n\neq 3$. 
\end{remark}
\begin{remark}\label{remark:GPR special case}
  In case $a_{I}^{J}=0$, where $a_I^J:=(f_I^J-f_J^I)/2$, the equations in (\ref{seq:thesystem}) take the form of the corresponding
  equations in \cite[Proposition~53]{GPR}, with the following caveats: In the momentum constraint equation, i.e. (\ref{eq:momcon}), the
  additional condition of constant mean curvature has to be imposed, which means that the first term on the right hand side of
  (\ref{eq:momcon}) vanishes. Second, in the lapse equation, i.e. (\ref{eq:lapseeq}), it is necessary to substitute the spatial scalar curvature
  from the Hamiltonian constraint in order to obtain our lapse equation from the lapse equation in \cite{GPR}. We leave the verifications of these
  statements to the reader. Note, however, that the main discrepancy between \cite{GPR} and (\ref{seq:thesystem}) is that in \cite{GPR}, the authors
  explicitly write out the spatial Ricci and scalar curvatures etc. Moreover, the relevant expressions for these quantities are to be found in
  \cite[Lemma~55]{GPR}. Note also that the second order evolution equation in \cite{GPR} for $\phi$ coincides with (\ref{eq:phi}). In addition, the
  evolution equation for $\bga^{I}_{JK}$ included in \cite{GPR} can be deduced by appealing to (\ref{eq:dtbgaIJK}) and the evolution equation for
  $e_I\phi$ can be deduced by computing the commutator $[e_0,e_I]$. 
\end{remark}
\begin{proof}
  The proof is to be found in Subsection~\ref{ssection:locexeinst}.
\end{proof}
Finally, we state the corresponding result on Cauchy stability.
\begin{thm}\label{thm:CauchystabEinstein}
  Given notation, assumptions and conclusions as in the statement of Theorem~\ref{thm:fosyslocalexistence}, let $e_{I}^{i}$, $\omega_{i}^{I}$,
  $\bk_{IJ}$, $\phi$ and $N$ be the solution obtained in the conclusions. Let, moreover,
  $\mI$ be the corresponding existence interval. Given $t_{1}\in\mI$, $m>n/2+1$ and $\e>0$, there is then a $\delta>0$ such that the following
  holds. Let $(\bM,\chh,\chkappa,\chphi_{0},\chphi_{1})$ be smooth CMC initial data for the Einstein-non-linear scalar field equations with potential $V$,
  satisfying $\tr_{\chh}\chkappa=1/t_{0}$, and let $\{\chxi_{I}\}$ be an orthonormal frame of $(\bM,\chh)$. Let $\che_{I}^{i}$, $\chom_{i}^{I}$, $\chk_{IJ}$,
  $\chphi$ and $\chN$ be the corresponding solution, obtained by applying Theorem~\ref{thm:fosyslocalexistence}. Let, moreover, $\chmI$ be the
  corresponding existence interval and $\che_{0}=\chN^{-1}\d_{t}$. Then, if
  \begin{subequations}\label{seq:Cauchycriterion}
    \begin{align}
      \textstyle{\sum}_{I,i}\|\che_{I}^{i}-e_{I}^{i}\|_{H^{m+1}(\bM_{t_{0}})}+\textstyle{\sum}_{I,i}\|\chom^{I}_{i}-\omega^{I}_{i}\|_{H^{m+1}(\bM_{t_{0}})} & < \delta,\\
      \textstyle{\sum}_{I,J}\|\bk_{IJ}-\chk_{IJ}\|_{H^{m+1}(\bM_{t_{0}})}+\textstyle{\sum}_{I,J}\|\d_{t}\bk_{IJ}-\d_{t}\chk_{IJ}\|_{H^{m}(\bM_{t_{0}})} & <\delta,\\
      \|\phi-\chphi\|_{H^{m+2}(\bM_{t_{0}})}+\|\d_{t}\phi-\d_{t}\chphi\|_{H^{m+1}(\bM_{t_{0}})}& \\
      +\|e_{0}\phi-\che_{0}\chphi\|_{H^{m+1}(\bM_{t_{0}})}+\|\d_{t}e_{0}\phi-\d_{t}\che_{0}\chphi\|_{H^{m}(\bM_{t_{0}})} & < \delta,
    \end{align}
  \end{subequations}  
  the interval $\chmI$ contains $t_{1}$ and (\ref{seq:Cauchycriterion}) holds with $t_{0}$ replaced by $t_{1}$ and $\delta$ replaced by $\e$.  
\end{thm}
\begin{proof}
  The proof is to be found at the end of Subsection~\ref{ssection:locexeinst}.
\end{proof}

\subsection{Outline}
In Section~\ref{section:modelsystem}, we develop a local existence theory for the model systems introduced in Subsection~\ref{ssection:modelsystem}.
This includes local existence and uniqueness, a continuation criterion, as well as Cauchy stability. One difficulty turns out to be the elliptic estimates
for $N_1$ and $N_2$ (solving (\ref{eq:themodelN}) and (\ref{eq:themodelNdot})). Note that, given $u$ and $v$, the equation (\ref{eq:themodelN}) is a linear
elliptic equation for $N_1$. Moreover, since $\zeta>0$, this equation has a unique solution and there are standard elliptic Schauder and Sobolev estimates
for $N_1$. However, for our applications, this information is not sufficient. We need to know how the constants appearing in the Schauder and Sobolev
estimates depend on $u$ and $v$ (and that the dependence is continuous with respect to appropriate norms). Unfortunately, we were unable to find a
statement of the estimates including sufficiently detailed information concerning the dependence of the constants for our applications. The relevant
Schauder and Sobolev estimates are to be found in Appendices~\ref{section:schaudermodeleq}, \ref{section:sobestellmodeleq} and
\ref{appendix:spec def sob}. To obtain the required regularity of the lapse function, as a
function of time, we also need an appropriate implicit function theorem argument; see Appendix~\ref{section:contdeponcoeffellmodeleq} for the
details. Once the local existence theory for solutions to the model system has been developed, we turn to applications to Einstein's equations.
%To begin with, we derive a second order equation for the Weingarten map in Section~\ref{section:second order eqs}. We also derive wave equations for
%the first derivatives of the scalar field and evolution equations for the spatial metric. 
In Section~\ref{section:applgentheorytoEFE}, we take the first steps in relating the model system to the equations we actually
wish to solve. In the end, the system we wish to solve is (\ref{seq:thesystem}). Unfortunately, this system cannot be solved directly. In
fact, we have to solve an auxiliary system which, as opposed to (\ref{seq:thesystem}), is second order in the coefficients of the second fundamental
form. The logic is then that we use initial data for (\ref{seq:thesystem}) to generate initial data for the auxiliary second order system.
The second order system can then be solved by appealing to the local existence theory for the model system. However, at the end we wish to return
to (\ref{seq:thesystem}). To this end we, in Section~\ref{section:applgentheorytoEFE} and using the auxiliary second order system, derive a system
of equations for expressions quantifying the violations of the CMC condition, of the Hamiltonian constraint, of the momentum constraint and of the
first order equation for the coefficients of the second fundamental form. We also demonstrate that given a solution to the auxiliary second order
system with appropriate initial data, the expressions quantifying the violations all vanish. This yields a solution to the original system
(\ref{seq:thesystem}), which is also a solution to the Einstein-non-linear scalar field system. Finally, in Section~\ref{section:thesystemlocexist},
we prove well posedness for the original system (\ref{seq:thesystem}).

\subsection{Acknowledgements}\label{ssection:acknoledgements}
The author would like to express his gratitude to Oliver Petersen for his help in the writing of Appendix~\ref{section:sobestellmodeleq}. 
This research was funded by Vetenskapsr\aa det (the Swedish Research Council), dnr. 2022-03053. 

\section{Model system}\label{section:modelsystem} 

\subsection{A simplified model system}
In order to solve the Cauchy problem associated with (\ref{seq:themodel}), the idea is to set up an iteration: we consider a sequence of solutions to
simplified versions of the model system and prove that the sequence converges to a solution of the actual system; in fact, the overall architecture of
the local existence argument is taken from \cite{majda}. In practice, this means that we, given
$u_{a}\in C^{\infty}(M,\rn{n_{1}})$, $v_{a}\in C^{\infty}(M,\rn{n_{2}})$ and $N_{i,a}\in C^{\infty}(M,\rn{})$, $i=1,2$, wish to solve the initial value problem for 
\begin{subequations}\label{seq:linmodsys}
  \begin{align}
    \d_{t}u_{b} = & f_{1}[u_{a},v_{a},N_{1,a}],\label{eq:ublinmodsys}\\
    \d_{t}^{2}v_{b} = & \bh^{ij}_{1}[u_{b},N_{1,b}]E_{i}E_{j}v_{b}+\g^{ij}[u_{b},v_{a},N_{1,b}]E_{i}E_{j}u_{b},\label{eq:vblinmodsys}\\
    & +f_{2}[u_{a},v_{a},N_{1,a},N_{2,a}]\nonumber\\
    \Delta_{\bh_{2}[u_{b}]}N_{1,b} = & \zeta[u_{b},v_{a}]N_{1,b}+f_{3}[u_{b},v_{a}],\label{eq:Nblinmodsys}\\
    \Delta_{\bh_{2}[u_{b}]}N_{2,b} = & \zeta[u_{b},v_{a}]N_{2,b}+f_{4}[u_{b},v_{a},N_{1,b}].\label{eq:Nbdotlinmodsys}    
  \end{align}
\end{subequations}
Note, to begin with, that there is no problem in solving (\ref{seq:linmodsys}). In fact, the following lemma holds.

\begin{lemma}\label{lemma:ubvbNbsol}
  Fix a model system in the sense of Definition~\ref{def:modelsystem}. Fix $u_{a}\in C^{\infty}(M,\rn{n_{1}})$, $v_{a}\in C^{\infty}(M,\rn{n_{2}})$ and
  $N_{i,a}\in C^{\infty}(M,\rn{})$, $i=1,2$. Fix, moreover, $t_{0}\in\mI$, $u_{0}\in C^{\infty}(\bM,\rn{n_{1}})$ and $v_{0,i}\in C^{\infty}(\bM,\rn{n_{2}})$,
  $i=0,1$. Then there is a unique solution $u_{b}\in C^{\infty}(M,\rn{n_{1}})$, $v_{b}\in C^{\infty}(M,\rn{n_{2}})$ and $N_{i,b}\in C^{\infty}(M,\rn{})$,
  $i=1,2$, to (\ref{seq:linmodsys}) such that $u_{b}(\cdot,t_{0})=u_{0}$, $v_{b}(\cdot,t_{0})=v_{0,0}$ and $(\d_{t}v_{b})(\cdot,t_{0})=v_{0,1}$.
\end{lemma}
\begin{proof}
  Given $u_{a}$, $v_{a}$, $N_{1,a}$ and $u_{0}$, it is clear that the initial value problem for (\ref{eq:ublinmodsys}) has a unique smooth solution
  on $M$. In (\ref{eq:vblinmodsys})--(\ref{eq:Nbdotlinmodsys}), $u_{b}$ can then be considered to be a given smooth function on $M$. Next,
  since $\zeta$ has a strictly positive lower bound and $u_{b}$ is a given smooth function, it is clear that for each $t\in\mI$, the equation
  (\ref{eq:Nblinmodsys}) has a unique smooth solution $N_{1,b}(\cdot,t)$; cf. Lemma~\ref{lemma:existence sol to elliptic PDE} and
  Remark~\ref{remark:coeff smooth sol smooth}. However, we also need to know that the dependence on $t$ is smooth. To this end, we appeal to
  Corollary~\ref{cor: Ck regularity}. In order to verify that the conditions are satisfied, note that for every $k,m\in\nn{}_0$ with
  $k>n/2+1$, we can think of
  \[
  t\mapsto [\bh_{2}[u_{b}](\cdot,t),\zeta[u_{b},v_{a}](\cdot,t),f_{3}[u_{b},v_{a}](\cdot,t)]
  \]
  as being a $C^{m}$ map from $\mI$ to space required in Corollary~\ref{cor: Ck regularity}. It then follows from Corollary~\ref{cor: Ck regularity}
  that $N_{1,b}$ is a $C^{m}$ map from $\mI$ to $H^{k+2}(\bM)$. Since $k>n/2+1$ and $m$ are arbitrary, we conclude that $N_{1,b}\in C^{\infty}(M)$. At this
  stage $u_{b}$ and
  $N_{1,b}$ have been determined. They are also smooth on $M$. That the equation (\ref{eq:vblinmodsys}) has a unique smooth solution, given initial data,
  then follows from, e.g., \cite[Theorem~12.19, p.~144]{RinCauchy}. Finally, since $u_{b}$, $v_{b}$ and $N_{1,b}$ are now given smooth functions, the proof
  of the fact that the equation (\ref{eq:Nbdotlinmodsys}) has a unique solution $N_{2,b}\in C^{\infty}(M)$ is identical to the proof in the case of $N_{1,b}$. 
\end{proof}

\subsection{Schauder estimates for the lapse functions}
In the proof of local existence, and in the derivation of a continuation criterion, we need Schauder estimates. The general form of the estimates
we need is standard. However, we need more detailed information concerning the dependence of the constants than is normally given. In
Appendix~\ref{section:schaudermodeleq} below, we therefore derive the relevant estimates for a model equation, including the needed information concerning
the constants. Here we state the conclusions obtained by applying the general estimates to (\ref{eq:Nblinmodsys}) and (\ref{eq:Nbdotlinmodsys}).
\begin{lemma}\label{lemma:schauder}
  Let $n_{1},n_{2}\in\nn{}$, $(\bM,\bh_{\refer})$ be a closed, connected and oriented Riemannian manifold, and assume $(\bM,\bh_{\refer})$ to have a
  global orthonormal frame
  $\{E_{i}\}$. Assume, in addition, $\mI$, $\bh_{2}^{ij}$, $\zeta$, $f_{3}$ and $f_{4}$ to have the properties described at the beginning of
  Subsection~\ref{ssection:modelsystem}. Let $\mJ\subset \mI$ be an interval and assume $u_{b}:\bM\times\mJ\rightarrow\rn{n_{1}}$ and
  $v_{a}:\bM\times\mJ\rightarrow\rn{n_{2}}$ to be such that $\d_{t}v_{a}$ is well defined on $\bM\times\mJ$; such that $u_{b}(\cdot,t)\in C^{2}$ for all
  $t\in\mJ$; such that $v_{a}(\cdot,t)\in C^{2}$ for all $t\in\mJ$; and such that $(\d_{t}v_{a})(\cdot,t)\in C^{1}$ for all $t\in\mJ$.  Finally, assume
  $N_{1,b}(\cdot,t)$ to be a $C^{2,1}$ solution to (\ref{eq:Nblinmodsys}) for all $t\in\mJ$, and $N_{2,b}(\cdot,t)$ to be a $C^{2,1}$ solution to
  (\ref{eq:Nbdotlinmodsys}) for all $t\in\mJ$. Let $\ro_{+}:=[0,\infty)$. Then there is a covering of $\bM$ (depending only on $(\bM,\bh_{\refer})$) and
  functions $\mB_{1}:\ro_{+}^{2}\rightarrow\ro_{+}$ and $\mB_{2}:\ro_{+}^{3}\rightarrow\ro_{+}$ that are continuous and increasing in each of their arguments,
  such that
  \begin{subequations}\label{seq:NibSchauder}
    \begin{align}
      \|N_{1,b}(\cdot,t)\|_{C^{2,1}} \leq & \mB_{1}[\|u_{b}(\cdot,t)\|_{C^{2}},\|v_{a}(\cdot,t)\|_{C^{1}}],\label{eq:NobSchauder}\\
      \|N_{2,b}(\cdot,t)\|_{C^{2,1}} \leq &
      \mB_{2}[\|u_{b}(\cdot,t)\|_{C^{2}},\|v_{a}(\cdot,t)\|_{C^{2}},\|(\d_{t}v_{a})(\cdot,t)\|_{C^{1}}]\label{eq:NtbSchauder}
    \end{align}
  \end{subequations}
  for all $t\in\mJ$. Here the functions $\mB_{i}$ only depend on bounds on $\bh_{2}^{ij}$, $\zeta$, $f_{3}$ and $f_{4}$ and
  their derivatives up to order two, one, one and one respectively; $(\bM,\bh_{\refer})$; the covering; the frame $\{E_{i}\}$; the $\lambda$ appearing in
  (\ref{eq:lambdabd}); and the $\lambda_{\zeta}$ appearing in (\ref{eq:lambdazetalowerandupperbd}). 
\end{lemma}
\begin{remark}\label{remark:def Hoelder spaces}
  In the statement of the lemma we use the notation $C^{k,\a}=C^{k,\a}(h_\refer)$ for the sake of brevity; see
  Appendix~\ref{section:hoelder spaces on manifolds} for a definition of H\"{o}lder spaces and norms.
\end{remark}
\begin{remark}
  In case $f_{3}$ is such that $f_{3}(\xi)=0$ for all $\xi$ with the property that the last $(n+1)n_{1}+n_{2}$ components of $\xi$ equal zero, then
  (\ref{eq:NobSchauder}) can be improved to
  \begin{equation}\label{eq:NobSchauderfz}
    \|N_{1,b}(\cdot,t)\|_{C^{2,1}} \leq \mB_{3}[\|u_{b}(\cdot,t)\|_{C^{2}},\|v_{a}(\cdot,t)\|_{C^{1}}]\cdot(\|u_{b}(\cdot,t)\|_{C^{2}}+\|v_{a}(\cdot,t)\|_{C^{1}}),
  \end{equation}
  where $\mB_{3}$ has the same dependence and properties as $\mB_{1}$ except that it depends on one more derivative of $f_{3}$. Similarly, in case,
  in addition, $f_{4}$ is such that $f_{4}(\xi)=0$ for all $\xi$ with the property that the last $(n+1)n_{1}+(n+2)n_{2}+(n+2)(n+1)/2$ components of
  $\xi$ equal zero, then (\ref{eq:NtbSchauder}) can be improved to
  \begin{equation}\label{eq:NtwobSchauderfz}
    \begin{split}
      \|N_{2,b}(\cdot,t)\|_{C^{2,1}} \leq & \mB_{4}[\|u_{b}(\cdot,t)\|_{C^{2}},\|v_{a}(\cdot,t)\|_{C^{2}},\|\d_{t}v_{a}(\cdot,t)\|_{C^{1}}]\\
      & \cdot(\|u_{b}(\cdot,t)\|_{C^{2}}+\|v_{a}(\cdot,t)\|_{C^{2}}+\|\d_{t}v_{a}(\cdot,t)\|_{C^{1}}),
    \end{split}    
  \end{equation}
  where $\mB_{4}$ has the same dependence and properties as $\mB_{2}$ except that it depends on one more derivative of $f_{3}$ and $f_{4}$.
\end{remark}
\begin{proof}
  We are here only interested in considering one $\bM_t:=\bM\times\{t\}$ at a time; i.e., we focus on $N_{1,b}(\cdot,t)$ etc. However,
  for the sake of brevity, we consistently omit the argument $(\cdot,t)$ in what follows. Since $\bh_{2}^{ij}$ only depends on $u$, the
  metric $\bh_{2}[u_{b}]$ is a $C^{2}$ metric on $\bM_t$. Since $\zeta$ and $f_{3}$ only depend on $u$, $E_{i}u$ and $v$, it is clear that $f_{3}[u_{b},v_{a}]$
  and $\zeta[u_{b},v_{a}]$ are both $C^{1}$. Finally, note that since $\bh_{2}^{ij}$ and $\zeta$ satisfy (\ref{eq:lambdabd}) and
  (\ref{eq:lambdazetalowerandupperbd}) respectively, (\ref{eq:hhrefabd}) below is satisfied. By an application of Theorem~\ref{thm:globalSchauder} below
  to (\ref{eq:Nblinmodsys}), we conclude that (\ref{eq:NobSchauder}) holds. In case $f_{3}$ is such that $f_{3}(\xi)=0$ for all $\xi$ with the property
  that the last $(n+1)n_{1}+n_{2}$ components of $\xi$ equal zero, we can rewrite $f_{3}[u_{b},v_{a}]$ as
  \begin{equation}\label{eq:fthreezerorepr}
    f_{3}[u_{b},v_{a}]=\textstyle{\int}_{0}^{1}(\d_{u}f_{3})[su_{b},sv_{a}]ds\cdot (u_{b},E_{1}u_{b},\dots,E_{n}u_{b})
    +\textstyle{\int}_{0}^{1}(\d_{v}f_{3})[su_{b},sv_{a}]ds\cdot v_{a}.
  \end{equation}
  Here $\d_{u}f_{3}$ denotes the derivative of $f_{3}$ with respect to all the $u$-related arguments; i.e., $u$, $E_{1}u,\dots,E_{n}u$. Applying
  Theorem~\ref{thm:globalSchauder} below to (\ref{eq:Nblinmodsys}), using the representation (\ref{eq:fthreezerorepr}) when estimating the last factor
  on the right hand side of (\ref{eq:uCtwooneglbd}), yields the improvement (\ref{eq:NobSchauderfz}). 

  Turning to $f_{4}$, note that it depends on $N_{1,b}$, $E_{i}N_{1,b}$ and $E_{i}E_{j}N_{1,b}$. However, bounds on these quantities in $C^{0,1}$ follow from
  (\ref{eq:NobSchauder}) (or (\ref{eq:NobSchauderfz}), in case $f_{3}$ satisfies the condition stated prior to (\ref{eq:NobSchauderfz})). Beyond these
  quantities, $f_{4}$ only depends on $u_{b}$, $v_{a}$, $E_{i}u_{b}$, $E_{i}v_{a}$ and $\d_{t}v_{a}$. These objects are clearly bounded in $C^{0,1}$ by
  $\|u_{b}\|_ {C^{2}}$, $\|v_{a}\|_{C^{2}}$ and $\|\d_{t}v_{a}\|_{C^{1}}$. Applying Theorem~\ref{thm:globalSchauder} to (\ref{eq:Nbdotlinmodsys}), keeping
  (\ref{eq:NobSchauder}) in mind, yields the conclusion that (\ref{eq:NtbSchauder}) holds. In case $f_{4}$ satisfies the condition stated prior to
  (\ref{eq:NtwobSchauderfz}), then there is a representation of $f_{4}$ analogous to (\ref{eq:fthreezerorepr}). The derivation of (\ref{eq:NtwobSchauderfz})
  is therefore similar to the derivation of (\ref{eq:NobSchauderfz}). There is, however, one difference: $f_{4}$ depends on up to two derivatives of $N_{1,b}$.
  This difference can, however, be handled by appealing to (\ref{eq:NobSchauderfz}), assuming $f_{3}$ satisfies the condition stated prior to
  (\ref{eq:NobSchauderfz}). We conclude that (\ref{eq:NtwobSchauderfz}) holds. 
\end{proof}

In the proof of local existence, we, unfortunately, also need to control $\d_{t}N_{1,b}$ in $C^{1}$.

\begin{lemma}
  With assumptions, notation and conclusions as in Lemma~\ref{lemma:ubvbNbsol}, there is a covering of $\bM$ (depending only on $(\bM,\bh_{\refer})$)
  and a function $\mC:\ro_{+}^{5}\rightarrow\ro_{+}$ which is continuous and increasing in each of its arguments, such that
  \begin{equation}\label{eq:dtNobest}
    \begin{split}
      & \|\d_{t}N_{1,b}(\cdot,t)\|_{C^{2,1}} \\
      \leq & \mC[\|u_{a}(\cdot,t)\|_{C^{2}},\|u_{b}(\cdot,t)\|_{C^{2}},\|v_{a}(\cdot,t)\|_{C^{2}},\|\d_{t}v_{a}(\cdot,t)\|_{C^{1}},\|N_{1,a}(\cdot,t)\|_{C^{2,1}}],
      \end{split}
  \end{equation}
  for all $t\in\mI$. Here the function $\mC$ only depends on bounds on
  $\bh_{2}^{ij}$, $\zeta$, $f_{1}$ and $f_{3}$ and their derivatives up to order three, two, two and two respectively; $(\bM,\bh_{\refer})$; the covering;
  the frame $\{E_{i}\}$; the $\lambda$ appearing in (\ref{eq:lambdabd}); and the $\lambda_{\zeta}$ appearing in (\ref{eq:lambdazetalowerandupperbd}).
\end{lemma}
\begin{remark}
  Remark~\ref{remark:def Hoelder spaces} is equally relevant here. 
\end{remark}
\begin{proof}
  As in the proof of Lemma~\ref{lemma:schauder}, we omit the argument $(\cdot,t)$ in the present proof. Differentiating (\ref{eq:Nblinmodsys}) yields
  \[
  \Delta_{\bh_{2}[u_{b}]}(\d_{t}N_{1,b}) = \zeta[u_{b},v_{a}]\d_{t}N_{1,b}+g[u_{a},u_{b},v_{a},N_{1,a},N_{1,b}],
  \]
  where $g$ is a smooth function depending only on the spacetime coordinates, $u_{a}$, $E_{i}u_{a}$, $u_{b}$, $E_{i}u_{b}$, $v_{a}$, $E_{i}v_{a}$,
  $\d_{t}v_{a}$, $N_{1,a}$, $E_{i}N_{1,a}$, $E_{i}E_{j}N_{1,a}$, $N_{1,b}$, $E_{i}N_{1,b}$ and $E_{i}E_{j}N_{1,b}$. In order to arrive at this conclusion
  we used the fact that (\ref{eq:ublinmodsys}) holds. Moreover, $g$ is a linear combination of terms built up of the following factors: components
  of $(\bh_{2}^{ij})^{(m)}[u_{b}]$,
  $\zeta^{(l)}[u_{b},v_{a}]$, $f_{1}^{(l)}[u_{a},v_{a},N_{1,a}]$ and $f_{3}^{(l)}[u_{b},v_{a}]$, where $m\in\{0,1,2\}$ and $l\in\{0,1\}$ (here $\zeta^{(l)}$
  denotes the $l$'th derivative of $\zeta$); the components of the inverse of the matrix with components $\bh_{2}^{ij}[u_{b}]$; and factors on which $g$
  depends; i.e., $u_{a}$, $E_{i}u_{a}$ etc. In particular, this means that there is a function $\mC_{g}:\ro_{+}^{6}\rightarrow\ro_{+}$ which is continuous
  and increasing in each of its arguments such that 
  \begin{equation*}
    \begin{split}
      & \|g[u_{a},u_{b},v_{a},N_{1,a},N_{1,b}]\|_{C^{0,1}}\\
      \leq & \mC_{g}[\|u_{a}\|_{C^{2}},\|u_{b}\|_{C^{2}},\|v_{a}\|_{C^{2}},\|\d_{t}v_{a}\|_{C^{1}},
      \|N_{1,a}\|_{C^{2,1}},\|N_{1,b}\|_{C^{2,1}}]
    \end{split}
  \end{equation*}
  for all $t\in\mI$. Here the function $\mC_{g}$ only depends on bounds on
  $\bh_{2}^{ij}$, $\zeta$, $f_{1}$ and $f_{3}$ and their derivatives up to order three, two, two and two respectively; $(\bM,\bh_{\refer})$; 
  the frame $\{E_{i}\}$; and the $\lambda$ appearing in (\ref{eq:lambdabd}). However, by appealing to (\ref{eq:NobSchauder}), dependence on
  $\|N_{1,b}\|_{C^{2,1}}$ can be eliminated. The estimates for $\zeta[u_{b},v_{a}]$ and $\bh_{2}[u_{b}]$ are the same as in the proof of
  Lemma~\ref{lemma:schauder}. Appealing to Theorem~\ref{thm:globalSchauder} now yields the desired conclusion.
\end{proof}

\subsection{High order estimates for the lapse functions} Next, we bound the $N_{i,b}$, $i=1,2$, with respect to Sobolev norms.

\begin{lemma}\label{lemma:lapseestimatehighnorm}
  Given the assumptions and notation introduced in Lemma~\ref{lemma:ubvbNbsol}, there is a covering of $\bM$ (depending only on $(\bM,\bh_{\refer})$)
  and a function $\mC_{1}:\ro_{+}^{2}\rightarrow\ro_{+}$ which is continuous and monotonically increasing in each of its arguments, such that
  \begin{equation}\label{eq:Nbmainestimate}
    \|N_{1,b}(\cdot,t)\|_{H^{m+2}}\leq \mC_{1}[\|u_{b}(\cdot,t)\|_{C^{2}},\|v_{a}(\cdot,t)\|_{C^{1}}]\left(\|u_{b}(\cdot,t)\|_{H^{m+1}}+\|v_{a}(\cdot,t)\|_{H^{m}}+1\right)
  \end{equation}
  for all $t\in\mI$. Here $\mC_{1}$ only depends on the covering; $(\bM,\bh_{\refer})$; $\{E_{i}\}$; $m$; $\lambda$; $\lambda_{\zeta}$; a bound on up to
  $\max\{1,m\}$ derivatives of $\zeta$ and $f_{3}$; and up to $\max\{2,m+1\}$ derivatives of $\bh_{2}^{ij}$. Next, there is a function
  $\mC_{2}:\ro_{+}^{3}\rightarrow\ro_{+}$ which is continuous and increasing in each of its arguments, such that
  \begin{equation}\label{eq:Nbdotmainestimate}
    \begin{split}
      \|N_{2,b}(\cdot,t)\|_{H^{m+2}} \leq &
      \mC_{2}[\|u_{b}(\cdot,t)\|_{C^{2}},\|v_{a}(\cdot,t)\|_{C^{2}},\|\d_{t}v_{a}(\cdot,t)\|_{C^{1}}]\\
      & \cdot\left(\|u_{b}(\cdot,t)\|_{H^{m+1}}+\|v_{a}(\cdot,t)\|_{H^{m+1}}+\|\d_{t}v_{a}(\cdot,t)\|_{H^{m}}+1\right)
    \end{split}
  \end{equation}
  for all $t\in\mI$. Here $\mC_{2}$ only depends on the covering; $(\bM,\bh_{\refer})$; $\{E_{i}\}$; $m$; $\lambda$; $\lambda_{\zeta}$; a bound on up to
  $\max\{1,m\}$ derivatives of $\zeta$, $f_{3}$ and $f_{4}$; and up to $\max\{2,m+1\}$ derivatives of $\bh_{2}^{ij}$.
\end{lemma}
\begin{remark}\label{remark:def Hoelder Sob spaces}
  In the proof and the statement of the lemma, we use the notation $C^{k}=C^{k}(h_\refer)$, $C^{k,\a}=C^{k,\a}(h_\refer)$ and $H^k=H^k(h_\refer)$ for the sake of
  brevity; see Appendix~\ref{section:hoelder spaces on manifolds} for a definition of $C^k$- and H\"{o}lder spaces and norms, and
  Appendix~\ref{appendix:sobolev} for a definition of Sobolev spaces and norms.
\end{remark}
\begin{remark}\label{remark:Nbmainestimateimp}
  If, in addition to the assumptions of the lemma, $m>n/2+1$ and $f_{3}$ is such that $f_{3}(\xi)=0$ for all $\xi$ with the property that the
  last $(n+1)n_{1}+n_{2}$ components of $\xi$ equal zero, then (\ref{eq:Nbmainestimate}) can be improved in the sense that the $+1$ inside the
  parenthesis on the right hand side can be removed. However, this improvement comes at the price of $\mC_{1}$ depending on one more derivative
  of $\zeta$, $f_{3}$ and $\bh_{2}^{ij}$. Assuming, in addition, $f_{4}$ to be such that $f_{4}(\xi)=0$ for all $\xi$ with the property that the
  last $(n+1)n_{1}+(n+2)n_{2}+(n+2)(n+1)/2$ components of $\xi$ equal zero, then (\ref{eq:Nbdotmainestimate}) can be improved in that the $+1$
  inside the parenthesis on the right can be removed. Again, this improvement comes at the price of $\mC_{2}$ depending on one more derivative
  of $\zeta$, $f_{3}$, $f_{4}$ and $\bh_{2}^{ij}$.
\end{remark}
\begin{proof}
  As in the proof of Lemma~\ref{lemma:schauder}, we omit the argument $(\cdot,t)$ in the present proof.
  We derive the desired estimates by appealing to Theorem~\ref{thm: elliptic estimate} with $\S=\bM$, $h_{\refer}=\bh_{\refer}$, $h=\bh_{2}[u_{b}]$,
  $a=\zeta[u_{b},v_{a}]$ and $f$ equal to $-f_{3}[u_{b},v_{a}]$ or $-f_{4}[u_{b},v_{a},N_{1,b}]$ in the case of (\ref{eq:Nbmainestimate}) and
  (\ref{eq:Nbdotmainestimate}) respectively. In order to verify that the conditions of Theorem~\ref{thm: elliptic estimate} are satisfied, note
  that (\ref{eq:lambdabd}) and (\ref{eq:lambdazetalowerandupperbd}) imply that (\ref{eq:hhrefequivellest}) holds, with $c_{1}=\lambda$ and
  $c_{2}=\lambda_\zeta$. The estimate (\ref{eq:Nbmainestimate}) follows by combining (\ref{eq:Hkptwonormu}), the norm equivalence stated in
  Subsection~\ref{ssection:conventionsframe} and the Moser estimates stated in (\ref{eq:moserest}). In order to obtain the improvement
  described in Remark~\ref{remark:Nbmainestimateimp},
  we need to appeal to (\ref{eq:mainHkptwoestell}). Moreover, we need to estimate the terms appearing in the parenthesis on the right hand side.
  When applying (\ref{eq:mainHkptwoestell}) to (\ref{eq:Nblinmodsys}), $f=-f_{3}[u_{b},v_{a}]$, and we can rewrite $f_3$ as in (\ref{eq:fthreezerorepr}).
  Applying $E_{\bfI}$ to (\ref{eq:fthreezerorepr}) results in a linear combination of terms, each of which consists of one integral of a derivative of
  $f_{3}$ (along $[su_{b},sv_{a}]$) and at least one factor consisting of a derivative of $u_{b}$ or $v_{a}$. In order to estimate such a term in $L^{2}$,
  we extract the integral in $L^{\infty}$ and apply Moser estimates, see (\ref{eq:moserest}), to what remains. This results in an estimate of the
  form (\ref{eq:Nbmainestimate}), where the $+1$ on the right hand side has been removed. However, the corresponding $\mC_{1}$ depends on one more
  derivative of $f_{3}$ (in what follows, we omit similar remarks). Next, we need to estimate the second term inside the parenthesis on the right hand
  side of (\ref{eq:mainHkptwoestell}). In our case, this term takes the form
  \begin{equation}\label{eq:sectermimpHkptNobest}
    \begin{split}
      \|N_{1,b}\|_{C^{1}}\|\bh_{2}[u_{b}]\|_{H^{m+1}} \leq & \|N_{1,b}\|_{C^{1}}\|\bh_{2}[u_{b}]-\bh_{2}[0]\|_{H^{m+1}}
      +\|N_{1,b}\|_{C^{1}}\|\bh_{2}[0]\|_{H^{m+1}}.
    \end{split}
  \end{equation}
  In the first term on the right hand side of (\ref{eq:sectermimpHkptNobest}), we can estimate the first factor by appealing to (\ref{eq:NobSchauder})
  and the second factor by an argument similar to the one used to estimate the first term inside the parenthesis on the right hand side of
  (\ref{eq:mainHkptwoestell}); i.e., $f_{3}[u_{b},v_{a}]$. The second factor in the second term on the right hand side of (\ref{eq:sectermimpHkptNobest})
  can be estimated by a constant depending only on $m$, $(\bM,\bh_{\refer})$, $\{E_{i}\}$ and up to $m+1$ derivatives of $\bh_{2}^{ij}$. In order to
  estimate the first factor in the second term, it is sufficient to appeal to (\ref{eq:NobSchauderfz}) and Sobolev embedding. The argument needed to
  estimate the third term inside the parenthesis on the right hand side of (\ref{eq:mainHkptwoestell}) is essentially identical to the argument in the
  case of the second term. The improvement of (\ref{eq:Nbmainestimate}) described in Remark~\ref{remark:Nbmainestimateimp} follows. 

  In order to prove (\ref{eq:Nbdotmainestimate}), we proceed in the same way. The only difference is that we, in this case, need to estimate
  $f_{4}[u_{b},v_{a},N_{1,b}]$ in $C^{0,1}$ and in $H^{m}$. Due to the dependence of $f_{4}$, estimating $f_{4}[u_{b},v_{a},N_{1,b}]$ in $C^{0,1}$ entails
  estimating $u_{b}$ in $C^{2}$, $v_{a}$ in $C^{2}$, $\d_{t}v_{a}$ in $C^{1}$ and $N_{1,b}$ in $C^{2,1}$. Moreover, the dependence of the function $\mC_{2}$ on
  these norms is continuous and increasing in each of the norms. However, the dependence on the $C^{2,1}$-norm of $N_{1,b}$ can, by appealing to
  (\ref{eq:NobSchauder}), be replaced by dependence on the $C^{2}$-norm of $u_{b}$ and the $C^{1}$-norm of $v_{a}$. In order to estimate $f_{4}$ in $H^{m}$
  we appeal to Moser estimates. This leads to a continuous dependence on $u_{b}$ in $C^{1}$, $v_{a}$ in $C^{1}$, $\d_{t}v_{a}$ in $C^{0}$ and $N_{1,b}$ in
  $C^{2}$ (where the dependence on the $C^{2}$-norm of $N_{1,b}$ can, by appealing to (\ref{eq:NobSchauder}), be replaced by continuous dependence on the
  $C^{2}$-norm on $u_{b}$ and the $C^{1}$-norm of $v_{a}$). Moreover, we obtain linear dependence on the $H^{m+1}$-norm of $u_{b}$, the $H^{m+1}$-norm of
  $v_{a}$, the $H^{m}$-norm of $\d_{t}v_{a}$ and the $H^{m+2}$-norm of $N_{1,b}$. However, the latter dependence can, by appealing to
  (\ref{eq:Nbmainestimate}), be replaced by linear dependence on the $H^{m+1}$-norm of $u_{b}$ and the $H^{m}$-norm of $v_{a}$. In order to obtain the improvement
  described in Remark~\ref{remark:Nbmainestimateimp}, we, again, need to appeal to (\ref{eq:mainHkptwoestell}). The argument is very similar to the
  proof of the improvement of (\ref{eq:Nbmainestimate}). We leave the details to the reader. Let us, however, make two observations. First, when estimating
  $f_{4}[u_{b},v_{a},N_{1,b}]$ in $H^{m}$ by rewriting it in analogy with (\ref{eq:fthreezerorepr}), we, in addition to the norms appearing in the
  second factor on the right hand side of (\ref{eq:Nbdotmainestimate}), need to estimate $\|N_{1,b}\|_{H^{m+2}}$. However, this can be done by
  appealing to the improvement of (\ref{eq:Nbmainestimate}) obtained in Remark~\ref{remark:Nbmainestimateimp}. Second, when estimating the second term inside
  the parenthesis on the right hand side of (\ref{eq:mainHkptwoestell}), we first estimate it in analogy with (\ref{eq:sectermimpHkptNobest}) and then we
  appeal to (\ref{eq:NtbSchauder}), (\ref{eq:NtwobSchauderfz}) and Sobolev embedding. The statements of the lemma and Remark~\ref{remark:Nbmainestimateimp} follow. 
\end{proof}

\subsection{Bounding solutions to the simplified model system in the high norm; energy estimates}
Next, we wish to estimate the following energy of the solution $u_{b}$ to (\ref{eq:ublinmodsys}):
\[
E_{1,m+1}[u_{b}](t):=\tfrac{1}{2}\textstyle{\sum}_{|\bfI|\leq m+1}\int_{\bM_{t}}|E_{\bfI}u_{b}|^{2}d\mu_{\bh_{\refer}}=\frac{1}{2}\|u_{b}\|_{H^{m+1}(\bM_{t})}^{2},
\]
where we use the notation introduced in Appendix~\ref{appendix:sobolev}, in particular Subsection~\ref{ssection:conventionsframe}. Moreover,
we, for the sake of brevity, write $H^k(\bM_t)$ instead of $H^k(\bM_t;\bbE)$, where $\bbE:=\{E_i\}_{i=1}^n$. The bound we derive for
$E_{1,m+1}[u_{b}]$ depends on $u_{a}$, $v_{a}$ and $N_{1,a}$. For that reason, it is convenient to introduce the notation
\begin{equation}\label{eq:mfNadef}
  \begin{split}
    \mfN_{a}(t) := & \|u_{a}\|_{H^{m+1}(\bM_{t})}+\|v_{a}\|_{H^{m+1}(\bM_{t})}+\|\d_{t}v_{a}\|_{H^{m}(\bM_{t})}\\
    & +\|N_{1,a}\|_{H^{m+2}(\bM_{t})}+\|N_{2,a}\|_{H^{m+2}(\bM_{t})}+\ioih.
  \end{split}
\end{equation}
Here $\ioih:=0$ in case $f_{1}(p,0)=0$, $f_{2}(p,0)=0$, $f_{3}(p,0)=0$ and $f_{4}(p,0)=0$ for all $p\in M$. If one of these conditions is violated, then
$\ioih:=1$. Making this distinction makes it possible to apply the estimates below both to prove local existence and to prove
Cauchy stability. We include $N_{2,a}$ in (\ref{eq:mfNadef}) since, even though the evolution of the norm of $u_{b}$ does not depend on this quantity, the
evolution of the norm of $v_{b}$, e.g., does. The (time dependent) constants appearing in the estimates also depend on
\begin{equation}\label{eq:smfnadef}
  \begin{split}
    \mfn_{a}(t) := & \|u_{a}(\cdot,t)\|_{C^{2}}+\|v_{a}(\cdot,t)\|_{C^{2}}+\|\d_{t}v_{a}(\cdot,t)\|_{C^{1}}\\
    & +\|N_{1,a}(\cdot,t)\|_{C^{2,1}}+\|N_{2,a}(\cdot,t)\|_{C^{2,1}},
  \end{split}  
\end{equation}
where we use the notation $C^{k,\a}=C^{k,\a}(\bh_{\refer})$; cf., e.g., (\ref{eq:dtNobest}), giving an estimate of $\d_{t}N_{1,b}$ which is needed to estimate
the time derivative of the energy associated with $v_{b}$. With this notation, there is a continuous and monotonically increasing function
$\mC_{a}:\ro_{+}\rightarrow\ro_{+}$ such that 
\begin{equation}\label{eq:dtEompoest}
  \left|\d_{t}E_{1,m+1}[u_{b}](t)\right|\leq \mC_{a}[\mfn_{a}(t)]\mfN_{a}(t)E_{1,m+1}^{1/2}[u_{b}](t)
\end{equation}
for all $t\in\mI$. Here $\mC_{a}$ only depends on bounds on $m$, $\{E_{i}\}$, $(\bM,\bh_{\refer})$, $f_{1}$ and its partial derivatives up to order $m+2-\ioih$.
In order to obtain this estimate, it is sufficient to appeal to (\ref{eq:ublinmodsys}) and Moser estimates; see (\ref{eq:moserest}). In case $\ioih=0$,
we rewrite $f_{1}[u_{a},v_{a},N_{1,a}]$ in analogy with (\ref{eq:fthreezerorepr}); this is what leads to an extra derivative in case $\ioih=0$. In practice, a better
estimate than (\ref{eq:dtEompoest}) holds, but this improvement will not be of use in our applications. Due to (\ref{eq:dtEompoest}), we obtain bounds on
$\|u_{b}\|_{H^{m+1}}$ and, assuming $m>n/2+1$, on $\|u_{b}\|_{C^{2}}$. Below, we appeal to these bounds in order to control $v_{b}$ and $N_{i,b}$, $i=1,2$.
For that reason, it is convenient to define
\begin{equation}\label{eq:mfnabdef}
  \mfN_{a,b}(t):=\mfN_{a}(t)+\|u_{b}\|_{H^{m+1}(\bM_{t})},\ \ \
  \mfn_{a,b}(t):=\mfn_{a}(t)+\|u_{b}(\cdot,t)\|_{C^{2}}.
\end{equation}
Combining this notation with (\ref{eq:dtNobest}) yields the conclusion that there is a continuous and monotonically increasing function
$\mC_{1,b}:\ro_{+}\rightarrow\ro_{+}$ such that
\begin{equation}\label{eq:dtNobitomfn}
  \|\d_{t}N_{1,b}(\cdot,t)\|_{C^{2,1}}\leq\mC_{1,b}[\mfn_{a,b}(t)],
\end{equation}
where $\mC_{1,b}$ only depends on bounds on $\bh_{2}^{ij}$, $\zeta$, $f_{1}$ and $f_{3}$ and their derivatives up to order three, two, two and two
respectively; $(\bM,\bh_{\refer})$; the covering; the frame $\{E_{i}\}$; the $\lambda$ appearing in (\ref{eq:lambdabd}); and the $\lambda_{\zeta}$
appearing in (\ref{eq:lambdazetalowerandupperbd}).

Next, we wish to estimate the energy for $v_{b}$. Consider, to this end, 
\begin{equation}\label{eq:Etwomponevbdef}
  E_{2,m+1}[v_{b}]:=\tfrac{1}{2}\textstyle{\sum}_{|\bfI|\leq m}\int_{\bM_{t}}(|\d_{t}E_{\bfI}v_{b}|^{2}+\bh^{ij}_{1}[u_{b},N_{1,b}]E_{i}E_{\bfI}v_{b}\cdot E_{j}E_{\bfI}v_{b}
  +|E_{\bfI}v_{b}|^{2})d\mu_{\bh_{\refer}}.
\end{equation}
Due to the second term on the right hand side of (\ref{eq:vblinmodsys}), it is not possible to derive straightforward energy estimates for this
quantity. However, 
\begin{equation}\label{eq:hEtwompodef}
  \hE_{2,m+1}[v_{b},u_{b}]:=E_{2,m+1}[v_{b}]
  +\textstyle{\sum}_{|\bfI|\leq m}\int_{\bM_{t}}\g^{ij}[u_{b},v_{a},N_{1,b}]E_{j}E_{\bfI}(u_{b})\cdot E_{i}E_{\bfI}v_{b}d\mu_{\bh_{\refer}}
\end{equation}
is an object for which we can estimate the time evolution. On the other hand, this quantity has the disadvantage that it does not dominate the
Sobolev norms we wish to control. For this reason, it is convenient to let
\[
\a_{\g}:=\textstyle{\sup}_{\xi\in M\times\rn{n_{1}+n_{2}+n+1}}\textstyle{\sum}_{i,j}\|\g^{ij}(\xi))\|;
\]
note that we assume $\g^{ij}$ to be globally bounded. Then 
\[
\big|\textstyle{\sum}_{|\bfI|\leq m}\int_{\bM_{t}}\g^{ij}E_{j}E_{\bfI}(u_{b})\cdot E_{i}E_{\bfI}v_{b}d\mu_{\bh_{\refer}}\big|\leq \tfrac{1}{2}E_{2,m+1}[v_{b}]
+C_{\lambda,n}\a_{\g}^{2}E_{1,m+1}[u_{b}],
\]
where $C_{\lambda,n}$ only depends on $\lambda$ and $n$. Introducing
\begin{equation}\label{eq:bgdef}
  \b_{\g}:=C_{\lambda,n}\a_{\g}^{2}+1,
\end{equation}
we conclude that if
\begin{equation}\label{eq:tEtwompo}
  \chE_{2,m+1}[v_{b},u_{b}]:=\hE_{2,m+1}[v_{b},u_{b}]+\b_{\g}E_{1,m+1}[u_{b}],
\end{equation}
then
\begin{equation}\label{eq:tEdominance}
  E_{1,m+1}[u_{b}](t)+\tfrac{1}{2}E_{2,m+1}[v_{b}](t)\leq \chE_{2,m+1}[v_{b},u_{b}](t).
\end{equation}
We derive the following estimate for $\chE_{2,m+1}[v_{b},u_{b}]$.

\begin{lemma}\label{lemma:energyestimatev}
  With assumptions, notation and conclusions as in Lemma~\ref{lemma:ubvbNbsol}, fix $m\in\nn{}$ with $m>n/2+1$ and define $\chE_{2,m+1}[v_{b},u_{b}]$ by
  (\ref{eq:tEtwompo}). Then there is a covering of $\bM$ (depending only on $(\bM,\bh_{\refer})$) and a continuous and monotonically increasing function
  $\mC_{m}:\ro_{+}\rightarrow\ro_{+}$ such that, given $t_{1}\in \mI$,
  \begin{equation}\label{eq:tEtwompodiffest}
    \begin{split}
      & |\chE_{2,m+1}[v_{b},u_{b}](t_{1})-\chE_{2,m+1}[v_{b},u_{b}](t_{0})|\\
      \leq & \big|\textstyle{\int}_{t_{0}}^{t_{1}}\mC_{m}\circ\mfn_{a,b}\cdot\big(E_{2,m+1}[v_{b}]+\mfN_{a,b}(\|v_{b}\|_{C^{1}}+1)E_{2,m+1}^{1/2}[v_{b}]
      +\mfN_{a,b}^{2}\big)dt\big|.
    \end{split}
  \end{equation}
  Here $\mC_{m}$ only depends on bounds on up to $\max\{3,m+2-\ioih\}$ derivatives of $\bh_{2}^{ij}$; $m+1-\ioih$ derivatives of $\zeta$, $f_{1}$, $f_{2}$
  and $f_{3}$;  and $m+1$ derivatives of $\g^{ij}$ and $\bh_{1}^{ij}$; $(\bM,\bh_{\refer})$; the covering; the frame $\{E_{i}\}$; $m$; $\lambda$; and
  $\lambda_{\zeta}$.
\end{lemma}
\begin{remark}
  Due to (\ref{eq:tEdominance}), we can replace $E_{2,m+1}$ with $\chE_{2,m+1}$ on the right hand side of (\ref{eq:tEtwompodiffest}). Moreover,
  since $m>n/2$, $\|v_{b}\|_{C^{1}}$ can, due to Sobolev embedding, be replaced by the square root of $\chE_{2,m+1}$. Combining these observations
  with (\ref{eq:tEtwompodiffest}), the prior bound on $\mfN_{a,b}$ and Gr\"{o}nwall's lemma yields a bound on $\chE_{2,m+1}$.
\end{remark}
\begin{proof}
  Time differentiating $E_{2,m+1}$ yields
  \begin{equation}\label{eq:dtEtomponefirststep}
    \begin{split}
      \d_{t}E_{2,m+1}[v_{b}]
      := & \textstyle{\sum}_{|\bfI|\leq m}\int_{\bM_{t}}(\d_{t}^{2}E_{\bfI}v_{b}\cdot\d_{t}E_{\bfI}v_{b}+(\d_{t}\bh_{1}^{ij}/2)E_{i}E_{\bfI}v_{b}\cdot E_{j}E_{\bfI}v_{b}\\
      & \phantom{\textstyle{\sum_{i=1}\int\bM_{t}}+}+\bh_{1}^{ij}E_{i}E_{\bfI}v_{b}\cdot E_{j}E_{\bfI}\d_{t}v_{b}
      +E_{\bfI}v_{b}\cdot E_{\bfI}\d_{t}v_{b})d\mu_{\bh_{\refer}}.
    \end{split}
  \end{equation}
  In order to estimate the contribution from the second term in the integrand, note that $\d_{t}\bh_{1}^{ij}$ is a linear combination of terms
  consisting of four types
  of factors: a first derivative of $\bh_{1}^{ij}$; $\d_{t}u_{b}$; $\d_{t}N_{1,b}$; and $E_{i}\d_{t}N_{1,b}$. Appealing to (\ref{eq:ublinmodsys}),
  (\ref{eq:dtNobitomfn}) and the assumptions, cf. Subsection~\ref{ssection:modelsystem}, it is clear that the contribution can be estimated by
  $\mC[\mfn_{a,b}(t)]E_{2,m+1}$, where $\mC:\ro_{+}\rightarrow\ro_{+}$ is continuous and monotonic. Moreover, $\mC$ only depends on bounds on $\bh_{2}^{ij}$,
  $\zeta$, $f_{1}$, $f_{3}$, $\g^{ij}$ and $\bh_{1}^{ij}$ and their derivatives up to order three, two, two, two, two, one and one respectively;
  $(\bM,\bh_{\refer})$; the covering; the frame $\{E_{i}\}$; $m$; the $\lambda$ appearing in (\ref{eq:lambdabd}); and the $\lambda_{\zeta}$ appearing in
  (\ref{eq:lambdazetalowerandupperbd}). Here we include dependence on, e.g., $\g^{ij}$, which is strictly speaking not necessary. The reason for this
  is that we wish to include all the dependence that will appear in the course of the proof, instead of adding dependence incrementally. 
  The contribution from the fourth term in the integrand of the right hand side of (\ref{eq:dtEtomponefirststep}) can be estimated by a numerical
  constant times $E_{2,m+1}[v_{b}]$. Next, consider
  \begin{equation}\label{eq:vbenspderfs}
    \begin{split}
      \textstyle{\int}_{\bM_{t}}\bh_{1}^{ij}E_{i}E_{\bfI}v_{b}\cdot E_{j}E_{\bfI}\d_{t}v_{b}d\mu_{\bh_{\refer}}
      = & -\textstyle{\int}_{\bM_{t}}E_{j}(\bh_{1}^{ij})E_{i}E_{\bfI}v_{b}\cdot E_{\bfI}\d_{t}v_{b}d\mu_{\bh_{\refer}}\\
      & -\textstyle{\int}_{\bM_{t}}\bh_{1}^{ij}E_{j}E_{i}E_{\bfI}v_{b}\cdot E_{\bfI}\d_{t}v_{b}d\mu_{\bh_{\refer}}\\
      & -\textstyle{\int}_{\bM_{t}}\bh_{1}^{ij}E_{i}E_{\bfI}v_{b}\cdot E_{\bfI}\d_{t}v_{b}\rodiv_{\bh_{\refer}}E_{j}d\mu_{\bh_{\refer}};
    \end{split}
  \end{equation}
  cf., e.g., the derivation of \cite[(B.7), p.~211]{RinWave}. 
  By (\ref{eq:NobSchauder}) and arguments similar to the above, the first term on the right hand side can be estimated by $\mC[\mfn_{a,b}(t)]E_{2,m+1}$,
  where $\mC$ has the same dependence as before. The contribution from the third term on the right hand side of (\ref{eq:vbenspderfs}) can be estimated
  by a constant times $E_{2,m+1}$, where the constant only depends on $\lambda$, $(\bM,\bh_{\refer})$ and $\{E_{i}\}$. In order to estimate the second term
  on the right hand side of (\ref{eq:vbenspderfs}), we need to consider $[\bh_{1}^{ij}E_{j}E_{i},E_{\bfI}]$. This expression is a sum of terms of the form
  $\lambda_{\bfI_{1},\bfI_{2}} E_{\bfI_{1}}(\bh_{1}^{ij})E_{\bfI_{2}}$,
  where $\lambda_{\bfI_{1},\bfI_{2}}$ only depends on the frame, $\bfI_{1}$ and $\bfI_{2}$; $|\bfI_{1}|+|\bfI_{2}|\leq |\bfI|+2$; $|\bfI_{1}|\geq 1$ if
  $|\bfI_{1}|+|\bfI_{2}|=|\bfI|+2$; and $|\bfI_{2}|\geq 1$. We thus need to estimate 
  \begin{equation}\label{eq:bhoneijvbprelest}
    \|E_{\bfI_{1}}(\bh_{1}^{ij})E_{\bfI_{2}}v_{b}\|_{L^{2}}.
  \end{equation}
  Expanding the first factor inside the norm leads to a linear combination of terms consisting of derivatives of $\bh_{1}^{ij}$ times spatial derivatives
  applied to $u_{b}$, $N_{1,b}$ and $E_{i}N_{1,b}$. Extracting the derivatives of $\bh_{1}^{ij}$ in the supremum norm, applying the Moser estimate (\ref{eq:moserest})
  to what remains and appealing to (\ref{eq:NobSchauder}) and (\ref{eq:Nbmainestimate}) (and, in case $\ioih=0$, to Remark~\ref{remark:Nbmainestimateimp}) yields
  \[
    \|E_{\bfI_{1}}(\bh_{1}^{ij})E_{\bfI_{2}}v_{b}\|_{L^{2}}\leq \mC_{m}[\mfn_{a,b}(t)]\big(\|v_{b}\|_{C^{1}}\mfN_{a,b}+E_{2,m+1}^{1/2}[v_{b}]\big),
  \]
  where $\mC_{m}$, in addition to the dependence of $\mC$, only depends on $m$; bounds on up to $m+1-\ioih$ derivatives of $\zeta$, $f_{1}$, $f_{2}$
  and $f_{3}$; bounds on up to $m+2-\ioih$ derivatives of $\bh_2^{ij}$; 
  and bounds on up to $m+1$ derivatives of $\g^{ij}$ and $\bh_{1}^{ij}$ (we here add dependence that will only be necessary later on). Adding up, 
  \begin{equation}\label{eq:fsEtwompovbest}
    \begin{split}
      & \big|\d_{t}E_{2,m+1}[v_{b}]
      -\textstyle{\sum}_{|\bfI|\leq m}\int_{\bM_{t}}E_{\bfI}(\d_{t}^{2}v_{b}-\bh_{1}^{ij}E_{i}E_{j}v_{b})\cdot\d_{t}E_{\bfI}v_{b}d\mu_{\bh_{\refer}}\big|\\
      \leq & \mC_{m}[\mfn_{a,b}(t)]\big(E_{2,m+1}[v_{b}]+\|v_{b}\|_{C^{1}}\mfN_{a,b}E_{2,m+1}^{1/2}[v_{b}]\big),
    \end{split}
  \end{equation}
  where $\mC_{m}$ has the dependence described above. Next, note that 
  \begin{equation}\label{eq:EbfIwoperven}
    \begin{split}
      & \textstyle{\sum}_{|\bfI|\leq m}\int_{\bM_{t}}E_{\bfI}(\d_{t}^{2}v_{b}-\bh_{1}^{ij}E_{i}E_{j}v_{b})\cdot\d_{t}E_{\bfI}v_{b}d\mu_{\bh_{\refer}}\\
      = & \textstyle{\sum}_{|\bfI|\leq m}\int_{\bM_{t}}E_{\bfI}(\g^{ij}E_{i}E_{j}u_{b})\cdot\d_{t}E_{\bfI}v_{b}d\mu_{\bh_{\refer}}
      +\sum_{|\bfI|\leq m}\int_{\bM_{t}}E_{\bfI}f_{2}\cdot \d_{t}E_{\bfI}v_{b}d\mu_{\bh_{\refer}}.
    \end{split}
  \end{equation}
  The second term can be estimated by appealing to Moser estimates; it is bounded in absolute value by $\mC_{m}[\mfn_{a,b}(t)]\mfN_{a}E_{2,m+1}^{1/2}[v_{b}]$,
  where $\mC_{m}$ has the dependence described above. Note that, in case $\ioih=0$, we first rewrite $f_{2}$ in analogy with (\ref{eq:fthreezerorepr}).
  In the first term on the right hand side of (\ref{eq:EbfIwoperven}), if not all the derivatives hit $u_{b}$, the corresponding expressions can
  estimated similarly to the estimate of (\ref{eq:bhoneijvbprelest}). What we are left with is a sum of terms of the form
  \begin{equation*}
    \begin{split}
      & \textstyle{\int}_{\bM_{t}}\g^{ij}E_{i}E_{j}E_{\bfI}u_{b}\cdot\d_{t}E_{\bfI}v_{b}d\mu_{\bh_{\refer}}\\
      = & -\textstyle{\int}_{\bM_{t}}E_{i}(\g^{ij})E_{j}E_{\bfI}u_{b}\cdot\d_{t}E_{\bfI}v_{b}d\mu_{\bh_{\refer}}
      -\int_{\bM_{t}}\g^{ij}E_{j}E_{\bfI}u_{b}\cdot\d_{t}E_{i}E_{\bfI}v_{b}d\mu_{\bh_{\refer}}\\
      & -\textstyle{\int}_{\bM_{t}}\g^{ij}E_{j}E_{\bfI}u_{b}\cdot\d_{t}E_{\bfI}v_{b}\cdot (\rodiv_{\bh_{\refer}}E_{i})d\mu_{\bh_{\refer}}.
    \end{split}
  \end{equation*}
  The first and last terms can be estimated similarly to the above. In order to estimate the second term, we also need to integrate in time. In fact,
  \begin{equation*}
    \begin{split}
      & -\textstyle{\int}_{t_{0}}^{t_{1}}\int_{\bM_{t}}\g^{ij}E_{j}E_{\bfI}u_{b}\cdot\d_{t}E_{i}E_{\bfI}v_{b}d\mu_{\bh_{\refer}}dt\\
      = & -\big[\textstyle{\int}_{\bM_{t}}\g^{ij}E_{j}E_{\bfI}u_{b}\cdot E_{i}E_{\bfI}v_{b}d\mu_{\bh_{\refer}}\big]_{t_{0}}^{t_{1}}
      +\textstyle{\int}_{t_{0}}^{t_{1}}\int_{\bM_{t}}\d_{t}(\g^{ij})E_{j}E_{\bfI}u_{b}\cdot E_{i}E_{\bfI}v_{b}d\mu_{\bh_{\refer}}dt\\
      & +\textstyle{\int}_{t_{0}}^{t_{1}}\int_{\bM_{t}}\g^{ij}E_{j}E_{\bfI}(\d_{t}u_{b})\cdot E_{i}E_{\bfI}v_{b}d\mu_{\bh_{\refer}}dt.
    \end{split}
  \end{equation*}
  In order to estimate the second term on the right hand side, we proceed as in the estimate of the contribution from the second term in the
  integrand on the right hand side of (\ref{eq:dtEtomponefirststep}). The third term on the right hand side can be estimated by appealing to
  (\ref{eq:ublinmodsys}) and Moser estimates; again, we first rewrite $f_{1}$ in analogy with (\ref{eq:fthreezerorepr}) in case $\ioih=0$. The
  first term gives rise to boundary terms that have to be included in the energy. This is what leads to the energy defined in (\ref{eq:hEtwompodef}).
  The above estimates lead to the conclusion that
  \begin{equation}\label{eq:hEtwompodiffest}
    \begin{split}
      & \big|\hE_{2,m+1}[v_{b},u_{b}](t_{1})-\hE_{2,m+1}[v_{b},u_{b}](t_{0})\big|\\
      \leq & \big|\textstyle{\int}_{t_{0}}^{t_{1}}\mC_{m}[\mfn_{a,b}(t)](E_{2,m+1}[v_{b}]+(\|v_{b}\|_{C^{1}}+1)\mfN_{a,b}E_{2,m+1}^{1/2}[v_{b}])dt\big|.
    \end{split}
  \end{equation}
  Combining (\ref{eq:hEtwompodiffest}) with (\ref{eq:dtEompoest}) and (\ref{eq:tEtwompo}) yields (\ref{eq:tEtwompodiffest}). The lemma follows.
\end{proof}

\subsection{Convergence in a low norm}\label{ssection:conlownorm}

Next, we estimate differences of solutions in a low norm. The relevant situation to consider is the following. Fix a model system.
Then, using notation and assumptions as in Subsection~\ref{ssection:modelsystem}, let $\mJ\subseteq\mI$, $M_{\mJ}:=\bM\times\mJ$, and
assume that
\begin{subequations}\label{seq:uabietcreg}
  \begin{align}
    u_{a,i},u_{b,i} \in & C^{1}(\mJ,C^{1}(\bM,\rn{n_{1}}))\cap C^{0}(\mJ,C^{2}(\bM,\rn{n_{1}})),\\
    v_{a,i},v_{b,i} \in & C^{2}(\mJ,C^{0}(\bM,\rn{n_{2}}))\cap C^{1}(\mJ,C^{1}(\bM,\rn{n_{2}}))\cap C^{0}(\mJ,C^{2}(\bM,\rn{n_{2}})),\\
    N_{1,a,i},N_{1,b,i} \in & C^{1}(\mJ,C^{1}(\bM))\cap C^{0}(\mJ,C^{2,1}(\bM)),\\
    N_{2,a,i},N_{2,b,i} \in & C^{0}(\mJ,C^{2,1}(\bM))
  \end{align}
\end{subequations}
for $i=1,2$. Assume, finally, that $u_{b,i},v_{b,i},N_{1,b,i},N_{2,b,i}$ are solutions to (\ref{seq:linmodsys}) corresponding to
$u_{a,i},v_{a,i},N_{1,a,i},N_{2,a,i}$ for $i=1,2$. In practice, we are going to apply the results of the present subsection in two different settings.
One setting is when proving uniqueness of solutions. The other setting is when proving convergence of an iteration. In the latter setting, the $u_{a,i}$
etc. are given smooth functions on $M$. Moreover, we specify smooth initial data $u_{0,i}$, $v_{0,0,i}$, $v_{0,1,i}$, $i=1,2$, as in the statement of
Lemma~\ref{lemma:ubvbNbsol}. Appealing to Lemma~\ref{lemma:ubvbNbsol} then yields corresponding solutions $u_{b,i}$, $v_{b,i}$, $N_{l,b,i}$. In what
follows, it is convenient to use the notation
\begin{equation}\label{eq:UVLdef}
  U_{a}:=u_{a,2}-u_{a,1},\ \ \
  V_{a}:=v_{a,2}-v_{a,1},\ \ \
  L_{l,a}:=N_{l,a,2}-N_{l,a,1},
\end{equation}
$l=1,2$. We similarly define $U_{b}:=u_{b,2}-u_{b,1}$ etc. In analogy with (\ref{eq:mfNadef}), we introduce the notation
\begin{equation}\label{eq:mNadef}
  \begin{split}
    \mN_{a}(t) := & \|U_{a}\|_{H^{1}(\bM_{t})}+\|V_{a}\|_{H^{1}(\bM_{t})}+\|\d_{t}V_{a}\|_{L^{2}(\bM_{t})}\\
    & +\|L_{1,a}\|_{H^{2}(\bM_{t})}+\|L_{2,a}\|_{H^{2}(\bM_{t})}.
  \end{split}
\end{equation}
Since the constants appearing in the estimates depend on the norms of both $u_{a,1}$ and $u_{a,2}$ etc., it is also convenient to use the notation
\begin{equation}\label{eq:mfladef}
  \begin{split}
    \mfl_{a}(t) := & \textstyle{\sum}_{i=1}^{2}\|u_{a,i}(\cdot,t)\|_{C^{2}}+\textstyle{\sum}_{i=1}^{2}\|v_{a,i}(\cdot,t)\|_{C^{2}}
    +\textstyle{\sum}_{i=1}^{2}\|\d_{t}v_{a,i}(\cdot,t)\|_{C^{1}}\\
    & +\textstyle{\sum}_{i=1}^{2}\|N_{1,a,i}(\cdot,t)\|_{C^{2,1}}+\textstyle{\sum}_{i=1}^{2}\|N_{2,a,i}(\cdot,t)\|_{C^{2,1}},
  \end{split}
\end{equation}
where $C^{k,\a}=C^{k,\a}(\bh_{\refer})$. Similarly to (\ref{eq:mfnabdef}), it is also convenient to introduce the notation
\begin{subequations}\label{seq:mNabmflab}
  \begin{align}
    \mN_{a,b}(t) := & \mN_{a}(t)+\|U_{b}\|_{H^{1}(\bM_{t})},\label{eq:mNabdef}\\
    \mfl_{a,b}(t) := & \mfl_{a}(t)+\textstyle{\sum}_{i=1}^{2}\|u_{b,i}(\cdot,t)\|_{C^{2}}.\label{eq:mflabdef}
    %+\textstyle{\sum}_{i=1}^{2}\|N_{1,b,i}\|_{C^{2,1}}+\textstyle{\sum}_{i=1}^{2}\|N_{2,b,i}\|_{C^{2,1}}
  \end{align}
\end{subequations}
%though it should be noted that we need neither the bounds on $N_{1,b,i}$ nor the bounds on $N_{2,b,i}$ until we estimate the difference of the
%lapse functions in Lemma~\ref{lemma:lapseestimates} below.

\textit{Estimates for the lapse function.} We begin by estimating $L_{1,b}$ and $L_{2,b}$.

\begin{lemma}\label{lemma:lapseestimates}
  Given the above assumptions and notation, there is a covering of $\bM$ (depending only on $(\bM,\bh_{\refer})$) and a continuous and monotonically
  increasing function $\mfC_{L}:\ro_{+}\rightarrow\ro_{+}$ such that 
  \begin{align}
    \|L_{1,b}\|_{H^{2}(\bM_{t})} \leq & \mfC_{L}[\mfl_{a,b}(t)]\left(\|U_{b}\|_{H^{1}(\bM_{t})}+\|V_{a}\|_{L^{2}(\bM_{t})}\right),\label{eq:Lobest}\\
    \|L_{2,b}\|_{H^{2}(\bM_{t})} \leq & \mfC_{L}[\mfl_{a,b}(t)]
    \left(\|U_{b}\|_{H^{1}(\bM_{t})}+\|V_{a}\|_{H^{1}(\bM_{t})}+\|\d_{t}V_{a}\|_{L^{2}(\bM_{t})}\right)\label{eq:Ltwobest}
  \end{align}
  for $t\in\mJ$. Here $\mfC_{L}$ only depends on bounds on $\bh_{2}^{ij}$, $\zeta$, $f_{3}$ and $f_{4}$ and their derivatives up to order two, one,
  one and one respectively; $(\bM,\bh_{\refer})$; the covering; the frame $\{E_{i}\}$; $\lambda$; and $\lambda_{\zeta}$.
\end{lemma}
\begin{proof}
  Due to Lemma~\ref{lemma:schauder}, there is a function $\mfC$, with the same properties and dependence as $\mfC_{L}$ appearing in the statement,
  such that
  \begin{equation}\label{eq:NiblSchaudest}
    \|N_{1,b,1}\|_{C^{2,1}}+\|N_{1,b,2}\|_{C^{2,1}}+\|N_{2,b,1}\|_{C^{2,1}}+\|N_{2,b,2}\|_{C^{2,1}}\leq\mfC[\mfl_{a,b}(t)]
  \end{equation}
  for all $t\in\mJ$; as before, we ignore the argument $(\cdot,t)$. Next, due to (\ref{eq:Nblinmodsys}),
  \begin{equation}\label{eq:DeltabhtwodiffLocase}
    \begin{split}
      \Delta_{\bh_{2}[u_{b,1}]}L_{1,b} = & \zeta[u_{b,2},v_{a,2}]L_{1,b}+(\Delta_{\bh_{2}[u_{b,1}]}-\Delta_{\bh_{2}[u_{b,2}]})N_{1,b,2}\\
      & +(\zeta[u_{b,2},v_{a,2}]-\zeta[u_{b,1},v_{a,1}])N_{1,b,1}+f_{3}[u_{b,2},v_{a,2}]-f_{3}[u_{b,1},v_{a,1}].
    \end{split}
  \end{equation}
  At this stage, we can appeal to Proposition~\ref{prop: bound on the inverse} and Remark~\ref{remark:basic Htwo est low reg},
  with $h=\bh_{2}[u_{b,1}]$, $a=\zeta[u_{b,2},v_{a,2}]$ and $u=L_{1,b}$. Keeping (\ref{eq:NiblSchaudest}) in mind, this yields (\ref{eq:Lobest}).
  Next, note that (\ref{eq:Nbdotlinmodsys}) yields
  \begin{equation}\label{eq:DeltabhtwodiffLtwocase}
    \begin{split}
      \Delta_{\bh_{2}[u_{b,1}]}L_{2,b} = & \zeta[u_{b,2},v_{a,2}]L_{2,b}+(\Delta_{\bh_{2}[u_{b,1}]}-\Delta_{\bh_{2}[u_{b,2}]})N_{2,b,2}\\
      & +(\zeta[u_{b,2},v_{a,2}]-\zeta[u_{b,1},v_{a,1}])N_{1,b,2}\\
      & +f_{4}[u_{b,2},v_{a,2},N_{1,b,2}]-f_{4}[u_{b,1},v_{a,1},N_{1,b,1}].
    \end{split}
  \end{equation}
  Again, we can appeal to Proposition~\ref{prop: bound on the inverse} and Remark~\ref{remark:basic Htwo est low reg},
  with $h=\bh_{2}[u_{b,1}]$, $a=\zeta[u_{b,2},v_{a,2}]$ and $u=L_{2,b}$. Keeping (\ref{eq:NiblSchaudest}) in mind, this yields an estimate
  analogous to (\ref{eq:Ltwobest}), but with the one difference that there is an additional term $\|L_{1,b}\|_{H^{2}(\bM_{t})}$ in the paranthesis
  on the right hand side. However, this term can be eliminated by appealing to (\ref{eq:Lobest}).
\end{proof}

\textit{Energy estimates.} In order to control $U_{b}$, we introduce the energy
\begin{equation}\label{eq:meonedef}
  \me_{1}:=\tfrac{1}{2}\textstyle{\sum}_{|\bfI|\leq 1}\int_{\bM}|E_{\bfI}U_{b}|^{2}d\mu_{\bh_{\refer}}.
\end{equation}
In order to control $V_{b}$, we, in analogy with the derivation of the bounds in the high norm, introduce three energies.
To begin with, let
\begin{equation}\label{eq:metwodef}
  \me_{2}:=\tfrac{1}{2}\textstyle{\int}_{\bM}[|\d_{t}V_{b}|^{2}+\bh^{ij}_{1}[u_{b,1},N_{1,b,1}]E_{i}V_{b}\cdot E_{j}V_{b}
  +|V_{b}|^{2}]d\mu_{\bh_{\refer}}.
\end{equation}
Due to the second term on the right hand side of (\ref{eq:vblinmodsys}), we are not able to derive estimates for this energy
directly. Instead, we have to consider
\begin{equation}\label{eq:hmetwodef}
  \hme_{2}:=\me_{2}+\textstyle{\int}_{\bM}\g^{ij}[u_{b,1},v_{a,1},N_{1,b,1}]E_{j}U_{b}\cdot E_{i}V_{b}d\mu_{\bh_{\refer}}.
\end{equation}
For this quantity, it is possible to derive energy estimates. However, it might be negative, and does not bound the norms we wish to
control. On the other hand, by a combination of Cauchy-Schwarz and Young's inequality, 
\[
  \big|\textstyle{\int}_{\bM}\g^{ij}[u_{b,1},v_{a,1},N_{1,b,1}]E_{j}U_{b}\cdot E_{i}V_{b}d\mu_{\bh_{\refer}}\big|\leq \tfrac{1}{2}\me_{2}+\mfb_{D}\me_{1},
\]
where $\mfb_{D}$ is a constant only depending on $\lambda$, $n$ and bounds on the $\g^{ij}$'s. Introducing
\begin{equation}\label{eq:chmetwodef}
  \chme_{2}:=\hme_{2}+(\mfb_{D}+1)\me_{1}
\end{equation}
then yields an energy which we can estimate and which is such that
\begin{equation}\label{eq:chmetwodom}
  \tfrac{1}{2}\me_{2}+\me_{1}\leq \chme_{2}.
\end{equation}

\begin{lemma}\label{lemma:estimatesofdifferences}
  Fix a model system and an interval $\mJ\subseteq \mI$. Let $u_{a,i}$, $v_{a,i}$, $N_{1,a,i}$, $N_{2,a,i}$ and
  $u_{b,i}$, $v_{b,i}$, $N_{1,b,i}$, $N_{2,b,i}$, $i=1,2$, have the regularity described by (\ref{seq:uabietcreg}). Assume moreover, that
  $u_{b,i}$, $v_{b,i}$, $N_{1,b,i}$, $N_{2,b,i}$ is a solution to (\ref{seq:linmodsys}) corresponding to $u_{a,i}$, $v_{a,i}$, $N_{1,a,i}$, $N_{2,a,i}$.
  Then, using the notation introduced in (\ref{eq:mNadef})--(\ref{seq:mNabmflab}) and (\ref{eq:meonedef})--(\ref{eq:chmetwodef}), there is a
  continuous and monotonically increasing function $\mfC:\ro_{+}\rightarrow\ro_{+}$ such that
  \begin{equation}\label{eq:meoneenest}
    |\me_{1}(t_{1})-\me_{1}(t_{0})|\leq \big|\textstyle{\int}_{t_{0}}^{t_{1}}\mfC[\mfl_{a}(t)]\mN_{a}(t)\me_{1}^{1/2}(t)dt\big|
  \end{equation}
  for all $t_{0},t_{1}\in\mJ$, where $\mfC$ only depends on bounds on $f_{1}$ and its partial derivatives up to order $2$. In addition, there
  is a continuous and monotonically increasing function $\chmfC:\ro_{+}\rightarrow\ro_{+}$ such that,
  \begin{equation}\label{eq:chmeenest}
    \begin{split}
      |\chme_{2}(t_{1})-\chme_{2}(t_{0})|
      \leq & \big|\textstyle{\int}_{t_{0}}^{t_{1}}\chmfC[\mfl_{a,b}(t)](\|v_{b,2}\|_{C^{2}}+\|\d_{t}N_{1,b,1}\|_{C^{1}}+1)[\me_{2}(t)+\mN_{a,b}^{2}(t)]dt\big|
    \end{split}    
  \end{equation}
  for all $t_{0},t_{1}\in\mJ$,
  where $\chmfC:\ro_{+}\rightarrow\ro_{+}$ only depends on bounds on $\bh_{2}^{ij}$, $\bh_{1}^{ij}$,
  $\g^{ij}$, $\zeta$, $f_{1}$, $f_{2}$, $f_{3}$ and $f_{4}$ and their derivatives up to order two, one, one, one, two, one, one and one respectively;
  $(\bM,\bh_{\refer})$; the covering; $\{E_{i}\}$; $\lambda$; and $\lambda_{\zeta}$.
\end{lemma}
\begin{remark}
  The estimate (\ref{eq:meoneenest}) gives a bound on $\mN_{a,b}$, which justifies the appearance of this quantity in (\ref{eq:chmeenest}).
\end{remark}
\begin{remark}
   Combining (\ref{eq:chmetwodom}) and (\ref{eq:chmeenest}) yields
  \begin{equation}\label{eq:chmetwoimpest}
    |\chme_{2}(t_{1})-\chme_{2}(t_{0})|\leq \big|\textstyle{\int}_{t_{0}}^{t_{1}}C(\chme_{2}+\mN_{a}^{2})dt\big|,
  \end{equation}
  where $C$ depends on $\mfl_{a,b}$, $\|\d_{t}N_{1,b,1}\|_{C^{1}}$ and $\|v_{b,2}\|_{C^{2}}$.
\end{remark}
\begin{proof}
  Time differentiating $\me_{1}$ introduced in (\ref{eq:meonedef}) and appealing to the equations for $u_{b,i}$ yields 
  \[
    |\d_{t}\me_{1}|\leq \mfC[\mfl_{a}]\mN_{a}\me_{1}^{1/2}
  \]
  for $t\in\mJ$, where $\mfC:\ro_{+}\rightarrow\ro_{+}$ is a continuous and monotonically increasing function depending only on bounds on
  $f_{1}$ and its partial derivatives up to order $2$. The estimate (\ref{eq:meoneenest}) follows.

  Next, we time differentiate the energy introduced in (\ref{eq:metwodef}). The term that results when the time derivative hits the last
  term in the integrand can be estimated by the energy. Consider the expression that results
  when the time derivative hits $\bh_{1}^{ij}[u_{b,1},N_{1,b,1}]$. Appealing to (\ref{eq:ublinmodsys}) leads to the conclusion that
  \[
  \|\d_{t}(\bh_{1}^{ij}[u_{b,1},N_{1,b,1}])\|_{C^{0}}\leq C(\|\d_{t}N_{1,b,1}\|_{C^{1}}+1),  
  \]
  where $C$ is a constant depending only on $(\bM,\bh_{\refer})$, $\{E_{i}\}$, bounds on $f_{1}$ and bounds on the derivative of
  $\bh_{1}^{ij}$. Next, it is natural to integrate by parts in analogy with (\ref{eq:vbenspderfs}). The analogue of the last term on the
  right hand side can be estimated by $C\me_{2}$, where $C$ only depends on $\lambda$, $(\bM,\bh_{\refer})$ and $\{E_{i}\}$. In order to estimate
  the first term on the right hand side, we need to estimate
  \[
  \|E_{j}(\bh_{1}^{ij}[u_{b,1},N_{1,b,1}])\|_{C^{0}}\leq C(\|u_{b,1}\|_{C^{1}}+\|N_{1,b,1}\|_{C^{2}}+1),
  \]
  where $C$ only depends on bounds on the first derivative of $\bh_{1}^{ij}$, $(\bM,\bh_{\refer})$ and $\{E_{i}\}$. Combining this
  estimate with (\ref{eq:NiblSchaudest}) yields the conclusion that there is a continuous and monotonically increasing function
  $\mfC:\ro_{+}\rightarrow\ro_{+}$ such that 
  \begin{equation*}
    \begin{split}
       \big|\d_{t}\me_{2}-\textstyle{\int}_{\bM}(\d_{t}^{2}V_{b}-\bh^{ij}_{1}[u_{b,1},N_{1,b,1}]E_{j}E_{i}V_{b})\cdot\d_{t}V_{b}d\mu_{\bh_{\refer}}\big|
      \leq & \mfC[\mfl_{a,b}](\|\d_{t}N_{1,b,1}\|_{C^{1}}+1)\me_{2}.
    \end{split}
  \end{equation*}
  Moreover, $\mfC$ only depends on bounds on $\bh_{2}^{ij}$, $\bh_{1}^{ij}$, $\zeta$, $f_{1}$, $f_{3}$ and $f_{4}$ and their derivatives up
  to order two, one, one, one, one and one respectively; $(\bM,\bh_{\refer})$; the covering; $\{E_{i}\}$; $\lambda$; and
  $\lambda_{\zeta}$. On the other hand,
  \begin{equation}\label{eq:diffVbrem}
    \begin{split}
      & \d_{t}^{2}V_{b}-\bh^{ij}_{1}[u_{b,1},N_{1,b,1}]E_{j}E_{i}V_{b}\\
      = & \g^{ij}[u_{b,2},v_{a,2},N_{1,b,2}]E_{i}E_{j}u_{b,2}-\g^{ij}[u_{b,1},v_{a,1},N_{1,b,1}]E_{i}E_{j}u_{b,1}\\
      & +f_{2}[u_{a,2},v_{a,2},N_{1,a,2},N_{2,a,2}]-f_{2}[u_{a,1},v_{a,1},N_{1,a,1},N_{2,a,1}]\\
      & +\big(\bh^{ij}_{1}[u_{b,2},N_{1,b,2}]-\bh^{ij}_{1}[u_{b,1},N_{1,b,1}]\big)E_{i}E_{j}v_{b,2}.
    \end{split}
  \end{equation}
  The second line on the right hand side can be estimated in $L^{2}$ by $C\mN_{a}$, where $C$ is a constant depending only on bounds on the first
  derivative of $f_{2}$, $(\bM,\bh_{\refer})$ and $\{E_{i}\}$. In order to estimate the third line on the right hand side, note that
  \[
  \|\bh^{ij}_{1}[u_{b,2},N_{1,b,2}]-\bh^{ij}_{1}[u_{b,1},N_{1,b,1}]\|_{2}\leq C\left(\|U_{b}\|_{L^{2}(\bM_{t})}+\|L_{1,b}\|_{H^{1}(\bM_{t})}\right),
  \]
  where $C$ is a constant depending only on bounds on the first derivative of $\bh_{1}^{ij}$. Combining this estimate with
  (\ref{eq:Lobest}) yields a continuous and monotonically increasing function $\mfC:\ro_{+}\rightarrow\ro_{+}$ such that the third line on
  the right hand side of (\ref{eq:diffVbrem}) can be estimated by
  \[
  \mfC[\mfl_{a,b}(t)]\|v_{b,2}\|_{C^{2}}\mN_{a,b}(t).
  \]
  Here $\mfC$ only depends on bounds on $\bh_{2}^{ij}$, $\bh_{1}^{ij}$, $\zeta$, $f_{3}$ and $f_{4}$ and their derivatives up to order two, one, one,
  one and one respectively; $(\bM,\bh_{\refer})$; the covering; $\{E_{i}\}$; $\lambda$; and $\lambda_{\zeta}$. Finally, the first
  term on the right hand side of (\ref{eq:diffVbrem}) can be written
  \begin{equation}\label{eq:Vbenprobterm}
    \begin{split}
      & \g^{ij}[u_{b,2},v_{a,2},N_{1,b,2}]E_{i}E_{j}u_{b,2}-\g^{ij}[u_{b,1},v_{a,1},N_{1,b,1}]E_{i}E_{j}u_{b,1}\\
      = & \left(\g^{ij}[u_{b,2},v_{a,2},N_{1,b,2}]-\g^{ij}[u_{b,1},v_{a,1},N_{1,b,1}]\right)E_{i}E_{j}u_{b,2}
      +\g^{ij}[u_{b,1},v_{a,1},N_{1,b,1}]E_{i}E_{j}U_{b}.
    \end{split}
  \end{equation}
  The first term on the right hand side can be estimated by $\mfC[\mfl_{a,b}]\mN_{a,b}$ in $L^{2}$ by an argument similar to the one required to estimate
  the third line on the right hand side of (\ref{eq:diffVbrem}). Here $\mfC:\ro_{+}\rightarrow\ro_{+}$ is a continuous and monotonically increasing
  function depending only on bounds on $\bh_{2}^{ij}$, $\g^{ij}$, $\zeta$, $f_{3}$ and $f_{4}$ and their derivatives up to order two, one, one, one and
  one respectively; $(\bM,\bh_{\refer})$; the covering; $\{E_{i}\}$; $\lambda$; and $\lambda_{\zeta}$. However, the second term on the right hand
  side of (\ref{eq:Vbenprobterm}) cannot be estimated directly in $L^{2}$. Instead we, in analogy with the proof of Lemma~\ref{lemma:energyestimatev},
  have to consider
  \begin{equation}\label{eq:prebadpartVben}
    \begin{split}
      & \textstyle{\int}_{\bM}\g^{ij}[u_{b,1},v_{a,1},N_{1,b,1}]E_{i}E_{j}U_{b}\cdot\d_{t}V_{b}d\mu_{\bh_{\refer}}\\
      = & -\textstyle{\int}_{\bM}E_{i}\left(\g^{ij}[u_{b,1},v_{a,1},N_{1,b,1}]\right)E_{j}U_{b}\cdot\d_{t}V_{b}d\mu_{\bh_{\refer}}\\
      & -\textstyle{\int}_{\bM}\g^{ij}[u_{b,1},v_{a,1},N_{1,b,1}]E_{j}U_{b}\cdot\d_{t}E_{i}V_{b}d\mu_{\bh_{\refer}}\\
      & -\textstyle{\int}_{\bM}\g^{ij}[u_{b,1},v_{a,1},N_{1,b,1}]E_{j}U_{b}\cdot\d_{t}V_{b}\cdot\rodiv_{\bh_{\refer}}E_{i}d\mu_{\bh_{\refer}}
    \end{split}
  \end{equation}
  in greater detail. Due to (\ref{eq:NiblSchaudest}), the first and the last term can be estimated by $\mfC[\mfl_{a,b}]\mN_{a,b}\me_{2}^{1/2}$, where
  $\mfC:\ro_{+}\rightarrow\ro_{+}$ is a continuous and monotonically increasing function depending only on bounds on $\bh_{2}^{ij}$, $\g^{ij}$, $\zeta$,
  $f_{3}$ and $f_{4}$ and their derivatives up to order two, one, one, one and one respectively; $(\bM,\bh_{\refer})$; the covering; 
  $\{E_{i}\}$; $\lambda$; and $\lambda_{\zeta}$. In order to handle the second term on the right hand side of (\ref{eq:prebadpartVben}), we need to
  integrate in time as well. Consider, to this end
  \begin{equation}\label{eq:Vbentimeint}
    \begin{split}
      & -\textstyle{\int}_{t_{0}}^{t_{1}}\int_{\bM}\g^{ij}[u_{b,1},v_{a,1},N_{1,b,1}]E_{j}U_{b}\cdot\d_{t}E_{i}V_{b}d\mu_{\bh_{\refer}}dt\\
      = & -\left[\textstyle{\int}_{\bM}\g^{ij}[u_{b,1},v_{a,1},N_{1,b,1}]E_{j}U_{b}\cdot E_{i}V_{b}d\mu_{\bh_{\refer}}\right]_{t_{0}}^{t_{1}}\\
      & +\textstyle{\int}_{t_{0}}^{t_{1}}\int_{\bM}\d_{t}(\g^{ij}[u_{b,1},v_{a,1},N_{1,b,1}])E_{j}U_{b}\cdot E_{i}V_{b}d\mu_{\bh_{\refer}}dt\\
      & +\textstyle{\int}_{t_{0}}^{t_{1}}\int_{\bM}\g^{ij}[u_{b,1},v_{a,1},N_{1,b,1}]E_{j}\d_{t}U_{b}\cdot E_{i}V_{b}d\mu_{\bh_{\refer}}dt.
    \end{split}
  \end{equation}
  The last term on the right hand side can be estimated by
  \[
  \big|\textstyle{\int}_{t_{0}}^{t_{1}}\int_{\bM}\mfC[\mfl_{a}(t)]\mN_{a}(t)\me_{2}^{1/2}(t)dt\big|,
  \]
  where $\mfC:\ro_{+}\rightarrow\ro_{+}$ is a continuous and monotonically increasing function depending only on $\lambda$, bounds on $\g^{ij}$
  and bounds on $f_{1}$ and its partial derivatives up to order $2$; note that, to estimate $E_{j}\d_{t}U_{b}$ in $L^{2}$, it is sufficient to proceed
  as in the proof of (\ref{eq:meoneenest}). The second term on the right hand side of (\ref{eq:Vbentimeint}) can be estimated by
  \[
  \big|\textstyle{\int}_{t_{0}}^{t_{1}}\int_{\bM}\mfC[\mfl_{a}(t)](\|\d_{t}N_{1,b,1}\|_{C^{1}}+1)\mN_{a,b}(t)\me_{2}^{1/2}(t)dt\big|,
  \]
  where $\mfC:\ro_{+}\rightarrow\ro_{+}$ is a continuous and monotonically increasing function depending only on $\lambda$, bounds on $f_{1}$
  and bounds on $\g^{ij}$ and its first derivative. Here we used the fact that $u_{b,1}$ satisfies (\ref{eq:ublinmodsys}). Including the first
  term on the right hand side of (\ref{eq:Vbentimeint}) in the energy yields
  \[
  |\hme_{2}(t_{1})-\hme_{2}(t_{0})|\leq \big|\textstyle{\int}_{t_{0}}^{t_{1}}\mfC[\mfl_{a,b}](\|v_{b,2}\|_{C^{2}}+\|\d_{t}N_{1,b,1}\|_{C^{1}}+1)
  (\me_{2}+\mN_{a,b}\me_{2}^{1/2})dt\big|,
  \]
  where $\mfC:\ro_{+}\rightarrow\ro_{+}$ is a continuous and monotonically increasing function depending only on bounds on $\bh_{2}^{ij}$, $\bh_{1}^{ij}$,
  $\g^{ij}$, $\zeta$, $f_{1}$, $f_{2}$, $f_{3}$ and $f_{4}$ and their derivatives up to order two, one, one, one, two, one, one and one respectively;
  $(\bM,\bh_{\refer})$; the covering; $\{E_{i}\}$; $\lambda$; and $\lambda_{\zeta}$. Combining this observation with (\ref{eq:meoneenest}) yields
  (\ref{eq:chmeenest}). 
\end{proof}

\subsection{Uniqueness}

Due to the estimates derived in the previous subsection, we are in a position to prove uniqueness of solutions.

\begin{lemma}\label{lemma:uniquemodelsol}
  Fix a model system and an interval $\mJ\subseteq \mI$ with non-empty interior. Assume that there are two solutions $u_{i}$, $v_{i}$ and $N_{l,i}$,
  $i=1,2$, $l=1,2$, to (\ref{seq:themodel}) on $M_{\mJ}:=\bM\times\mJ$ with the following regularity:
  \begin{subequations}\label{seq:uabietcregunique}
    \begin{align}
      u_{i}\in & C^{1}(\mJ,C^{1}(\bM,\rn{n_{1}}))\cap C^{0}(\mJ,C^{2}(\bM,\rn{n_{1}})),\\
      v_{i}\in & C^{2}(\mJ,C^{0}(\bM,\rn{n_{2}}))\cap C^{1}(\mJ,C^{1}(\bM,\rn{n_{2}}))\cap C^{0}(\mJ,C^{2}(\bM,\rn{n_{2}})),\\
      N_{1,i}\in & C^{1}(\mJ,C^{1}(\bM))\cap C^{0}(\mJ,C^{2,1}(\bM)),\\
      N_{2,i}\in & C^{0}(\mJ,C^{2,1}(\bM))
    \end{align}
  \end{subequations}
  for $i=1,2$. Assume, moreover, that $t_{0}\in \mJ$ and that $u_{1}(\cdot,t_{0})=u_{2}(\cdot,t_{0})$,
  $v_{1}(\cdot,t_{0})=v_{2}(\cdot,t_{0})$ and $(\d_{t}v_{1})(\cdot,t_{0})=(\d_{t}v_{2})(\cdot,t_{0})$. Then the two solutions coincide on $M_{\mJ}$. 
\end{lemma}
\begin{proof}
  Using the notation introduced at the beginning of Subsection~\ref{ssection:conlownorm}, $u_{a,i}=u_{b,i}=u_{i}$, and similarly for $v$, $N_{1}$
  and $N_{2}$. We therefore use the notation $U_{a}=U_{b}=u_{2}-u_{1}=:U$, and similarly for $V$, $L_{1}$ and $L_{2}$; cf. (\ref{eq:UVLdef}). Moreover,
  we let $\mN(t)$ denote the right hand side of (\ref{eq:mNadef}) with $U_{a}$ replaced by $U$ etc. 
  Then $\mN_{a,b}\leq 2\mN$, using the notation introduced in (\ref{eq:mNabdef}). We also let $\mfl$ denote the right hand side of
  (\ref{eq:mfladef}). Again $\mfl_{a,b}\leq 2\mfl$, where $\mfl_{a,b}$ is introduced in (\ref{eq:mflabdef}). Next, note that
  combining the definition of $\me_{1}$ and $\chme_{2}$ with (\ref{eq:chmetwodom}) and the definition of $\mN$ yields
  \begin{equation}\label{eq:mNest}
    [\mN(t)]^{2}\leq C_{1}\big[\chme_{2}(t)+\|L_{1}(\cdot,t)\|_{H^{2}(\bM)}^{2}+\|L_{2}(\cdot,t)\|_{H^{2}(\bM)}^{2}\big]\leq \mfC[\mfl(t)]\chme_{2}(t)
  \end{equation}
  where $C_{1}$ only depends on the $\lambda$ and $n$; and $\mfC$ has the same properties and dependence as
  the function $\mfC_{L}$ appearing in the statement of Lemma~\ref{lemma:lapseestimates}. Combining (\ref{eq:mNest}) with (\ref{eq:chmetwodom}) and
  (\ref{eq:chmeenest}) yields
  \begin{equation}\label{eq:chmetwounique}
    |\chme_{2}(t_{1})-\chme_{2}(t_{0})|\leq \big|\textstyle{\int}_{t_{0}}^{t_{1}}C(t)\chme_{2}(t)dt\big|,
  \end{equation}
  where $C$ depends on the model system and the solutions. Combining (\ref{eq:chmetwounique}) with the fact that $\chme_{2}(t_{0})=0$ yields the
  conclusion that $\chme_{2}(t)=0$ for all $t\in\mJ$, so that $\mN(t)=0$ for all $t\in\mJ$ due to (\ref{eq:mNest}). This means that $u_{1}=u_{2}$,
  $v_{1}=v_{2}$ and $N_{l,1}=N_{l,2}$ for $l=1,2$. The lemma follows. 
\end{proof}

\subsection{Local existence}
Next, we wish to prove local existence. Fix, to this end, a model system and introduce the following inner product. Let $\mJ\subset \mI$
be an interval, $M_\mJ:=\bM\times\mJ$, $u\in C^{0}(M_\mJ,\rn{n_{1}})$, $v\in C^{0}(M_\mJ,\rn{n_{2}})$ and $N_{1}\in C^{0}(M_\mJ)$.
Assume, moreover, that $N_{1}$ is differentiable with respect to the $\bM$-variables and that $E_{i}N_{1}\in C^{0}(M_\mJ)$, $i=1,\dots,n$. Next, let
$\xi_{i}\in H^{m+1}(\bM,\rn{n_{1}})$, $\eta_{i}\in H^{m+1}(\bM,\rn{n_{2}})$ and $\chi_{i}\in H^{m}(\bM,\rn{n_{2}})$, $i=1,2$, and
\begin{equation}\label{eq:Htinnerprod}
  \begin{split}
    & \ldr{(\xi_{1},\eta_{1},\chi_{1}),(\xi_{2},\eta_{2},\chi_{2})}_{H,t}\\
    = & \tfrac{1}{2}\textstyle{\sum}_{|\bfI|\leq m}\textstyle{\int}_{\bM_{t}}
    (E_{\bfI}\chi_{1}\cdot E_{\bfI}\chi_{2}+\bh^{ij}_{1}[u,N_{1}](\cdot,t)E_{i}E_{\bfI}\eta_{1}\cdot E_{j}E_{\bfI}\eta_{2}
    +E_{\bfI}\eta_{1}\cdot E_{\bfI}\eta_{2})d\mu_{\bh_{\refer}}\\
    & +\tfrac{1}{2}\b_{\g}\textstyle{\sum}_{|\bfI|\leq m+1}\textstyle{\int}_{\bM_{t}}E_{\bfI}\xi_{1}\cdot E_{\bfI}\xi_{2}d\mu_{\bh_{\refer}}\\
    & +\tfrac{1}{2}\textstyle{\sum}_{|\bfI|\leq m}\textstyle{\int}_{\bM_{t}}\g^{ij}[u,v,N_{1}](\cdot,t)E_{j}E_{\bfI}\xi_{1}\cdot E_{i}E_{\bfI}\eta_{2}d\mu_{\bh_{\refer}}\\
    & +\tfrac{1}{2}\textstyle{\sum}_{|\bfI|\leq m}\textstyle{\int}_{\bM_{t}}\g^{ij}[u,v,N_{1}](\cdot,t)E_{j}E_{\bfI}\xi_{2}\cdot E_{i}E_{\bfI}\eta_{1}d\mu_{\bh_{\refer}}
  \end{split}
\end{equation}
for $t\in \mJ$, where $\b_{\g}$ is defined by (\ref{eq:bgdef}); when we speak of this inner product in what follows, the choice of the functions $u$, $v$
and $N_{1}$ should be clear from the context. For each $t\in \mJ$, (\ref{eq:Htinnerprod}) defines an inner product on
$H:=H^{m+1}(\bM,\rn{n_{1}})\times H^{m+1}(\bM,\rn{n_{2}})\times H^{m}(\bM,\rn{n_{2}})$. To justify this statement, note that $\ldr{\cdot,\cdot}_{H,t}$ is
symmetric and linear (over the real numbers) in both arguments. Moreover, due to
(\ref{eq:tEdominance}), it is clear that there is a constant depending only on the $\lambda$ appearing in (\ref{eq:lambdabd}) and bounds on the
$\g^{ij}$, say $C_{\lambda,\g}>1$, such that if $\xi=(u_{1},v_{1,0},v_{1,1})$, then 
\begin{equation}\label{eq:Htequiv}
  C_{\lambda,\g}^{-1}\ldr{\xi,\xi}_{H,t}\leq \|u_{1}\|_{H^{m+1}(\bM)}^{2}+\|v_{1,0}\|_{H^{m+1}(\bM)}^{2}+\|v_{1,1}\|_{H^{m}(\bM)}^{2}
  \leq C_{\lambda,\g}\ldr{\xi,\xi}_{H,t}
\end{equation}
for all $t\in\mJ$. This means that $\ldr{\cdot,\cdot}_{H,t}$ has the required positive definiteness properties. In fact, $H$ is a Hilbert space
with this inner product, and this inner product is equivalent to the standard one.

\begin{prop}\label{prop:localexistence}
  Fix a model system, $m>n/2+2$, $\Ui_{0}\in H^{m+1}(\bM,\rn{n_{1}})$, $\Vi_{0,0}\in H^{m+1}(\bM,\rn{n_{2}})$ and
  $\Vi_{0,1}\in H^{m}(\bM,\rn{n_{2}})$. Given a compact interval $\mJ:=[T_{1},T_{2}]\subset \mI$, there is a covering of $\bM$ (depending only on
  $(\bM,\bh_{\refer})$) and a $T>0$ (depending only on the model system, $\mJ$, $m$, $(\bM,\bh_{\refer})$, $\{E_{i}\}$, the covering and an upper bound on
  the sum of the $H^{m+1}$-norms of $\Ui_{0}$ and $\Vi_{0,0}$ and the $H^{m}$-norm of $\Vi_{0,1}$) such that if $t_{0}\in \mJ$ and
  $\mI_{T}:=[t_{0}-T,t_{0}+T]$, then there is a unique solution $u,v,N_{1},N_{2}$ to (\ref{seq:themodel}) on $M_{T}:=\bM\times\mI_{T}$ such that
  \begin{subequations}\label{seq:uabietcreglocalexistence}
    \begin{align}
      u\in & C^{1}(\mI_{T},C^{1}(\bM,\rn{n_{1}}))\cap C^{0}(\mI_{T},C^{2}(\bM,\rn{n_{1}})),\\
      v\in & C^{2}(\mI_{T},C^{0}(\bM,\rn{n_{2}}))\cap C^{1}(\mI_{T},C^{1}(\bM,\rn{n_{2}}))\cap C^{0}(\mI_{T},C^{2}(\bM,\rn{n_{2}})),\\
      N_{1}\in & C^{1}(\mI_{T},C^{1}(\bM))\cap C^{0}(\mI_{T},C^{2,1}(\bM)),\\
      N_{2}\in & C^{0}(\mI_{T},C^{2,1}(\bM))
    \end{align}
  \end{subequations}
  and (\ref{eq:uvdtvindata}) hold. Furthermore, if $\kappa_{0}$ is the smallest integer strictly larger than $n/2$,
  \begin{subequations}\label{seq:regularityoneuvNi}
    \begin{align}
      u\in & \textstyle{\bigcap}_{j=1}^{m+1-\kappa_{0}}C^{j}\left(\mI_{T},H^{m+2-j}(\bM,\rn{n_{1}})\right),\\
      v\in & \textstyle{\bigcap}_{j=0}^{m-\kappa_{0}}C^{j}\left(\mI_{T},H^{m+1-j}(\bM,\rn{n_{2}})\right),\\
      N_{1}\in & \textstyle{\bigcap}_{j=1}^{m-\kappa_{0}}C^{j}\left(\mI_{T},H^{m+3-j}(\bM)\right),\\
      N_{2}\in & \textstyle{\bigcap}_{j=0}^{m-1-\kappa_{0}}C^{j}\left(\mI_{T},H^{m+2-j}(\bM)\right).
    \end{align}
  \end{subequations}  
  Given these functions, define
  \[
  \chmfE_{2,m+1}[v,u](t) = \ldr{(u(\cdot,t),v(\cdot,t),(\d_{t}v)(\cdot,t)),(u(\cdot,t),v(\cdot,t),(\d_{t}v)(\cdot,t))}_{H,t},
  \]
  where the inner product is defined by (\ref{eq:Htinnerprod}), and the functions $u$, $v$ and $N_{1}$ appearing in (\ref{eq:Htinnerprod}) are the
  ones determined by the solution. Then there is a continuous and monotonically increasing function $\mfC:\ro_{+}\rightarrow\ro_{+}$ such that 
  \begin{equation}\label{eq:chmfEmainbound}
    \chmfE_{2,m+1}[v,u](t_{1})+1\leq \left[\chmfE_{2,m+1}[v,u](t_{0})+1\right]\exp\big(\big|\textstyle{\int}_{t_{0}}^{t_{1}}\mfC[\ell(t)]dt\big|\big)
  \end{equation}
  for $t_{1}\in\mI_{T}$. Here $\mfC$ only depends on $m$, $(\bM,\bh_{\refer})$, $\{E_{i}\}$, the covering and the model system. Moreover,
  \begin{equation}\label{eq:elldef}
    \ell(t):=\|u(\cdot,t)\|_{C^{2}}+\|v(\cdot,t)\|_{C^{2}}+\|\d_{t}v(\cdot,t)\|_{C^{1}}.
  \end{equation}
  Finally, the bound on $\chmfE_{2,m+1}[v,u](t)$ for $t\in\mI_{T}$ can be assumed to only depend on the model system, $m$, $(\bM,\bh_{\refer})$,
  $\{E_{i}\}$, the covering and an upper bound on the sum of the $H^{m+1}$-norms of $\Ui_{0}$ and $\Vi_{0,0}$ and the $H^{m}$-norm of $\Vi_{0,1}$.
\end{prop}
\begin{remark}
  The only reason for the dependence of $T$ on $\mJ$ is to ensure that $\mI_T\subseteq\mI$.
\end{remark}
\begin{remark}
  It should be sufficient to assume $m>n/2+1$. However, we expect the corresponding proof to lead to technical complications that we do not wish to
  address here.
\end{remark}
\begin{proof}
  The uniqueness statement is an immediate consequence of Lemma~\ref{lemma:uniquemodelsol}. In the proof we therefore focus on the existence statement,
  the additional regularity properties of solutions and the energy estimate (\ref{eq:chmfEmainbound}). 
  
  \textit{Setting up the iteration.} Let $\Ui_{0,l}\in C^{\infty}(\bM,\rn{n_{1}})$ and $\Vi_{0,0,l},\Vi_{0,1,l}\in C^{\infty}(\bM,\rn{n_{2}})$, $l\in\nn{}_0$, be
  such that if
  \[
  \|\Ui_{0}\|_{H^{m+1}}+\|\Vi_{0,0}\|_{H^{m+1}}+\|\Vi_{0,1}\|_{H^{m}}\leq C_{0},
  \]
  then
  \begin{equation}\label{eq:UilVilbd}
    \|\Ui_{0,l}\|_{H^{m+1}}+\|\Vi_{0,0,l}\|_{H^{m+1}}+\|\Vi_{0,1,l}\|_{H^{m}}\leq C_{0}+1
  \end{equation}
  and
  \[
    \lim_{l\rightarrow\infty}[\|\Ui_{0,l}-\Ui_{0}\|_{H^{m+1}}+\|\Vi_{0,0,l}-\Vi_{0,0}\|_{H^{m+1}}+\|\Vi_{0,1,l}-\Vi_{0,1}\|_{H^{m}}]=0
  \]
  (in what follows, we impose additional restrictions on the sequence). 
  Let now $u_{0}(\bx,t):=\Ui_{0,0}(\bx)$, $v_{0}(\bx,t):=\Vi_{0,0,0}(\bx)$, $N_{1,0}:=1$ and $N_{2,0}:=0$. Note that $u_{0}$, $v_{0}$, $N_{1,0}$ and $N_{2,0}$ are
  smooth functions on $M:=\bM\times \mI$. Given that $u_{l}$, $v_{l}$, $N_{1,l}$ and $N_{2,l}$ have been defined, define $u_{l+1}$, $v_{l+1}$,
  $N_{1,l+1}$ and $N_{2,l+1}$ by solving (\ref{seq:linmodsys}) with $u_{a}=u_{l}$, $v_{a}=v_{l}$, $N_{1,a}=N_{1,l}$, $N_{2,a}=N_{2,l}$, initial data for $u_{l+1}$
  given by $\Ui_{0,l+1}$, initial data for $v_{l+1}$ given by $\Vi_{0,0,l+1}$ and initial data for $\d_{t}v_{l+1}$ given by $\Vi_{0,1,l+1}$. That
  the solution exists and is smooth follows from Lemma~\ref{lemma:ubvbNbsol}. 

  \textit{Bounds in the high norm.} The first goal is to prove that the sequence $(u_{l},v_{l},N_{1,l},N_{2,l})$ is bounded in the appropriate norm
  on an interval of the form $\mI_{T}:=[-T+t_{0},t_{0}+T]$ for any $t_{0}\in \mJ$ and some suitable choice of $T>0$. The relevant expression we wish
  to bound is
  \begin{equation}\label{eq:mfNldef}
    \begin{split}
      \mfN_{l}(t) := & \|u_{l}\|_{H^{m+1}(\bM_{t})}+\|v_{l}\|_{H^{m+1}(\bM_{t})}+\|\d_{t}v_{l}\|_{H^{m}(\bM_{t})}\\
      & +\|N_{1,l}\|_{H^{m+2}(\bM_{t})}+\|N_{2,l}\|_{H^{m+2}(\bM_{t})}+1.
    \end{split}
  \end{equation}
  By Sobolev embedding and the fact that $m>n/2+1$, a bound on $\mfN_{l}(t)$ yields a bound on
  \begin{equation}\label{eq:nuldef}
    \begin{split}
      \nu_{l}(t) := & \|u_{l}(\cdot,t)\|_{C^{2}}+\|v_{l}(\cdot,t)\|_{C^{2}}+\|\d_{t}v_{l}(\cdot,t)\|_{C^{1}}\\
      & +\|N_{1,l}(\cdot,t)\|_{C^{2,1}}+\|N_{2,l}(\cdot,t)\|_{C^{2,1}}+1.
    \end{split}
  \end{equation}
  In fact, we obtain bounds on $N_{i,l}$ in $C^{3}(\bh_\refer)$, but we choose the $C^{2,1}$-norm since it fits better with Schauder estimates.
  Assume, inductively, that, for some $0\leq l\in\nn{}_0$, there is a constant $\mM$ such that
  \begin{equation}\label{eq:indassmfNj}
    \mfN_{j}(t)\leq\mM
  \end{equation}
  for $t\in \mI_{T}$ and $0\leq j\leq l$. Note that for any $t_{0}\in \mJ$ and any $T$ such that $[T_{1}-T,T_{2}+T]\subset\mI$, this statement holds for
  $l=0$, assuming $\mM$ to be large enough. Moreover, for $l=0$, the bound only depends on $(\bM,\bh_{\refer})$ and $C_{0}$ due to (\ref{eq:UilVilbd})
  and the definitions of $u_{0}$, $v_{0}$, $N_{1,0}$ and $N_{2,0}$. Since $m>n/2+1$, combining (\ref{eq:indassmfNj}) with Sobolev embedding yields
  $\nu_{j}(t)\leq C\mM$ for $t\in\mI_{T}$. Let $0\leq j\leq l$. Then, due to (\ref{eq:dtEompoest}),
  \[
    E_{1,m+1}[u_{j+1}](t)+1\leq E_{1,m+1}[u_{j+1}](t_{0})+1+\big|\textstyle{\int}_{t_{0}}^{t}K_{1}(E_{1,m+1}[u_{j+1}](s)+1)ds\big|,
  \]
  where $K_{1}$ is a constant depending only on $\mM$, $(\bM,\bh_{\refer})$, $\{E_{i}\}$, $m$ and bounds on $f_{1}$ and its partial derivatives up to
  order $m+1$. Appealing to Gr\"{o}nwall's lemma yields 
  \[
    E_{1,m+1}[u_{j+1}](t)+1\leq (E_{1,m+1}[u_{j+1}](t_{0})+1)e^{K_{1}|t-t_{0}|}
  \]
  for all $t\in\mI_{T}$ and $j\in\{0,\dots,l\}$. Combining this estimate with (\ref{eq:UilVilbd}) yields
  \begin{equation}\label{eq:Eomponeulpobd}
    E_{1,m+1}[u_{j+1}](t)+1\leq (C_{0}^{2}+2)e^{K_{1}T}
  \end{equation}
  for all $t\in\mI_{T}$ and $j\in\{0,\dots,l\}$. Assuming $T$ to be small enough that $K_{1}T\leq \ln 2$, we obtain a uniform bound on $u_{j+1}$ in
  $H^{m+1}(\bM_{t})$ for $t\in\mI_{T}$ and $j\in\{0,\dots,l\}$ which only depends on $C_{0}$. By Sobolev embedding, we thereby obtain a uniform bound
  on $u_{j+1}$ in $C^{2}(\bM_{t})$ for $t\in\mI_{T}$ and $j\in\{0,\dots,l\}$ which only depends on $C_{0}$,
  $(\bM,\bh_{\refer})$, the frame $\{E_{i}\}$ and $m$. At this stage, we can appeal to Lemma~\ref{lemma:energyestimatev}. Note that the $\mfN_{a,b}$
  appearing in this lemma is, in our setting, bounded by a constant depending only on $C_{0}$ and $\mM$, assuming $K_{1}T\leq \ln 2$. By Sobolev
  embedding, this implies that $\mfn_{a,b}$ is bounded by a constant depending only on $(\bM,\bge_{\refer})$, the frame $\{E_{i}\}$, $m$, $C_{0}$ and
  $\mM$. Due to Sobolev embedding and (\ref{eq:tEdominance}), the estimate (\ref{eq:tEtwompodiffest}) implies that
  \[
  \left|\chE_{2,m+1}[v_{j+1},u_{j+1}](t_{1})-\chE_{2,m+1}[v_{j+1},u_{j+1}](t_{0})\right|
  \leq  \big|\textstyle{\int}_{t_{0}}^{t_{1}}K_{2}(\chE_{2,m+1}[v_{j+1},u_{j+1}]+1)dt\big|,
  \]
  where $K_{2}$ only depends on $(\bM,\bh_{\refer})$, the frame $\{E_{i}\}$, the covering mentioned in Lemma~\ref{lemma:energyestimatev}, the model system,
  $m$, $C_{0}$ and $\mM$. In analogy with the derivation of (\ref{eq:Eomponeulpobd}), we obtain
  \begin{equation}\label{eq:chEtmponeulpobd}
    \chE_{2,m+1}[v_{j+1},u_{j+1}](t)\leq C_{\g,\lambda}(C_{0}^{2}+1)
  \end{equation}
  for all $t\in\mI_{T}$, assuming $T$ to be such that $\max\{K_{1},K_{2}\}T\leq \ln 2$. Here $C_{\g,\lambda}$ is a constant depending only on $n$, $m$,
  $\lambda$ and the global bound on the $\|\g^{ij}\|$. Combining the above conclusions with (\ref{eq:lambdabd})
  yields
  \begin{equation}\label{eq:ulpovlpointerm}
    \|u_{k}\|_{H^{m+1}(\bM_{t})}+\|v_{k}\|_{H^{m+1}(\bM_{t})}+\|\d_{t}v_{k}\|_{H^{m}(\bM_{t})}\leq C_{\g,\lambda}\ldr{C_{0}}
  \end{equation}
  for all $t\in\mI_{T}$ and $k\in\{1,\dots,l+1\}$, assuming $T$ to be such that $\max\{K_{1},K_{2}\}T\leq \ln 2$, where $C_{\g,\lambda}$ has the same
  dependence as in (\ref{eq:chEtmponeulpobd}). Due to the choice of $u_{0}$ and $v_{0}$, the estimate (\ref{eq:ulpovlpointerm}) also holds for
  $k=0$. Combining (\ref{eq:ulpovlpointerm}) with Sobolev embedding yields
  \begin{equation}\label{eq:ulpovlpointermsob}
  \|u_{k}(\cdot,t)\|_{C^{2}}+\|v_{k}(\cdot,t)\|_{C^{2}}+\|\d_{t}v_{k}(\cdot,t)\|_{C^{1}}\leq C_{2}\ldr{C_{0}}
  \end{equation}
  for all $t\in\mI_{T}$ and $k\in\{0,\dots,l+1\}$, assuming $T$ to be such that $\max\{K_{1},K_{2}\}T\leq \ln 2$, where $C_{2}$ only depends on
  $(\bM,\bh_{\refer})$, $\{E_{i}\}$, $m$, $\lambda$ and the global bound on the $\|\g^{ij}\|$. What remains to bound (in $H^{m+2}$) is $N_{1,l+1}$ and
  $N_{2,l+1}$. This can be achieved by appealing to Lemma~\ref{lemma:lapseestimatehighnorm} with $u_{b}=u_{l+1}$ and $v_{a}=v_{l}$. Combining
  Lemma~\ref{lemma:lapseestimatehighnorm} with (\ref{eq:ulpovlpointerm}) and (\ref{eq:ulpovlpointermsob}) yields
  \[
  \|N_{1,l+1}\|_{H^{m+2}(\bM_{t})}+\|N_{2,l+1}\|_{H^{m+2}(\bM_{t})}\leq C_{3}\ldr{C_{0}}
  \]
  for all $t\in\mI_{T}$, assuming $T$ to be such that $\max\{K_{1},K_{2}\}T\leq \ln 2$, where $C_{3}$ only depends on $C_{0}$, $(\bM,\bh_{\refer})$, the
  covering, $\{E_{i}\}$, $m$ and the system. To summarise, 
  \[
  \mfN_{l+1}(t)\leq C_{\mfN}
  \]
  for all $t\in\mI_{T}$, assuming $\max\{K_{1},K_{2}\}T\leq \ln 2$, where $C_{\mfN}$ only depends on $(\bM,\bh_{\refer})$, the covering, $\{E_{i}\}$, the model
  system, $m$ and $C_{0}$. Due to the observations made below (\ref{eq:indassmfNj}), we can thus fix a constant $\mM$ with the same dependence as
  $C_{\mfN}$ and deduce that (\ref{eq:indassmfNj}) holds for all $j$ and all $t\in\mI_{T}$, where $T\leq 1/K$ for some constant $K$ with the same type
  of dependence as $C_{\mfN}$.

  \textit{Convergence in the low norm.} Define, for $l\geq 1$, $\mN_{l}$ in analogy with (\ref{eq:mNadef}), where $U_{a}$ is replaced by $u_{l}-u_{l-1}$
  etc. Our next goal is to prove that there is a constant $\mK_{0}$ such that
  \begin{equation}\label{eq:mNlest}
    \mN_{l}(t)\leq \mK_{0}2^{-l}
  \end{equation}
  for all $t\in\mI_{T}$ and $l\geq 1$. Note that for $l\leq 2$, this estimate holds due to (\ref{eq:indassmfNj}), where $\mK_{0}$ only depends on
  $(\bM,\bh_{\refer})$, the covering, $\{E_{i}\}$, the model system, $m$ and $C_{0}$. Assume now, inductively, that (\ref{eq:mNlest}) holds for some
  $l\geq 1$. Appealing to (\ref{eq:meoneenest}) then yields the conclusion that
  \begin{equation}\label{eq:meolpoprelest}
    |\me_{1,l+1}(t_{1})-\me_{1,l+1}(t_{0})|\leq \big|\textstyle{\int}_{t_{0}}^{t_{1}}\mK_{1}\mK_{0}2^{-l}\me_{1,l+1}^{1/2}dt\big|
  \end{equation}
  for $t_{1}\in\mI_{T}$, where $\me_{1,l+1}$ is given by the right hand side of (\ref{eq:meonedef}) with $U_{b}=u_{l+1}-u_{l}$. Moreover, $\mK_{1}$ is a
  constant depending only on $(\bM,\bh_{\refer})$, the covering, $\{E_{i}\}$, the model system, $m$ and $C_{0}$. From (\ref{eq:meolpoprelest}), it follows
  that
  \[
  \left(\me_{1,l+1}(t_{1})+4^{-l-1}\right)\leq \left(\me_{1,l+1}(t_{0})+4^{-l-1}\right)e^{\mK_{1}\mK_{0}T}
  \]
  for $t_{1}\in\mI_{T}$. Assume, to begin with, $T$ to be small enough, depending only on $\mK_{0}$ and $\mK_{1}$, that $\mK_{1}\mK_{0}T\leq\ln 2$. Then
  \begin{equation}\label{eq:meolpoest}
    \me_{1,l+1}(t_{1})\leq 2\me_{1,l+1}(t_{0})+(2\mK_{1}\mK_{0}T)4^{-l-1}
  \end{equation}
  for $t_{1}\in\mI_{T}$. In particular, making appropriate assumptions concerning the initial sequence, the first term on the right hand side can be
  assumed to be smaller than $4^{-l-1}$. Assuming $T$ to be small enough that $2\mK_{1}\mK_{0}T\leq 1$, we conclude that $2\me_{1,l+1}(t)\leq 4^{-l}$ for
  $t_{1}\in\mI_{T}$. This means that the $\mN_{a,b}$ appearing in, e.g., (\ref{eq:chmeenest}) can be estimated by $(\mK_{0}+1)2^{-l}$. Combining this
  observation with (\ref{eq:dtNobest}), (\ref{eq:chmetwodom}) and (\ref{eq:chmeenest}) yields
  \[
  |\chme_{2,l+1}(t_{1})-\chme_{2,l+1}(t_{0})|\leq \big|\textstyle{\int}_{t_{0}}^{t_{1}}\mK_{2}\ldr{\mK_{0}}^{2}\big(\chme_{2,l+1}(t)+4^{-l-1}\big)dt\big|
  \]
  for $t_{1}\in\mI_{T}$, where $\mK_{2}$ is a constant depending only on $(\bM,\bh_{\refer})$, the covering, $\{E_{i}\}$, $m$, the model system and $C_{0}$.
  Here $\chme_{2,l+1}$ is defined by (\ref{eq:metwodef}), (\ref{eq:hmetwodef}) and (\ref{eq:chmetwodef}), where $\me_{1}$ is replaced by $\me_{1,l+1}$,
  $U_{b}=u_{l+1}-u_{l}$, $V_{b}=v_{l+1}-v_{l}$, $u_{b,1}=u_{l}$, $v_{a,1}=v_{l-1}$ and $N_{1,b,1}=N_{1,l}$. In particular, 
  \[
  \left(\chme_{2,l+1}(t_{1})+4^{-l-1}\right)\leq \left(\chme_{2,l+1}(t_{0})+4^{-l-1}\right)e^{\mK_{2}\ldr{\mK_{0}}^{2}T}
  \]
  for $t_{1}\in\mI_{T}$. In analogy with (\ref{eq:meolpoest}), we deduce that
  \[
  \chme_{2,l+1}(t_{1})\leq 2\chme_{2,l+1}(t_{0})+(2\mK_{2}\ldr{\mK_{0}}^{2}T)4^{-l-1}
  \]
  for $t_{1}\in\mI_{T}$, assuming $T$ to be small enough that $\mK_{2}\ldr{\mK_{0}}^{2}T\leq\ln 2$. Combining this estimate with (\ref{eq:metwodef})
  and (\ref{eq:chmetwodom}), it follows that
  \begin{equation}\label{eq:vdifflpoest}
    \|v_{l+1}-v_{l}\|_{H^{1}(\bM_{t})}^{2}+\|\d_{t}v_{l+1}-\d_{t}v_{l}\|_{L^{2}(\bM_{t})}^{2}\leq \mK_{3}\ldr{\mK_{0}}^{2}(\chme_{2,l+1}(t_{0})+4^{-l-1}T)
  \end{equation}
  for $t\in\mI_{T}$, where $\mK_{3}$ has the same dependence as $\mK_{2}$. Note that this estimate also holds with $l$ replaced by $l-1$,
  assuming $l\geq 2$; note that since we know (\ref{eq:mNlest}) to hold for $l=1$ and $l=2$, we know (\ref{eq:vdifflpoest}) to hold for $l=1$ and
  $l=2$. Assume now that $l\geq 2$. Then, combining (\ref{eq:meolpoest}); (\ref{eq:vdifflpoest}); (\ref{eq:vdifflpoest}) with $l$ replaced by $l-1$;
  and Lemma~\ref{lemma:lapseestimates} yields
  \[
  [\mN_{l+1}(t)]^{2}\leq \mK_{4}\ldr{\mK_{0}}^{2}(\chme_{2,l+1}(t_{0})+\chme_{2,l}(t_{0})+4^{-l-1}T)
  \]
  for $t\in\mI_{T}$, where $\mK_{4}$ has the same dependence as $\mK_{2}$; to obtain this estimate we also used (\ref{eq:chmetwodom}). Assuming
  $T$ to be small enough, the bound depending only on $\mK_{0}$ and $\mK_{4}$, and assuming the initial sequence to have been chosen appropriately,
  we conclude that (\ref{eq:mNlest}) holds with $l$ replaced by $l+1$ and $\mK_{0}$ replaced by $1$. If $\mK_{0}'$ is the constant appearing in
  (\ref{eq:mNlest}) for $1\leq l\leq 2$ (note that this constant can be assumed to have the same dependence as $\mK_{2}$), let $\mK_{0}:=\max\{\mK_{0}',1\}$.
  Due to the dependence of $\mM$, noted above, we conclude that we only need to require $T$ to be small enough, the bound depending only on
  $(\bM,\bh_{\refer})$, the covering, $\{E_{i}\}$, the model system, $m$ and $C_{0}$ (since we here tacitly assume $T$ to be small enough that $\mI_{T}\subset\mI$,
  $T$ also, strictly speaking, depends on $\mJ$). This ensures that both the bounds (\ref{eq:indassmfNj}) and (\ref{eq:mNlest}) hold for all $j$ and $l$
  respectively. 
  
  We conclude that $u_{l}$, $v_{l}$, $N_{1,l}$ and $N_{2,l}$ converge to limits $u$, $v$, $N_{1}$ and $N_{2}$ respectively, satisfying
  \begin{align*}
    & u\in C^{0}(\mI_{T},H^{1}(\bM)),\ \ \
    v\in C^{0}(\mI_{T},H^{1}(\bM))\cap C^{1}(\mI_{T},L^{2}(\bM))\\
    & N_{1},N_{2}\in C^{0}(\mI_{T},H^{2}(\bM)).
  \end{align*}
  Below we refer to this convergence as \textit{convergence in the low norm}. Next we wish to combine this convergence with boundedness in the high
  norm, i.e. (\ref{eq:indassmfNj}), and interpolation in order to obtain convergence in intermediate norms. To this end, recall that the
  $H^k(\bM;\bbE)$ and $H^k(\bh_\refer)$ norms are equivalent, see Subsection~\ref{ssection:conventionsframe}, and that (\ref{eq:HmHpmequiv}) holds.
  Here the $\|\cdot\|_{(s)}$ norm is defined by (\ref{eq:specsobsp}) with $h=\bh_\refer$. Combining convergence in the low
  norm with (\ref{eq:indassmfNj}) and standard interpolation estimates, see (\ref{eq:Hsinterpol}), yields the conclusion that
  for every $0<s<m$, $u_{l}$, $v_{l}$, $N_{1,l}$ and $N_{2,l}$ converge to limits $u$, $v$, $N_{1}$ and $N_{2}$ respectively, satisfying
  \begin{align*}
    & u\in C^{0}(\mI_{T},H_{(s+1)}(\bh_{\refer};\rn{n_1})),\ \ \
    v\in C^{0}(\mI_{T},H_{(s+1)}(\bh_{\refer};\rn{n_2}))\cap C^{1}(\mI_{T},H_{(s)}(\bh_{\refer};\rn{n_2}))\\
    & N_{1},N_{2}\in C^{0}(\mI_{T},H_{(s+2)}(\bh_\refer));
  \end{align*}  
  the Sobolev spaces $H_{(s)}$ and associated norms are introduced in Appendix~\ref{ssection:specsob} below. In particular, we have this regularity
  for $s=m-1$. This means that $u_{l}$, $v_{l}$, $N_{1,l}$ and $N_{2,l}$ converge to limits $u$, $v$, $N_{1}$ and $N_{2}$ respectively,
  satisfying
  \begin{align*}
    & u\in C^{0}(\mI_{T},H^{m}(\bM)),\ \ \
    v\in C^{0}(\mI_{T},H^{m}(\bM))\cap C^{1}(\mI_{T},H^{m-1}(\bM))\\
    & N_{1},N_{2}\in C^{0}(\mI_{T},H^{m+1}(\bM)).
  \end{align*}
  Combining this observation with the equations satisfied by the sequence yields  
  \[
  u\in C^{1}(\mI_{T},H^{m}(\bM)),\ \ \ v\in C^{2}(\mI_{T},H^{m-2}(\bM)).
  \]
  Combining the above observations with the fact that $u$, $v$ and $N_{1}$ solve (\ref{eq:themodelN}) (this is a consequence
  of the above convergence and Sobolev embedding) and Corollary~\ref{cor: Ck regularity} (with $k=m-1$) yields
  \[
  N_{1}\in C^{1}(\mI_{T},H^{m+1}(\bM)).
  \]
  Combining the above observations with Sobolev embedding, the solution satisfies (\ref{seq:uabietcreglocalexistence}).
  
  \textit{Bounding the solution in the high norm.} Next, we wish to prove that $u(\cdot,t)\in H^{m+1}(\bM,\rn{n_{1}})$ for all $t\in\mI_{T}$ etc. Define,
  to this end, 
  \[
  \hu(i,t):=\ldr{u(\cdot,t),\varphi_{i}}=\textstyle{\int}_{\bM}u(\cdot,t)\varphi_{i}d\mu_{\bh_{\refer}},
  \]
  where we use the notation introduced in Appendix~\ref{ssection:specsob} in the case that $(\S,h)=(\bM,\bh_{\refer})$. We introduce similar notation
  for $v$ and the $N_{j}$, $j=1,2$. Note that $\hu_{l}(i,t)\rightarrow\hu(i,t)$ for
  all $t\in \mI_{T}$. This means that 
  \begin{equation*}
    \begin{split}
      & \textstyle{\sum}_{i=0}^{N}\ldr{\lambda_{i}}^{2(m+1)}|\hu(i,t)|^{2}\\
      = & \textstyle{\lim}_{l\rightarrow\infty}\textstyle{\sum}_{i=0}^{N}\ldr{\lambda_{i}}^{2(m+1)}|\hu_{l}(i,t)|^{2}
      \leq \limsup_{l\rightarrow\infty}\|u_{l}(\cdot,t)\|_{(m+1)}^{2}\\
      \leq & C\textstyle{\limsup}_{l\rightarrow\infty}\|u_{l}\|_{H^{m+1}(\bM_{t})}^{2}\leq C\mM^{2}
    \end{split}
  \end{equation*}  
  for all $N\in\nn{}_0$, where $C$ only depends on $m$, $\{E_{i}\}$ and $(\bM,\bh_{\refer})$.
  Due to this and similar arguments, we conclude that $u(\cdot,t)\in H^{m+1}(\bM,\rn{n_{1}})$, $v(\cdot,t)\in H^{m+1}(\bM,\rn{n_{2}})$,
  $(\d_{t}v)(\cdot,t)\in H^{m}(\bM,\rn{n_{2}})$ and $N_{i}(\cdot,t)\in H^{m+2}(\bM)$, $i=1,2$. Moreover,
  \[
  \|u(\cdot,t)\|_{H^{m+1}(\bM)}+\|v(\cdot,t)\|_{H^{m+1}(\bM)}+\|(\d_{t}v)(\cdot,t)\|_{H^{m}(\bM)}
  +\textstyle{\sum}_{i=1}^{2}\|N_{i}(\cdot,t)\|_{H^{m+2}(\bM)}\leq C\mM
  \]
  for all $t\in\mI_{T}$, where $C$ only depends on $m$, $\{E_{i}\}$ and $(\bM,\bh_{\refer})$.

  \textit{Weak continuity.} Let $f$ be an element in the dual of $H_{(m+1)}(\bh_\refer;\rn{n_{1}})$. Then $f_{\chi}=f$, for some function
  $\chi\in H_{(-m-1)}(\bh_\refer;\rn{n_{1}})$, where $f_{\chi}:=\ldr{\cdot,\chi}$; cf. Appendix~\ref{ssection:specsob}, in particular (\ref{eq:fchidef md}).
  We wish to prove that $t\mapsto f(u(\cdot,t))$ defines a continuous function of $t$. Let, to this end, $\chi_{j}\in C^{\infty}(\bM,\rn{n_{1}})$
  be such that $\chi_{j}\rightarrow\chi$ in $H_{(-m-1)}(\bh_\refer;\rn{n_{1}})$. Consider
  \begin{equation*}
    \begin{split}
      |f[u(\cdot,t)]-f[u_{l}(\cdot,t)]| = & |\ldr{u(\cdot,t)-u_{l}(\cdot,t),\chi}|\\
      \leq & C\mM\|\chi-\chi_{j}\|_{(-m-1)}+\big|\textstyle{\int}_{\bM}[u(\cdot,t)-u_{l}(\cdot,t)]\chi_{j}d\mu_{\bh_{\refer}}\big|.
    \end{split}
  \end{equation*}
  Given $\e>0$, choose $j$ large enough that the first term on the right hand side is smaller than $\e/2$. Given this $j$, choose $l$ large enough
  that the second term is smaller than $\e/2$ for all $t\in\mI_{T}$. Then the left hand side is smaller than $\e$ for all $t\in\mI_{T}$. In particular
  $f[u_{l}(\cdot,t)]$ converges uniformly to $f[u(\cdot,t)]$. This means that the map $t\mapsto f[u(\cdot,t)]$ is continuous; i.e., $u$ considered as
  a map from $\mI_{T}$ to $H^{m+1}(\bM,\rn{n_{1}})$ is weakly continuous. Similarly $v$ can be considered to be a weakly continuous map from $\mI_{T}$ to
  $H^{m+1}(\bM,\rn{n_{2}})$ and $\d_{t}v$ can be considered to be a weakly continuous map from $\mI_{T}$ to $H^{m}(\bM,\rn{n_{2}})$. Finally, the
  $N_{i}$ can be considered to be weakly continuous maps from $\mI_{T}$ to $H^{m+2}(\bM,\rn{n_{2}})$.

  \textit{Energy inequality.} In order to prove (\ref{eq:chmfEmainbound}), it is of interest to return to Lemma~\ref{lemma:energyestimatev}. Let
  \begin{equation*}
    \begin{split}
      & \chF_{2,m+1}[v_{l+1},u_{l+1}](t)\\
      := & \ldr{[u_{l+1}(\cdot,t),v_{l+1}(\cdot,t),(\d_{t}v_{l+1})(\cdot,t)],[u_{l+1}(\cdot,t),v_{l+1}(\cdot,t),(\d_{t}v_{l+1})(\cdot,t)]}_{H,t}, 
    \end{split}
  \end{equation*}  
  where $\ldr{\cdot,\cdot}_{H,t}$ is introduced in (\ref{eq:Htinnerprod}) and the functions $u$, $v$ and $N_{1}$ appearing in (\ref{eq:Htinnerprod}) are
  part of the solution obtained above. We define $\chE_{2,m+1}[v_{l+1},u_{l+1}]$ in exactly the same way, but with $u$, $v$ and $N_{1}$ appearing in
  (\ref{eq:Htinnerprod}) replaced by $u_{l+1}$, $v_{l}$ and $N_{1,l+1}$ respectively; note that this definition coincides with the one given in
  (\ref{eq:tEtwompo}). Due to the convergence and boundedness derived so far, we know that
  \begin{equation}\label{eq:chFtmpochEtmpolim}
    \lim_{l\rightarrow\infty}(\chF_{2,m+1}[v_{l+1},u_{l+1}]-\chE_{2,m+1}[v_{l+1},u_{l+1}])=0.
  \end{equation}
  In fact, the convergence is uniform on $\mI_{T}$. Next, note that Lemma~\ref{lemma:lapseestimatehighnorm} yields bounds on the $N_{i,l}$ in terms of
  $v_{l-1}$ and $u_{l}$. More specifically, for $l\geq 3$, there is a continuous and monotonically increasing function $\mC_{1}:\ro_{+}\rightarrow \ro_{+}$
  such that 
  \begin{equation}\label{eq:NilHmptwosumest}
    \begin{split}
      \textstyle{\sum}_{i=1}^{2}\|N_{i,l}(\cdot,t)\|_{H^{m+2}(\bM)}^{2}\leq & \mC_{1}[\ell_{1,l}(t)][\chE_{2,m+1}[v_{l},u_{l}](t)+\chE_{2,m+1}[v_{l-1},u_{l-1}](t)+1]
    \end{split}
  \end{equation}
  for $t\in\mI_{T}$. Here
  \[
  \ell_{1,l}(t):=\|u_{l}(\cdot,t)\|_{C^{2}}+\|v_{l-1}(\cdot,t)\|_{C^{2}}+\|\d_{t}v_{l-1}(\cdot,t)\|_{C^{1}}.
  \]
  Moreover, $\mC_{1}$ only depends on $m$, the covering, $(\bM,\bh_{\refer})$, $\{E_{i}\}$ and the model system. Next, we wish to apply
  Lemma~\ref{lemma:energyestimatev}. Note, to this end, that $\mfN_{a,b}$ appearing in the statement of Lemma~\ref{lemma:energyestimatev}
  can in the present setting be estimated according to
  \begin{equation*}
    \begin{split}
      \mfN_{a,b}^{2}(t) \leq & \mC_{2}[\ell_{1,l}(t)](\chE_{\rotot,m+1,l}+1)
    \end{split}
  \end{equation*}
  for $t\in\mI_{T}$, where
  \[
  \chE_{\rotot,m+1,l}(t):=\chE_{2,m+1}[v_{l+1},u_{l+1}](t)+\chE_{2,m+1}[v_{l},u_{l}](t)+\chE_{2,m+1}[v_{l-1},u_{l-1}](t).
  \]
  Moreover, $\mC_{2}$ has the same dependence and properties as $\mC_{1}$. In order to obtain this estimate, we appealed to
  (\ref{eq:NilHmptwosumest}). Next, note that in the present setting, Lemma~\ref{lemma:schauder} yields
  \[
  \mfn_{a,b}(t)\leq \mB[\ell_{1,l+1}(t)+\ell_{1,l}(t)],
  \]
  where $\mB:\ro_{+}\rightarrow\ro_{+}$ is a continuous and monotonically increasing function depending only on the covering, $(\bM,\bh_{\refer})$,
  $\{E_{i}\}$ and the model system. Combining the above observations with Lemma~\ref{lemma:energyestimatev} yields a continuous and monotonically
  increasing function $\mC_{E}:\ro_{+}\rightarrow\ro_{+}$ such that 
  \begin{equation}\label{eq:chElpoest}
    \begin{split}
      & \chE_{2,m+1}[v_{l+1},u_{l+1}](t_{1})\\
      \leq & \chE_{2,m+1}[v_{l+1},u_{l+1}](t_{0})+\big|\textstyle{\int}_{t_{0}}^{t_{1}}
      \mC_{E}[\ell_{l}(t)](\chE_{\rotot,m+1,l}(t)+1)dt\big|
    \end{split}    
  \end{equation}
  for $t\in\mI_{T}$, where
  \[
  \ell_{l}(t):=\ell_{1,l+1}(t)+\ell_{1,l}(t)+\|v_{l+1}(\cdot,t)\|_{C^{1}}.
  \]
  Moreover, $\mC_{E}$ only depends on $m$, the covering, $(\bM,\bh_{\refer})$, $\{E_{i}\}$ and the model system. Next, let
  $\chmfE_{2,m+1}[v,u]$ be defined as in the statement of the lemma. Moreover, let
  \[
  \chmfF_{2,m+1}:=\limsup_{l\rightarrow\infty}\chF_{2,m+1}[v_{l},u_{l}].
  \]
  Combining this notation with (\ref{eq:chFtmpochEtmpolim}), (\ref{eq:chElpoest}), Fatou's lemma and an argument similar to Lebesgue's dominated
  convergence theorem yields
  \begin{equation}\label{eq:chmfFestimate}
    \chmfF_{2,m+1}(t_{1})\leq \chmfE_{2,m+1}[v,u](t_{0})+\big|\textstyle{\int}_{t_{0}}^{t_{1}}\mfC[\ell(t)][\chmfF_{2,m+1}(t)+1]dt\big|
  \end{equation}  
  for $t\in\mI_{T}$, where $\mfC$ is a function with the properties given in the statement of the lemma. Moreover, $\ell(t)$ is defined by
  (\ref{eq:elldef}). Combining (\ref{eq:chmfFestimate}) with Gr\"{o}nwall's lemma yields
  \begin{equation}\label{eq:chmfFGrLe}
    \chmfF_{2,m+1}(t_{1})+1\leq [\chmfE_{2,m+1}[v,u](t_{0})+1]\exp\big(\big|\textstyle{\int}_{t_{0}}^{t_{1}}\mfC[\ell(t)]dt\big|\big)
  \end{equation}
  for $t_{1}\in\mI_{T}$.

  Next, let $\ldr{\cdot,\cdot}_{H,t}$ be the inner product introduced in (\ref{eq:Htinnerprod}), where the functions $u$, $v$ and $N_{1}$ appearing in
  (\ref{eq:Htinnerprod}) are part of the solution obtained above. Then
  \begin{equation}\label{eq:chmfEmfFest}
    \begin{split}
      \chmfE_{2,m+1}[v,u](t) = & \ldr{(u(\cdot,t),v(\cdot,t),(\d_{t}v)(\cdot,t)),(u(\cdot,t),v(\cdot,t),(\d_{t}v)(\cdot,t))}_{H,t}\\
      = & \textstyle{\lim}_{l\rightarrow\infty}\ldr{(u(\cdot,t),v(\cdot,t),(\d_{t}v)(\cdot,t)),(u_{l}(\cdot,t),v_{l}(\cdot,t),(\d_{t}v_{l})(\cdot,t))}_{H,t}\\
      \leq & \chmfE_{2,m+1}^{1/2}[v,u](t)\chmfF_{2,m+1}^{1/2}(t),
    \end{split}
  \end{equation}
  where the second equality is due to the weak convergence of $u_{l}(\cdot,t)$ to $u(\cdot,t)$ etc., and we used the fact that
  \[
  \chmfF_{2,m+1}(t)
  =\limsup_{l\rightarrow\infty}\ldr{(u_{l}(\cdot,t),v_{l}(\cdot,t),(\d_{t}v_{l})(\cdot,t)),(u_{l}(\cdot,t),v_{l}(\cdot,t),(\d_{t}v_{l})(\cdot,t))}_{H,t}. 
  \]
  Due to (\ref{eq:chmfEmfFest}), $\chmfE_{2,m+1}[v,u](t)\leq \chmfF_{2,m+1}[v,u](t)$, so that (\ref{eq:chmfFGrLe}) yields (\ref{eq:chmfEmainbound}). 

  \textit{Strong continuity.} Next, we wish to prove strong continuity of the solution. Due to uniqueness, it is sufficient to prove continuity
  at $t=t_{0}$. In other words, it is sufficient to prove that
  \[
  \lim_{t\rightarrow t_{0}}\left(\|(u,v)(\cdot,t)-(u,v)(\cdot,t_{0})\|_{H^{m+1}(\bM)}%+\|v(\cdot,t)-v(\cdot,t_{0})\|_{H^{m+1}(\bM)}
  +\|(\d_{t}v)(\cdot,t)-(\d_{t}v)(\cdot,t_{0})\|_{H^{m}(\bM)}\right)=0.
  \]
  Let, to this end,
  \[
  \xi(t):=[u(\cdot,t),v(\cdot,t),(\d_{t}v)(\cdot,t)],\ \ \
  \eta:=(\Ui_{0},\Vi_{0,0},\Vi_{0,1}).
  \]
  Then
  \begin{equation}\label{eq:Htzdiffexp}
    \ldr{\xi(t)-\eta,\xi(t)-\eta}_{H,t_{0}} = \ldr{\xi(t),\xi(t)}_{H,t_{0}}-2\ldr{\xi(t),\eta}_{H,t_{0}}+\ldr{\eta,\eta}_{H,t_{0}}.
  \end{equation}
  Note that the last term on the right hand side is $\chmfE_{2,m+1}[v,u](t_{0})$. Due to the weak continuity of the solution, the second term
  converges to $-2\chmfE_{2,m+1}[v,u](t_{0})$. In order to calculate the limit of the first term, note first that
  \[
  \lim_{t\rightarrow t_{0}}[\ldr{\xi(t),\xi(t)}_{H,t}-\ldr{\xi(t),\xi(t)}_{H,t_{0}}]=0.
  \]
  Moreover, due to (\ref{eq:chmfEmainbound}),
  \[
  \lim_{t\rightarrow t_{0}}\ldr{\xi(t),\xi(t)}_{H,t}\leq \chmfE_{2,m+1}[v,u](t_{0}).
  \]
  Combining these observations with (\ref{eq:Htequiv}) yields strong continuity at $t_{0}$ and thus, by the above argument, strong continuity
  on $\mI_{T}$. More specifically,
  \begin{equation}\label{eq:uvbasstrongreg}
    u\in C^{0}(\mI_{T},H^{m+1}(\bM,\rn{n_{1}})),\ \ 
    v\in C^{0}(\mI_{T},H^{m+1}(\bM,\rn{n_{2}})),\ \ 
    \d_{t}v\in C^{0}(\mI_{T},H^{m}(\bM,\rn{n_{2}})).
  \end{equation}
  Note, however, that $\d_{t}v$ here denotes the time derivative of $v$ considered as a function on $M_{T}$ and not the time derivative
  of $v$ considered as a function from $\mI_{T}$ to $H^{m}(\bM,\rn{n_2})$. Nevertheless, from (\ref{eq:uvbasstrongreg}), it can be deduced that
  \begin{equation}\label{eq:vbasstrongreg}
    v\in C^{1}(\mI_{T},H^{m}(\bM,\rn{n_{2}}));
  \end{equation}
  we leave the details to the reader. In what follows, we take steps of this type without further comment. Combining (\ref{eq:uvbasstrongreg}),
  (\ref{eq:vbasstrongreg}), (\ref{eq:themodelN}) and Corollary~\ref{cor: Ck regularity} yields
  \begin{equation}\label{eq:Nonebasstrongreg}
    N_{1}\in C^{0}(\mI_{T},H^{m+2}(\bM)).
  \end{equation}
  Combining (\ref{eq:uvbasstrongreg}), (\ref{eq:Nonebasstrongreg}), (\ref{eq:themodelNdot}) and Corollary~\ref{cor: Ck regularity} yields
  \begin{equation}\label{eq:Ntwobasstrongreg}
    N_{2}\in C^{0}(\mI_{T},H^{m+2}(\bM)).
  \end{equation}
  
  \textit{Differentiability with respect to time.} In order to derive higher order regularity in time, note that combining
  (\ref{eq:uvbasstrongreg}) and (\ref{eq:Nonebasstrongreg}) with (\ref{eq:themodelu}) yields the conclusion that
  \begin{equation}\label{eq:dtustrongreg}
    u\in C^{1}(\mI_{T},H^{m+1}(\bM,\rn{n_{1}})).
  \end{equation}
  Due to (\ref{eq:uvbasstrongreg}), (\ref{eq:Nonebasstrongreg}), (\ref{eq:Ntwobasstrongreg}) and (\ref{eq:themodelv}), it similarly follows that
  \begin{equation}\label{eq:vCtwostrong}
    v\in C^{2}(\mI_{T},H^{m-1}(\bM,\rn{n_{2}})). 
  \end{equation}
  Next, combining (\ref{eq:vbasstrongreg}) and (\ref{eq:dtustrongreg}) with (\ref{eq:themodelN}) and Corollary~\ref{cor: Ck regularity} yields
  \begin{equation}\label{eq:NoneConestrong}
    N_{1}\in C^{1}(\mI_{T},H^{m+2}(\bM)).
  \end{equation}
  Combining (\ref{eq:vbasstrongreg}), (\ref{eq:dtustrongreg}), (\ref{eq:vCtwostrong}) and (\ref{eq:NoneConestrong}) with
  (\ref{eq:themodelNdot}) and  Corollary~\ref{cor: Ck regularity} yields
  \begin{equation}\label{eq:NtwoConestrong}
    N_{2}\in C^{1}(\mI_{T},H^{m+1}(\bM)).
  \end{equation}
  Continuing as above, an inductive argument yields (\ref{seq:regularityoneuvNi}). The proposition follows. 
\end{proof}

\subsection{Continuation criterion, smooth solutions}\label{ssection:cont crit smooth sol}

Next, we derive a continuation criterion and prove the existence of smooth solutions, given smooth initial data.
The statements and proofs of the results are very similar to the corresponding statements and results of
\cite[Section~9.4, pp.~87--89]{RinCauchy}.
\begin{lemma}\label{lemma:contcriterionfdreg}
  Fix a model system, $m>n/2+2$, $\Ui_{0}\in H^{m+1}(\bM,\rn{n_{1}})$, $\Vi_{0,0}\in H^{m+1}(\bM,\rn{n_{2}})$ and $\Vi_{0,1}\in H^{m}(\bM,\rn{n_{2}})$.
  Let $t_{0}\in\mI$ and $T>0$ be such that $t_0+T\in\mI$. Assume that there is a solution $u,v,N_{1},N_{2}$ to (\ref{seq:themodel}) on
  $\mI_{T}:=[t_{0},t_{0}+T]$ satisfying
  (\ref{seq:uabietcreglocalexistence}) and (\ref{eq:uvdtvindata}). Then $u$, $v$, $N_{1}$ and $N_{2}$ satisfy (\ref{seq:regularityoneuvNi}).
  Moreover, (\ref{eq:chmfEmainbound}) is satisfied for $t_{1}\in \mI_{T}$. Let $T_{m}$ be the supremum of the times $T$ such that there is a
  solution on $[t_{0},t_{0}+T]$ satisfying (\ref{eq:uvdtvindata}) and (\ref{seq:uabietcreglocalexistence})--(\ref{seq:regularityoneuvNi}). Then
  either $T_{m}=t_{+}-t_0$ or
  \begin{equation}\label{eq:Tmcharacterization}
    \lim_{t\uparrow T_{m}}\sup_{t_{0}\leq\tau\leq t_{0}+t}\left(\|u(\cdot,\tau)\|_{C^{2}}+\|v(\cdot,\tau)\|_{C^{2}}
      +\|\d_{t}v(\cdot,\tau)\|_{C^{1}}\right)=\infty,
  \end{equation}
  where $t_{\pm}$ are defined by $\mI=(t_{-},t_{+})$.
\end{lemma}
\begin{remark}
  Due to the characterisation of $T_{m}$, it is clear that $T_{m}$ is independent of $m$. Note also that a similar argument
  gives a similar conclusion in the opposite time direction. 
\end{remark}
\begin{proof}
  Consider the estimate (\ref{eq:chmfEmainbound}). Due to Proposition~\ref{prop:localexistence}, we know that this estimate holds on an interval
  containing $t_{0}$ in its interior. Let $\msA_{+}$ be the set of $0\leq T_{*}\leq T$ such that
  \begin{subequations}\label{seq:uvreg}
    \begin{align}
      u\in & C^{1}\left([t_{0},t_{0}+T_{*}],H^{m+1}(\bM,\rn{n_{1}})\right),\\
      v\in & \textstyle{\bigcap}_{j=0}^{1}C^{j}\left([t_{0},t_{0}+T_{*}],H^{m+1-j}(\bM,\rn{n_{2}})\right)
    \end{align}
  \end{subequations}  
  and such that (\ref{eq:chmfEmainbound}) holds for all $t_{1}\in [t_{0},t_{0}+T_{*}]$. Due to Proposition~\ref{prop:localexistence}, we know that there is a
  $T_{*}>0$ such that $T_{*}\in\msA_{+}$. This means that $\msA_{+}$ is a non-empty and connected subset of $[0,T]$. We need to prove that it is open and
  closed. In order to prove openness,
  assume $T_{*}\in\msA_+$. Then we can apply Proposition~\ref{prop:localexistence} to the initial data induced at $t=t_{0}+T_{*}$. This leads to a solution
  in a neighbourhood, say $I_{\e}:=[t_{0}+T_{*}-\e,t_{0}+T_{*}+\e]$ for some $\e>0$, of $t=t_{0}+T_{*}$ (which coincides with the original solution on their
  common existence interval due to uniqueness; cf. Lemma~\ref{lemma:uniquemodelsol}). 
  This means that (\ref{seq:uvreg}) holds with $T_{*}$ replaced by
  $T_{*}+\min\{\e,T-T_{*}\}$ and that (\ref{eq:chmfEmainbound}) holds for $t_{1}\in [t_{0},t_{0}+T_{*}+\min\{\e,T-T_{*}\}]$. Thus $\msA_{+}$ is open. To
  prove that it is closed, assume $T_{*}$ to belong to the closure of $\msA_{+}$. Due to (\ref{eq:chmfEmainbound}), which holds for any $t_{1}<t_{0}+T_{*}$,
  it is clear that there is, on the set $[t_{0},t_{0}+T_{*})$, a uniform bound on the $H^{m+1}$-norms of $u(\cdot,t)$ and $v(\cdot,t)$ and the $H^{m}$-norm
  of $(\d_{t}v)(\cdot,t)$. Due to the dependence of the existence time on the data, the equation etc. in Proposition~\ref{prop:localexistence}, it is
  clear that we obtain a solution beyond $t_{0}+T_{*}$, that (\ref{seq:uvreg}) holds and that (\ref{eq:chmfEmainbound}) holds for
  $t_{1}\in [t_{0},t_{0}+T_{*}]$. This means that $\msA_{+}=[0,T]$.
  Since $u$, $v$, $N_{1}$ and $N_{2}$ locally have the regularity (\ref{seq:regularityoneuvNi}), it is clear that they have this regularity on $\mI_{T}$. 
  
  Finally, to prove the characterisation of $T_{m}$, assume that $T_{m}<t_{+}-t_0$, but that (\ref{eq:Tmcharacterization}) does not hold. This means that
  $\|u(\cdot,t)\|_{C^{2}}$, $\|v(\cdot,t)\|_{C^{2}}$ and $\|\d_{t}v(\cdot,t)\|_{C^{1}}$ are uniformly bounded on $[t_{0},t_{0}+T_{m})$.
  Combining this observation with (\ref{eq:chmfEmainbound}) yields the conclusion that $\chmfE_{2,m+1}[v,u]$ is uniformly bounded on this interval.
  However, by the local existence result, Proposition~\ref{prop:localexistence}, we can then extend the solution beyond $t_{0}+T_{m}$. This contradicts
  the definition of $T_{m}$. 
\end{proof}

Finally, we are in a position to prove Corollary~\ref{cor:localexmodsys}.
\begin{proof}[Proof of Corollary~\ref{cor:localexmodsys}]
  Given smooth initial data, we can apply Proposition~\ref{prop:localexistence} and Lemma~\ref{lemma:contcriterionfdreg} for a fixed $m>n/2+2$. This
  yields a solution with an existence interval as characterised in the statement of the corollary. Moreover, since the initial data are smooth, we can
  apply Lemma~\ref{lemma:contcriterionfdreg} for every $m>n/2+2$. This means that the solution satisfies (\ref{seq:regularityoneuvNi}) for every
  $m>n/2+2$. In particular, the solution is smooth. The uniqueness follows from Lemma~\ref{lemma:uniquemodelsol}.
\end{proof}

\subsection{Stability}\label{ssection:stabilitymodelsystem}
Next, we prove Cauchy stability of solutions to (\ref{seq:themodel}); i.e., Proposition~\ref{prop:CauchyStability}.

\begin{proof}[Proof of Proposition~\ref{prop:CauchyStability}]
  In what follows, we refer to the solution $u$, $v$, $N_{1}$, $N_{2}$ to (\ref{seq:themodel}) as the \textit{background solution}. Given initial data
  as in (\ref{eq:roperid}), we refer to the corresponding solution (denoted $u_{\roper}$, $v_{\roper}$, $N_{1,\roper}$ and $N_{2,\roper}$) to (\ref{seq:themodel})
  satisfying (\ref{eq:roperidcond}) as the \textit{perturbed solution}. Denote the corresponding existence intervals obtained by appealing to Corollary~\ref{cor:localexmodsys}
  by $\mI_{\robg}$ and $\mI_{\roper}$. We are interested in controlling the distance between the two solutions. It is therefore of interest to
  derive a system of equations for $\mfu:=u_{\roper}-u$, $\mfv:=v_{\roper}-v$, $\mfN_{1}:=N_{1,\roper}-N_{1}$ and $\mfN_{2}:=N_{2,\roper}-N_{2}$.
  Compute
  \begin{equation}\label{eq:dtmfuleq}
    \d_{t}\mfu=\mff_{1}[\mfu,\mfv,\mfN_{1}],
  \end{equation}
  where
  \begin{equation}
    \begin{split}
      \mff_{1}[\mfu,\mfv,\mfN_{1}] = & \textstyle{\int}_{0}^{1}(\d_{u}f_{1})[s\mfu+u,\mfv+v,\mfN_{1}+N_{1}]ds \cdot \mfu\\
      & +\textstyle{\int}_{0}^{1}(\d_{v}f_{1})[u,s\mfv+v,\mfN_{1}+N_{1}]ds \cdot \mfv\\
      & +\textstyle{\int}_{0}^{1}(\d_{N_{1}}f_{1})[u,v,s\mfN_{1}+N_{1}]ds \cdot (\mfN_{1},E_{1}\mfN_{1},\dots,E_{n}\mfN_{1}).
    \end{split}
  \end{equation}
  Here $\d_{N_{1}}f_{1}$ denotes the derivative of $f_{1}$ with respect to all arguments corresponding to $N_{1}$; i.e., $N_{1}$, $E_{1}N_{1},\dots,E_{n}N_{1}$.  
  In particular, $\mff_{1}$ is such that $\mff_{1}[0,0,0]=0$ (setting the third variable to zero here also entails putting the $E_{i}\mfN_{1}$ arguments
  to zero). We would like to think of (\ref{eq:dtmfuleq}) as fitting into the framework introduced at the beginning of Subsection~\ref{ssection:modelsystem}.
  However, $\mff_{1}$ then has to satisfy the conditions imposed on $f_{1}$. To this end, we interpret the occurrences of $u$, $v$ and $N_{1}$ as given
  functions (contributing to the (non-explicit) $t$ and $\bx$ dependence). In order to obtain the desired bounds on the functions appearing in the system,
  it is therefore natural to replace $\mI$ appearing in the definition of a model system by an open interval $\mI_{0}$, where $\mI_{0}$ has compact closure
  contained in $\mI_{\robg}$. We can also assume that $t_{0},t_{\rofin}\in\mI_{0}$. However, even so, $\mff_{1}$ and its derivatives are typically not bounded.
  Nevertheless, this can be arranged by replacing the last factor in the first term on the right hand side (i.e., $\mfu$) by $\phi\circ\mfu$, where
  $\phi:=C_{0}^{\infty}(\rn{n_{1}},\rn{n_{1}})$ is such that $\phi(\xi)=\xi$ for $|\xi|\leq 1$; and similarly for the last factors in the second and third terms
  on the right hand side. This leads to a function $\mfg_{1}$, satisfying the requirements introduced in Subsection~\ref{ssection:modelsystem}, such that
  $\mfg_{1}[0,0,0]=0$, and such that if $|\mfu|\leq 1$, $|\mfv|\leq 1$, $|\mfN_{1}|\leq 1$ and $|E_{i}\mfN_{1}|\leq 1$, $i=1,\dots,n$, then $\mfu$, $\mfv$ and
  $\mfN_{1}$ solve (\ref{eq:dtmfuleq}) iff they solve
  \[
    \d_{t}\mfu=\mfg_{1}[\mfu,\mfv,\mfN_{1}].
  \]
  Next, (\ref{eq:themodelv}) yields
  \begin{equation*}
    \begin{split}
      \d_{t}^{2}\mfv = & \bh_{1}^{ij}[\mfu+u,\mfN_{1}+N_{1}]E_{i}E_{j}\mfv
      +(\bh_{1}^{ij}[\mfu+u,\mfN_{1}+N_{1}]-\bh_{1}^{ij}[u,N_{1}])E_{i}E_{j}v\\
      & +\g^{ij}[\mfu+u,\mfv+v,\mfN_{1}+N_{1}]E_{i}E_{j}\mfu\\
      & +(\g^{ij}[\mfu+u,\mfv+v,\mfN_{1}+N_{1}]-\g^{ij}[u,v,N_{1}])E_{i}E_{j}u\\
      & +f_{2}[\mfu+u,\mfv+v,\mfN_{1}+N_{1},\mfN_{2}+N_{2}]-f_{2}[u,v,N_{1},N_{2}].
    \end{split}
  \end{equation*}
  By arguments similar to the above, this equation is equivalent to 
  \[
  \d_{t}^{2}\mfv=\bmfh_{1}^{ij}[\mfu,\mfN_{1}]E_{i}E_{j}\mfv+\eta^{ij}[\mfu,\mfv,\mfN_{1}]E_{i}E_{j}\mfu
  +\mfg_{2}[\mfu,\mfv,\mfN_{1},\mfN_{2}],
  \]
  assuming $|\mfu|\leq 1$, $|E_{i}\mfu|\leq 1$, $|\mfv|\leq 1$, $|E_{i}\mfv|\leq 1$, $|\d_{t}\mfv|\leq 1$, $|\mfN_{1}|\leq 1$, $|\mfN_{2}|\leq 1$,
  $|E_{i}\mfN_{1}|\leq 1$, $|E_{i}\mfN_{2}|\leq 1$, $|E_{i}E_{j}\mfN_{1}|\leq 1$ and $|E_{i}E_{j}\mfN_{2}|\leq 1$, $i,j=1,\dots,n$, where $\bmfh_{1}^{ij}$,
  $\eta^{ij}$ and $\mfg_{2}$ have the same properties as $\bh_{1}^{ij}$, $\g^{ij}$ and $f_{2}$, introduced at the beginning of
  Subsection~\ref{ssection:modelsystem}, respectively. Moreover, $\mfg_{2}[0,0,0,0]=0$.

  Turning to the lapse equations, (\ref{eq:themodelN}) yields
  \begin{equation*}
    \begin{split}
      & \Delta_{\bh_{2}[\mfu+u]}\mfN_{1}+(\Delta_{\bh_{2}[\mfu+u]}-\Delta_{\bh_{2}[u]})N_{1}\\
      = & \zeta[\mfu+u,\mfv+v]\mfN_{1}+(\zeta[\mfu+u,\mfv+v]-\zeta[u,v])N_{1}\\
      & +f_{3}[\mfu+u,\mfv+v]-f_{3}[u,v].
    \end{split}
  \end{equation*}
  Again, due to arguments similar to the above, this equation is equivalent to
  \begin{equation}\label{eq:DeltamfNonel}
    \Delta_{\bmfh_{2}[\mfu]}\mfN_{1}=\theta[\mfu,\mfv]\mfN_{1}+\mfg_{3}[\mfu,\mfv],
  \end{equation}
  assuming $|\mfu|\leq 1$, $|\mfv|\leq 1$ and $|E_{i}\mfu|\leq 1$, $i=1,\dots,n$, where $\bmfh_{2}^{ij}$, $\theta$ and $\mfg_{3}$ have the same properties
  as $\bh_{2}^{ij}$, $\zeta$ and $f_{3}$, introduced at the beginning of Subsection~\ref{ssection:modelsystem}, respectively. Moreover, $\mfg_{3}[0,0]=0$.
  Finally, (\ref{eq:themodelNdot}) yields an equation which, due to arguments similar to the above, is equivalent to
  \[
  \Delta_{\bmfh_{2}[\mfu]}\mfN_{2}=\theta[\mfu,\mfv]\mfN_{2}+\mfg_{4}[\mfu,\mfv,\mfN_{1}],
  \]
  assuming $|\mfu|\leq 1$, $|\mfv|\leq 1$, $|E_{i}\mfu|\leq 1$, $|E_{i}\mfv|\leq 1$, $|\d_t\mfv|\leq 1$, $|\mfN_{1}|\leq 1$, $|E_{i}\mfN_{1}|\leq 1$ and
  $|E_{i}E_{j}\mfN_{1}|\leq 1$, $i,j=1,\dots,n$, where $\bmfh_{2}^{ij}$ and $\theta$ are the same functions as in the case of (\ref{eq:DeltamfNonel}) and
  $\mfg_{4}$ has the same properties as $f_{4}$, introduced at the beginning of Subsection~\ref{ssection:modelsystem}. Moreover, $\mfg_{4}[0,0,0]=0$. To
  summarise, we obtain the following system
  \begin{subequations}\label{seq:stabsys}
    \begin{align}
      \d_{t}\mfu = & \mfg_{1}[\mfu,\mfv,\mfN_{1}],\\
      \d_{t}^{2}\mfv = & \bmfh_{1}^{ij}[\mfu,\mfN_{1}]E_{i}E_{j}\mfv+\eta^{ij}[\mfu,\mfv,\mfN_{1}]E_{i}E_{j}\mfu+\mfg_{2}[\mfu,\mfv,\mfN_{1},\mfN_{2}],\\
      \Delta_{\bmfh_{2}[\mfu]}\mfN_{1} = & \theta[\mfu,\mfv]\mfN_{1}+\mfg_{3}[\mfu,\mfv],\\
      \Delta_{\bmfh_{2}[\mfu]}\mfN_{2} = & \theta[\mfu,\mfv]\mfN_{2}+\mfg_{4}[\mfu,\mfv,\mfN_{1}].
    \end{align}
  \end{subequations}
  Note here that $\bmfh_{1}^{ij}$, $\bmfh_{2}^{ij}$ and $\theta$ satisfy (\ref{eq:lambdabd}) and (\ref{eq:lambdazetalowerandupperbd}) with the same values
  of $\lambda$ and $\lambda_{\zeta}$ as $\bh_{1}^{ij}$, $\bh_{2}^{ij}$ and $\zeta$.  
  Since this system satisfies the conditions described at the beginning of Subsection~\ref{ssection:modelsystem} (with $\mI$ replaced by $\mI_{0}$),
  Corollary~\ref{cor:localexmodsys} applies. Given smooth initial data at $t_{0}\in\mI_{0}$, say $\mfu_{0}$, $\mfv_{0,0}$ and $\mfv_{0,1}$, we thus obtain
  a smooth solution $\mfu$, $\mfv$, $\mfN_{1}$ and $\mfN_{2}$ to (\ref{seq:stabsys}). Denote the corresponding existence interval by $\mI_{\max}\subset\mI_{0}$.
  Restricting the system given by (\ref{seq:stabsys}) to the subinterval $\mI_{\max}$, it still satisfies the conditions described at the beginning of
  Subsection~\ref{ssection:modelsystem} (with $\mI$ replaced by $\mI_{\max}$). The advantage of taking this perspective is that we can then think of the
  solution $\mfu$, $\mfv$, $\mfN_{1}$ and $\mfN_{2}$ as arising from an application of Lemma~\ref{lemma:ubvbNbsol}; i.e., $u_{a}=u_{b}=\mfu$,
  $v_{a}=v_{b}=\mfv$ etc. Moreover, we can appeal to (\ref{eq:NobSchauderfz}) and (\ref{eq:NtwobSchauderfz}) in order to conclude that there is a
  continuous and monotonically increasing function $\mB:\ro_{+}\rightarrow\ro_{+}$ such that
  \begin{equation}\label{eq:mfNonetwobSchauderest}
    \begin{split}
      \|\mfN_{1}(\cdot,t)\|_{C^{2,1}}+\|\mfN_{2}(\cdot,t)\|_{C^{2,1}} \leq & \mB[\mfl(t)]\mfl(t)
    \end{split}    
  \end{equation}
  for all $t\in\mI_{\max}$, where
  \begin{equation}\label{eq:mfldef}
    \mfl(t):=\|\mfu(\cdot,t)\|_{C^{2}}+\|\mfv(\cdot,t)\|_{C^{2}}+\|\d_{t}\mfv(\cdot,t)\|_{C^{1}}.
  \end{equation}
  Here $\mB$ only depends on bounds on $\bmfh_{2}^{ij}$, $\theta$, $\mfg_{3}$ and $\mfg_{4}$ and
  their derivatives up to order two, one, two and two respectively; $(\bM,\bh_{\refer})$; the covering; $\{E_{i}\}$; the $\lambda$ appearing in
  (\ref{eq:lambdabd}); and the $\lambda_{\zeta}$ appearing in (\ref{eq:lambdazetalowerandupperbd}). Returning to (\ref{eq:smfnadef}) and
  (\ref{eq:mfnabdef}), it is thus clear that $\mfn_{a,b}$ is bounded by the right hand side of (\ref{eq:mfNonetwobSchauderest}) (with a trivial modification
  of $\mB$). Next, Lemma~\ref{lemma:lapseestimatehighnorm} and Remark~\ref{remark:Nbmainestimateimp} yield the conclusion that there is a
  continuous and monotonically increasing function $\mC:\ro_{+}\rightarrow\ro_{+}$ such that
  \begin{equation}\label{eq:mfNonetwoSobest}
    \begin{split}
      \|\mfN_{1}\|_{H^{m+2}(\bM_{t})}+\|\mfN_{2}\|_{H^{m+2}(\bM_{t})} \leq & \mC[\mfl(t)](\|\mfu\|_{H^{m+1}(\bM_{t})}+\|\mfv\|_{H^{m+1}(\bM_{t})}+\|\d_{t}\mfv\|_{H^{m}(\bM_{t})})
    \end{split}    
  \end{equation}
  for all $t\in\mI_{\max}$. Here $\mC$ only depends on the covering; $(\bM,\bh_{\refer})$; $\{E_{i}\}$; $m$; $\lambda$; $\lambda_{\zeta}$; a bound on up to
  $m+1$ derivatives of $\theta$, $\mfg_{3}$ and $\mfg_{4}$; and up to $m+2$ derivatives of $\bmfh_{2}^{ij}$.
  Returning to (\ref{eq:mfNadef}) and
  (\ref{eq:mfnabdef}), it is thus clear that $\mfN_{a,b}$ is bounded by the right hand side of (\ref{eq:mfNonetwoSobest}) (with a trivial modification
  of $\mC$). Returning to (\ref{eq:tEtwompodiffest}) with this information at hand yields a continuous and monotonically increasing function
  $\mC_{m}$ such that 
  \begin{equation}\label{eq:mftEtwompodiffest}
    \begin{split}
      |\chE_{2,m+1}[\mfv,\mfu](t_{1})-\chE_{2,m+1}[\mfv,\mfu](t_{0})|
      \leq & \big|\textstyle{\int}_{t_{0}}^{t_{1}}\mC_{m}[\mfl(t)]\chE_{2,m+1}[\mfv,\mfu](t)dt\big|.
    \end{split}
  \end{equation}
  Here $\mC_{m}$ only depends on bounds on up to $m+2$ derivatives of $\theta$, $\mfg_{1}$, $\mfg_{2}$, $\mfg_{3}$, $\mfg_{4}$, $\bmfh_{1}^{ij}$,
  $\bmfh_{2}^{ij}$ and $\eta^{ij}$; $(\bM,\bh_{\refer})$; the covering; $\{E_{i}\}$; $m$; $\lambda$; and $\lambda_{\zeta}$. In what follows, it is of interest
  to note that, due to Sobolev embedding, there is a constant $K_{\mfl}\geq 1$, depending only on $\lambda$, $\{E_{i}\}$ and $(\bM,\bh_{\refer})$, such that
  \begin{equation}\label{eq:mflSobemb}
    \mfl(t)\leq K_{\mfl}\chE_{2,m+1}^{1/2}[\mfv,\mfu](t)
  \end{equation}
  for all $t\in\mI_{\max}$. Combining (\ref{eq:mftEtwompodiffest}) with Gr\"{o}nwall's lemma yields
  \begin{equation}\label{eq:CauchyEnergyEstimate}
    \chE_{2,m+1}[\mfv,\mfu](t_{1})\leq \chE_{2,m+1}[\mfv,\mfu](t_{0})\exp\big(\big|\textstyle{\int}_{t_{0}}^{t_{1}}\mC_{m}[\mfl(t)]dt\big|\big)
  \end{equation}
  for all $t\in\mI_{\max}$. Let $C_{m}:=\mC_{m}(1)$ and let $t_{\rofin}$ be as in the statement of the proposition. Assume that $t_{\rofin}\geq t_{0}$ (the
  argument in the case that $t_{\rofin}\leq t_{0}$ is similar and left to the reader) and that 
  \begin{equation}\label{eq:chEinconCauchy}
    \chE_{2,m+1}[\mfv,\mfu](t_{0})\leq (3K_{\mfl})^{-2}\exp(-C_{m}|t_{\rofin}-t_{0}|).
  \end{equation}
  Let $\msA$ be the set of $t\in [t_{0},t_{\rofin}]$ such that $[t_{0},t]\subset\mI_{\max}$ and such that
  \begin{equation}\label{eq:chECauchyboot}
    \chE_{2,m+1}[\mfv,\mfu](s)\leq (2K_{\mfl})^{-2}
  \end{equation}
  for all $s\in [t_{0},t]$. Due to (\ref{eq:chEinconCauchy}), it is clear that $\msA$ is non-empty. It is also connected by definition. What remains
  is to prove that it is open and closed. In order to prove that it is open, assume $t\in\msA$. Then $[t_{0},t]\subset\mI_{\max}$ by definition. In
  particular, the solution exists beyond $t$. Moreover, combining (\ref{eq:mflSobemb}) and (\ref{eq:chECauchyboot}) yields
  \[
    \mfl(s)\leq 1/2
  \]
  for all $s\in [t_{0},t]$. By continuity, there is therefore an $\e>0$ such that $[t_{0},t+\e)\subset\mI_{\max}$ and $\mfl(s)\leq 1$ for all
  $s\in [t_{0},t+\e)$. This means, in particular, that $\mC_{m}[\mfl(s)]\leq C_{m}$ for all $s\in [t_{0},t+\e)$. Combining this observation
  with (\ref{eq:CauchyEnergyEstimate}) and (\ref{eq:chEinconCauchy}) yields
  \[
    \chE_{2,m+1}[\mfv,\mfu](s)\leq (3K_{\mfl})^{-2}\exp(-C_{m}|t_\rofin-t_{0}|)
    \exp\big(\big|\textstyle{\int}_{t_{0}}^{s}\mC_{m}[\mfl(t)]dt\big|\big)\leq (3K_{\mfl})^{-2}
  \]
  for all $s\in [t_{0},t+\e)\cap [t_0,t_\rofin]$. This means that the bootstrap assumption (\ref{eq:chECauchyboot}) is satisfied for all
  $s\in [t_{0},t+\e)\cap [t_0,t_\rofin]$. Thus
  $\msA$ is open. Finally, we need to prove that it is closed. Assume, to this end that $t_{k}\in\msA$ is an increasing sequence converging to
  $t_{*}\in\mI_{0}$. There are two possibilities. Either $[t,t_{*}]\subset \mI_{\max}$ or $t_{*}\notin \mI_{\max}$. In the former case, $t_{*}\in\msA$
  due to continuity. In case $t_{*}\notin \mI_{\max}$, we have to have
  \begin{equation}\label{eq:contcritCauchy}
    \textstyle{\limsup}_{t\rightarrow t_{*}}\mfl(t)=\infty
  \end{equation}
  due to the Corollary~\ref{cor:localexmodsys}. On the other hand, since $t_{k}\in\msA$, combining (\ref{eq:mflSobemb}) and (\ref{eq:chECauchyboot})
  for $t\in [t_{0},t_{k}]$ yields the conclusion that $\mfl(t)\leq 1/2$ for $t\in [t_{0},t_{k}]$. Since $t_{k}\rightarrow t_{*}$, this means that
  $\mfl(t)\leq 1/2$ for all $t\in [t_{0},t_{*})$. This conclusion contradicts (\ref{eq:contcritCauchy}). Thus $t_{*}\in \msA$, and $\msA$ is closed.
  This means that $\msA=[t_{0},t_{\rofin}]$. By the above arguments, it also follows that $\mfl(t)\leq 1/2$ for all $t\in [t_{0},t_{\rofin}]$. Combining
  this observation with (\ref{eq:CauchyEnergyEstimate}) yields the conclusion that
  \begin{equation}\label{eq:CauchymapLipschitz}
    \chE_{2,m+1}[\mfv,\mfu](t)\leq \chE_{2,m+1}[\mfv,\mfu](t_{0})\exp\left(C_{m}|t_{\rofin}-t_{0}|\right)
  \end{equation}
  for all $t\in [t_{0},t_{\rofin}]$. Combining this observation with (\ref{eq:mfNonetwobSchauderest}) and (\ref{eq:mflSobemb}), it is clear that
  there is an $\e_{0}>0$ such that if $\chE_{2,m+1}[\mfv,\mfu](t_{0})\leq\e_{0}$, then
  \[
  \mfl(t)+\|\mfN_{1}(\cdot,t)\|_{C^{2,1}}+\|\mfN_{2}(\cdot,t)\|_{C^{2,1}}\leq 1
  \]
  for all $t\in [t_{0},t_{\rofin}]$. This means, in particular, that the solution $\mfu$, $\mfv$, $\mfN_{1}$ and $\mfN_{2}$ to (\ref{seq:stabsys})
  corresponds to a solution $u_{\roper}=u+\mfu$, $v_{\roper}=v+\mfv$, $N_{1,\roper}=N_{1}+\mfN_{1}$ and $N_{2,\roper}=N_{2}+\mfN_{2}$ to
  (\ref{seq:themodel}). Combining this observation with (\ref{eq:CauchymapLipschitz}), it is clear that given $\e>0$, there is a $\delta>0$
  such that the statement of the proposition holds. 
\end{proof}

\section{Gauge choice and constraint propagation}\label{section:applgentheorytoEFE}
%\section{Second order equations in gauges with vanishing shift vector field}\label{section:second order eqs}

Our next goal is to apply the general theory to gauge fixed versions of Einstein's equations. In order to clarify our gauge choice, let, to begin with,
$(M,g)$ be a Lorentz manifold, where $M:=\bM\times \mI$, $\bM$ is an $n$-dimensional connected (and typically oriented and closed) manifold and $\mI$ is an
open interval. Assume, moreover, the metrics induced on the leaves $\bM_t:=\bM\times\{t\}$, say $\bge$, to be Riemannian. Then the
\textit{lapse function} $N$ and the \textit{shift vector field} $\chi$ are defined by
\begin{equation}\label{eq:LapseShiftdef}
  \d_{t}=NU+\chi,
\end{equation}
where $\d_{t}$ denotes differentiation with respect to the second variable in the division $\bM\times \mI$, $U$ denotes the future directed unit normal
to the leaves $\bM_{t}$ and $\chi$ is required to be tangential to $\bM_{t}$. The gauge choices we are interested in here are all such that the shift
vector field vanishes. This means that the spacetime metric takes the form
\begin{equation}\label{eq:g zero shift}
  g=-N^2dt\otimes dt+\bge,
\end{equation}
where the lapse function $N$ is strictly positive. In what follows, we typically assume $\bM$ to have a global frame $\{E_i\}$, though this is not
absolutely necessary. Next, recall that if $\chi=0$, then
\begin{equation}\label{dtbkijgeo}
  \d_{t}\bk_{ij}=\dmfk_{ij},
\end{equation}
where
\begin{equation}\label{eq:dmfk def}
  \dmfk_{ij}:=\bna_{i}\bna_{j}N+2N\bk_{i}^{\phantom{i}m}\bk_{mj}-N(\tr_{\bge}\bk)\bk_{ij}+NR_{ij}-N\bR_{ij};
\end{equation}
see, e.g., the first displayed equation after \cite[(263), p.~64]{RinGeo}. Here the indices refer to a fixed frame on $\bM$, say $\{E_{i}\}$.
It is important to note that (\ref{dtbkijgeo}) holds for any foliation as above. In particular, it holds irrespective of whether Einstein's equations
(and additional conditions, such as the CMC condition) are satisfied or not. In (\ref{eq:dmfk def}), $\bna$ denotes the Levi-Civita connection of the
Riemannian metric $\bge$ induced on the leaves of the foliation; $\bk$ denotes the induced second fundamental form; $R_{ij}$ denotes the spatial
components of the spacetime Ricci tensor; and $\bR_{ij}$ denotes the components of the Ricci tensor associated with $\bge$. Specialising to the
Einstein-non-linear scalar field setting with a CMC foliation, the equation we actually wish to solve is
\begin{equation}\label{eq:dtbkij}
  \d_{t}\bk_{ij}=\dka_{ij},
\end{equation}
where
\begin{equation}\label{eq:dkaijdefinition}
  \dka_{ij}:=\bna_{i}\bna_{j}N+2N\bk_{i}^{\phantom{i}m}\bk_{mj}-Nt^{-1}\bk_{ij}+N\bna_{i}\phi\bna_{j}\phi
  +\tfrac{2N}{n-1}(V\circ\phi)\bge_{ij}-N\bR_{ij}.
\end{equation}
To arrive at (\ref{eq:dtbkij}) from (\ref{dtbkijgeo}), we made the substitutions $\tr_{\bge}\bk=1/t$ (which is the CMC condition) and
\[
  \mathrm{Ric}=d\phi\otimes d\phi+\tfrac{2}{n-1}(V\circ\phi)g,
\]
which are a consequence of the Einstein-non-linear scalar field equations. In order to relate (\ref{dtbkijgeo}) and (\ref{eq:dtbkij}), it
is convenient to introduce
\begin{equation}\label{eq:msDijdef}
  \msD_{ij}:=N^{-1}(\dmfk_{ij}-\dka_{ij})=-\vartheta\bk_{ij}+R_{ij}-\bna_{i}\phi\bna_{j}\phi-\tfrac{2}{n-1}(V\circ\phi)\bge_{ij},
\end{equation}
where $\theta:=\tr_{\bge}\bk$ and $\vartheta:=\theta-1/t$. Next, we introduce an orthonormal frame with respect to which we express the equations. 

\subsection{Orthonormal frame}
The orthonormal frame we use here is defined by an evolution equation; see (\ref{eq:evoluteqframe}) below. We begin by deriving conditions on the
evolution equation that ensure that the frame remains orthonormal. 

\begin{lemma}\label{lemma:frameorth}
  Let $(M,g)$ be a Lorentz manifold, where $M=\bM\times \mI$ and $\mI$ is an open interval. Assume $\d_{t}$ to be timelike and orthogonal to the leaves
  $\bM_{t}:=\bM\times\{t\}$ of the foliation, and let $N>0$ be such that $g(\d_{t},\d_{t})=-N^{2}$; cf. (\ref{eq:LapseShiftdef}). Assume that there
  is a global frame $\{E_{i}\}$ of the tangent space of $\bM$. Let $\{e_{I}\}$, $I=1,\dots,n:=\rodim{\bM}$, be a frame for the leaves of the foliation.
  Assume the frame to be orthonormal on one leaf and to satisfy the evolution equations
  \begin{equation}\label{eq:evoluteqframe}
    \d_{t}e_{I}^{i}=f_{I}^{J}e_{J}^{i},
  \end{equation}
  for some functions $f_{I}^{J}$, where $e_{I}^{i}$ is defined by the condition that $e_{I}=e_{I}^{i}E_{i}$. Then the frame $\{e_{I}\}$ is orthonormal
  on all the leaves if and only if the components of the symmetric part of the matrix with components $f_{I}^{J}$ are given by $-N\bk_{IJ}$;
  i.e., if
  \begin{equation}\label{eq:fIJsymmpart}
    f_{I}^{J}+f_{J}^{I}=-2N\bk_{IJ},
  \end{equation}
  where $\bk$ denotes the second fundamental form of the leaves of the foliation and $\bk_{IJ}:=\bk(e_{I},e_{J})$. 
\end{lemma}
\begin{remark}\label{remark:fIJbreakdown}
  It is important to note that any choice of antisymmetric part of the matrix with components $f_{I}^{J}$ is compatible with the frame $\{e_{I}\}$
  being orthonormal. In what follows, we write $f_{I}^{J}=s_{I}^{J}+a_{I}^{J}$, where $s_{I}^{J}=s_{J}^{I}$ and $a_{I}^{J}=-a_{J}^{I}$. By the lemma,
  $s_{I}^{J}=-N\bk_{IJ}$.
\end{remark}
\begin{remark}
  If $\{\omega^{I}\}$ is the frame dual to $\{e_{I}\}$ and $\omega^{I}_{i}$ is defined by the condition that
  $\omega^{I}=\omega^{I}_{i}\eta^{i}$, where $\{\eta^{i}\}$ is the frame dual to $\{E_{i}\}$, then $\de_{I}^{J}=\omega^{J}_{i}e^{i}_{I}$, so that
  \[
    \d_{t}\omega^{I}_{i}=-f^{I}_{J}\omega^{J}_{i}.
  \]  
\end{remark}
%\begin{remark}\label{remark:bkIJ symm fIJ}%NEW REMARK
%  The right hand side of (\ref{eq:fIJsymmpart}) can be simplified to $-2N\bk_{IJ}$. However, at the end, we write down and solve an
%  equation for an object we still denote $\bk_{IJ}$, but which is not a priori the second fundamental form of the leaves of the foliation.
%  In fact, this object is not even a priori symmetric. This is the motivation for writing the right hand side of (\ref{eq:fIJsymmpart})
%  in such a way that it is symmetric irrespective of the symmetry properties of the object $\bk_{IJ}$. 
%\end{remark}
\begin{proof}
  By time differentiating the condition $\de_{IJ}=g(e_{I},e_{J})$, it can be verified that the condition that the frame remains
  orthonormal is equivalent to (\ref{eq:fIJsymmpart}). The lemma follows. 
\end{proof}

Next, define $\bga_{IJ}^{K}$ by the relation $[e_{I},e_{J}]=\bga_{IJ}^{K}e_{K}$. It is of interest to derive evolution equations for $\bga_{IJ}^{K}$.
\begin{lemma}\label{lemma:Jacobi identity}%NEW FORMULATION
  Given assumptions and notation as in Lemma~\ref{lemma:frameorth}, define $\bga_{IJ}^{K}$ as above. Then
  \begin{equation}\label{eq:dtbgaIJK}
    \begin{split}
      \d_{t}\bga_{IJ}^{K} = & e_{I}(f_{J}^{K})-e_{J}(f_{I}^{K})+f_{I}^{M}\bga_{MJ}^{K}+f_{J}^{M}\bga_{IM}^{K}-f_{M}^{K}\bga_{IJ}^{M}.
    \end{split}
  \end{equation}  
\end{lemma}
\begin{remark}\label{remark:depofbgaIJK}
  Note also that $\bga_{IJ}^{K}$ is polynomial in $\eta_{ij}^k$, $\omega^{I}_{i}$, $e_{J}^{j}$ and $E_{k}e_{L}^{l}$ (and of first order in the last
  expression). Here $\eta_{ij}^k$ is defined by $[E_i,E_j]=\eta_{ij}^kE_k$. 
\end{remark}
\begin{proof}
  Note that $0=[\d_t,[e_{I},e_{J}]]+[e_{I},[e_{J},\d_t]]+[e_{J},[\d_t,e_{I}]]$ due to the Jacobi identity. Moreover, $[e_J,\d_t]=-f_J^Me_M$ and 
  \begin{equation*}
    \begin{split}
      [\d_t,[e_{I},e_{J}]] = & [\d_t,\bga_{IJ}^{K}e_{K}]
      = (\d_t\bga_{IJ}^{K}+\bga_{IJ}^Mf_M^K)e_{K}.
    \end{split}
  \end{equation*}
  Combining these observations yields (\ref{eq:dtbgaIJK}).
\end{proof}

Our next goal is to derive a second order equation for $\bk_{IJ}$. To this end, we first need to derive a first order equation.
However, 
\[
  \d_{t}\bk_{IJ}=\d_{t}(\bk_{ij}e_{I}^{i}e_{J}^{j})=(\d_{t}\bk_{ij})e_{I}^{i}e_{J}^{j}+f_{I}^{M}\bk_{MJ}+\bk_{IM}f_{J}^{M}.
\]
If (\ref{eq:dtbkij}) holds (and note, very carefully, that we most of the time do not assume (\ref{eq:dtbkij}) to hold, we only deduce it
at the very end)
\begin{equation}\label{eq:dtbkIJfv}
  \begin{split}
    \d_{t}\bk_{IJ} = & \dka_{IJ}+f_{I}^{M}\bk_{MJ}+\bk_{IM}f_{J}^{M}\\
    = & \bna_{I}\bna_{J}N+2N\bk_{IM}\bk_{MJ}-Nt^{-1}\bk_{IJ}+Ne_{I}(\phi)e_{J}(\phi)\\
    &+\tfrac{2N}{n-1}(V\circ\phi)\de_{IJ}-N\bR_{IJ}+f_{I}^{M}\bk_{MJ}+\bk_{IM}f_{J}^{M};
  \end{split}
\end{equation}
here we use the convention that we sum over repeated indices, even if both are subscripts. In order to derive a second order equation for
$\bk_{IJ}$, we need to differentiate (\ref{eq:dtbkIJfv}). Most of the terms that result do not contribute to the symbol. However, we do need
to compute $\d_{t}\bR_{IJ}$.
\begin{lemma}\label{lemma:dtbRIJ}
  Given assumptions and notation as in Lemma~\ref{lemma:frameorth}, assume $\mI\subset (0,\infty)$, assume $\phi$ to be a smooth function on $M$ and
  assume the frame $\{e_{I}\}$ to be orthonormal on all the leaves of the foliation. Let $\bR_{IJ}$ denote the components of the spatial Ricci curvature
  with respect to the orthonormal frame $\{e_{I}\}$, let $\theta:=\tr_{\bge}\bk$ and let $\vartheta:=\theta-1/t$. Let, moreover,
  \begin{equation}\label{eq:msCIdef}
    \msC(e_{I}):=(\bna_{e_{L}}\bk)(e_{L},e_{I})-e_{I}(\theta)-e_{0}(\phi)e_{I}(\phi),
  \end{equation}
  where $e_{0}:=N^{-1}\d_{t}$. Then
  \begin{equation}\label{eq:dtbRIJ}
    \d_{t}\bR_{IJ}-N\bna_{I}\msC_{J}-N\bna_{J}\msC_{I}-N\bna_{I}\bna_{J}\vartheta=-e_{M}e_{M}(N\bk_{IJ})+\Upsilon_{IJ}+\mfR_{IJ},
  \end{equation}
  where $\mfR_{IJ}$ is polynomial in $f_{L}^{M}$, $e_{L}^{i}$, $\omega^{L}_{i}$, $N$, $N^{-1}$, $\bk_{LM}$, $\d_{t}\phi$, $E_{i}f_{L}^{M}$, $E_{i}e_{L}^{j}$,
  $E_{i}N$, $E_{i}\bk_{LM}$, $E_{i}\phi$, $E_{i}\d_{t}\phi$, $E_{i}E_{j}\phi$ and $E_{i}E_{j}N$. Moreover, $\Upsilon_{IJ}$ is a sum of terms, each of which
  can be written as a polynomial in $f_{L}^{M}$, $e_{L}^{i}$, $\omega^{L}_{i}$ and $\bk_{LM}$ multiplied by a single factor of the form $E_{i}E_{j}e_{K}^{k}$. 
\end{lemma}
\begin{remark}
  Note that if the Einstein-non-linear scalar field equations are satisfied, then $\msC$ introduced in (\ref{eq:msCIdef}) vanishes. Moreover, if
  the CMC condition $\theta=1/t$ holds, then $\vartheta=0$. 
\end{remark}
\begin{remark}\label{remark:Ei str coeff conv}
  In the statement of the lemma, as well as in the proof, we consistently omit reference to dependence on spatial derivatives of the structure
  coefficients of $\{E_i\}$; i.e., $E_{\bfI}\eta^i_{jk}$, where we use the notation introduced in Remark~\ref{remark:depofbgaIJK} and
  Subsection~\ref{ssection:conventionsframe}. In particular, it is tacitly understood that the coefficients of the polynomials $\mfR_{IJ}$ and
  $\Upsilon_{IJ}$ are allowed to depend on expressions of the form $E_{\bfI}\eta^i_{jk}$. 
\end{remark}
\begin{remark}\label{remark:symmetry}%NEW REMARK
  The expressions appearing in (\ref{eq:dtbRIJ}) are largely geometric. However, we later use them in a different, non-geometric, context
  (in particular, below, the notation $\bk_{IJ}$ is used for variables that are only a posteriori demonstrated to be the components of the second
  fundamental form with respect to a frame). In
  this later context, the symmetry properties of the last objects on the right hand side of (\ref{eq:dtbRIJ}) are important. It is therefore important
  to note that in the current context, we can replace $\Upsilon_{IJ}$ and $\mfR_{IJ}$ with $\Upsilon_{(IJ)}$ and $\mfR_{(IJ)}$ respectively, where the
  parentheses denote symmetrisation, since all the terms on the left hand side and the first term on the right hand side of (\ref{eq:dtbRIJ}) are
  symmetric in the geometric context. 
\end{remark}
\begin{proof}
  It is convenient to begin by expressing the spatial Ricci curvature in terms of the structure coefficients $\bga_{IJ}^{K}$ and the connection
  coefficients $\bGa_{IJ}^{K}$. Here $\bGa^{K}_{IJ}$ is defined by the condition that
  \[
    \bna_{e_{I}}e_{J}=\bGa_{IJ}^{K}e_{K}.
  \]
  Due to the Koszul formula, 
  \begin{equation}\label{eq:Koszulexpr}
    \bGa_{IJ}^{K}=\ldr{\bna_{e_{I}}e_{J},e_{K}}=\tfrac{1}{2}(-\bga_{JK}^{I}+\bga_{KI}^{J}+\bga_{IJ}^{K}),
  \end{equation}
  where $\ldr{\cdot,\cdot}:=\bge$ and $[e_{I},e_{J}]=\bga_{IJ}^{K}e_{K}$. Note, for future reference, that $\d_{t}\bga_{IJ}^{K}$ and $\d_{t}\bGa_{IJ}^{K}$
  are polynomials in $f_{L}^{M}$, $E_{i}f_{L}^{M}$, $e_{L}^{i}$, $\omega^{L}_{i}$ and $E_{i}e_{L}^{j}$ (keeping the convention introduced in
  Remark~\ref{remark:Ei str coeff conv} in mind); see (\ref{eq:dtbgaIJK}), (\ref{eq:Koszulexpr}) and
  Remark~\ref{remark:depofbgaIJK}. Moreover, the polynomials are of order one in $E_{i}f_{L}^{M}$ and $E_{i}e_{L}^{j}$. Next, note that 
  \begin{equation}\label{eq:bRIJformula}
    \bR_{IJ}=e_{I}(\bGa^{J}_{MM})-e_{M}(\bGa^{J}_{IM})+\bGa^{L}_{MM}\bGa_{IL}^{J}-\bGa_{IL}^{M}\bGa^{J}_{LM}-\bga^{M}_{IL}\bGa^{J}_{ML},
  \end{equation}
  where we sum over all repeated indices. When time differentiating (\ref{eq:bRIJformula}), it is clear that the only terms that cannot immediately be
  included in $\mfR_{IJ}$ arise when the time derivative hits the first two terms on the right hand side of (\ref{eq:bRIJformula}). In short, we need
  to focus on 
  \begin{equation}\label{eq:bLIJdef}
    \bL_{IJ}:=f_{I}^{L}e_{L}(\bGa^{J}_{MM})-f_{M}^{L}e_{L}(\bGa^{J}_{IM})+e_{I}(\d_{t}\bGa^{J}_{MM})-e_{M}(\d_{t}\bGa^{J}_{IM}). 
  \end{equation}
  Due to Remark~\ref{remark:depofbgaIJK} and (\ref{eq:Koszulexpr}), it is clear that the first two terms on the right hand side give rise to expressions
  that can be included in either $\Upsilon_{IJ}$ or $\mfR_{IJ}$. We therefore need to focus on the last two terms.  Note, to this end, that 
  $\bGa^{J}_{MM}=\bga^{M}_{JM}$. Combining this observation with (\ref{eq:dtbgaIJK}) yields
  \begin{equation}\label{eq:dtbGaJMM}
    \begin{split}
      \d_{t}\bGa^{J}_{MM} = & \d_{t}\bga^{M}_{JM}=e_{J}(f_{M}^{M})-e_{M}(f_{J}^{M})+f_{J}^{K}\bga_{KM}^{M}\\
      = & -e_{J}(N\theta)+e_{M}(N\bk_{MJ})-e_{M}(a_{J}^{M})+f_{J}^{K}\bga_{KM}^{M},
    \end{split}
  \end{equation}
  where $\theta:=\tr_{\bge}\bk$. Rewriting $\theta=\vartheta+1/t$, $-e_{J}(N\theta)=-e_{J}(N)\theta-Ne_{J}(\vartheta)$, so that applying $e_{I}$ to the first
  term on the far right hand side of (\ref{eq:dtbGaJMM}) yields $-N\bna_{I}\bna_{J}\vartheta$ plus expressions that can be included in $\mfR_{IJ}$. The terms
  that result when applying $e_{I}$ to the last term on the far right hand side can, similarly to the first two terms on the right hand side of
  (\ref{eq:bLIJdef}), be included in either $\Upsilon_{IJ}$ or $\mfR_{IJ}$. What remains to consider is $e_{M}(N\bk_{MJ})$ and $-e_{M}(a_{J}^{M})$, or, more
  precisely, $e_{I}e_{M}(N\bk_{MJ})$ and $-e_{I}e_{M}(a_{J}^{M})$; cf. the second last term on the right hand side of (\ref{eq:bLIJdef}). Next, consider
  \begin{equation}\label{eq:dtbGaJIM}
    \begin{split}
      \d_{t}\bGa^{J}_{IM} = & \tfrac{1}{2}(-\d_{t}\bga_{MJ}^{I}+\d_{t}\bga_{JI}^{M}+\d_{t}\bga_{IM}^{J})\\
      = & e_{M}(N\bk_{IJ})-e_{J}(N\bk_{IM})+e_{I}(a_{M}^{J})+\tfrac{1}{2}f_{M}^{K}(\bga_{IK}^{J}-\bga_{KJ}^{I})+\tfrac{1}{2}f_{J}^{K}(\bga^{M}_{KI}-\bga_{MK}^{I})\\
      & +\tfrac{1}{2}f_{I}^{K}(\bga_{JK}^{M}+\bga_{KM}^{J})+\tfrac{1}{2}f^{I}_{K}\bga^{K}_{MJ}
      -\tfrac{1}{2}f_{K}^{M}\bga^{K}_{JI}-\tfrac{1}{2}f_{K}^{J}\bga^{K}_{IM},
    \end{split}
  \end{equation}
  where we appealed to Lemmas~\ref{lemma:frameorth} and \ref{lemma:Jacobi identity}. Similarly to the above, applying $e_{M}$ to the last six terms on
  the far right hand side results in expressions that can be included in either $\Upsilon_{IJ}$ or $\mfR_{IJ}$. Combining (\ref{eq:bLIJdef}),
  (\ref{eq:dtbGaJMM}) and (\ref{eq:dtbGaJIM}) with the above observations yields the conclusion that the important terms in $\d_{t}\bR_{IJ}$ are
  \begin{equation}\label{eq:leadingorderterms}
    \begin{split}
      & -N\bna_{I}\bna_{J}\vartheta +e_{I}e_{M}(N\bk_{MJ})-e_{I}e_{M}(a_{J}^{M})\\
      & -e_{M}e_{M}(N\bk_{IJ})+e_{M}e_{J}(N\bk_{IM})-e_{M}e_{I}(a_{M}^{J})\\
      = & -e_{M}e_{M}(N\bk_{IJ})-N\bna_{I}\bna_{J}\vartheta+e_{I}e_{M}(N\bk_{MJ})\\
      & +e_{J}e_{M}(N\bk_{IM})+\bga_{MI}^{K}e_{K}(a_{J}^{M})+\bga_{MJ}^{K}e_{K}(N\bk_{IM}).
    \end{split}
  \end{equation}
  The last two terms give rise to expressions that can be included in $\mfR_{IJ}$. We therefore only need to focus on the first four terms.
  However, we wish to keep the first two terms; cf. (\ref{eq:dtbRIJ}). This leaves the third and fourth terms on the right hand side of
  (\ref{eq:leadingorderterms}). On the other hand, the sum of these terms can, to leading order, be expressed as a symmetrised gradient of
  $\msC$. Compute, in order to justify this statement, that
  \begin{equation}\label{eq:divbkMMJ}
    e_{M}(\bk_{MJ})=\msC(e_{J})+e_{J}(\theta)+e_{0}(\phi)e_{J}(\phi)+\bGa_{MM}^{K}\bk_{KJ}+\bGa_{MJ}^{K}\bk_{MK}. 
  \end{equation}
  From this equality, we can deduce two things. First, $\msC(e_{L})$ can be included in $\mfR_{IJ}$. Second, applying $e_{I}$ to (\ref{eq:divbkMMJ})
  yields
  \begin{equation*}
    \begin{split}
      e_{I}e_{M}(\bk_{MJ}) = & \bna_{I}\msC_{J}+\bna_{I}\bna_{J}\theta+\msC_{K}\bGa^{K}_{IJ}+\bGa_{IJ}^{K}e_{K}(\theta)\\
      & +e_{I}[e_{0}(\phi)e_{J}(\phi)+\bGa_{MM}^{K}\bk_{KJ}+\bGa_{MJ}^{K}\bk_{MK}].
    \end{split}
  \end{equation*}  
  All the terms on the right hand side but the first two give rise to expressions that can be included in $\Upsilon_{IJ}$ or $\mfR_{IJ}$. Keeping
  in mind that $\bna_{J}\theta=\bna_{J}\vartheta$ and that $\bna_{I}\bna_{J}\vartheta=\bna_{J}\bna_{I}\vartheta$, the lemma follows. 
\end{proof}

Next, we wish to calculate the time derivative of $\dvka_{IJ}$, where
\begin{equation}\label{eq:dvkaIJ}
  \dvka_{IJ}:=\dka_{IJ}+f_{I}^{M}\bk_{MJ}+\bk_{IM}f_{J}^{M}.
\end{equation}
Note that $\dvka_{IJ}$ is given by the far right hand side of (\ref{eq:dtbkIJfv}). 

\begin{lemma}\label{lemma:dtdvka}
  Given assumptions and notation as in Lemma~\ref{lemma:dtbRIJ}, let $\dvka_{IJ}$ be given by (\ref{eq:dvkaIJ}). Then 
  \begin{equation}\label{eq:dtdvkaIJ}
    \d_{t}\dvka_{IJ}+N^{2}\bna_{I}\msC_{J}+N^{2}\bna_{J}\msC_{I}+N^{2}\bna_{I}\bna_{J}\vartheta=Ne_{M}e_{M}(N\bk_{IJ})-(\d_{t}N)\bR_{IJ}+\hOm_{IJ}+\hmfK_{IJ},
  \end{equation}
  where $\hmfK_{IJ}$ is polynomial in $f_{L}^{M}$, $e_{L}^{i}$, $\omega^{L}_{i}$, $N$, $N^{-1}$, $\bk_{LM}$, $\d_{t}N$, $\d_{t}\phi$, $\d_{t}\bk_{LM}$,
  $\d_{t}f_{L}^{M}$, $E_{i}f_{L}^{M}$, $E_{i}e_{L}^{j}$, $E_{i}N$, $E_{i}\bk_{LM}$, $E_{i}\phi$, $E_{i}\d_{t}N$, $E_{i}\d_{t}\phi$, $E_{i}E_{j}\phi$, $E_{i}E_{j}N$,
  $E_{i}E_{j}\d_{t}N$ and $V^{(l)}\circ\phi$ (where the dependence on the last expression is of first order). Moreover, $\hOm_{IJ}$ is a
  sum of terms, each of which can be written as a polynomial in $N$, $f_{L}^{M}$, $e_{L}^{i}$, $\omega^{L}_{i}$ and $\bk_{LM}$ multiplied by
  a single factor of the form $E_{i}E_{j}e_{L}^{k}$. 
\end{lemma}
%\begin{remark}
%  The dependence of $\hmfK_{IJ}$ on $\bk_{LM}$, $\d_{t}\bk_{LM}$ and $E_{i}\bk_{LM}$ is, strictly speaking, superfluous, since
%  $2\bk_{LM}=-N^{-1}(f^{L}_{M}+f^{M}_{L})$. A similar statement applies to the dependence of $\hOm_{IJ}$ on $\bk_{LM}$. 
%\end{remark}
\begin{remark}\label{remark:t and x dependence}
  In the statement and the proof, explicit dependence on $t$ is tacit, and we use the convention introduced in
  Remark~\ref{remark:Ei str coeff conv}. 
\end{remark}
\begin{remark}\label{remark:symmetrisation dt dvka}
  As in Remark~\ref{remark:symmetry}, we can replace $\hOm_{IJ}$ and $\hmfK_{IJ}$ with $\hOm_{(IJ)}$ and $\hmfK_{(IJ)}$ respectively.
\end{remark}
\begin{proof}
  We compute the time derivative of $\dvka_{IJ}$ by considering one term at a time on the far right hand side of (\ref{eq:dtbkIJfv}).
  Note that the first term on the right hand side can be written
  \[
  \bna_{I}\bna_{J}N=e_{I}(e_{J}N)-\bGa_{IJ}^{K}e_{K}(N)=e_{I}^{i}E_{i}(e_{J}^{k})E_{k}N+e_{I}^{i}e_{J}^{k}E_{i}E_{k}N-\bGa_{IJ}^{K}e_{K}(N).
  \]
  Time differentiating this expression, keeping (\ref{eq:dtbGaJIM}) in mind, leads to terms that can be included in $\hmfK_{IJ}$. Time differentiating
  the remaining terms (except for the sixth) on the far right hand side of (\ref{eq:dtbkIJfv}) also leads to terms that can be included in $\hmfK_{IJ}$.
  The term that result when the time derivative hits $N$ in the sixth term is included explicitly on the right hand side. Finally,
  \[
  -N\d_{t}\bR_{IJ}+N^{2}\bna_{I}\msC_{J}+N^{2}\bna_{J}\msC_{I}+N^{2}\bna_{I}\bna_{J}\vartheta=Ne_{M}e_{M}(N\bk_{IJ})-N\Upsilon_{IJ}-N\mfR_{IJ},
  \]
  where we appealed to (\ref{eq:dtbRIJ}). The lemma follows. 
\end{proof}

\subsection{Equations for deviations from the CMC condition and from Einstein's equations}

In the end, we wish to solve the Einstein-non-linear scalar field equations. Moreover, we wish to prove that the foliation we construct has leaves
of constant mean curvature. We thus need to take the step from the equations we actually solve to the equations we wish to solve. This involves, to
begin with, the problem of deriving equations for the relevant deviations and then proving that, given appropriate initial conditions, the deviations
all vanish. We begin by deriving an evolution equation for $\msD_{ij}$ defined in (\ref{eq:msDijdef}). 

\begin{lemma}\label{lemma:ezmsDij}
  Given assumptions and notation as in Lemma~\ref{lemma:dtbRIJ}, let $\msD_{ij}$ be defined by (\ref{eq:msDijdef}) and $\bmsC$ be defined by
  $\bmsC:=\msC+\tfrac{1}{2}d\vartheta$. Assume, moreover, that 
  \begin{equation}\label{eq:assumedeqbkIJ}
    \begin{split}
      e_{0}^{2}(\bk_{IJ}) = & N^{-1}e_{M}e_{M}(N\bk_{IJ})-N^{-2}(\d_{t}N)\bR_{IJ}+N^{-2}\hOm_{(IJ)}\\
      & +N^{-2}\hmfK_{(IJ)}-e_{0}(\ln N)N^{-1}\dvka_{IJ}+N^{-1}f_{I}^{K}\msD_{KJ}+N^{-1}f^{K}_{J}\msD_{IK},
    \end{split}  
  \end{equation}
  where $\hOm_{IJ}$ and $\hmfK_{IJ}$ are the expressions arising in Lemma~\ref{lemma:dtdvka} and $\dvka_{IJ}$ is defined by (\ref{eq:dvkaIJ}). Then
  \begin{equation}\label{eq:msDevoeq}
    e_{0}(\msD_{ij})=\bna_{i}\bmsC_{j}+\bna_{j}\bmsC_{i},
  \end{equation}
  where the lower script Latin indices refer to the frame $\{E_{i}\}$.
\end{lemma}
\begin{remark}
  For future reference, it is convenient to note that (\ref{eq:msDevoeq}) implies that
  \begin{equation}\label{eq:eztrbgemsD}
    e_{0}(\tr_{\bge}\msD)=-2\bk^{ij}\msD_{ij}+2\bna^{i}\bmsC_{i}.
  \end{equation}
\end{remark}
\begin{remark}
    Assuming that the initial data for $\d_{t}\bk_{IJ}$ at some time, say $t=t_{0}$, is given by
    \begin{equation}\label{eq:iddtbkIJ}
      \d_{t}\bk_{IJ}|_{t=t_{0}}=\dvka_{IJ}|_{t=t_{0}},
    \end{equation}
    then (\ref{eq:geometrickneq}) below yields the conclusion that
    \begin{equation}\label{eq:msDvanin}
      \msD_{IJ}|_{t=t_{0}}=0.
    \end{equation}
\end{remark}
\begin{proof}
  Note, to begin with, that, for purely geometric reasons,
  \begin{equation}\label{eq:geometrickneq}
    e_{0}(\bk_{IJ})=N^{-1}\dvka_{IJ}+\msD_{IJ}.
  \end{equation}
  Time differentiating (\ref{eq:geometrickneq}) and appealing to (\ref{eq:dtdvkaIJ}) and Remark~\ref{remark:symmetrisation dt dvka} yields
  \begin{equation*}
    \begin{split}
      e_{0}^{2}(\bk_{IJ}) = & N^{-1}e_{M}e_{M}(N\bk_{IJ})-N^{-2}(\d_{t}N)\bR_{IJ}+N^{-2}\hOm_{(IJ)}+N^{-2}\hmfK_{(IJ)}-\bna_{I}\msC_{J}-\bna_{J}\msC_{I}\\
      & -\bna_{I}\bna_{J}\vartheta+e_{0}(\msD_{IJ})-e_{0}(\ln N)N^{-1}\dvka_{IJ}.
    \end{split}
  \end{equation*}
  Combining this equation with (\ref{eq:assumedeqbkIJ}) yields
  \begin{equation}\label{eq:msDevoeqpre}
    e_{0}(\msD_{IJ})=N^{-1}f_{I}^{K}\msD_{KJ}+N^{-1}\msD_{IK}f^{K}_{J}+\bna_{I}\msC_{J}+\bna_{J}\msC_{I}+\bna_{I}\bna_{J}\vartheta.
  \end{equation}
  Note that this can be viewed as an evolution equation for $\msD_{IJ}$. With respect to the frame $\{E_{i}\}$, this equation reads (\ref{eq:msDevoeq}).
\end{proof}

For future reference, it is convenient to make the following observation.

\begin{lemma}\label{lemma:thebkIJeq}
  Given assumptions and notation as in Lemma~\ref{lemma:dtbRIJ}, (\ref{eq:assumedeqbkIJ}) is equivalent to 
  \begin{equation}\label{eq:thebkIJeq}
    \d_{t}^{2}\bk_{IJ} = N^{2}e^{l}_{K}e^{m}_{K}E_{l}E_{m}\bk_{IJ}+\Omega_{IJ}+\mfK_{IJ},
  \end{equation}
  where $\mfK_{IJ}$ is polynomial in $f_{I}^{J}$, $e_{I}^{i}$, $\omega^{I}_{i}$, $N$, $N^{-1}$, $\bk_{IJ}$, $\d_{t}N$, $\d_{t}\phi$, $\d_{t}\bk_{IJ}$,
  $\d_{t}f_{I}^{J}$, $E_{i}f_{I}^{J}$, $E_{i}e_{I}^{j}$, $E_{i}N$, $E_{i}\bk_{IJ}$, $E_{i}\phi$, $E_{i}\d_{t}N$, $E_{i}\d_{t}\phi$, $E_{i}E_{j}\phi$, $E_{i}E_{j}N$,
  $E_{i}E_{j}\d_{t}N$ and $V^{(l)}\circ\phi$ (where the dependence on the last expression is of first order). Moreover, $\Omega_{IJ}$ is a sum of terms,
  each of which can be written as a polynomial in $N$, $f_{I}^{J}$, $e_{I}^{i}$, $\omega^{I}_{i}$ and $\bk_{IJ}$ multiplied by a single factor of the
  form $E_{i}E_{j}e_{I}^{k}$.
\end{lemma}
\begin{remark}
  Remark~\ref{remark:t and x dependence} is equally relevant in the present setting.
\end{remark}
\begin{remark}\label{remark:thebkIJeq antisymm}%NEW REMARK
  Considering (\ref{eq:thebkIJeq}) as an equation for variables $\bk_{IJ}$, where the $\bk_{IJ}$ are not a priori the components of the second
  fundamental form with respect to a frame (so that, in particular, $\bk_{IJ}$ does not necessarily equal $\bk_{JI}$), the sum of the last two
  terms on the right hand side can be written as a sum of a symmetric term and a sum of terms, each of which contains either a factor of the
  form $\bk_{[LM]}$ or a factor of the form $\d_t\bk_{[LM]}$, where we use the notation introduced in (\ref{eq:symm asymm}). 
\end{remark}
\begin{proof}
  Multiply (\ref{eq:assumedeqbkIJ}) by $N^{2}$. The second last term on the right hand side of the resulting equation is given by
  \begin{equation}\label{eq:lasttwoterms}
    Nf_{I}^{K}\msD_{KJ}=f_{I}^{K}(\d_{t}\bk_{KJ}-\dvka_{KJ}),
  \end{equation}
  where $\dvka_{IJ}$ is given by the far right hand side of (\ref{eq:dtbkIJfv}); see (\ref{eq:geometrickneq}). Considering the far right hand
  side of (\ref{eq:dtbkIJfv}), it is clear that (\ref{eq:lasttwoterms}) gives rise to terms that can all be included in either $\mfK_{IJ}$ or
  $\Omega_{IJ}$. A similar argument yields the same conclusion for $N^{2}$ times the last term on the right hand side of (\ref{eq:assumedeqbkIJ}).
  Next, writing $N^{2}e_{0}^{2}(\bk_{IJ})=\d_{t}^{2}\bk_{IJ}-N^{-1}(\d_{t}N)\d_{t}\bk_{IJ}$ etc., (\ref{eq:assumedeqbkIJ}) yields
  \begin{equation}\label{eq:primebkeq}
    \d_{t}^{2}\bk_{IJ} = N^{2}e^{l}_{K}e^{m}_{K}E_{l}E_{m}\bk_{IJ}-(\d_{t}N)\bR_{IJ}-e_{0}(N)\dvka_{IJ}+\Omega_{IJ}'+\mfK_{IJ}',
  \end{equation}
  where $\Omega_{IJ}'$ and $\mfK_{IJ}'$ have the same dependence as $\mfK_{IJ}$ and $\Omega_{IJ}$ above. The only expression in $-e_{0}(N)\dvka_{IJ}$
  which cannot immediately be included into $\mfK_{IJ}$ is
  \[
  -e_{0}(N)(-N\bR_{IJ})=(\d_{t}N)\bR_{IJ}.
  \]
  Combining this observation with (\ref{eq:primebkeq}) yields the conclusion of the lemma.

  Next, in order to justify Remark~\ref{remark:thebkIJeq antisymm}, we divide $\bk_{IJ}$ into a symmetric and antisymmetric part:
  $\bk_{IJ}=\bk_{(IJ)}+\bk_{[IJ]}$, where
  \begin{equation}\label{eq:symm asymm}
    \bk_{(IJ)}:=\tfrac{1}{2}(\bk_{IJ}+\bk_{JI}),\ \ \
    \bk_{[IJ]}:=\tfrac{1}{2}(\bk_{IJ}-\bk_{JI}).
  \end{equation}
  Next, let $b_{IJ}$ be symmetric; i.e., $b_{IJ}=b_{JI}$. Then $f_I^Kb_{KJ}+f_J^Kb_{IK}$ is symmetric in $I$ and $J$. If $b_{IJ}$ is not symmetric, then
  \[
  f_I^Kb_{KJ}+f_J^Kb_{IK}=(f_I^Kb_{(KJ)}+f_J^Kb_{(IK)})+(f_I^Kb_{[KJ]}+f_J^Kb_{[IK]}).
  \]
  Here the first term on the right hand side is symmetric, and each term in the second term on the right hand side contains a factor of $b_{[LM]}$.
  Applying this observation with $b_{IJ}=\d_t\bk_{IJ}-\dvka_{IJ}$ and recalling that the first, fourth, fifth and sixth terms on the far right hand
  side of (\ref{eq:dtbkIJfv}) are symmetric by construction, it follows that the sum of the last two terms on the right hand side of
  (\ref{eq:assumedeqbkIJ}) can be written as a sum of one expression which is symmetric in $I$ and $J$ and one expression which is a sum of terms,
  each of which contains a factor of the form $\bk_{[LM]}$ or $\d_t\bk_{[LM]}$. The third and fourth terms on the right hand side of
  (\ref{eq:assumedeqbkIJ}) are explicitly symmetric. The second term is symmetric by construction. What remains is to consider the antisymmetric
  part of $\dvka_{IJ}$. However, considering the far right hand side of (\ref{eq:dtbkIJfv}), it is clear that each expression in $\dvka_{[IJ]}$
  contains a factor of the form $\bk_{[LM]}$. 
\end{proof}

\textbf{The lapse equation.} Next, we derive an evolution equation for $\vartheta=\tr_{\bge}\bk-t^{-1}$. However, to obtain such an equation, we need
to impose an equation for the lapse function. Moreover, the evolution equation for $\vartheta$ will depend on a function quantifying the deviation
from the Hamiltonian constraint. 

\begin{lemma}
  Given assumptions and notation as in Lemma~\ref{lemma:dtbRIJ}, let
  \begin{equation}\label{eq:msHdef}
    \msH:=-\bS-t^{-2}+|\bk|_{\bge}^{2}+[e_{0}(\phi)]^{2}+|\bna\phi|_{\bge}^{2}+2V\circ\phi,
  \end{equation}
  where $\bS$ denotes the spatial scalar curvature. Assume, moreover, $N$ to satisfy
  \begin{equation}\label{eq:imposedlapseeq}
    -N^{-1}t^{-2}=N^{-1}\Delta_{\bge}N-|\bk|_{\bge}^{2}-[e_{0}(\phi)]^{2}+\tfrac{2}{n-1}V\circ\phi.
  \end{equation}
  Then
  \begin{equation}\label{eq:varthetaevoeq}
    e_{0}(\vartheta)=-t^{-1}\vartheta+\msD_{MM}+\msH.
  \end{equation}
\end{lemma}
\begin{remark}
  Using the notation (\ref{eq:msHdef}), the normal normal component of Einstein's equations reads $\msH+t^{-2}-\theta^{2}=0$ in the
  Einstein-non-linear scalar field setting; i.e., $\msH+t^{-2}-\theta^{2}=0$ is equivalent to the Hamiltonian constraint being satisfied.
\end{remark}
\begin{proof}
  Using the notation $\theta:=\tr_{\bge}\bk$, note that
  \[
  e_{0}(\theta)=N^{-1}\Delta_{\bge}N-\theta^{2}+\bge^{ij}R_{ij}-\bS;
  \]
  see \cite[(258), p.~63]{RinGeo}. Note that this identity is geometric; it is not based on the assumption that Einstein's equations are
  satisfied. Using the notation introduced in (\ref{eq:msDijdef}), it follows that
  \[
  \bge^{ij}R_{ij}=R_{MM}=\msD_{MM}+\theta\vartheta+|\bna\phi|_{\bge}^{2}+\tfrac{2n}{n-1}V\circ\phi.
  \]
  Moreover, combining the last two equalities with (\ref{eq:msHdef})
  yields
  \[
  e_{0}(\theta)=N^{-1}\Delta_{\bge}N-|\bk|_{\bge}^{2}-[e_{0}(\phi)]^{2}+\tfrac{2}{n-1}V\circ\phi-t^{-1}\vartheta+\msD_{MM}+\msH.
  \]
  In terms of $\vartheta$, this equality reads
  \[
  e_{0}(\vartheta)-N^{-1}t^{-2}=N^{-1}\Delta_{\bge}N-|\bk|_{\bge}^{2}-[e_{0}(\phi)]^{2}+\tfrac{2}{n-1}V\circ\phi-t^{-1}\vartheta+\msD_{MM}+\msH.
  \]
  Combining this (general, geometric) equality with the imposed equation on $N$, i.e. (\ref{eq:imposedlapseeq}), yields the conclusion that
  (\ref{eq:varthetaevoeq}) holds. The lemma follows. 
\end{proof}

\textbf{Equation for the scalar field.} Concerning the scalar field, we assume that it satisfies
\begin{equation}\label{eq:modscfiedeq}
  \Box_{g}\phi+\vartheta e_{0}(\phi)-V'\circ\phi=0.
\end{equation}
This means that the stress energy tensor associated with the scalar field, i.e.,
\[
T:=d\phi\otimes d\phi-\left[\tfrac{1}{2}\nabla^{\a}\phi\nabla_{\a}\phi+V\circ\phi\right]g
\]
satisfies
\begin{equation}\label{eq:divTsf}
  \rodiv_{g} T=-\vartheta e_{0}(\phi)d\phi.
\end{equation}
For future reference, it is convenient to make the following observation.

\begin{lemma}\label{lemma:phiess}
  Given assumptions and notation as in Lemma~\ref{lemma:dtbRIJ}, (\ref{eq:modscfiedeq}) is equivalent to
  \begin{equation}\label{eq:phiess}
    \d_{t}^{2}\phi = N^{2}e_{K}^{i}e_{K}^{j}E_{i}E_{j}\phi+\Phi,
  \end{equation}
  where $\Phi$ is polynomial in $N$, $N^{-1}$, $e^{i}_{I}$, $\omega^{I}_{i}$, $\d_{t}N$, $\d_{t}\phi$, $E_{i}N$, $E_{i}\phi$, $E_{j}e_{I}^{i}$ and $V'\circ\phi$
  (and of order at most one in the last expression). Moreover,
  \begin{equation}\label{eq:Ekphiess}
    \d_{t}^{2}(E_{k}\phi)=N^{2}e_{K}^{i}e_{K}^{j}E_{i}E_{j}(E_{k}\phi)+\Xi_{k}+\Phi_{k},
  \end{equation}
  where $\Phi_{k}$ is polynomial in $N$, $N^{-1}$, $e^{i}_{I}$, $\omega^{I}_{i}$, $\d_{t}N$, $\d_{t}\phi$, $E_{i}N$, $E_{i}\phi$, $E_{j}e_{I}^{i}$,
  $E_{i}\d_{t}N$, $E_{i}\d_{t}\phi$, $E_{i}E_{j}N$, $E_{i}E_{j}\phi$, $V'\circ\phi$ and $V''\circ\phi$ (and of order at most one in the last two
  expressions). Moreover, $\Xi_{k}$  is a sum of terms, each of which can be written as a polynomial in $N$, $e_{I}^{i}$, $\omega^{I}_{i}$ and
  $E_{i}\phi$ multiplied by a single factor of the form $E_{i}E_{j}e_{I}^{k}$. Finally,
  \begin{equation}\label{eq:ezphiess}
    \d_{t}^{2}(e_{0}\phi)=N^{2}e_{K}^{i}e_{K}^{j}E_{i}E_{j}(e_{0}\phi)+\Psi,
  \end{equation}
  where $\Psi$ is polynomial in $N$, $N^{-1}$, $e^{i}_{I}$, $\omega^{I}_{i}$, $f_{J}^{K}$, $\d_{t}N$, $\d_{t}\phi$, $E_{i}N$, $E_{i}\phi$, $E_{j}e_{I}^{i}$,
  $E_{i}f_{J}^{K}$, $E_{i}\d_{t}N$, $E_{i}\d_{t}\phi$, $\d_{t}e_{0}(\phi)$, $E_{i}E_{j}N$, $E_{i}E_{j}\phi$, $V'\circ\phi$ and
  $V''\circ\phi$ (and of order at most one in the last two expressions). 
\end{lemma}
\begin{remark}
  Remark~\ref{remark:t and x dependence} is equally relevant in the present setting.
\end{remark}
\begin{proof}
  Note, first of all, that (\ref{eq:modscfiedeq}) is equivalent to
  \begin{equation}\label{eq:phi}
    e_{0}^{2}(\phi) = e_{K}e_{K}(\phi)-t^{-1}e_{0}(\phi)+e_{K}(\ln N)e_{K}(\phi)-\bga_{KJ}^{J}e_{K}(\phi)-V'\circ\phi,
  \end{equation}
  where $\bga_{KJ}^{L}$ is defined by $[e_{K},e_{J}]=\bga_{KJ}^{L}e_{L}$; we leave the verification of this statement to the reader. Noting that
  $N^{2}e_{0}^{2}(\phi)=\d_{t}^{2}\phi-N^{-1}(\d_{t}N)\d_{t}\phi$ etc., (\ref{eq:phi}) can be written
  \begin{equation}\label{eq:phiadapt}
    \begin{split}
      \d_{t}^{2}\phi = & N^{2}e_{K}^{i}e_{K}^{j}E_{i}E_{j}\phi+N^{-1}(\d_{t}N)\d_{t}\phi+N^{2}e_{K}(e_{K}^{i})E_{i}\phi-t^{-1}N\d_{t}\phi\\
      & +Ne_{K}(N)e_{K}(\phi)-N^{2}\bga_{KJ}^{J}e_{K}(\phi)-N^{2}V'\circ\phi.
    \end{split}
  \end{equation}
  We conclude that (\ref{eq:phiess}) holds. Applying $E_{k}$ to (\ref{eq:phiadapt}) yields (\ref{eq:Ekphiess}). Next, we wish to apply 
  $e_{0}$ to (\ref{eq:phi}). Note, to this end, that
  \begin{equation}\label{eq:ezeKcomm}
    [e_{0},e_{K}]=N^{-1}f_{K}^{J}e_{J}+N^{-1}e_{K}(N)e_{0}.
  \end{equation}
  It can therefore be computed that   
  \begin{equation*}
    \begin{split}
      e_{0}e_{K}e_{K}(\phi) = & e_{K}e_{K}(e_{0}\phi)+N^{-1}f_{K}^{J}e_{J}e_{K}(\phi)+N^{-1}e_{K}(N)e_{0}e_{K}(\phi)\\
      & +e_{K}[N^{-1}f_{K}^{J}e_{J}(\phi)+N^{-1}e_{K}(N)e_{0}(\phi)].
    \end{split}
  \end{equation*}
  After multiplication by $N^{2}$, the second and fourth terms on the right hand side can be included in $\Psi$; and after multiplying by $N^{2}$ and
  appealing to (\ref{eq:ezeKcomm}), the third term can also be included in $\Psi$. Multiplying the first term on the right hand side by
  $N^{2}$ yields the first term on the right hand side of (\ref{eq:ezphiess}) plus terms that can be included in $\Psi$. Next, $N^{2}e_{0}^{3}(\phi)$
  gives rise to the term on the left hand side of (\ref{eq:ezphiess}) plus terms that can be included in $\Psi$. Applying $N^{2}e_{0}$ to the second
  and last terms on the right hand side of (\ref{eq:phi}) yields terms that can be included in $\Psi$. Appealing to (\ref{eq:dtbgaIJK})
  and (\ref{eq:ezeKcomm}), the terms that result when applying $N^{2}e_{0}$ to the third and fourth terms on the right hand side of
  (\ref{eq:phi}) can be included in $\Psi$. 
\end{proof}

\textbf{Equations for $\msD$ and $\msH$.} Next, note that if $G$ denotes the Einstein tensor, i.e., if $G=\mathrm{Ric}-Sg/2$, where $S$ denotes the
scalar curvature of $(M,g)$, then the Bianchi identities imply that $\rodiv_{g} G=0$. Introduce
\begin{equation}\label{eq:msRdef}
  \msR:=G-T.
\end{equation}
Then $\rodiv_{g}\msR=\vartheta e_{0}(\phi)d\phi$ due to the Bianchi identities combined with (\ref{eq:divTsf}), where we assume (\ref{eq:modscfiedeq})
to hold. Before deducing evolution equations for $\msH$ and $\bmsC$ from this equality, let us relate $\msR$ to $\msH$, $\bmsC$, $\vartheta$ and $\msD$.

\begin{lemma}
  Given assumptions and notation as in Lemma~\ref{lemma:dtbRIJ}, let $\msR$ be defined by (\ref{eq:msRdef}). Then
  \begin{equation}\label{eq:msRzzmsRzI}    
    \msR(e_{0},e_{0})=-\tfrac{1}{2}\msH+\tfrac{1}{2}(\theta+t^{-1})\vartheta,\ \ \
    \msR(e_{0},e_{I})=\msC(e_{I})=\bmsC_{I}-\tfrac{1}{2}e_{I}(\vartheta).
  \end{equation}
  Moreover,
  \begin{equation}\label{eq:msRIJdecomp}
    \msR_{IJ}=\msD_{IJ}+\vartheta\bk_{IJ}-\tfrac{1}{2}\msH \de_{IJ}-\tfrac{1}{2}\vartheta^{2}\de_{IJ}-(\tr_{\bge}\msD)\de_{IJ}.
  \end{equation}
\end{lemma}
\begin{proof}
  Due to the definitions of $\msR$, $\msC$, $\bmsC$ and $\msH$, and the Gauß and Codazzi equations
  (cf., e.g., \cite[Proposition~13.3, p. 149]{RinCauchy}), the equalities (\ref{eq:msRzzmsRzI}) hold.
  Next, $\msR_{IJ}=\msD_{IJ}+\vartheta\bk_{IJ}$ for $I\neq J$. It can also be calculated that
  \[
  \msR_{II}=\msD_{II}+\vartheta\bk_{II}+\tfrac{1}{n-1}\tr_{g}\msR
  \]
  (no summation on $I$). Summing over $I$ in this equality yields
  \[
  \tr_{g}\msR-\tfrac{1}{2}\msH+\tfrac{1}{2}(\theta+t^{-1})\vartheta=\tr_{g}\msR+\msR(e_{0},e_{0})=\tr_{\bge}\msD+\vartheta\theta+\tfrac{n}{n-1}\tr_{g}\msR.
  \]
  Combining the above observations yields (\ref{eq:msRIJdecomp}).
\end{proof}
Next, we derive evolution equations for $\bmsC$ and $\msH$.
\begin{lemma}
  Given assumptions and notation as in Lemma~\ref{lemma:dtbRIJ}, assume $\phi$ to satisfy (\ref{eq:modscfiedeq}). Assume, moreover,
  (\ref{eq:assumedeqbkIJ}) and (\ref{eq:imposedlapseeq}) to hold. Then
  \begin{equation}\label{eq:msHevoeq}
    e_{0}(\msH) = \bna_{M}\bna_{M}\vartheta+2\bna_{I}(\ln N)\bna_{I}\vartheta-2\bna_{M}\bmsC_{M}+2\bk_{IJ}\msD_{IJ}
    +A_{\msH}^{I}\bmsC_{I}+B_{\msH}\vartheta,
  \end{equation}
  where $A_{\msH}^{I}$ and $B_{\msH}$ are smooth functions on the domain of definition of the solution. Moreover,
  \begin{equation}\label{eq:ezmsE}
    e_{0}(\msE) = \bna_{M}\bna_{M}\vartheta+2\bna_{I}(\ln N)\bna_{I}\vartheta+A_{\msH}^{I}\bmsC_{I}+B_{\msH}\vartheta,
  \end{equation}
  where
  \begin{equation}\label{eq:msEdef}
    \msE:=\msH+\tr_{\bge}\msD.
  \end{equation}
  Finally,
  \begin{equation}\label{eq:ezbmsCI}
    e_{0}(\bmsC_{I})=\bna_{M}\msD_{MI}-\tfrac{1}{2}\bna_{I}\msD_{MM}+\bk_{IJ}\bna_{J}\vartheta+A^{LM}_{I}\msD_{LM}+A^{M}_{I}\bmsC_{M}+A_{I}\msE+B_{I}\vartheta,
  \end{equation}
  where $A_{I}^{LM}$ etc. are smooth functions on the domain of definition of the solution.
\end{lemma}
\begin{proof}
  As we noted in connection with (\ref{eq:msRdef}), $\rodiv_{g}\msR=\vartheta e_{0}(\phi)d\phi$. In terms of the frame $\{e_{\a}\}$, this equality
  can be written
  \begin{equation}\label{eq:Bianchiidentitiesevolutioneq}
    \begin{split}
      \vartheta e_{0}(\phi)e_{\a}(\phi) = & -(\nabla_{e_{0}}\msR)(e_{0},e_{\a})+(\nabla_{e_{M}}\msR)(e_{M},e_{\a})\\
      = & -e_{0}[\msR(e_{0},e_{\a})]+\msR(\nabla_{e_{0}}e_{0},e_{\a})+\msR(e_{0},\nabla_{e_{0}}e_{\a})+e_{M}[\msR(e_{M},e_{\a})]\\
      & -\msR(\nabla_{e_{M}}e_{M},e_{\a})-\msR(e_{M},\nabla_{e_{M}}e_{\a}).
    \end{split}
  \end{equation}  
  Introducing the notation $\G_{\a\b}^{\g}$ by requiring that $\nabla_{e_{\a}}e_{\b}=\G_{\a\b}^{\g}e_{\g}$, the $0$-component of
  (\ref{eq:Bianchiidentitiesevolutioneq})
  reads
  \begin{equation*}
    \begin{split}
      \vartheta [e_{0}(\phi)]^{2} = & -e_{0}[-\msH/2+(\theta+t^{-1})\vartheta/2]+2\G_{00}^{I}(\bmsC_{I}-\bna_{I}\vartheta/2)
      +e_{M}(\bmsC_{M}-\bna_{M}\vartheta/2)\\
      & -\G_{MM}^{0}[-\msH/2+(\theta+t^{-1})\vartheta/2]-\G_{MM}^{I}(\bmsC_{I}-e_{I}(\vartheta)/2)-\G_{M0}^{I}\msR_{MI},
    \end{split}
  \end{equation*}
  where we appealed to (\ref{eq:msRzzmsRzI}). Appealing to (\ref{eq:varthetaevoeq}) when $e_{0}$ hits $\vartheta$ and ignoring terms that
  contain an undifferentiated factor $\vartheta$ or $\bmsC_{I}$, this equality can be written
  \begin{equation*}
    \begin{split}
      \dots = & e_{0}(\msH)/2-(\theta+t^{-1})\msH/2-(\theta+t^{-1})\msD_{MM}/2-\bna_{I}(\ln N)\bna_{I}\vartheta+\bna_{M}\bmsC_{M}\\
      & -\bna_{M}\bna_{M}\vartheta/2-\G_{MM}^{I}\bna_{I}\vartheta/2+\theta\msH/2+\G_{MM}^{I}\bna_{I}\vartheta/2\\
      & -\bk_{IJ}\msD_{IJ}+\theta\msH/2+\theta\msD_{MM}+\dots,
    \end{split}
  \end{equation*}
  where the dots signify terms that contain an undifferentiated factor $\vartheta$ or $\bmsC_{I}$. Adding up the terms that contain a
  factor $\msH$ results in $\msH\vartheta/2$, which contains a factor $\vartheta$. Adding up the terms that contain a
  factor $\msD_{MM}$ results in $\vartheta\msD_{MM}/2$, which contains a factor $\vartheta$. To conclude, (\ref{eq:msHevoeq}) holds.
  Combining this equality with (\ref{eq:eztrbgemsD}) (which holds if (\ref{eq:assumedeqbkIJ}) is satisfied; cf. Lemma~\ref{lemma:ezmsDij})
  yields (\ref{eq:ezmsE}). 
  
  Next, the $I$-component of
  (\ref{eq:Bianchiidentitiesevolutioneq}) reads
  \begin{equation*}
    \begin{split}
      \vartheta e_{0}(\phi)e_{I}(\phi) = & -e_{0}(\bmsC_{I}-e_{I}(\vartheta)/2)+\G_{00}^{J}\msR_{JI}+\G_{0I}^{0}[-\msH/2+(\theta+t^{-1})\vartheta/2]\\
      & +\G_{0I}^{J}(\bmsC_{J}-e_{J}(\vartheta)/2)+e_{M}(\msR_{MI})-\G_{MM}^{0}(\bmsC_{I}-e_{I}(\vartheta)/2)\\
      & -\G_{MM}^{J}\msR_{JI}-\G_{MI}^{0}(\bmsC_{M}-e_{M}(\vartheta)/2)-\G_{MI}^{J}\msR_{MJ}.
    \end{split}
  \end{equation*}
  Ignoring terms that contain an undifferentiated factor $\vartheta$, $\msE$, $\msD_{IJ}$ or $\bmsC_{I}$ (note that, since $\msH=\msE-\msD_{MM}$,
  we also ignore terms containing an undifferentiated factor $\msH$; and due to (\ref{eq:msRIJdecomp}), we can ignore terms containing an
  undifferentiated factor $\msR_{IJ}$) yields
  \begin{equation}\label{eq:ezbmsCIfstep}
    \begin{split}
      \dots = & -e_{0}(\bmsC_{I})+e_{0}e_{I}(\vartheta)/2-\G_{0I}^{J}e_{J}(\vartheta)/2+e_{M}(\msR_{MI})\\
      & +\G_{MM}^{0}e_{I}(\vartheta)/2+\G_{MI}^{0}e_{M}(\vartheta)/2+\dots,
    \end{split}
  \end{equation}
  where the dots signify the ignored terms. Note that
  \begin{equation*}
    \begin{split}
      e_{0}e_{I}(\vartheta) = & (\nabla_{e_{0}}e_{I}-\nabla_{e_{I}}e_{0})(\vartheta)+e_{I}e_{0}(\vartheta)\\
      = & \G_{0I}^{0}e_{0}(\vartheta)+\G_{0I}^{J}e_{J}(\vartheta)-\G_{I0}^{J}e_{J}(\vartheta)+e_{I}e_{0}(\vartheta).
    \end{split}
  \end{equation*}
  Appealing to (\ref{eq:varthetaevoeq}) yields
  \begin{equation*}
    \begin{split}
      e_{0}e_{I}(\vartheta)/2-\G_{0I}^{J}e_{J}(\vartheta)/2 = &
      -\bk_{IJ}e_{J}(\vartheta)/2-t^{-1}e_{I}(\vartheta)/2+e_{I}(\msD_{MM})/2+e_{I}(\msH)/2+\dots,
    \end{split}
  \end{equation*}
  where the dots signify the ignored terms. Combining this equality with (\ref{eq:msRIJdecomp}) and (\ref{eq:ezbmsCIfstep}) yields
  (\ref{eq:ezbmsCI}).   
\end{proof}

\textbf{Wave equations.} Combining (\ref{eq:varthetaevoeq}), (\ref{eq:msDevoeq}), (\ref{eq:msHevoeq}),
(\ref{eq:ezmsE}) and (\ref{eq:ezbmsCI}) yields a homogeneous system of equations for $\vartheta$, $\msD$, $\msH$,
$\msE$ and $\bmsC$. It is of interest to use this system to prove that if the initial data for $\vartheta$, $\msD$, $\msH$, $\msE$ and
$\bmsC$ vanish, then these quantities vanish on the time interval for which the solution is defined. However, it is not immediately
obvious how to do so. In fact, we need to derive higher order equations for some of the quantities of interest. 

\begin{lemma}
  Given assumptions and notation as in Lemma~\ref{lemma:dtbRIJ}, assume $\phi$ to satisfy (\ref{eq:modscfiedeq}). Assume, moreover,
  (\ref{eq:assumedeqbkIJ}) and (\ref{eq:imposedlapseeq}) to hold. Then
  \begin{equation}\label{eq:ezsqvarth}
    e_{0}^{2}(\vartheta)=\bna_{M}\bna_{M}\vartheta+A^{I}_{\vartheta}\bmsC_{I}+B_{\vartheta}^{\a}e_{\a}(\vartheta)+B_{\vartheta}\vartheta
  \end{equation}
  for some functions $A^{I}_{\vartheta}$, $B_{\vartheta}^{\a}$ and $B_{\vartheta}$ which are smooth on the domain of definition of the solution.
  Moreover,
  \begin{equation}\label{eq:ezsqdtvartheta}
    e_{0}^{2}(\d_{t}\vartheta)=\bna_{M}\bna_{M}\d_{t}\vartheta+\mfL_{\vartheta}[E_{\leq 2}\vartheta,E_{\leq 1}\d_{t}\vartheta,\bmsC,\d_{t}\bmsC],
  \end{equation}
  where $E_{\leq l}\vartheta$ signifies all functions of the form $E_{\bfI}\vartheta$ for $|\bfI|\leq l$. Moreover $\mfL_{\vartheta}$ is linear in
  the arguments appearing inside the bracket. Finally, using analogous notation,
  \begin{equation}\label{eq:ezsqbmsC}
    e_{0}^{2}(\bmsC_{I})=\bna_{M}\bna_{M}\bmsC_{I}
    +\mfL_{\bmsC}[E_{\leq 1}\vartheta,E_{\leq 1}\d_{t}\vartheta,\d_{t}^{2}\vartheta,\d_{t}\bmsC,E_{\leq 1}\bmsC,E_{\leq 1}\msD,\msE].
  \end{equation}
\end{lemma}
\begin{remark}
  Note that (\ref{eq:ezsqvarth}), (\ref{eq:ezsqdtvartheta}) and (\ref{eq:ezsqbmsC}) are quite similar to
  \cite[(14.75), pp.~4420--4421]{rasq}.
\end{remark}
\begin{proof}  
  Note, to begin with, that (\ref{eq:varthetaevoeq}) implies $e_{0}(\vartheta)=\msE-t^{-1}\vartheta$,
  where we used the notation introduced in (\ref{eq:msEdef}). Applying $e_{0}$ to this equality and appealing to (\ref{eq:ezmsE}) yields
  (\ref{eq:ezsqvarth}). Note that this equality can be used to express $e_{0}^{2}(\vartheta)$ in terms of spatial derivatives of $\vartheta$
  of order at most two, time derivatives of $\vartheta$ of order at most one and $\bmsC_{I}$. Differentiating (\ref{eq:ezsqvarth}) with respect
  to time yields (\ref{eq:ezsqdtvartheta}). Finally, applying $e_{0}$ to (\ref{eq:ezbmsCI}) and appealing to (\ref{eq:msDevoeq}),
  (\ref{eq:ezmsE}) and (\ref{eq:ezsqvarth}) yields (\ref{eq:ezsqbmsC}).  
\end{proof}

\textbf{Energy}. In order to prove that we have a solution to Einstein's equations with a CMC foliation, we need to demonstrate that $\vartheta$,
$\msH$ etc. vanish. Here we do so by means of an energy. In order to arrive at the relevant expression for the energy, it is convenient to first
consider the natural energies associated with (\ref{eq:ezsqdtvartheta}) and (\ref{eq:ezsqbmsC}). Let, to this end,
\[
\me_{\vartheta}(t):=\textstyle{\int}_{\bM_{t}}\big[|e_{0}\d_{t}\vartheta|^{2}+|\bna \d_{t}\vartheta|_{\bge}^{2}\big]d\mu_{\bge},\ \ \
\me_{\msC}(t):=\textstyle{\int}_{\bM_{t}}\big[\textstyle{\sum}_{I}|e_{0}\bmsC_{I}|^{2}+\textstyle{\sum}_{I,J}|\bna_{I}\bmsC_{J}|^{2}\big]d\mu_{\bge}.
\]
We also need an energy for the remaining expressions, involving fewer derivatives. Introduce, to this end,
\begin{equation}\label{eq:meroremdef}
  \me_{\rorem}(t):=\textstyle{\int}_{\bM_{t}}\big[|e_{0}\vartheta|^{2}+|\bna\vartheta|_{\bge}^{2}+\textstyle{\sum}_{I}|\bmsC_{I}|^{2}
    +\textstyle{\sum}_{I,J}|\msD_{IJ}|^{2}+\vartheta^{2}+\msH^{2}+\msE^{2}\big]d\mu_{\bge}.
\end{equation}
In the end, it will be convenient to combine these energies into
\begin{equation}\label{eq:meeta}
  \me_{\eta}(t):=\me_{\vartheta}(t)+\me_{\msC}(t)+\eta^{2}\me_{\rorem}(t)
\end{equation}
for some $0<\eta\in\rn{}$. However, as will become clear in a moment, this is not sufficient. 

\begin{lemma}\label{lemma:enstepone}
  Given assumptions and notation as in Lemma~\ref{lemma:dtbRIJ}, assume $\phi$ to satisfy (\ref{eq:modscfiedeq}). Assume, moreover,
  (\ref{eq:assumedeqbkIJ}) and (\ref{eq:imposedlapseeq}) to hold. Then there are non-negative continuous functions $C_{\vartheta}$, $C_{\msC}$ and
  $C_{\rorem}$ on $\mI$ and smooth functions $f^{ij}$ and $h^{LJKi}$ on $M$, for $i,j,L,J,K\in\{1,\dots,n\}$ such that
  \begin{subequations}\label{seq:methmscrem}
    \begin{align}
      \big|\d_{t}\me_{\vartheta}(t)-\textstyle{\int}_{\bM_{t}}f^{ij}E_{i}E_{j}(\vartheta)\d_{t}^{2}\vartheta d\mu_{\bge}\big| \leq & C_{\vartheta}(t)\me_{1}(t),\label{eq:meth}\\
      \big|\d_{t}\me_{\msC}(t)-\textstyle{\int}_{\bM_{t}}h^{LJKi}E_{i}(\msD_{LJ})\d_{t}\bmsC_{K}d\mu_{\bge}\big| \leq & C_{\vartheta}(t)\me_{1}(t),\label{eq:memsC}\\
      |\d_{t}\me_{\rorem}(t)| \leq & C_{\rorem}(t)\me_{1}(t)\label{eq:merorem}      
    \end{align}
  \end{subequations}  
  for all $t\in \mI$; here $\me_{1}$ denotes the energy introduced in (\ref{eq:meeta}) with $\eta=1$. 
\end{lemma}
\begin{remark}
  The functions appearing in the statement of the lemma are allowed to depend on the solution and the coefficients of the equations.
\end{remark}
\begin{proof}
  When differentiating the energies introduced prior to the statement of the lemma, we ignore terms that can be estimated by $C(t)\me_{1}(t)$ for some
  continuous function $C$ on $\mI$; cf. (\ref{seq:methmscrem}). Note, in particular, that when $\d_{t}$ hits $d\mu_{\bge}$ in any of the energies
  introduced above, then the resulting expression can be estimated in this way. Next, note that
  \[
  \d_{t}|e_{0}\d_{t}\vartheta|^{2}=2e_{0}^{2}(\d_{t}\vartheta)\d_{t}^{2}\vartheta
  \]
  and that
  \[
  \d_{t}|\bna \d_{t}\vartheta|_{\bge}^{2}=\d_{t}[\bge^{ij}E_{i}(\d_{t}\vartheta)E_{j}(\d_{t}\vartheta)]
  =(\d_{t}\bge^{ij})E_{i}(\d_{t}\vartheta)E_{j}(\d_{t}\vartheta)+2\bge^{ij}E_{i}(\d_{t}^{2}\vartheta)E_{j}(\d_{t}\vartheta).
  \]
  After integration, the first term on the right hand side can be estimated by $C\me_{1}$. Moreover, after partial integration, the second term gives
  rise to a term in the integrand of the form $-2\d_{t}^{2}\vartheta\Delta_{\bge}(\d_{t}\vartheta)$. To summarise,
  \[
  \d_{t}\me_{\vartheta}(t)
  =2\textstyle{\int}_{\bM_{t}}\big[e_{0}^{2}(\d_{t}\vartheta)-\Delta_{\bge}(\d_{t}\vartheta)\big]\d_{t}^{2}\vartheta d\mu_{\bge}+\dots,
  \]
  where the dots signify terms that can be estimated by $C\me_{1}$. Combining this observation with (\ref{eq:ezsqdtvartheta}) and the definition of the
  energy, we conclude that there are functions $f^{ij}$ such that (\ref{eq:meth}) holds. By a similar argument,
  \begin{equation*}
    \begin{split}
      \d_{t}\me_{\msC}(t) = &
      2\textstyle{\sum}_{I}\textstyle{\int}_{\bM_{t}}\big[e_{0}^{2}(\bmsC_{I})-\bnabla_J\bnabla_J\bmsC_{I}\big]\d_{t}\bmsC_{I}d\mu_{\bge}+\dots,
    \end{split}
  \end{equation*}
  where the dots signify terms that can be estimated by $C\me_1$. Combining this observation with (\ref{eq:ezsqbmsC}) yields smooth 
  functions $h^{IJKi}$ such that (\ref{eq:memsC}) holds.

  Next, due to (\ref{eq:ezsqvarth}) and the definition of the energy,
  \[
  \d_{t}\textstyle{\int}_{\bM_{t}}\big[|e_{0}\vartheta|^{2}+|\bna \vartheta|_{\bge}^{2}\big]d\mu_{\bge}=\dots.
  \]
  Since $\d_{t}\bmsC_{I}$ is included in the energy, it is immediate that
  \[
  \d_{t}\textstyle{\int}_{\bM_{t}}\textstyle{\sum}_{I}|\bmsC_{I}|^{2}d\mu_{\bge}=\dots.
  \]
  Next, consider
  \[
  \d_{t}\textstyle{\sum}_{I,J}\textstyle{\int}_{\bM_{t}}|\msD_{IJ}|^{2}d\mu_{\bge}
  =\d_{t}\textstyle{\int}_{\bM_{t}}\bge^{ij}\bge^{kl}\msD_{ik}\msD_{jl}d\mu_{\bge}=\dots,
  \]
  where we appealed to (\ref{eq:msDevoeq}) and the definition of the energy. Since $\d_t\vartheta$ is included in the energy,
  \[
  \d_{t}\textstyle{\int}_{\bM_{t}}\vartheta^{2}d\mu_{\bge}=\dots.
  \]
  Finally, the arguments concerning the contributions
  from the time derivatives of $\msH$ and $\msE$ are similar, due to (\ref{eq:msHevoeq}), (\ref{eq:ezmsE}), the definition of the energy and
  the fact that $\bna_{M}\bna_{M}\vartheta=\Delta_{\bge}\vartheta$ can be expressed in terms of $\d_{t}^{2}\vartheta$ by appealing to (\ref{eq:ezsqvarth}).
\end{proof}

Next, we need to estimate the contributions from the terms arising from $f^{ij}$ and $h^{LJKi}$. Introduce, to this end,
\begin{equation}\label{eq:hmethmsCdef}
  \hme_{\vartheta}(t):=\me_{\vartheta}(t)+\textstyle{\int}_{\bM_{t}}f^{ij}E_{j}(\vartheta)E_{i}\d_{t}\vartheta d\mu_{\bge},\ \ \
  \hme_{\msC}(t):=\me_{\msC}(t)+\textstyle{\int}_{\bM_{t}}h^{LJKi}\msD_{LJ}E_{i}\bmsC_{K}d\mu_{\bge}. 
\end{equation}

\begin{lemma}
  Given assumptions and notation as in Lemma~\ref{lemma:enstepone}, there are continuous functions $\hC_{\vartheta}$ and $\hC_{\msC}$ on $\mI$ such
  that if $t_{0},t_{1}\in \mI$, then
  \begin{equation}\label{eq:hmethmsCest}
    |\hme_{\vartheta}(t_{1})-\hme_{\vartheta}(t_{0})| \leq  \big|\textstyle{\int}_{t_{0}}^{t_{1}}\hC_{\vartheta}(t)\me_{1}(t)dt\big|,\ \ \
    |\hme_{\msC}(t_{1})-\hme_{\msC}(t_{0})| \leq  \big|\textstyle{\int}_{t_{0}}^{t_{1}}\hC_{\msC}(t)\me_{1}(t)dt\big|.
  \end{equation}
\end{lemma}
\begin{proof}
  Integration by parts yields
  \begin{equation}\label{eq:firstbadtermEFEsat}
    \textstyle{\int}_{t_{0}}^{t_{1}}\textstyle{\int}_{\bM_{t}}f^{ij}E_{i}E_{j}(\vartheta)\d_{t}^{2}\vartheta d\mu_{\bge}dt
    =-\textstyle{\int}_{t_{0}}^{t_{1}}\textstyle{\int}_{\bM_{t}}f^{ij}E_{j}(\vartheta)E_{i}\d_{t}^{2}\vartheta d\mu_{\bge}dt+\dots,
  \end{equation}
  where the dots signify terms that can be estimated by
  \[
  \big|\textstyle{\int}_{t_{0}}^{t_{1}}C(t)\me_1(t)dt\big|.
  \]
  Integrating by parts in the first term on the right hand side of (\ref{eq:firstbadtermEFEsat}) yields
  \[
  -\textstyle{\int}_{\bM_{t_{1}}}f^{ij}E_{j}(\vartheta)E_{i}\d_{t}\vartheta d\mu_{\bge}
  +\textstyle{\int}_{\bM_{t_{0}}}f^{ij}E_{j}(\vartheta)E_{i}\d_{t}\vartheta d\mu_{\bge}+\dots.
  \]
  Combining these observations with (\ref{eq:meth}) and the definition of $\hme_{\vartheta}$ yields the first estimate in (\ref{eq:hmethmsCest}).
  Similarly, integration by parts yields
  \begin{equation}\label{eq:secondbadtermEFEsat}
    \textstyle{\int}_{t_{0}}^{t_{1}}\textstyle{\int}_{\bM_{t}}h^{IJKi}E_{i}(\msD_{IJ})\d_{t}\bmsC_{K}d\mu_{\bge}dt
    =-\textstyle{\int}_{t_{0}}^{t_{1}}\textstyle{\int}_{\bM_{t}}h^{IJKi}\msD_{IJ}\d_{t}E_{i}\bmsC_{K}d\mu_{\bge}dt+\dots.
  \end{equation}
  Integrating by parts with respect to $t$ in the first term on the right hand side of (\ref{eq:secondbadtermEFEsat}) yields
  \[
  -\textstyle{\int}_{\bM_{t_{1}}}h^{IJKi}\msD_{IJ}E_{i}\bmsC_{K}d\mu_{\bge}+\textstyle{\int}_{\bM_{t_{0}}}h^{IJKi}\msD_{IJ}E_{i}\bmsC_{K}d\mu_{\bge}+\dots,
  \]
  where we appealed to (\ref{eq:msDevoeq}). Combining these observations with (\ref{eq:memsC}) and the definition of $\hme_{\msC}$ yields the
  second estimate in (\ref{eq:hmethmsCest}).
\end{proof}

The estimates (\ref{eq:hmethmsCest}) are of the desired form. However, the energies do not control the quantities we wish to prove vanish.
On the other hand, this can easily be remedied.

\begin{lemma}
  Given assumptions and notation as in Lemma~\ref{lemma:enstepone}, let $K$ be a compact subinterval of $\mI$. If $1\leq\eta\in\rn{}$ is large
  enough and
  \[
  \hme_{\eta}(t):=\hme_{\vartheta}(t)+\hme_{\msC}(t)+\eta^{2}\me_{\rorem}(t),
  \]
  cf. (\ref{eq:meroremdef}) and (\ref{eq:hmethmsCdef}), then, using the notation introduced in (\ref{eq:meeta}),
  \begin{equation}\label{eq:mehmeequiv}
    \me_{\eta}(t)/2\leq\hme_{\eta}(t)\leq 2\me_{\eta}(t)
  \end{equation}
  for all $t\in K$. Moreover, there is a constant $C_{K}$ such that if $t_{0},t_{1}\in K$, then
  \begin{equation}\label{eq:hmediff}
    |\hme_{\eta}(t_{1})-\hme_{\eta}(t_{0})|\leq C_{K}\big|\textstyle{\int}_{t_{0}}^{t_{1}}\hme_{\eta}(t)dt\big|.    
  \end{equation}
\end{lemma}
\begin{remark}\label{remark:sotoEin}
  Assume that $\vartheta|_{t_{0}}=0$, $\msH|_{t_{0}}=0$, $\bmsC_{I}|_{t_{0}}=0$ and $\msD_{IJ}|_{t_{0}}=0$ for some $t_{0}\in \mI$. Then $\msE|_{t_{0}}=0$
  due to (\ref{eq:msEdef}). Moreover, $\d_{t}\vartheta|_{t_{0}}=0$ due to (\ref{eq:varthetaevoeq}), so that $\d_{t}^{2}\vartheta|_{t_{0}}=0$ due to
  (\ref{eq:ezsqvarth}). Finally, due to (\ref{eq:ezbmsCI}), $\d_{t}\bmsC_{I}|_{t_{0}}=0$. Due to these observations, $\hme_{\eta}(t_{0})=0$, so that
  $\hme_{\eta}(t)=0$ for all $t$ in the existence interval of the solution; this is an immediate consequence of (\ref{eq:hmediff}) and Gr\"{o}nwall's
  lemma. Combining this fact with (\ref{eq:msRzzmsRzI}) and (\ref{eq:msRIJdecomp}) yields the conclusion
  that $\msR=0$, so that $G=T$; cf. the definition of $\msR$, (\ref{eq:msRdef}). Thus Einstein's equations are fulfilled. Next, since $\vartheta=0$ and
  (\ref{eq:modscfiedeq}) holds, we know that $\Box_{g}\phi-V'\circ\phi=0$; i.e., the non-linear scalar field equations are satisfied. Finally, since
  $\vartheta=0$, the hypersurface $\bM_t$ has constant mean curvature $1/t$. 
\end{remark}
\begin{proof}
  Due to the definition of $\hme_{\vartheta}$ and $\hme_{\msC}$ and the compactness of $K$, the estimate (\ref{eq:mehmeequiv}) follows immediately by
  choosing $\eta$ to be large enough. Given this estimate, (\ref{eq:hmediff}) follows from (\ref{eq:merorem}), (\ref{eq:hmethmsCest}) and the fact
  that $\eta\geq 1$. 
\end{proof}

\section{Local existence and Cauchy stability}\label{section:thesystemlocexist}

The overall logic behind proving local existence of solutions to (\ref{seq:thesystem}) is to prove local existence of solutions to the following
system:
\begin{subequations}\label{seq:thesystemso}
  \begin{align}
    \d_{t}e_{I}^{i} = & f_{I}^{J}e_{J}^{i},\label{eq:eIievoso}\\
    \d_{t}\omega_{i}^{I} = & -f_{J}^{I}\omega_{i}^{J},\label{eq:omIievoso}\\
    \d_{t}^{2}\bk_{IJ} = & N^{2}e^{l}_{K}e^{m}_{K}E_{l}E_{m}\bk_{IJ}+\Omega_{IJ}+\mfK_{IJ},\label{eq:bkIJsosys}\\
    \d_{t}^{2}\phi = & N^{2}e^{l}_{K}e^{m}_{K}E_{l}E_{m}\phi+\Phi,\label{eq:scalarfieldfeso}\\
    \Delta_{\bge}(N-1) = & \zeta (N-1)+\mfN,\label{eq:lapseeqso}    
  \end{align}
\end{subequations}
where $f_I^J$ is given by (\ref{eq:fIJthmdef}); $\Omega_{IJ}$ and $\mfK_{IJ}$ are the expressions obtained in Lemma~\ref{lemma:thebkIJeq}; and
$\Phi$ is the expression obtained in Lemma~\ref{lemma:phiess}. Moreover,
\begin{equation}\label{eq:zetamfNdef}
  \zeta:=\bk_{IJ}\bk_{IJ}+(N^{-1}\d_t\phi)^{2}-\tfrac{2}{n-1}V\circ\phi,\ \ \
  \mfN:=\bk_{IJ}\bk_{IJ}+(N^{-1}\d_t\phi)^{2}-\tfrac{2}{n-1}V\circ\phi-t^{-2}.
\end{equation}
At this point we think of (\ref{seq:thesystemso}) as being devoid of geometric content (with the one exception that we, for convenience,
use $\Delta_{\bge}$ to denote the Laplace operator associated with the metric $\bge$ which is defined by the requirement that $\{e_I\}$ is
an orthonormal frame, where $e_I:=e_I^iE_i$). In particular, $\bk_{IJ}$ need not be symmetric etc. 

\textbf{Initial data for the first order system.}
Assume now that we have CMC initial data for (\ref{seq:thesystem}). This means $e_{I}^{i}|_{t_{0}}$, $\omega^{I}_{i}|_{t_{0}}$, $\bk_{IJ}|_{t_{0}}$,
$N|_{t_{0}}>0$, $\phi|_{t_{0}}$ and $(U\phi)|_{t_{0}}$ are given, where $(U\phi)|_{t_{0}}=(e_0\phi)|_{t_{0}}$ denotes the future directed unit normal
derivative of the scalar field at the initial hypersurface. Moreover, for $t=t_{0}$, these data satisfy $\bk_{II}|_{t_{0}}=1/t_{0}$,
$\bk_{IJ}|_{t_{0}}=\bk_{JI}|_{t_{0}}$, (\ref{eq:lapseeq}), (\ref{eq:momcon}) and (\ref{eq:Hamcon}), where
\begin{equation}\label{eq:bgebkitof}
  e_{I}|_{t_{0}}:=e^{i}_{I}|_{t_{0}}E_{i};\ 
  \omega^{I}|_{t_{0}}:=\omega^{I}_{i}|_{t_{0}}\eta^{i};\ 
  \bge|_{t_{0}}:=\textstyle{\sum}_{I}\omega^{I}|_{t_{0}}\otimes\omega^{I}|_{t_{0}};\ 
  \bk|_{t_{0}}:=\bk_{IJ}|_{t_{0}}\omega^{I}|_{t_{0}}\otimes\omega^{J}|_{t_{0}}.
\end{equation}
Moreover, $\theta|_{t_0}=\bk_{II}|_{t_{0}}$ and $\bna$ and $\bS$ on $\bM_{t_{0}}$ are the Levi-Civita connection and the scalar curvature associated
with $\bge|_{t_{0}}$. In (\ref{eq:bgebkitof}),
$\{\eta^{i}\}$ is the dual basis of the fixed frame $\{E_{i}\}$. In addition, we require $\zeta$ appearing in (\ref{eq:lapseeqso}) to satisfy
$\zeta|_{t_{0}}>0$; note that (\ref{eq:lapseeq}) is the same as (\ref{eq:lapseeqso}). This means that $N|_{t_{0}}$ is uniquely determined by
(\ref{eq:lapseeqso}) and the initial data for the other unknowns, so that it does not have to be included in the initial data. Next, note
that $(\mfN/\zeta)|_{t_{0}}=1-1/(t_{0}^{2}\zeta|_{t_{0}})$, so that
\begin{equation}\label{eq:Ntzlowerbd}
N|_{t_{0}}\geq 1-(\mfN/\zeta)|_{t_{0}}=1/(t_{0}^{2}\zeta|_{t_{0}})>0
\end{equation}
by the maximum principle. Note also that it is understood that
$\omega^{I}_{i}|_{t_{0}}e_{J}^{i}|_{t_{0}}=\de^{I}_{J}$, so that the initial data for $\omega^{I}_{i}$ are determined by the initial data for $e^{i}_{I}$. To
summarise, CMC initial data for (\ref{seq:thesystem}) are given by $e_{I}^{i}|_{t_{0}}$, $\bk_{IJ}|_{t_{0}}$, $\phi|_{t_{0}}$ and $(U\phi)|_{t_{0}}$ satisfying
$\bk_{II}|_{t_{0}}=1/t_{0}$, $\bk_{IJ}|_{t_{0}}=\bk_{JI}|_{t_{0}}$, (\ref{eq:momcon}), (\ref{eq:Hamcon}) and the condition that $\zeta|_{t_{0}}>0$. Given these
initial data, we wish to solve (\ref{seq:thesystem}) locally. 

\textbf{From initial data for the first order system to initial data for the second order system.} 
Given initial data for (\ref{seq:thesystem}), we define initial data for (\ref{seq:thesystemso}) by the data for (\ref{seq:thesystem}) complemented by
$\d_{t}\bk_{KJ}|_{t_{0}}$, given by the right hand side of (\ref{eq:dtbkIJ}), restricted to $\bM_{t_{0}}$. Thus $\d_{t}\bk_{KJ}|_{t_{0}}=\d_{t}\bk_{JK}|_{t_{0}}$.
Due to this choice of initial data,
$\vartheta|_{t_{0}}=0$, $\msH|_{t_{0}}=0$, $\msC|_{t_{0}}=0$ and $\msD_{KJ}|_{t_{0}}=0$ (the latter equality follows from the fact that (\ref{eq:dtbkIJ}) is
satisfied for $t=t_{0}$). Given these initial data for (\ref{seq:thesystemso}), we solve (\ref{seq:thesystemso}) (for which we, for the moment, assume
that there is a good local existence theory). In solving (\ref{seq:thesystemso}) and in the arguments below, it is understood that
\[
\omega^{K}:=\omega^{K}_{i}\eta^{i},\ \ \
\bge:=\textstyle{\sum}_{K}\omega^{K}\otimes\omega^{K},\ \ \ \bk:=\bk_{KJ}\omega^{K}\otimes\omega^{J},\ \ \ g:=-N^{2}dt\otimes dt+\bge. 
\]
Assume the existence interval for the solution to be $\mI\subset (0,\infty)$, yielding a solution on $M=\bM\times \mI$. Assume, moreover, that
$\zeta>0$ on $M$, where $\zeta$ is introduced in (\ref{eq:zetamfNdef}). Then, by the argument given in connection with (\ref{eq:Ntzlowerbd}), $N>0$
on $M$. Then $(M,g)$ is a Lorentz manifold and $g(\d_{t},\d_{t})=-N^{2}$. Next, by our choices, $\bk_{IJ}|_{t_{0}}$ and $\d_{t}\bk_{IJ}|_{t_{0}}$ are
symmetric. Combining this observation with (\ref{eq:bkIJsosys}) and Remark~\ref{remark:thebkIJeq antisymm}, it is clear that $\bk_{[IJ]}$ satisfies a
homogeneous wave equation with vanishing initial data, so that $\bk_{IJ}$ is symmetric. Since $\{e_{K}\}$, where $e_{K}:=e_{K}^{i}E_{i}$, is an orthonormal
frame, Lemma~\ref{lemma:frameorth} applies and implies that $f_{K}^{J}+f_{J}^{K}=-2N\bkappa_{KJ}$, where $\bkappa_{KJ}$ are the components of the second
fundamental form with respect to the frame $\{e_I\}$. On the other hand, due to (\ref{eq:fIJthmdef}) and the symmetry of $\bk_{IJ}$, we also know that 
$f_{K}^{J}+f_{J}^{K}=-2N\bk_{KJ}$. Thus $\bk_{IJ}$ are the components of the second fundamental form with respect to the frame $\{e_I\}$.

\begin{remark}\label{remark:sotososotofo}
  Due to the above observations, in our setting, the conditions of Lemma~\ref{lemma:dtbRIJ} are satisfied. Next, note that (\ref{eq:assumedeqbkIJ})
  is equivalent to
  (\ref{eq:bkIJsosys}); see Lemma~\ref{lemma:thebkIJeq}. Moreover, (\ref{eq:imposedlapseeq}) is an immediate consequence of (\ref{eq:lapseeqso}).
  Finally, (\ref{eq:modscfiedeq}) is equivalent to (\ref{eq:scalarfieldfeso}); see Lemma~\ref{lemma:phiess}. To conclude, the conditions of
  Lemma~\ref{lemma:enstepone} are fulfilled. Due to Remark~\ref{remark:sotoEin}, we conclude that $\msH=0$, $\msC=0$, $\msD=0$ and $\vartheta=0$
  on $M$. Since $\msH=0$, we conclude that (\ref{eq:Hamcon}) holds on $M$; since $\msC=0$, we conclude that (\ref{eq:momcon}) holds on all of
  $M$; and since $\msD=0$, we conclude that (\ref{eq:dtbkIJ}) holds on all of $M$. Combining these observations with the fact that
  (\ref{seq:thesystemso}) holds yields the conclusion that (\ref{seq:thesystem}) holds on $M$. Finally, since $\vartheta=0$, we conclude that
  $\tr_{\bge}\bk=1/t$ on the interval of existence. This demonstrates local existence of solutions to (\ref{seq:thesystem}). However, it is also of
  interest to note that the solutions to (\ref{seq:thesystem}) solve the Einstein-non-linear scalar field system; see Remark~\ref{remark:sotoEin}.
  What remains is to demonstrate local existence of solutions to
  (\ref{seq:thesystemso}). The idea is to do so by appealing to Corollary~\ref{cor:localexmodsys}. However, the system (\ref{seq:thesystemso})
  does not immediately fall into the framework considered in this corollary. The purpose of the next subsection is to reformulate
  (\ref{seq:thesystemso}) in such a way that it does.
\end{remark}

\subsection{Adapting the equations to the general framework}
Before rewriting (\ref{seq:thesystemso}) so that it fits into the general framework, we need to derive an equation for the time derivative of the
lapse function. Note, to this end, that the lapse equation (\ref{eq:lapseeq}) reads
\begin{equation}\label{eq:lapse}
  e_{K}e_{K}(N)-t^{-2}(N-1)=\bga^{J}_{KJ}e_{K}(N)+N\big(\bk_{IJ}\bk_{IJ}-t^{-2}+[e_{0}(\phi)]^{2}-\tfrac{2}{n-1}V\circ\phi\big).
\end{equation}

\begin{lemma}\label{lemma:Ndoteq}
  Given that (\ref{eq:lapseeq}) holds, $\Ndot:=\d_{t}N$ satisfies
  \begin{equation}\label{eq:lapsedot}
    \begin{split}
      e_{K}e_{K}(\Ndot) -t^{-2}\Ndot = & \bga^{J}_{KJ}e_{K}(\Ndot)+\Ndot\big(\bk_{IJ}\bk_{IJ}-t^{-2}+[e_{0}(\phi)]^{2}-\tfrac{2}{n-1}V\circ\phi\big)
      +\mfR_{N},
    \end{split}
  \end{equation}
  where $\mfR_{N}$ is a polynomial in $N$, $E_{i}N$, $E_{i}E_{j}N$, $e_{K}^{i}$, $\omega^{K}_{i}$, $E_{i}e_{K}^{j}$, $f_{K}^{J}$, $E_{i}f_{K}^{J}$,
  $\bk_{JK}$, $\d_{t}\bk_{JK}$, $\d_{t}\phi$, $e_{0}\phi$, $\d_{t}e_{0}\phi$, and $V'\circ\phi$ (which is of first order in the last expression).
  The function $\mfR_{N}$ is also allowed to depend on the spacetime coordinates and this dependence is smooth, assuming $\mI\subset (0,\infty)$.
\end{lemma}
\begin{remark}
  The equation (\ref{eq:lapsedot}) can also be written
  \begin{equation}\label{eq:lapsedotabs}
    \Delta_{\bge}\Ndot =  \zeta\Ndot+\mfR_{N}.
  \end{equation}
  %where $\mfR_{N}$ is a polynomial in $N$, $E_{i}N$, $E_{i}E_{j}N$, $e_{K}^{i}$, $\omega^{K}_{i}$, $E_{i}e_{K}^{j}$, $f_{K}^{J}$, $E_{i}f_{K}^{J}$,
  %$\bk_{JK}$, $\d_{t}\bk_{JK}$, $\d_{t}\phi$, $e_{0}\phi$, $\d_{t}e_{0}\phi$, $V\circ\phi$ and $V'\circ\phi$ (which is of first order in the last
  %two expressions).
\end{remark}
\begin{proof}
  Note, to begin with, that $[\d_{t},e_{K}]=f_{K}^{J}e_{J}$, so that
  \[
  \d_{t}e_{K}e_{K}(N)=e_{K}e_{K}(\Ndot)+f_{K}^{J}e_{J}e_{K}(N)+e_{K}(f_{K}^{J}e_{J}N).
  \]
  %Next,
  %\[
  %  \d_{t}[-t^{-2}(N-1)]=-t^{-2}\Ndot+2t^{-3}(N-1).
  %\]
  Moreover,
  \[
  \d_{t}[\bga^{J}_{KJ}e_{K}(N)]=(\d_{t}\bga^{J}_{KJ})e_{K}(N)+\bga^{J}_{KJ}f_{K}^{L}e_{L}(N)+\bga^{J}_{KJ}e_{K}(\Ndot),
  \]
  where we tacitly expand $\d_{t}\bga^{J}_{KJ}$ by appealing to (\ref{eq:dtbgaIJK}). Adding up these observations yields the desired conclusion. 
\end{proof}

\subsection{Local existence}\label{ssection:locexeinst}

Next, we wish to prove local existence of solutions to (\ref{seq:thesystem}) by proving that (\ref{seq:thesystemso}) can be rewritten so that
Corollary~\ref{cor:localexmodsys} applies. In order to do this, we need to explain what the variables $u$, $v$, $N_{1}$ and $N_{2}$ are in the
present setting. In our setting, $u$ collects all the functions $e_{I}^{i}$ and $\omega^{I}_{i}$; $v$ collects all the functions $\bk_{IJ}$,
$\phi$, $e_{0}\phi$ and $E_{i}\phi$; $N_{1}=N-1$; and $N_{2}=\d_{t}N$. Clearly, an equation of the form (\ref{eq:themodelu}) holds, assuming we make
appropriate assumptions concerning $a_{I}^{J}$; cf. Remark~\ref{remark:fIJbreakdown}. As in the statement of Theorem~\ref{thm:fosyslocalexistence},
$f_{I}^{J}=-N\bk_{IJ}+a_I^J$. The contribution of the first term to $f_{1}$ only depends on $v$ and $N_{1}$. Assuming $a_{K}^{J}$ to only depend
on $e_{I}^{i}$, $\omega^{I}_{i}$, $\bk_{IJ}$, $\phi$, $e_{0}\phi$, $E_{i}\phi$, $N$ and $E_{i}N$, then $f_{1}$ has the right dependence. However, it does
not a priori have the desired boundedness properties. On the other
hand, that can be arranged at a later stage; given initial data, we choose $f_{1}$ so that it coincides with the function arising from the equations
in a neighbourhood of the initial data, but is such that it is constant outside a compact set. We describe the detailed modifications in the proof
of local existence. Turning to (\ref{eq:themodelv}), note that $\phi$ satisfies (\ref{eq:phiess}). Since $\Phi$ on the right hand side of
(\ref{eq:phiess}) can be included in $f_{2}$, we can interpret (\ref{eq:phiess}) as a subsystem of (\ref{eq:themodelv}). Next, consider
(\ref{eq:Ekphiess}). To begin with, $\Phi_{k}$ can be included in $f_{2}$. Moreover, $\Xi_{k}$ can be included in the second term on the right hand
side of (\ref{eq:themodelv}). By similar arguments, it can be verified that $\Psi$ appearing in (\ref{eq:ezphiess}) can be included in $f_{2}$, so that
(\ref{eq:ezphiess}) can be considered to be a subsystem of (\ref{eq:themodelv}). Next, let us turn to (\ref{eq:thebkIJeq}). By arguments similar to
the above, $\mfK_{IJ}$ can be included in $f_{2}$ and $\Omega_{IJ}$ can be included in the second term on the right hand side of (\ref{eq:themodelv}).
Finally, note that the first terms on the right hand sides of (\ref{eq:phiess}), (\ref{eq:Ekphiess}), (\ref{eq:ezphiess}) and (\ref{eq:thebkIJeq})
are such that they are included in the first term on the right hand side of (\ref{eq:themodelv}). Finally, (\ref{eq:lapseeqso}) is of the
form (\ref{eq:themodelN}) and (\ref{eq:lapsedotabs}) is of the form (\ref{eq:themodelNdot}). In fact, $\zeta$ and $f_{3}$ only depend on $v$.
With this information at hand, we are in a position to prove Theorem~\ref{thm:fosyslocalexistence}.

\begin{proof}[Proof of Theorem~\ref{thm:fosyslocalexistence}]
  As mentioned at the beginning of the present subsection, we intend to prove the statement by appealing to Corollary~\ref{cor:localexmodsys}.
  Moreover, the relation between the variables in (\ref{seq:thesystem}) and the variables in (\ref{seq:themodel}) is clarified at the beginning
  of the present subsection. However, we also need to specify initial data for $u$ and $v$. Note, to this end, that given initial data as in
  the statement of the theorem, $N|_{t_{0}}$ is uniquely determined by (\ref{eq:lapseeq}). The initial data for $u$ are already given in the
  statement of the theorem. Moreover, $\bk_{IJ}|_{t_{0}}$, $\phi|_{t_{0}}$, $(\d_{t}\phi)|_{t_{0}}:=N|_{t_{0}}(e_{0}\phi)|_{t_{0}}$, $(E_{i}\phi)|_{t_{0}}$,
  $(\d_{t}E_{i}\phi)|_{t_{0}}$ and $(e_{0}\phi)|_{t_{0}}$ are all given. In addition, $(\d_{t}\bk_{IJ})|_{t_{0}}$ is given by (\ref{eq:dtbkIJ}). The only
  initial datum left to be specified is $(\d_{t}e_{0}\phi)|_{t_{0}}$. However, this quantity is obtained by restricting the right hand side of
  (\ref{eq:phi}) to $t=t_{0}$ and then multiplying the result with $N|_{t_{0}}$. This means that $u|_{t_{0}}$, $v|_{t_{0}}$ and $(\d_{t}v)|_{t_{0}}$
  have all been specified. Moreover, as described at the beginning of the present subsection, the equations we wish to solve take the form
  (\ref{seq:themodel}), and the functions $f_{1}$ etc. have the dependence described in connection with (\ref{seq:themodel}); we refer to this
  as the \textit{original system}. However, the bounds such as (\ref{eq:lambdabd}) etc. also have to be satisfied in order for us to be allowed
  to appeal to Corollary~\ref{cor:localexmodsys}. In order to ensure (\ref{eq:lambdabd}), let $\bh_{1}^{ij}=N^2e_{K}^{i}e_{K}^{j}$ and
  $\bh_{2}^{ij}=e_{K}^{i}e_{K}^{j}$ in a
  neighbourhood of the initial values $N|_{t_{0}}$ and $e_{K}^{i}|_{t_{0}}$, but such that $\bh_{l}^{ij}$ are constant outside a compact set. This means
  that there is a $\lambda>0$ such that (\ref{eq:lambdabd}) is fulfilled. Moreover, the $\bh_{l}^{ij}$ are smooth functions, all of whose derivatives
  are bounded. We can modify the remaining ingredients of (\ref{seq:themodel}) similarly. We refer to the resulting system of equations as the
  \textit{modified system}. Applying Corollary~\ref{cor:localexmodsys} to the modified system yields a local solution to (\ref{seq:themodel}). By
  continuity, this yields a local solution to the original system in an open interval containing the initial time $t_{0}$. 
  
  Next, let us prove uniqueness of solutions to the original system. Assume, to this end, that we have two solutions, say $\sigma_{1}$ and
  $\sigma_{2}$, to the original system on two open intervals, say $\mI_{1}$ and $\mI_{2}$, containing $t_{0}$. Assume, moreover, that the associated
  $\mC_i$, cf. (\ref{eq:mCdef}), satisfy $\mC_i(t)\in\rn{}$ for $t\in\mI_i$. Let $K\subset \mI_{1}\cap \mI_{2}$ be a compact interval
  containing $t_{0}$ in its interior. Then we can modify the original system in such a way that the resulting modified system satisfies the conditions
  of Corollary~\ref{cor:localexmodsys} and coincides with the original system in a neighbourhood of the range of the solutions on $\bM\times K$.
  Then Corollary~\ref{cor:localexmodsys} applies and yields the conclusion that the solutions coincide on $\bM\times K$. Since $K$ was arbitrary,
  we conclude that the solutions coincide on their common domain of definition. This makes it possible to define the maximal interval of existence
  as the union of all the open intervals $\mI$ such that there is a solution on $\bM\times \mI$ with the property that $\mC(t)\in\rn{}$ for $t\in\mI$.

  Let $\mI$ be the maximal interval of existence for solutions to the original system. Next, we wish to prove the continuation criterion. Assume,
  to this end, that $\mI=(t_{-},t_{+})$, that $t_{-}>0$ and that $\mC(t)$ is bounded on $(t_{-},t_{0}]$. This means that $\|u(\cdot,t)\|_{C^{2}}$,
  $\|v(\cdot,t)\|_{C^{2}}$ and $\|\d_{t}v(\cdot,t)\|_{C^{1}}$ are uniformly bounded on $(t_{-},t_{0}]$. Moreover, $\zeta$ is uniformly bounded away from zero
  on $\bM\times (t_{-},t_{0}]$. Combining this observation with the maximum principle, cf. (\ref{eq:Ntzlowerbd}), yields the conclusion that $N_{1}$
  is uniformly bounded away from zero. Next, note that the original $\zeta$ and $f_3$ are bounded in $C^1$ and the original $h_2$ is bounded in
  $C^2$. Moreover, the bounds are uniform for $t\in (t_{-},t_{0}]$. This means that $N_1$ is bounded in $C^{2,1}$, uniformly for $t\in (t_{-},t_{0}]$;
  see Theorem~\ref{thm:globalSchauder}. As a consequence, the original $f_4$ is bounded in $C^{0,1}$, uniformly for $t\in (t_{-},t_{0}]$. Applying
  Theorem~\ref{thm:globalSchauder} to the equation for $N_2$ then yields the conclusion that $N_2$ is bounded in $C^{2,1}$, uniformly for
  $t\in (t_{-},t_{0}]$. For these
  reasons, it is possible to construct a modified system satisfying the conditions of Corollary~\ref{cor:localexmodsys} and such that it coincides
  with the original system on an open neighbourhood of the closure of the range of the solution on $\bM\times (t_{-},t_{0}]$. However, then
  Corollary~\ref{cor:localexmodsys} ensures that it is possible to extend the solution to the modified system beyond $t_{-}$. Moreover, by continuity,
  the solution to the modified system is a solution to the original system in a neighbourhood of $t_{-}$. This contradicts the definition of $t_{-}$
  and proves the continuation criterion to the past. The argument to the future is similar.
  
  Our next goal is to prove that the solution to the original system is a solution to (\ref{seq:thesystem}). The first step is to prove that
  we have a solution to (\ref{seq:thesystemso}). In order to do so, we need to demonstrate that if $v_\phi$, $v_{\phi,i}$ and $v_{\phi,0}$ are the
  components of $v$ corresponding to $\phi$, $E_i\phi$ and $e_0\phi$ respectively, then $E_iv_\phi=v_{\phi,i}$ and $e_0v_\phi=v_{\phi,0}$. We also need
  to demonstrate that $\d_{t}N_{1}=N_{2}$. In order to prove these statements, let us first eliminate the ambiguity inherent in deriving the original
  system; in the derivation, we can think of $E_{i}\phi$ as $E_{i}v_\phi$ or as $v_{\phi,i}$ etc. In the equation for $u$, there is no ambiguity, since
  $f_1$ is only allowed to depend on $u$, $v$, $N_1$ and $E_iN_1$. In other words, in the case of $f_1$, occurrences of $E_i\phi$ and $e_0\phi$ have
  to be interpreted as $v_{\phi,i}$ and $v_{\phi,0}$ respectively. In fact, whenever $f_I^J$ appears, $E_i\phi$ and $e_0\phi$ should be interpreted as
  $v_{\phi,i}$ and $v_{\phi,0}$ respectively (this also applies to the subsystem for $v_{\phi,0}$ and the equation for $N_2$). Next, consider the equation
  for $v_\phi$, considered as a subsystem of the equation for $v$. It arises from (\ref{eq:phiadapt}). Here, $N$ and $\d_tN$ are replaced by $N_1+1$ and
  $N_2$ respectively; and $e^i_I$ and $\omega^I_i$ are replaced by components of $u$. However, every occurrence of $\phi$ is replaced by $v_\phi$. When
  deriving the subsystem for the $v_{\phi,i}$ (corresponding to system of wave equations for $E_i\phi$ obtained by applying $E_k$ to (\ref{eq:phiadapt}),
  see (\ref{eq:Ekphiess})), we interpret $E_i\d_t\phi$ and $\d_t E_i\phi$ as $\d_tv_{\phi,i}$; and $E_jE_i\phi$ as $E_jv_{\phi,i}$. However, occurrences of
  single derivatives, such as $E_i\phi$ and $\d_t\phi$ should be replaced by $E_iv_\phi$ and $\d_tv_\phi$ respectively. In particular, in the subsystem
  of equations for $v_{\phi,i}$, only $u$, $N_1$, $N_2$, $v_\phi$ and $v_{\phi,i}$ occur. Moreover, it follows that $E_iv_\phi$ satisfies the same system of
  wave equations as $v_{\phi,i}$ with the same initial data. This means that $E_iv_\phi=v_{\phi,i}$. Next, consider the equation for $N_1$. It arises from
  (\ref{eq:lapse}). When rewriting this equation in terms of the $u$, $v$ and $N_1$ variables, there is no ambiguity, since $\zeta$ and $f_3$ are only
  allowed to depend on $u$, $v$ and $E_iu$. In particular, we interpret $e_0\phi$ as $v_{\phi,0}$. When
  differentiating (\ref{eq:lapse}) with respect to
  time, time derivatives of $e_I^i$ occur. This leads to expressions of the form $f_I^Je_J^i$, and we interpret $f_I^J$ as described above. 
  On the other hand, when differentiating (\ref{eq:lapse}), $\d_t\phi$ is interpreted as $\d_tv_{\phi}$. Due to these choices, $\d_tN_1$ and
  $N_2$ satisfy the same equation. This means that $N_2=\d_tN_1$, and we define $N:=N_1+1$. Finally, we wish to prove that $e_0v_\phi=v_{\phi,0}$.
  The equation for $v_{\phi,0}$ is derived using (\ref{eq:phi}) as a starting point. Moreover, (\ref{eq:phi}) can be written (\ref{eq:phiadapt}). However,
  once we have rewritten (\ref{eq:phiadapt}) as an equation for $v_\phi$ (in terms of $u$, $v$, $N_1$ and $N_2$), it is no longer obvious what the
  relation between (\ref{eq:phi}) and (\ref{eq:phiadapt}) is. However, now that we know that $\d_tN=\d_tN_1=N_2$, the equation for $v_\phi$ implies that
  (\ref{eq:phi}) holds with $\phi$ replaced by $v_\phi$, and where it is understood that $e_0=N^{-1}\d_t$. When deriving the subsystem corresponding to
  $v_{\phi,0}$, we interpret each occurrence of $f_I^J$ as before. Note that this
  observation also applies when derivatives of $f_I^J$ appear. Moreover, second spatial derivatives of $\phi$ that appear explicitly (as opposed to
  implicitly due to derivatives of $f_I^J$) are interpreted as first derivatives of $v_{\phi,i}$. Moreover, whenever mixed second order space and time
  derivatives (effectively $\d_t E_i\phi$) appear explicitly, we interpret them as $\d_tv_{\phi,i}$.
  In particular, if $e_0e_I\phi$ or $e_Ie_0\phi$ occur, we expand the corresponding expression, and then apply the above convention. If single
  derivatives of $\phi$ occur explicitly, we interpret them as single derivatives applied to $v_\phi$. Only in the expressions corresponding to the
  left hand side and first term on the right hand side of (\ref{eq:themodelv}) and in remaining terms involving second time derivatives (such as
  $e_0^2\phi$) do we introduce $v_{\phi,0}$; in particular, we write $-t^{-1}e_0^2\phi$ as $-t^{-1}N_1^{-1}\d_tv_{\phi,0}$. Due to these choices;
  the fact that $\d_tN=\d_tN_1=N_2$; and the fact that $E_iv_\phi=v_{\phi,i}$, it follows that $e_0v_\phi$ and $v_{\phi,0}$ satisfy the same equation
  with the same initial data. This means that $v_{\phi,0}=e_0v_\phi$. 
  
  To conclude, extracting $e_{I}^{i}$, $\omega_{i}^{I}$, $\bk_{IJ}$, $\phi$ and $N$ from the solution to the original system yields
  a solution to (\ref{seq:thesystemso}). Due to Remark~\ref{remark:sotososotofo} and the observations made at the beginning of the present section,
  we conclude that the solution to the original system yields a solution to (\ref{seq:thesystem}) on the maximal existence interval. Moreover, this
  solution is a solution to the Einstein-non-linear scalar field equations with leaves of constant mean curvature. Finally, $\bk_{IJ}$ are the
  components of the second fundamental form of the hypersurfaces $\bM_{t}$ with respect to the orthonormal frame $\{e_{I}\}$.
  
  What remains to be demonstrated is uniqueness of solutions to (\ref{seq:thesystem}). In order to do so, we prove that solutions to
  (\ref{seq:thesystem}) are also solutions to (\ref{seq:thesystemso}), which yields solutions to what we refer to as the original system in the
  present proof. Considering (\ref{eq:dkaijdefinition}) and (\ref{eq:dvkaIJ}), it is clear that (\ref{eq:dtbkIJ}) is the same as
  $\d_{t}\bk_{IJ}=\dvka_{IJ}$. Combining this observation with (\ref{eq:dtdvkaIJ}) and Remark~\ref{remark:symmetrisation dt dvka} yields
  \[
  \d_{t}^{2}\bk_{IJ}=\d_{t}\dvka_{IJ}=Ne_{M}e_{M}(N\bk_{IJ})-(\d_{t}N)\bR_{IJ}+\hOm_{(IJ)}+\hmfK_{(IJ)}, 
  \]
  where we used the fact that $\msC=0$, see (\ref{eq:momcon}), and the fact that $\vartheta=0$. This equation implies that
  \[
  e_{0}^{2}(\bk_{IJ})=N^{-1}e_{M}e_{M}(N\bk_{IJ})-N^{-2}(\d_{t}N)\bR_{IJ}+N^{-2}\hOm_{(IJ)}+N^{-2}\hmfK_{(IJ)}-e_{0}(\ln N)N^{-1}\d_{t}\bk_{IJ}. 
  \]
  Since $\d_{t}\bk_{IJ}=\dvka_{IJ}$ and $\msD_{IJ}=0$ in the current setting, this implies that (\ref{eq:assumedeqbkIJ}) holds. This
  means that (\ref{eq:bkIJsosys}) holds; see Lemma~\ref{lemma:thebkIJeq}. This means that all the equations in (\ref{seq:thesystemso})
  are satisfied. Appealing to Lemmas~\ref{lemma:phiess} and \ref{lemma:Ndoteq} and proceeding as above demonstrates that we have a
  solution to what is referred to as the original system in the present proof. Moreover, the initial data for this solution is determined
  by the initial data for (and the equations of) (\ref{seq:thesystem}). Due to the uniqueness of solutions to the original system, noted
  above, we conclude that solutions to (\ref{seq:thesystem}) are unique. 
\end{proof}

Finally, we demonstrate Cauchy stability in the case of (\ref{seq:thesystem}).

\begin{proof}[Proof of Theorem~\ref{thm:CauchystabEinstein}]
  The statement is an immediate consequence of Proposition~\ref{prop:CauchyStability}, given the relation between the different systems, explained in
  detail at the beginning of the present subsection and in the proof of Theorem~\ref{thm:fosyslocalexistence}.
\end{proof}

\appendix

\section{Sobolev spaces}\label{appendix:sobolev}

In order not to interrupt the presentation of the main body of the text with technical details, we here summarise some of the notation as well
as some of the observations concerning Sobolev spaces and estimates that we need in this article. 

\subsection{Measurability}
Given a smooth manifold $M$, we begin by defining a $\sigma$-algebra of measurable subsets; cf. \cite[Definition~13.22, p.~204]{RinTop}.
\begin{definition}\label{def:measurability}
  Let $M$ be a smooth $n$-dimensional manifold. Then a subset $A\subseteq M$ is \textit{measurable} if, for every choice of local coordinates
  $(\sfx,U)$, the set $\sfx(A\cap U)$ is a Lebesgue measurable subset of $\rn{n}$. If $X$ is a topological space, $f:M\rightarrow X$ is said to
  be measurable if $f^{-1}(V)$ is measurable for every open set $V\subseteq X$. 
\end{definition}
\begin{remark}
  Note that if $U,V\subseteq\rn{n}$ and $\phi:U\rightarrow V$ is $C^1$, then $\phi$ maps Lebesgue measurable sets to Lebesgue measurable sets;
  cf. \cite[Lemma~13.21, p.~204]{RinTop}. 
\end{remark}
\begin{remark}
  The collection $\msA$ of measurable subsets of $M$ is a $\sigma$-algebra; cf. \cite[Lemma~13.24, p.~204]{RinTop}. Moreover, due to the
  definition (in particular, the properties of the Lebesgue measure), $\msA$ includes all the Borel sets. Finally, if $(\sfx_i,U_i)$ is a countable
  collection of coordinate charts such that the $\{U_i\}$ cover $M$, and if $\sfx_i(A\cap U_i)$ is measurable for all $i$, then $A\in\msA$;
  cf. \cite[Lemma~13.24, p.~204]{RinTop}. 
\end{remark}
Let $(M,h)$ be a smooth Riemannian manifold and $\msA$ be the $\sigma$-algebra of measurable subsets of $M$. We then wish to associate a measure
$\mu_h:\msA\rightarrow [0,\infty]$  with $h$. We do so in several steps. Assume first that $A\in\msA$ and that there are local coordinates $(\sfx,U)$
such that $A\subset U$. We then define
\begin{equation}\label{eq:muh A local}
  \mu_h(A):=\textstyle{\int}_{\sfx(A)} |\det \sfh_\sfx|^{1/2}\circ\sfx^{-1}d\mu,
\end{equation}
where $\mu$ is the standard Lebesgue measure on $\rn{n}$. Moreover, $\sfh_\sfx$ is the matrix valued function on $U$ with components
\begin{equation}\label{eq:sfh sfx def}
  (\sfh_\sfx)_{ij}:=h(\d_{\sfx^i},\d_{\sfx^j}).
\end{equation}
If $(\sfy,V)$ are local coordinates with $A\subseteq V$, then
\[
\textstyle{\int}_{\sfx(A)} |\det \sfh_\sfx|^{1/2}\circ\sfx^{-1}d\mu=\textstyle{\int}_{\sfy(A)} |\det \sfh_\sfy|^{1/2}\circ\sfy^{-1}d\mu
\]
due to the change of variables formula; cf. \cite[Theorem~7.26, pp.~153-154]{rudin}. In particular, the right hand side of (\ref{eq:muh A local}) is independent of the
local coordinates. Moreover, $\mu_h$, given by (\ref{eq:muh A local}), defines a measure on $U$; cf. \cite[Theorem~1.29, p.~23]{rudin}. In particular, considered as
a measure on $U$, $\mu_h$ is countably additive. Next, let $P_i\subseteq M$, $i\in\nn{}$, be such that each $P_i\in\msA$;
$P_i\cap P_j=\varnothing$ for $i\neq j$ and the union of the $P_i$ equals $M$ (i.e., $\{P_i\}$ is a \textit{partition} of $M$); and each $P_i$ is
contained in a coordinate chart, say $(\sfx_i,U_i)$; we leave it to the reader to verify that such covers exist. Define, for any $A\in\msA$,
\begin{equation}\label{eq:muh A}
  \mu_h(A):=\textstyle{\sum}_i\mu_h(A\cap P_i).
\end{equation}
In order to verify that the result is independent of the partition, let $\{Q_i\}$ be another partition of $M$ with analogous properties. Then, due
to the properties we have already verified,
\[
\textstyle{\sum}_i\mu_h(A\cap P_i)=\textstyle{\sum}_i\textstyle{\sum}_j\mu_h(A\cap P_i\cap Q_j)
=\textstyle{\sum}_j\textstyle{\sum}_i\mu_h(A\cap P_i\cap Q_j)=\textstyle{\sum}_j\mu_h(A\cap Q_j),
\]
where we appealed to \cite[Corollary, p.~23]{rudin} in the second step. This means that $\mu_h$ is well defined. We leave it to the reader to verify that
$\mu_h$ is a measure. 

Let $f\in C_c(M)$; i.e., $f$ is a continuous function with compact support in $M$. Assume, moreover, that there is a coordinate chart $(\sfx,U)$ such
that $\supp f\subset U$. Assume, finally, $M$ to be oriented; $(\sfx,U)$ to be a positively oriented chart; and $\Vol_h$ to be the volume form
associated with $h$; cf. \cite[Proposition~15.29, p.~389]{Lee}. We then leave it to the reader to verify that
\begin{equation}\label{eq:int dmuh int vol form}
  \textstyle{\int}_M fd\mu_h=\textstyle{\int}_Mf\Vol_h;
\end{equation}
cf., e.g., \cite[Chapter~16]{Lee} and \cite[Theorem~1.29, p.~23]{rudin}. Combining this observation with a partition of unity implies that
(\ref{eq:int dmuh int vol form}) holds for all $f\in C_c(M)$.

\subsection{$L^q$-spaces} 
Let $r,s\in\nn{}_0$ and $T^{(r,s)}TM$ denote the bundle of mixed tensors of type $(r,s)$; see, e.g.,
\cite[p.~316]{Lee}. A \textit{section} of this bundle is a function $\psi:M\rightarrow T^{(r,s)}TM$ such that the base point of $\psi(p)$ is $p$;
i.e., if $\pi:T^{(r,s)}TM\rightarrow M$ is the natural projection to the base point, then $\pi\circ\psi=\roId_M$. Let $\mT^{(r,s)}(M)$ denote the set
of smooth sections of $T^{(r,s)}TM$. By the definition of measurability, it is clear that a section is measurable if and only if the components of
the section with respect to local coordinates are measurable; cf.
Definition~\ref{def:measurability}. A contraction of a tensor product of two measurable sections is thus measurable. Since $\msA$ includes the
Borel sets, continuous sections are measurable. In what follows, we denote the measurable sections by $\mM^{(r,s)}(M)$.

Let $(M,h)$ be a Riemannian manifold and $u,v\in T^{(r,s)}TM$ with $\pi(u)=\pi(v)$. Then, if $p:=\pi(u)$, 
\[
(u,v)_{h(p)}:=h^{i_1j_1}(p)\cdots h^{i_sj_s}(p)h_{l_1m_1}(p)\cdots h_{l_rm_r}(p)u^{l_1\cdots l_r}_{i_1\cdots i_s}v^{m_1\cdots m_r}_{j_1\cdots j_s}. 
\]
If $\phi$ and $\psi$ are sections of the same type, we define $(\phi,\psi)_h$ to be the function whose value at $p\in M$ is $(\phi(p),\psi(p))_{h(p)}$.
Note that if $\phi$ and $\psi$ are measurable, then $(\phi,\psi)_h$ is a measurable function due to the above observations. Moreover,
$(\phi,\phi)_h\geq 0$. Let $1\leq q\in\ro$. In what follows we say that $\phi\in \mL^q[T^{(r,s)}TM;h]$ if $\phi\in\mM^{(r,s)}(M)$ and
\[
\textstyle{\int}_M|\phi|_h^{q}d\mu_h<\infty,
\]
where $|\phi|_h:=(\phi,\phi)_h^{1/2}$. Moreover, we say that $\phi\in \mL^\infty[T^{(r,s)}TM;h]$ if $\phi\in\mM^{(r,s)}(M)$ and
\[
  \mathrm{ess\ sup}_{M}|\phi|_h<\infty.
\]
Next, we say that $\phi,\psi\in\mM^{(r,s)}(M)$ are equivalent if the set on which they differ has measure zero. Let $L^q[T^{(r,s)}TM;h]$ denote the set
of equivalence classes of elements of $\mL^q[T^{(r,s)}TM;h]$. If $u\in L^q[T^{(r,s)}TM;h]$, $1\leq q<\infty$ and $\phi$ is an element of the equivalence
class $u$, we define
\[
  \|u\|_{L^q(h)}:=\big(\textstyle{\int}_M|\phi|_h^{q}d\mu_h\big)^{1/q}.
\]
If $u\in L^\infty[T^{(r,s)}TM;h]$ and $\phi$ is an element of the equivalence class $u$, we define
\[
  \|u\|_{L^\infty(h)}:=\mathrm{ess\ sup}_{M}|\phi|_h.
\]
We define $\mL^q_{\roloc}[T^{(r,s)}TM;h]$ and $L^q_{\roloc}[T^{(r,s)}TM;h]$ similarly. For the sake of brevity, we often use the simplified notation
$L^2(h)$ etc. Moreover, we will typically not distinguish between functions and equivalence classes of functions that differ on a set of measure
zero. 

\subsection{Weak derivatives}
Assume that $M$ is oriented and that $u\in L^1_{\roloc}[T^{(r,s)}TM;h]$. Then $u$ is said to be \textit{weakly differentiable} if there is
an element, say $D^hu\in L^1_{\roloc}[T^{(r,s+1)}TM;h]$, such that for every $\phi\in \mT^{(r,s+1)}(M)$ with compact support,
\[
\textstyle{\int}_M (D^hu,\phi)_hd\mu_h=-\textstyle{\int}_M(u,\rodiv_h\phi)_hd\mu_h,
\]
where
\[
(\rodiv_h\phi)_{i_1\cdots i_s}^{j_1\dots j_r}=\nabla^{h,i}\phi_{ii_1\cdots i_s}^{j_1\dots j_r}
\]
and $\nabla^h$ denotes the Levi-Civita connection of $(M,h)$. If $u$ is smooth, then $D^hu=\nabla^hu$; this follows from the fact that if
$X$ is a smooth vector field with compact support on $M$, then
\[
\textstyle{\int}_M\rodiv_hX\, d\mu_h=\textstyle{\int}_M\rodiv_hX\, \Vol_h=0.
\]
Higher weak derivatives are defined iteratively and the $k$'th weak derivative is denoted $(D^h)^ku$. Again, if $u$ is smooth, $(D^h)^ku=(\nabla^h)^ku$. 
We leave it to the reader to verify that $u$ is $k$ times weakly differentiable if and only if the components of $u\circ\sfx^{-1}$ with respect to
$\sfx$ are $k$ times weakly differentiable on $\sfx(U)$ for any choice of local coordinates $(\sfx,U)$. In particular, whether a function has a weak
derivative or not does not depend on the choice of Riemannian metric $h$ (though the weak derivative $D^hu$ of course does). Note also that, with
respect to local coordinates, $(D^h)^ku$ is given by replacing derivatives by weak derivatives in the expression for $(\nabla^h)^ku$ with respect to
local coordinates.

\subsection{Sobolev spaces}\label{ssection:sob sp}
Let $1\leq q\in\ro$, $k\in\nn{}_0$ and $u\in L^q[T^{(r,s)}TM;h]$. If $u$ is $k$ times weakly differentiable and
\[
\textstyle{\int}_M|(D^h)^ju|_h^qd\mu_h<\infty
\]
for all $0\leq j\leq k$, then we say that $u\in W^{k,q}[T^{(r,s)}TM;h]$ and use the notation
\[
\|u\|_{W^{k,q}(h)}:=\big(\textstyle{\sum}_{j=0}^k\textstyle{\int}_M|(D^h)^ju|_h^qd\mu_h\big)^{1/q}.
\]
Next, let $k\in\nn{}_0$ and $u\in L^\infty[T^{(r,s)}TM;h]$. If $u$ is $k$ times weakly differentiable and
\[
\|(D^h)^ju\|_{L^\infty(h)}<\infty
\]
for all $0\leq j\leq k$, then we say that $u\in W^{k,\infty}[T^{(r,s)}TM;h]$ and use the notation
\[
\|u\|_{W^{k,\infty}(h)}:=\textstyle{\sum}_{j=0}^k\|(D^h)^ju\|_{L^\infty(h)}.
\]
Define, in addition, $H^{k}[T^{(r,s)}TM;h]:=W^{k,2}[T^{(r,s)}TM;h]$ and, for $u,v\in H^{k}[T^{(r,s)}TM;h]$, 
\[
(u,v)_{H^k(h)}:=\textstyle{\sum}_{j=0}^k\textstyle{\int}_M((D^h)^ju,(D^h)^jv)_hd\mu_h.
\]
For the sake of brevity, we often simply write $H^{k}(h)$.

\textit{Sobolev spaces on closed manifolds.} Let $(M,h)$ be a closed and oriented Riemannian manifold, $1\leq q\in\ro$ and $k\in\nn{}_0$. Let $\bbP$ be a
finite partition of unity $\phi_i$, $i=1,\dots,l$, such that $\supp\phi_i\subset U_i$, where $(\sfx_i,U_i)$ are local coordinates. Define, for
$u\in W^{k,q}[T^{(r,s)}TM;h]$,
\[
\|u\|_{W^{k,q}(M;\bbP)}:=\big(\textstyle{\sum}_{i=1}^l\sum_{i_1,\dots,i_s}\sum_{j_1,\dots,j_r}\sum_{|\a|\leq k}
\int_{\sfx_i(U_i)}(\phi_i|\d^\a u^{j_1\dots j_r}_{i_1\dots i_s}|^q)\circ\sfx_i^{-1}
d\mu\big)^{1/q}.
\]
We leave it as an exercise to verify that the right hand side is finite for $u\in W^{k,q}[T^{(r,s)}TM;h]$. We also leave it as an exercise to verify
that $\|\cdot\|_{W^{k,q}(h)}$ and $\|u\|_{W^{k,q}(M;\bbP)}$ are equivalent, irrespective of the choice of $h$ and $\bbP$, but with constants depending on
$h$ and $\bbP$. Finally, we leave it to the reader to verify that $W^{k,q}[T^{(r,s)}TM;h]$ and $H^{k}[T^{(r,s)}TM;h]$ are Banach and Hilbert spaces,
respectively, and that $\mT^{(r,s)}(M)$ is dense in $W^{k,q}[T^{(r,s)}TM;h]$ for $q<\infty$.

\subsection{Conventions concerning frames}\label{ssection:conventionsframe}
Let $M$ be a manifold. If $U\subset M$ is open; $\{E_i\}_{i=1}^l$ are smooth vector fields on $U$; $\psi\in C^{m}(U)$; $k\leq m$; and
$\bfI=(i_{1},\dots,i_{k})$, where $i_j\in\{1,\dots,l\}$ for $j=1,\dots,k$, then $E_{\bfI}\psi:=E_{i_{1}}\cdots E_{i_{k}}\psi$. If $\bfI=(i_{1},\dots,i_{k})$,
we also use the notation $|\bfI|=k$. In case $\bfI$ is empty, $E_{\bfI}$ denotes multiplication by $1$ and we use the convention that $|\bfI|=0$.

In case $(M,h)$ is a closed $n$-dimensional Riemannian manifold; $\bbE:=\{E_i\}_{i=1}^n$ is a global frame on $M$; and $u\in H^k(h)$, we can define
$E_{\bfI}u$ for $|\bfI|\leq k$ in terms of the weak derivatives of $u$ (as well as the structure coefficients of the frame etc.). For $u\in H^k(h)$,
we then use the notation 
\begin{equation}\label{eq:psiHmdef}
  \|u\|_{H^{k}(M;\bbE)}:=\big(\textstyle{\sum}_{|\bfI|\leq k}\int_{M}|E_{\bfI}u|^{2}d\mu_{h}\big)^{1/2}.
\end{equation}
We leave it to the reader to verify that this norm is equivalent to $\|\cdot\|_{H^k(h)}$.

\subsection{Moser estimates}
In the proof of local existence, one essential tool is Moser estimates. In case $(M,h)$ is a closed Riemannian manifold; $\bbE:=\{E_i\}_{i=1}^n$ is
a global frame on $M$; $\psi_{1},\dots,\psi_{l}\in C^{\infty}(M)$; and $|\bfI_{1}|+\dots+|\bfI_{l}|=m$, the relevant estimate reads
\begin{equation}\label{eq:moserest}
  \|E_{\bfI_{1}}\psi_{1}\cdots E_{\bfI_{l}}\psi_{l}\|_{L^{2}(h)}\leq C\textstyle{\sum}_{j=1}^{l}\|\psi_{j}\|_{H^{m}(M;\bbE)}\prod_{i\neq j}\|\psi_{i}\|_{L^{\infty}(h)},
\end{equation}
where the constant only depends on an upper bound on $|\rodiv_hE_{i}|$, $n$ and $m$ (note that if $\bbE$ is an orthonormal frame with respect to $h$,
then $|\rodiv_hE_{i}|$ can be estimated in terms of the structure coefficients of $\bbE$). This estimate is a special case of \cite[Corollary~B.8, p.~215]{RinWave}. 

\subsection{Modified Moser estimates}\label{ssection:mod moser est}
When deriving elliptic estimates, we, for the sake of completeness, do not want to assume that we have a global frame. Let $(M,h)$ be a closed
Riemannian manifold and $U\subseteq M$ be an open subset on which there is an orthonormal frame $\{E_i\}$ with respect to $h$. Let $\zeta\in C^\infty_c(U)$.
Then we can consider $E_{\zeta,i}:=\zeta E_i$ to be a smooth vector field on $M$. Moreover, the vector fields $E_{\zeta,1},\dots,E_{\zeta,n}$ satisfy the conditions
imposed on the vector fields $W_i$ in \cite[Section~B.1, p.~209]{RinWave} (in the present setting the $W_i$ have no time dependence). Given $\bfI=(i_1,\dots,i_l)$,
we use the notation $E_{\zeta,\bfI}:=E_{\zeta,i_1}\cdots E_{\zeta,i_l}$. Given $u\in H^l(h)$, we use the notation
\[
  \|\mathbf{E}^{\leq l}_{\zeta}u\|_{L^2(h)}:=\big(\textstyle{\sum}_{|\bfI|\leq l}\textstyle{\int}_M|E_{\zeta,\bfI}u|^2d\mu_h\big)^{1/2}.
\]
In case $\psi_{1},\dots,\psi_{l}\in C^{\infty}(M)$ and $|\bfI_{1}|+\dots+|\bfI_{l}|=m$, \cite[Corollary~B.8, p.~215]{RinWave} then implies that
\begin{equation}\label{eq:moserest mod}
  \|E_{\zeta,\bfI_{1}}\psi_{1}\cdots E_{\zeta,\bfI_{l}}\psi_{l}\|_{L^{2}(h)}\leq C\textstyle{\sum}_{j=1}^{l}\|\mathbf{E}^{\leq m}_{\zeta}\psi_j\|_{L^2(h)}\prod_{i\neq j}\|\psi_{i}\|_{L^{\infty}(h)},
\end{equation}
where the constant only depends on an upper bound on $|\rodiv_hE_{\zeta,i}|$, $n$ and $m$ (due to (\ref{eq:rodiv hrefer Ezetai}) below, it follows that
$|\rodiv_hE_{\zeta,i}|$ only depends on $\zeta$ and $\{E_i\}$). %This estimate is a special case of \cite[Corollary~B.8, p.~215]{RinWave}. 

\subsection{Rellich-Kondrakhov theorem, Poincaré inequality}
Next, we recall the Rellich-Kondrakhov theorem.
\begin{thm}[Rellich-Kondrakhov Theorem]\label{thm:Rellich Kondrakhov}
  Let $(M,h)$ be a closed Riemannian manifold of dimension $n\geq 2$. Then, for every $p,q\geq 1$, $k\in\nn{}_0$ and $l\in\nn{}$ satisfying
  $1/p>1/q-l/n$, the embedding $H^{k+l,q}(h)\hookrightarrow H^{k,p}(h)$ is compact. 
\end{thm}
\begin{remark}\label{remark:Rellich Kondrakhov dim one}
  We are also interested in the case $n=1$ and $M=\sn{1}$. In this case, the embedding from $H^1(\sn{1})$ into $C^{0,\a}(\sn{1})$ is compact for
  $\a<1/2$; see, e.g, \cite[Theorem~1.4.13, p.~48]{baer}. In particular, this means that the embedding from $H^1(\sn{1})$ into $L^2(\sn{1})$ is
  compact. Note also that $W^{1,1}(\sn{1})$ is continuously embedded into $C^0(\sn{1})$ and that $W^{1,1}(\sn{1})$ is compactly embedded into
  $L^1(\sn{1})$; see \cite[Proposition~1.2.5, p.~27]{baer}. Combining these observations with \cite[Proposition~1.2.5, p.~27]{baer} yields the
  conclusion that $W^{1,1}(\sn{1})$ is compactly embedded into $L^p(\sn{1})$ for any $p\in[1,\infty)$. 
\end{remark}
\begin{proof}
  The proof of this statement can be found, e.g., in \cite[Theorem~3.6, p.~24]{hebey} and \cite[Theorem~1.2.12, p.~34]{baer}. 
\end{proof}

\section{H\"{o}lder spaces on manifolds}\label{section:hoelder spaces on manifolds}

Let $(M,g)$ be a closed Riemannian manifold and $\gamma:I\rightarrow M$ be a smooth curve. Then $\mT^{(r,s)}(\gamma)$ is the set of smooth functions
$T:I\rightarrow T^{(r,s)}TM$ such that the base point of $T(t)$ is $\gamma(t)$. We refer to this set as the set of \textit{$(r,s)$-tensor fields on $\gamma$}.
Note, in particular, that if $S\in \mT^{(r,s)}(M)$, then $S\circ\g\in \mT^{(r,s)}(\g)$. Conversely, if $t_0\in I$, $\g'(t_0)\neq 0$ and $T\in \mT^{(r,s)}(\gamma)$,
we can locally write $T$ as a restriction of a smooth tensor field on an open neighbourhood of $\g(t_0)$ to $\g$. Next, we have a notion of induced covariant derivative,
just as in the case of vector fields on $\g$; cf., e.g., \cite[Proposition~18, p.~65]{oneill}. In fact, if $\nabla$ denotes the Levi-Civita connection of $(M,g)$,
there is a unique function taking $T\in \mT^{(r,s)}(\gamma)$ to $T'\in \mT^{(r,s)}(\gamma)$, called the \textit{induced covariant derivative}, such that 
\begin{enumerate}
\item $(aT_1+bT_2)'=aT_1'+bT_2'$ for all $a,b\in\ro$ and $T_1,T_2\in \mT^{(r,s)}(\gamma)$,
\item $(hT)'=h'T+hT'$ for all $h\in C^\infty(I)$ and $T\in \mT^{(r,s)}(\gamma)$,
\item $(S\circ\g)'(t)=\nabla_{\g'(t)}S$ for all $t\in I$ and $S\in \mT^{(r,s)}(M)$.
\end{enumerate}
Moreover, the induced covariant derivative has the property that
\begin{equation}\label{eq:Tone Ttwo prod der}
  \tfrac{d}{dt}(T_1,T_2)_g=(T_1',T_2)_g+(T_1,T_2')_g
\end{equation}
for $T_1,T_2\in \mT^{(r,s)}(\gamma)$. The proof of these statements is similar to the proof of \cite[Proposition~18, p.~65]{oneill}. Writing the equation $T'=0$
in local coordinates, it is clear that it is a linear homogeneous system of ODE's for the components of $T$ with respect to the local coordinates. If $a\in I$
and $\xi\in T^{(r,s)}TM$ with base point $\g(a)$, there is thus a unique $T\in\mT^{(r,s)}(\g)$ such that $T'=0$ and $T(a)=\xi$. If $\g$ is injective
and $b\in I$, we then refer to $T(b)$ as the \textit{parallel transport of $\xi$ from $\g(a)$ to $\g(b)$ along $\g$}. Next, let
\[
  \msO_g:=\{(p,q)\in M\times M\, |\, 0<d_g(p,q)<\roinj(M,g)\},
\]
where $\roinj(M,g)$ denotes the injectivity radius of $(M,g)$ and $d_g$ denotes the topological metric on $M$ induced by $g$. If $(p,q)\in\msO_g$, then there is
an $r>0$ satisfying $d_g(p,q)<r<\roinj(M,g)$. This means that the exponential map based at $p$ is a diffeomorphism from the open ball of radius $r$ in
$T_pM$ onto its image. In particular, there is a corresponding geodesic $\sigma_{pq}$ with the property that $\sigma_{pq}(0)=p$ and $\sigma_{pq}(1)=q$; in fact,
$\sigma_{pq}$ is the unique length minimising geodesic from $p$ to $q$ in $M$; see \cite[Proposition~16, p.~134]{oneill}.
If $(p,q)\in\msO_g$ and $\xi\in T^{(r,s)}TM$ with base point $p$, we, in what follows, denote by $P^{(r,s)}[\sigma_{pq}](\xi)$ the parallel transport of $\xi$ from $p$ to
$q$ along $\sigma_{pq}$. Let $k\in\nn{}_0$, $\a\in (0,1]$, $r,s\in\nn{}_0$ and assume $S$ to be a $C^k$ section of $T^{(r,s)}TM$. Then we define
\begin{align}
  \|S\|_{C^{k}(g)} := & \textstyle{\sup}_{p\in M}\sum_{j\leq k}|((\nabla^g)^jS)(p)|_{g},\label{eq:Ck g def}\\
  [S]_{C^{k,\a}(g)} := & \textstyle{\sup}_{(p,q)\in \msO_g}\tfrac{|P^{(r,s+k)}[\sigma_{pq}](((\nabla^g)^kS)(p))-((\nabla^g)^kS)(q)|_g}{[d_g(p,q)]^\a},\label{eq:Ckalpha brackets g}\\
  \|S\|_{C^{k,\a}(g)} := & \|S\|_{C^{k}(g)}+[S]_{C^{k,\a}(g)}. \label{eq:Ckalpha g}
\end{align}
If $\|S\|_{C^{k,\a}(g)}$ is bounded, we say that $S\in C^{k,\a}(g)$. When deriving Schauder estimates, it is convenient to use an equivalent norm, expressed in terms of
local coordinates. Next, we prove that there are appropriate coverings by coordinate neighbourhoods that can be used to define such a norm. 

\begin{lemma}\label{lemma:local coord Holder}
  Let $(M,g)$ be a closed Riemannian manifold. Given $\e\in (0,1/2)$, there is a finite cover $\{U_i\}_{i=1}^{l_\e}$ of $M$ and local coordinates
  $\sfx_i:U_i\rightarrow\rn{n}$ with the following properties: the Christoffel symbols of $g$ with respect to the coordinates $(\sfx_i,U_i)$ are
  bounded in absolute value by $\e$ on $U_i$ for all $i\in\{1,\dots,l_\e\}$;
  \begin{equation}\label{eq:sfgxi minus delta}
    \|\sfg_{\sfx_i}\circ\sfx_i^{-1}-\roId\|_{C^0(\sfx_i(U_i))}\leq \e,\ \ \
    \|\sfg_{\sfx_i}^{-1}\circ\sfx_i^{-1}-\roId\|_{C^0(\sfx_i(U_i))}\leq \e
  \end{equation}
  for all $i\in\{1,\dots,l_\e\}$; the estimate
  \begin{equation}\label{eq:dpq sfxi p q equiv}
    (1-\e)d(p,q)\leq |\sfx_i(p)-\sfx_i(q)|\leq (1+\e)d(p,q)
  \end{equation}
  holds for all $p,q\in U_i$ and $i\in\{1,\dots,l_\e\}$; for any $k\in\nn{}_0$, there is a constant $C_k\in\ro$ such that
  \begin{equation}\label{eq:sfg sfx Ck bd}
    \|\sfg_{\sfx_i}\circ\sfx_i^{-1}\|_{C^k(\sfx_i(U))}\leq C_k
  \end{equation}
  for all $i\in\{1,\dots,l_\e\}$; and there is an $r>0$ such that for any $p\in M$, the ball $B_r(p)$ is contained in one of the neighbourhoods
  $U_i$.
\end{lemma}
\begin{remark}
  The coordinates $(\sfx_i,U_i)$ that we construct in the proof are normal coordinates. However, this is not of importance in what follows. In fact,
  we wish to work with coordinates having the properties listed in the lemma but which are not necessarily normal coordinates. 
\end{remark}
\begin{remark}
  For the sake of brevity, we use the notation $d:=d_g$ in the statement of the lemma and the proof. The notation $\sfg_{\sfx_i}$ is introduced in
  (\ref{eq:sfh sfx def}). Note also that the last statement of the lemma is an immediate consequence of the assumptions and the Lebesgue Lemma;
  cf., e.g., \cite[Lemma~9.11, p.~28]{bredon}. However, we state it as a property of the covering for the sake of convenience. 
\end{remark}
\begin{proof}
  One essential feature of the coverings we wish to construct is the estimate (\ref{eq:dpq sfxi p q equiv}). A natural class of coordinate systems
  in which to obtain it is normal coordinates. If $\sfx$ are normal coordinates, based at $p$, on a ball $B_r(p)$, then
  \begin{equation}\label{eq:dpq ito lc}
    d(p,q)=|\sfx(q)|=|\sfx(q)-\sfx(p)|
  \end{equation}
  for $q\in B_r(p)$; see below. Since a convex neighbourhood is a normal neighbourhood of each of its points, it is natural, as a next step, to use
  the corresponding family of normal neighbourhoods. However, we only obtain (\ref{eq:dpq ito lc}) if one of the points $p$, $q$ is the base point
  of the normal coordinates. Next, it is therefore natural to try to obtain uniform control of the family of normal coordinates in (a subset of) a convex
  neighbourhood in terms of one of the normal coordinate systems.

  \textit{Constructing a family of normal neighbourhoods.} To begin with, let $(\sfy_i,V_i)$, $i=1,\dots,l$, be a finite covering of $M$ by local
  coordinate systems. Let
  $\{E_{i,k}\}_{k=1}^n$ be an orthonormal frame on $V_i$ with respect to $g$ for $i=1,\dots,l$. For each $x\in M$, there is a convex neighbourhood of $x$
  with respect to $g$, say $\msC$, contained in one of the $V_i$; cf. \cite[Proposition~7, p.~130]{oneill}. Note that the map taking $p,q\in\msC$ to
  $\sigma'_{pq}(0)$ is a smooth map from $\msC\times\msC$ to $TM$, where $\sigma_{pq}$ is the unique geodesic in $\msC$ such that $\sigma_{pq}(0)=p$ and
  $\sigma_{pq}(1)=q$; cf. \cite[Lemma~9, p.~131]{oneill}. For this reason, we obtain a smooth family of normal coordinates, say $\sfx_{p}$ for
  $p\in\msC$, defined as follows. For $p,q\in\msC$, let
  \[
    \sfx_{p}(q):=[\ldr{\sigma'_{pq}(0),E_{i,1}|_{p}},\dots,\ldr{\sigma'_{pq}(0),E_{i,n}|_{p}}],
  \]
  where $\ldr{\cdot,\cdot}:=g$. Next, note that if $p\in\msC$ and $\rho>0$ are such that $B_{\rho}(p)\subset \msC$, then $B_{\rho}(p)$ is a normal
  neighbourhood of $p$ (since $\msC$ is a normal neighbourhood of $p$, by definition, and the relevant subset of $T_{p}M$ corresponding to $B_{\rho}(p)$
  is starshaped). Moreover, if $B_{\rho}(p)$ is normal, then $d(p,q)=|\sigma'_{pq}(0)|_{g}=|\sfx_{p}(q)|$ for all $q\in B_{\rho}(p)$; cf.
  \cite[Proposition~16, p.~134]{oneill}. Here $|\cdot|$ represents the standard norm in $\rn{n}$.

  \textit{Uniform control of the family in terms of one coordinate system.} Next, note that if $o\in\msC$, then
  \[
    \sfx_{p}(q)=\sfx_{p}(q)-\sfx_{p}(p)=\sfx_{p}\circ\sfx_{o}^{-1}\circ\sfx_{o}(q)-\sfx_{p}\circ\sfx_{o}^{-1}\circ\sfx_{o}(p).
  \]
  Assume now that $B_{2r}(o)\subset\msC$ and that $p,q\in B_{2r}(o)$ for some $r>0$, so that the straight line between $\sfx_{o}(p)$ and $\sfx_{o}(q)$ is
  a subset of the ball of radius $2r$ and center at the origin in $\rn{n}$. Let $\xi(t):=t\sfx_{o}(q)+(1-t)\sfx_{o}(p)$. Then
  \[
    \sfx_{p}\circ\sfx_{o}^{-1}\circ\sfx_{o}(q)-\sfx_{p}\circ\sfx_{o}^{-1}\circ\sfx_{o}(p)=\textstyle{\int}_{0}^{1}D(\sfx_{p}\circ\sfx_{o}^{-1})(\xi(t))dt
    [\sfx_{o}(q)-\sfx_{o}(p)].
  \]
  Next, note that for $p=o$, $D(\sfx_{p}\circ\sfx_{o}^{-1})$ is the identity mapping on $\rn{n}$. Letting $\g$ be a length minimising geodesic with
  $\g(0)=o$ and $\g(1)=p$, it then follows that 
  \[
    D(\sfx_{p}\circ\sfx_{o}^{-1})-\mathrm{Id}=\textstyle{\int}_{0}^{1}\tfrac{d}{dt}\left[D(\sfx_{\g(t)}\circ\sfx_{o}^{-1})\right]dt.
  \]
  However, since $\sfx_{p}(q)$ is smooth as a function of both $p$ and $q$, there is, for each pair of compact subsets $K_0,K_1\subset\msC$, a constant $C$ such
  that
  \begin{equation}\label{eq:Dvariation}
    \big\|\tfrac{d}{dt}\big[D(\sfx_{\g(t)}\circ\sfx_{o}^{-1})\big](r)\big\|\leq C|\g'(t)|_{g},
  \end{equation}
  assuming $\g(t)\in K_0$ and $r\in K_1$, where $\|\cdot\|$ denotes the standard norm on $L(\rn{n})$. This means that, if $\g(t)\in K_0$ for $t\in [0,1]$ and $r\in K_1$, 
  \begin{equation}\label{eq:bsfxpbsfxoest}
    \|D(\sfx_{p}\circ\sfx_{o}^{-1})(r)-\mathrm{Id}\|\leq Cd(o,p).
  \end{equation}
  
  \textit{Constructing a covering.} Let $x\in M$ and let $\msC$ be a convex neighbourhood of $x$ as above. In particular, $\msC\subset V_i$ for some
  $i\in\{1,\dots,l\}$. Let now $r_{0}\in (0,1)$ be such that $K_1:=\overline{B}_{2r_{0}}(x)$ is compact and contained in $\msC$ and let $K_0:=\overline{B}_{r_{0}}(x)$.
  Note that $K_0$ is a compact subset of $\msC$. Let $C\in (0,\infty)$ be the associated constant such that (\ref{eq:Dvariation}) holds for $o=x$,
  $\g(t)\in K_0$ and $r\in K_1$. Next, let $r\leq r_{0}$ be such that $rC<\e$. Letting $o=x$ in the above arguments, and using the notation $\sfx:=\sfx_{x}$, 
  \[
    \sfx_{p}(q)=[\sfx(q)-\sfx(p)]+\textstyle{\int}_{0}^{1}\left[D(\sfx_{p}\circ\sfx^{-1})(\xi(t))-\mathrm{Id}\right]dt\ [\sfx(q)-\sfx(p)]
  \]
  for $p,q\in B_{2r}(x)$. Appealing to (\ref{eq:bsfxpbsfxoest}) and the condition on $r$ yields, in addition,
  \begin{equation}\label{eq:EuclRiemequiv}
    (1-\e)|\sfx(q)-\sfx(p)|\leq d(p,q)\leq (1+\e)|\sfx(q)-\sfx(p)|
  \end{equation}
  for $p\in K:=\overline{B}_{r}(x)$ and $q\in B_{2r}(x)$. Next, by decreasing $r$, if necessary, we can assume the absolute value of the Christoffel
  symbols of $g$ with respect to $\sfx$ to be bounded by $\e$ on $B_{r}(x)$; cf. \cite[Proposition~33, p.~73]{oneill}. Similarly, we can assume estimates
  analogous to (\ref{eq:sfgxi minus delta}) and (\ref{eq:sfg sfx Ck bd}) to hold. For each $x\in M$, we thus obtain local coordinates
  $(\sfx_x,B_{r(x)}(x))$. By compactness of $M$, this covering can be reduced to a finite covering with the properties stated in the lemma. 
\end{proof}

Given a covering of the type constructed in Lemma~\ref{lemma:local coord Holder}, we can define alternate H\"{o}lder norms.

\begin{definition}\label{def:Hoelder bbU}
  Let $(M,g)$ be a closed Riemannian manifold; let $\bbU:=\{U_i\}$ be a covering of the type constructed in Lemma~\ref{lemma:local coord Holder};
  and let $\sfx_i$ be the corresponding local coordinates on $U_i$. Define
  \begin{subequations}\label{seq: C k alpha M bbU}
    \begin{align}
      \|T\|_{C^{k}(M;\bbU)} := &
      \sup_{i\in\{1,\dots,l_\e\}}\sup_{\underset{i_1,\dots,i_s}{\overset{j_1,\dots,j_r}{}}}\|{}^{\sfx_i}T^{j_1\cdots j_r}_{i_1\cdots i_s}\circ\sfx_i^{-1}\|_{C^k(\sfx_i(U_i))},
      \label{eq:Ck bbU}\\
            [T]_{C^{k,\a}(M;\bbU)} := & \sup_{i\in\{1,\dots,l_\e\}}\sup_{\underset{i_1,\dots,i_s}{\overset{j_1,\dots,j_r}{}}}\sup_{|\b|=k}\sup_{\underset{p,q\in U_i}{\overset{p\neq q}{}}}
            \tfrac{|\d^\b({}^{\sfx_i}T^{j_1\cdots j_r}_{i_1\cdots i_s}\circ\sfx_i^{-1})(\sfx_i(p))-\d^\b({}^{\sfx_i}T^{j_1\cdots j_r}_{i_1\cdots i_s}
              \circ\sfx_i^{-1})(\sfx_i(q))|}{|\sfx_i(p)-\sfx_i(q)|^\a},\label{eq:Ckalpha sn bbU}\\
            \|T\|_{C^{k,\a}(M;\bbU)} := & \|T\|_{C^{k}(M;\bbU)}+[T]_{C^{k,\a}(M;\bbU)}\label{eq:Ckalpha bbU}
    \end{align}
  \end{subequations}  
  for $T\in C^{k,\a}(g)$. 
\end{definition}
\begin{remark}
  In the definition, we use the convention that if $(\sfx,U)$ are local coordinates and $T$ is a section of $T^{(r,s)}TM$, then
  \[
    {}^{\sfx}T^{j_1\cdots j_r}_{i_1\cdots i_s}:=T(\d_{\sfx^{i_1}},\dots,\d_{\sfx^{i_s}},d\sfx^{j_1},\dots,d\sfx^{j_r}).
  \]
\end{remark}
Given this definition, we can prove the following equivalence.
\begin{lemma}\label{lemma:equiv Hoelder norms}
  Let $(M,g)$ be a closed Riemannian manifold; let $\bbU:=\{U_i\}$ be a covering of the type constructed in Lemma~\ref{lemma:local coord Holder};
  and let $\sfx_i$ be the corresponding local coordinates on $U_i$. Then, given the notation introduced in
  Definition~\ref{def:Hoelder bbU}, there is a constant $C>1$ such that
  \begin{equation}\label{eq:Holder norms equiv}
    C^{-1}\|S\|_{C^{k,\a}(M;\bbU)}\leq \|S\|_{C^{k,\a}(g)}\leq C\|S\|_{C^{k,\a}(M;\bbU)}
  \end{equation}
  for all $S\in C^{k,\a}(g)$. 
\end{lemma}
\begin{remark}
  Due to (\ref{eq:Holder norms equiv}), it is clear that $C^{k,\a}(g)$ is complete; i.e., that it is a Banach space. 
\end{remark}
\begin{proof}
  In what follows, we, for the sake of brevity, write $\nabla$ and $d$ instead of $\nabla^g$ and $d_g$, respectively. We leave it to the reader to
  verify that the norms $\|\cdot\|_{C^{k}(M;\bbU)}$ and $\|\cdot\|_{C^{k}(g)}$ are equivalent. In order to estimate $[S]_{C^{k,\a}(g)}$, see
  (\ref{eq:Ckalpha brackets g}), in terms of $\|S\|_{C^{k,\a}(M;\bbU)}$, note first that if $d(p,q)\geq r>0$ (where $r$ is the number whose existence is
  guaranteed in Lemma~\ref{lemma:local coord Holder}), then
  \begin{equation*}
    \begin{split}
      \tfrac{|P^{(r,s+k)}[\sigma_{pq}](\nabla^kS(p))-(\nabla^kS)(q)|_g}{[d(p,q)]^\a}
      \leq & r^{-\a}(|P^{(r,s+k)}[\sigma_{pq}](\nabla^kS(p))|_g+|(\nabla^kS)(q)|_g)\\
      = & r^{-\a}(|(\nabla^kS)(p)|_g+|(\nabla^kS)(q)|_g)\leq C_k\|S\|_{C^k(M;\bbU)},
    \end{split}
  \end{equation*}
  where we, in the second last step, used the fact that, due to (\ref{eq:Tone Ttwo prod der}), parallel transport preserves norms and in the last step used
  the fact that the different $C^k$-norms are equivalent. In other words, we can assume that $q\in B_r(p)$. This means that there is an $i\in \{1,\dots,l_\e\}$
  such that $p,q\in U_i$. This, in its turn, means that $d(p,q)\geq |\sfx_i(p)-\sfx_i(q)|/2$; cf. (\ref{eq:dpq sfxi p q equiv}). It is therefore natural to
  focus on the parallel transport of $\nabla^kS$. To this end, we solve the equation $T'=0$ along $\gamma$; i.e., 
  \begin{equation}\label{eq:par trans nabla k S}
    \tfrac{d}{dt}T^{j_1\cdots j_r}_{i_1\cdots i_{s+k}}+\dot{\g}^m\textstyle{\sum}_{a=1}^r{}^{\sfx_i}\Gamma_{ml}^{j_a}\circ\g\cdot T^{j_1\cdots l\cdots j_r}_{i_1\cdots i_{s+k}}
    -\dot{\g}^m\textstyle{\sum}_{b=1}^{s+k}\ {}^{\sfx_i}\Gamma_{mi_b}^l\circ\g\cdot T^{j_1\cdots j_r}_{i_1\cdots l\cdots i_{s+k}}=0,
  \end{equation}
  where $\g:=\sigma_{pq}$; and ${}^{\sfx_i}\Gamma^m_{jk}$ are the Christoffel symbols of $g$ with respect to the coordinates $(\sfx_i,U_i)$. We also impose the
  initial condition
  \[
  T^{j_1\cdots j_r}_{i_1\cdots i_{s+k}}(0)=(\nabla^kS)(\d_{\sfx^{i_1}_i}|_p,\dots,\d_{\sfx^{i_{s+k}}_i}|_p,d\sfx^{j_1}_i|_p,\dots,d\sfx^{j_r}_i|_p).
  \]
  In order to estimate the integral of the second and third terms on the left hand side of (\ref{eq:par trans nabla k S}), note that
  $|{}^{\sfx_i}\Gamma^m_{jk}|\leq \e$ by assumption. Note also that, along the curve,
  \[
    |T^{j_1\cdots j_r}_{i_1\cdots i_{s+2}}|\leq C|(\nabla^kS)(p)|_g,
  \]
  where $C$ is a constant depending only on $r$ and $s$. Here we used the fact that, due to (\ref{eq:Tone Ttwo prod der}), the $g$-norm of $T$ is preserved along the curve and
  the fact that (\ref{eq:sfgxi minus delta}) holds. Similarly, $|\dot{\g}^m|\leq C|\dot{\g}|_g$ for some numerical constant $C$. Combining the above observations
  yields the conclusion that
  \begin{equation*}
    \begin{split}
      |T^{j_1\cdots j_r}_{i_1\cdots i_{s+2}}(1)-T^{j_1\cdots j_r}_{i_1\cdots i_{s+2}}(0)| \leq & C_{n,r,s}\e|(\nabla^kS)(p)|_g\textstyle{\int}_0^1 |\dot{\g}(\tau)|_gd\tau\\
      = & C_{n,r,s}\e |(\nabla^kS)(p)|_gd(p,q),
    \end{split}
  \end{equation*}  
  where $C_{n,r,s}\in\ro$ is a constant depending only on $n$, $r$ and $s$. Thus
  \[
    |{}^{\sfx_i}[P^{(r,s+k)}[\sigma_{pq}](\nabla^kS(p))]^{j_1\cdots j_r}_{i_1\cdots i_{s+k}}-{}^{\sfx_i}(\nabla^kS)^{j_1\cdots j_r}_{i_1\cdots i_{s+k}}(p)|
    \leq C_{n,r,s}\e |(\nabla^kS)(p)|_gd(p,q),
  \]
  so that 
  \begin{equation*}
    \begin{split}
      & |{}^{\sfx_i}[P^{(r,s+k)}[\sigma_{pq}](\nabla^kS(p))]^{j_1\cdots j_r}_{i_1\cdots i_{s+k}}-{}^{\sfx_i}(\nabla^kS)^{j_1\cdots j_r}_{i_1\cdots i_{s+k}}(q)|\\
      \leq & |{}^{\sfx_i}(\nabla^kS)^{j_1\cdots j_r}_{i_1\cdots i_{s+k}}(p)-{}^{\sfx_i}(\nabla^kS)^{j_1\cdots j_r}_{i_1\cdots i_{s+k}}(q)|
      +C_{n,r,s}\e |(\nabla^kS)(p)|_gd(p,q).
    \end{split}
  \end{equation*}
  Dividing by $[d(p,q)]^\a$, the second term on the right hand side can be estimated by a constant times $\|S\|_{C^{k}(M;\bbU)}$. What remains to be estimated is thus
  \[
    \tfrac{|{}^{\sfx_i}(\nabla^kS)^{j_1\cdots j_r}_{i_1\cdots i_{s+k}}(p)-{}^{\sfx_i}(\nabla^kS)^{j_1\cdots j_r}_{i_1\cdots i_{s+k}}(q)|}{|\sfx_i(p)-\sfx_i(q)|^\a}.
  \]  
  We leave it to the reader to verify that this expression can be estimated by $\|S\|_{C^{k,\a}(M;\bbU)}$. There is thus a constant $C_{k,\a}$ such that the second
  inequality in (\ref{eq:Holder norms equiv}) holds. The proof of the first inequality is similar. We leave the details to the reader. 
\end{proof}

\section{Schauder estimates, elliptic model equation}\label{section:schaudermodeleq}

In order to derive a continuation criterion for solutions to (\ref{seq:themodel}), we need Schauder estimates. Moreover, we need to have detailed
information concerning the dependence of the constants (appearing in the estimates) on the unknowns in (\ref{seq:themodel}). To this end, we 
quote Gilbarg and Trudinger's classical text \cite{giatr}, more specifically \cite[Theorem~6.2, p.~90]{giatr}.
To be able to do so easily in the proofs to follow, it is convenient to introduce the relevant notation. Let $\Omega$ be a proper open
subset of $\rn{n}$. For $x,y\in\Omega$, define $d_{x}:=\mathrm{dist}(x,\partial\Omega)$ (where the distance is measured with the standard Euclidean
metric) and $d_{x,y}:=\min\{d_{x},d_{y}\}$. For $k\in\nn{}_0$, $u\in C^{k}(\Omega)$ and $\sigma\in\rn{}$, define
\begin{align*}
  [u]_{k,0;\Omega}^{(\sigma)} := & \sup_{x\in\Omega, |\b|=k}d_{x}^{k+\sigma}|(\d^{\b}u)(x)|,\\
  |u|_{k,0;\Omega}^{(\sigma)} := & \textstyle{\sum}_{j=0}^{k}[u]_{j,0;\Omega}^{(\sigma)}.
\end{align*}
In case $u\in C^{k,\alpha}(\Omega)$ for some $k\in\nn{}_0$ and $0<\alpha\leq 1$, let $\msO_\Omega:=\{(x,y)\in\Omega\times\Omega\, |\, x\neq y\}$ and
\begin{align*}
  [u]_{k,\alpha;\Omega}^{(\sigma)} := & \sup_{(x,y)\in\msO_\Omega}\sup_{|\b|=k}d_{x,y}^{k+\alpha+\sigma}\tfrac{|\d^{\b}u(x)-\d^{\b}u(y)|}{|x-y|^{\alpha}},\\
  |u|_{k,\alpha;\Omega}^{(\sigma)} := & |u|_{k,0;\Omega}^{(\sigma)}+[u]_{k,\alpha;\Omega}^{(\sigma)}.
\end{align*}
In case $\sigma=0$, we use the notation $*$ instead of $(0)$ in the above norms and semi-norms; e.g., $|u|_{k,\alpha;\Omega}^{*}:=|u|_{k,\alpha;\Omega}^{(0)}$ etc.
The statement of \cite[Theorem~6.2, p.~90]{giatr} is then the following:
\begin{thm}[Gilbarg and Trudinger, Theorem~6.2, p.~90]\label{thm:schaudergilantru}%[\cite[Theorem~6.2, p.~90]{giatr}]
  Let $\Omega$ be a proper and non-empty open subset of $\rn{n}$, $\a\in (0,1]$ and let $u\in C^{2,\a}(\Omega)$ be a bounded solution in $\Omega$ of the
  equation
  \begin{equation}\label{eq:GT elliptic eqn}
    a^{ij}\d_{i}\d_{j}u+b^{i}\d_{i}u+cu=F,
  \end{equation}
  where $F\in C^{\a}(\Omega)$ and there are positive constants $\lambda$ and $\Lambda$ such that the coefficients satisfy
  \begin{equation}\label{eq:aijlowbd}
    a^{ij}\xi_{i}\xi_{j}\geq \lambda |\xi|^{2}
  \end{equation}
  for all $x\in\Omega$ and $\xi\in\rn{n}$, and
  \[
  |a^{ij}|^{*}_{0,\a;\Omega},\   |b^{i}|^{(1)}_{0,\a;\Omega},\   |c|^{(2)}_{0,\a;\Omega} \leq \Lambda.
  \]
  Then
  \[
  |u|_{2,\a;\Omega}^{*}\leq C(|u|_{0;\Omega}+|F|^{(2)}_{0,\a;\Omega}),
  \]
  where $C=C(n,\a,\lambda,\Lambda)$ and $|u|_{0;\Omega}:=\sup_{x\in\Omega}|u(x)|$.
\end{thm}
\begin{remark}\label{remark:cont dep of constants}
  In our applications, it will be of interest to keep in mind that the dependence of $C$ on $\Lambda$ can be assumed to be continuous and such that
  $C$ is monotonically increasing with $\Lambda$. To justify the statement that $C$ can be assumed to increase monotonically with $\Lambda$, note that
  $C$ can be replaced by
  \[
  C'(n,\a,\lambda,\Lambda):=\textstyle{\inf}_{\eta\geq \Lambda}C(n,\a,\lambda,\eta),
  \]
  which is clearly monotonically increasing. In order to justify the statement concerning continuity, assume $f:(0,\infty)\rightarrow (0,\infty)$
  to be an increasing function. Then $f$ is Lebesgue integrable on compact subintervals of $(0,\infty)$, since it is measurable
  and bounded on compact intervals ($f(t)\leq f(b)$ for all $t\in [a,b]$). Define the function $g:(0,\infty)\rightarrow (0,\infty)$ by the condition
  \[
  g(t)=\textstyle{\int}_{t}^{t+1}f(s)ds.
  \]
  Then $g(t)\geq f(t)$ by construction. Moreover, $g$ is increasing, since $f$ is increasing. Finally, $g$ is continuous, since $f$ is uniformly bounded
  on compact subsets. Applying this observation to $C'(n,\a,\lambda,\cdot)$ yields the desired conclusion. 
\end{remark}

Next, we state the Schauder estimates we need in our applications. 

\begin{thm}\label{thm:globalSchauder}
  Let $(\S,h_{\refer})$ be a closed Riemannian manifold and define H\"{o}lder spaces as at the beginning of
  Appendix~\ref{section:hoelder spaces on manifolds} with $(M,g)=(\S,h_{\refer})$. Let $c_{1},c_2>1$ be constants. Then there is a function
  $\mfS:\ro_{+}^{3}\rightarrow\ro_{+}$, depending only on $c_{1}$, $c_2$, $(\S,h_{\refer})$ and a covering of $(\S,h_{\refer})$ constructed in the
  proof. Moreover, $\mfS$ is continuous, monotonically increasing in each of its arguments, and such that the following holds. Let
  $h\in C^{1,1}$ be a Riemannian metric on $\S$; $X\in C^{0,1}$ be a vector field; and $f,a\in C^{0,1}$ be real valued functions. Assume that 
  \begin{equation}\label{eq:hhrefabd}
    c_{1}^{-1}h_{\refer}\leq h\leq c_{1}h_{\refer},\ \ \
    c_2^{-1}\leq a\leq c_2
  \end{equation}  
  and that $u\in C^{2,1}$ is a solution to
  \begin{equation}\label{eq:Modelelliptic}
    -\Delta_{h}u+Xu+au=f.
  \end{equation}
  Then
  \begin{equation}\label{eq:uCtwooneglbd}
    \|u\|_{C^{2,1}}\leq \mfS(\|h\|_{C^{1,1}},\|X\|_{C^{0,1}},\|a\|_{C^{0,1}})\cdot \|f\|_{C^{0,1}}.
  \end{equation}
\end{thm}
\begin{remark}
  In the statement and proof, we, for the sake of brevity, use the notation $C^{k,\a}:=C^{k,\a}(h_{\refer})$ and $\ro_{+}:=[0,\infty)$. 
\end{remark}
\begin{proof}
  The idea of the proof is to apply Theorem~\ref{thm:schaudergilantru} in local coordinates and to deduce the desired conclusion by appealing
  to Lemma~\ref{lemma:equiv Hoelder norms}. We therefore start by constructing appropriate local coordinates. We first apply
  Lemma~\ref{lemma:local coord Holder} with $(M,g)=(\S,h_{\refer})$. Then $\{B_{r/2}(p)\}_{p\in\S}$ is an open covering of $\S$, where $r>0$ is as
  in the statement of Lemma~\ref{lemma:local coord Holder}. This means that there is a finite subcovering, say $\{B_{r/2}(p_i)\}_{i=1}^l$. By
  construction, there are local coordinates $\sfx_i$ on $U_i:=B_r(p_i)$ for each $i$. Moreover, $(\sfx_i,U_i)$ have the properties listed in
  Lemma~\ref{lemma:local coord Holder} (with a different $r>0$). Then, due to Lemma~\ref{lemma:equiv Hoelder norms}, we know that, for each
  $i\in\{1,\dots,l\}$,
  \[
    u\circ\sfx_i^{-1}\in C^{2,1}(\Omega_i),\ {}^{\sfx_i}h_{jm}\circ\sfx_i^{-1}\in C^{1,1}(\Omega_i),\ 
    {}^{\sfx_i}X^j\circ\sfx_i^{-1},\ a\circ\sfx_i^{-1},\ f\circ\sfx_i^{-1}\in C^{0,1}(\Omega_i),
  \]
  where $\Omega_i:=\sfx_i(B_r(p_i))$. Moreover, for $p\in B_r(p_i)$ and $\xi\in\rn{n}$,
  \[
    {}^{\sfx_i}h^{jm}(p)\xi_j\xi_m\geq \tfrac{1}{c_1}\, {}^{\sfx_i}h^{jm}_{\refer}(p)\xi_j\xi_m\geq \tfrac{1}{2c_1}|\xi|^2,
  \]
  where we used (\ref{eq:hhrefabd}) in the first step and (\ref{eq:sfgxi minus delta}) in the second step. This means that we are in a
  position to apply Theorem~\ref{thm:schaudergilantru}. More specifically, $\Omega:=\Omega_i$ (where $\Omega_i$ by construction is open,
  non-empty and bounded), $\a=1$ and $a^{jm}:={}^{\sfx_i}h^{jm}\circ\sfx_i^{-1}$, so that (\ref{eq:aijlowbd}) holds with $\lambda=1/(2c_1)$.
  In order to determine $b^j$, note that (\ref{eq:Modelelliptic}) can be written
  \[
  {}^{\sfx_i}h^{jm}\d_{\sfx^{j}_i}\d_{\sfx^{m}_i}u-{}^{\sfx_i}\Gamma^{j}\d_{\sfx^{j}_i}u-{}^{\sfx_i}X^j\d_{\sfx^{j}_i}u-au=-f
  \]
  on $B_{r}(p_i)$, where ${}^{\sfx_i}\Gamma^{j}={}^{\sfx_i}h^{kl}\, {}^{\sfx_i}\Gamma^{j}_{kl}$ are the contracted Christoffel symbols. In other words,
  the function $u\circ\sfx_i^{-1}$ satisfies (\ref{eq:GT elliptic eqn}) on $\Omega_i$ with $a^{lm}$ as above,
  \begin{equation}\label{eq:bj def}
    b^j:=-{}^{\sfx_i}\Gamma^{j}\circ\sfx_i^{-1}-{}^{\sfx_i}X^j\circ\sfx_i^{-1},
  \end{equation}
  $c:=-a\circ\sfx_i^{-1}$ and $F:=-f\circ\sfx_i^{-1}$. Next, note that the distance from a point $p\in B_r(p_i)$ to $\d B_r(p_i)$ (with respect to
  the topological metric induced by $h_{\refer}$ on $\Sigma$, say $d_{\refer}$) is bounded from above by $r$. Due to (\ref{eq:dpq sfxi p q equiv}),
  we conclude that if $d_x:=\mathrm{dist}(x,\d\Omega_i)$ (where $\mathrm{dist}$ denotes the distance with respect to the Euclidean metric on
  $\rn{n}$), then $d_x\leq 2r$ for $x\in\Omega_i$. By a similar argument, if $\omega_i:=\sfx_i(B_{r/2}(p_i))$, then $d_x\geq r/4$ for $x\in\omega_i$. Next, 
  \[
  |a^{jm}|_{0,1;\Omega_i}^*=|a^{jm}|_{0,0;\Omega_i}^*+[a^{jm}]_{0,1;\Omega_i}^*\leq 2c_1+\textstyle{\sup}_{(x,y)\in\msO_i}d_{x,y}\tfrac{|a^{jm}(x)-a^{jm}(y)|}{|x-y|},
  \]
  where we appealed to (\ref{eq:sfgxi minus delta}), (\ref{eq:hhrefabd}) and used the notation $\msO_i:=\{(x,y)\in\Omega_i\times\Omega_i\, |\, x\neq y\}$.
  In order to estimate the second term on the far right hand side, let $(x,y)\in\msO_i$ and define, for $t\in[0,1]$,
  \[
  \xi_{lm}(t):=t\cdot {}^{\sfx_i}h_{lm}\circ\sfx_i^{-1}(x)+(1-t)\cdot{}^{\sfx_i}h_{lm}\circ\sfx_i^{-1}(y).
  \]
  Due to (\ref{eq:sfgxi minus delta}) and (\ref{eq:hhrefabd}), the eigenvalues of the matrix with components $\xi_{lm}(t)$ are bounded from below
  by $1/(2c_1)$ and from above by $2c_1$. The matrix with components $\xi_{lm}(t)$ is therefore invertible and the eigenvalues of the inverse
  satisfy the same bounds. Let $\xi^{lm}(t)$ denote the components of the inverse of the matrix with components $\xi_{lm}(t)$. Then
  \begin{equation*}
    \begin{split}
      a^{lm}(x)-a^{lm}(y) = & \xi^{lm}(1)-\xi^{lm}(0)
      =-\textstyle{\int}_{0}^{1}\xi^{lj}(t)\xi^{mk}(t)\tfrac{d}{dt}\xi_{jk}(t)dt\\
      = & -\textstyle{\int}_{0}^{1}\xi^{lj}(t)\xi^{mk}(t)dt [{}^{\sfx_i}h_{jk}\circ\sfx_i^{-1}(x)-{}^{\sfx_i}h_{jk}\circ\sfx_i^{-1}(y)].
    \end{split}
  \end{equation*}

  Combining the above observations with Lemma~\ref{lemma:equiv Hoelder norms} and similar (but much simpler) arguments for $c$ and $F$, it is clear that
  \[
  |a^{jm}|_{0,1;\Omega_i}^*\leq C\|h\|_{C^{0,1}},\ \ \
  |c|_{0,1;\Omega_i}^{(2)}\leq C\|a\|_{C^{0,1}},\ \ \
  |F|_{0,1;\Omega_i}^{(2)}\leq C\|f\|_{C^{0,1}},  
  \]
  where the constants only depends on $r$ and $c_1$. Turning to $b^j$ introduced in (\ref{eq:bj def}), the argument concerning the second term
  on the right hand side is similar to the above. Finally, consider
  \[
    {}^{\sfx_i}\Gamma^{j}\circ\sfx_i^{-1}=\tfrac{1}{2}a^{kl}a^{jm}[\d_k({}^{\sfx_i}h_{lm}\circ\sfx_i^{-1})
    +\d_l({}^{\sfx_i}h_{km}\circ\sfx_i^{-1})-\d_m({}^{\sfx_i}h_{kl}\circ\sfx_i^{-1})].
  \]  
  By arguments similar to the above, it can be verified that
  \[
  |b^{j}|_{0,1;\Omega_i}^{(1)}\leq C\|h\|_{C^{0,1}}\|h\|_{C^{1}}+C\|h\|_{C^{1,1}}+C\|X\|_{C^{0,1}}
  \]
  where the constants only depends on $r$ and $c_1$.

  By the above estimates and observations, we are allowed to apply Theorem~\ref{thm:schaudergilantru}. This yields
  \[
  |u\circ\sfx_i^{-1}|_{2,1;\Omega_i}^*\leq C(|u\circ\sfx_i^{-1}|_{0;\Omega_i}+|F|_{0,1;\Omega_i}^{(2)}),
  \]
  where $C$ only depends on $n$, $c_1$, $r$ and continuously and monotonically increasingly (cf. Remark~\ref{remark:cont dep of constants})
  on $\|h\|_{C^{1,1}}$, $\|X\|_{C^{0,1}}$ and $\|a\|_{C^{0,1}}$. Next, $\|u\|_{C^{0}}\leq c_2\|f\|_{C^{0}}$ due to the maximum principle. To summarise, 
  \[
  |u\circ\sfx_i^{-1}|_{2,1;\Omega_i}^*\leq C\|f\|_{C^{0,1}},
  \]
  where $C$ only depends on $n$, $c_1$, $c_2$, $r$ and continuously and monotonically increasingly on $\|h\|_{C^{1,1}}$, $\|X\|_{C^{0,1}}$ and
  $\|a\|_{C^{0,1}}$. Since $d_x\geq r/4$ for $x\in\omega_i$, this estimate allows us to bound $\|u\|_{C^{2,1}(\S;\bbV)}$, where $\bbV:=\{V_i\}_{i=1}^l$ is
  given by $V_i:=B_{r/2}(p_i)$; cf. (\ref{seq: C k alpha M bbU}). The desired estimate then follows by an application of
  Lemma~\ref{lemma:equiv Hoelder norms}.
\end{proof}

\section{Sobolev estimates, elliptic model equation}\label{section:sobestellmodeleq}

Next, we need Sobolev estimates of solutions to the elliptic PDE's appearing in (\ref{seq:themodel}). The estimates are essentially standard. However,
we here need detailed information concerning how the constants appearing in the estimates depend on the coefficients of the equation. The main result
of the present appendix is the estimate (\ref{eq:Hkptwonormu}) below.  

\begin{thm}[The elliptic estimates] \label{thm: elliptic estimate}
  Let $(\S, h_\refer)$ be a closed, connected and oriented Riemannian manifold of dimension $n$; let $k \in \nn{}_0$; and let $c_1,c_2 > 1$ be real
  numbers. Then there are
  polynomials $\mfD_{1}:\ro_{+}^{2}\rightarrow\ro_{+}$ and $\mfD_{2}:\ro_{+}^{4}\rightarrow\ro_{+}$ with non-negative coefficients; and a continuous function
  $\mfD_{3}:\ro_{+}^{3}\rightarrow\ro_{+}$ which is monotonically increasing in each of its arguments. Here $\mfD_{i}$, $i=1,2,3$, only depend on $c_{1}$,
  $c_{2}$, $n$, $k$ and $(\S,h_{\refer})$. Finally, let $h$ and $a$ be a Riemannian metric and a smooth function on $\S$, respectively, such that
  \begin{equation}\label{eq:hhrefequivellest}
    c_{1}^{-1}h_\refer\leq h\leq  c_1 h_\refer, \ \ \
    c_{2}^{-1}\leq  a\leq c_2.
  \end{equation}
  Then the following holds. First, if $f\in H^k$, and $u\in H^1$ is a weak solution to 
  \[
  (-\Delta_{h}+a)u=f,
  \]
  then $u\in H^{k+2}$. Second, for such $f$ and $u$, 
  \begin{equation}\label{eq:mainHkptwoestell}
    \norm{u}_{H^{k+2}} \leq \mfD_{1}[\norm{h}_{W^{2, \infty}},\norm{a}_{W^{1, \infty}}]
    \left( \norm{f}_{H^k} + \norm{u}_{W^{1, \infty}} \norm{h}_{H^{k+1}}  + \norm{u}_{L^\infty} \norm{a}_{H^k} \right).
  \end{equation}
  Third, if, in addition to the above, $k>n/2$, then
  \begin{equation}\label{eq:Hkestcrude}
    \|u\|_{H^{k+2}} \leq \mfD_{2}[\norm{h}_{W^{2, \infty}},\norm{a}_{W^{1, \infty}},\|a\|_{H^{k}},\|h\|_{H^{k+1}}]\cdot \|f\|_{H^{k}}.
  \end{equation}
  Fourth, if, in addition to the above assumptions (except for the assumption that $k>n/2$), $u\in C^{2,1}$ and $f\in C^{0,1}$, then  
  \begin{equation}\label{eq:Hkptwonormu}    
    \norm{u}_{H^{k+2}} \leq \mfD_{3}[\norm{h}_{W^{2, \infty}},\norm{a}_{W^{1, \infty}},\norm{f}_{C^{0,1}}]
    \cdot (\norm{f}_{H^k}+\|h\|_{H^{k+1}}+\|a\|_{H^{k}}).
  \end{equation}  
\end{thm}
\begin{remark}
  It is understood that we use the metric $h_{\refer}$ to define all the spaces and corresponding norms. 
\end{remark}
\begin{remark}
  When we say that the $\mfD_i$ depend on $(\S,h_{\refer})$, we include dependence of coverings of $\S$, partitions of unity, local orthonormal
  frames with respect to $h_\refer$ etc. constructed in the proof.  
\end{remark}
\begin{remark}
  It is also useful to keep in mind that the following, more primitive, estimate holds
  \begin{equation}\label{eq:u Hk plus two Lu Hk}
    \|u\|_{H^{k+2}}\leq C\|Lu\|_{H^{k}},
  \end{equation}
  where $C$ only depends on $c_{1}$, $c_{2}$, $n$, $k$, $(\S,h_{\refer})$, the $C^{l_k+1}$-norm of $h$, where $l_k:=\max\{k,1\}$, and the $C^{l_k}$-norm
  of $a$. This estimate is analogous to (\ref{eq:Hkestcrude}). However, it is better in the sense that we do not require $k>n/2$ and it is worse
  in the sense that we require control of $h$ and $a$ in $C^{l_k+1}$ and $C^{l_k}$ respectively. 
\end{remark}
The proof of Theorem \ref{thm: elliptic estimate} is to be found at the end of Subsection~\ref{ssection:loc fr comm} and rests on several lemmas.
We begin with some basic observations.

\subsection{Basic estimates and observations}
We begin by noting that the volume forms and measures associated with different Riemannian metrics are equivalent as far as estimates are concerned. 

\begin{lemma}[Equivalence of volume forms] \label{le: equivalence volume forms}
  Let $(\S, h_\refer)$ be a closed, connected and oriented Riemannian manifold of dimension $n$. Assume that $h$ is a Riemannian metric on $\S$ and
  that $c_1^{-1} h_\refer\leq h\leq c_1 h_\refer$ for some constant $c_1 > 1$. Then there is a constant $c$, depending only on $c_1$ and $n$, such that
  \begin{equation}\label{eq:phi int h h ref equiv}
    c^{-1}\textstyle{\int}_\S|\phi|\Vol_{h_\refer}\leq \textstyle{\int}_\S|\phi|\Vol_h\leq c \textstyle{\int}_\S|\phi|\Vol_{h_\refer}
  \end{equation}
  for all $\phi\in C(M)$. Similarly, if $\phi\in L^1(h_\refer)$, then $\phi\in L^1(h)$ and
  \begin{equation}\label{eq:phi int h h ref L one equiv}
    c^{-1}\textstyle{\int}_\S|\phi|d\mu_{h_\refer}\leq \textstyle{\int}_\S|\phi|d\mu_h\leq c \textstyle{\int}_\S|\phi| d\mu_{h_\refer}.
  \end{equation}
\end{lemma}
\begin{proof}
  Using a finite partition of unity, it is sufficient to prove (\ref{eq:phi int h h ref equiv}) for functions
  $\phi$ with support in a coordinate neighbourhood, say $U$. Let $\{E_i\}$ be an orthonormal frame on $U$ with respect to $h_\refer$. In
  order to prove the second inequality in (\ref{eq:phi int h h ref equiv}), it is sufficient to prove that the determinant of
  $h_{ij} := h(E_i, E_j)$ is bounded by the determinant of $h_\refer(E_i, E_j) = \de_{ij}$ (i.e., $1$) times a factor that only depends on $c_1$ and
  $n$. However,
  \begin{align*}
    \abs{\det(h_{ij})}
    &= \left| \textstyle{\sum}_{\s \in S_n}  \mathrm{sgn}(\s) \prod_{i = 1}^n h_{i\s(i)} \right| 
    \leq n! \cdot \max_{i,j = 1, \hdots, n}\abs{h_{ij}}^n\leq n! \cdot c_1^n.
  \end{align*}
  The first inequality in (\ref{eq:phi int h h ref equiv}) is proven similarly. The estimate (\ref{eq:phi int h h ref L one equiv}) follows from
  (\ref{eq:int dmuh int vol form}), (\ref{eq:phi int h h ref equiv}) and a density argument. 
\end{proof}

\begin{lemma}[Equivalence of Sobolev norms] \label{le: equivalence Sobolev norms}
  Let $(\S, h_\refer)$ be a closed, connected and oriented Riemannian manifold of dimension $n$. Assume that $h$ is a Riemannian metric on $\S$
  satisfying $c_1^{-1} h_\refer\leq h\leq c_1 h_\refer$ for some constant $c_1 > 1$. Then there is a constant $c>1$, depending only on
  $c_1$ and $n$, such that
  \begin{equation}\label{eq:Hoh Hohref}
    c^{-1} \norm{u}_{H^1(h_\refer)}\leq \norm{u}_{H^1(h)}\leq c \norm{u}_{H^1(h_\refer)}
  \end{equation}
  for all $u \in H^1(h)$.
\end{lemma}
\begin{proof}
  Due to Lemma~\ref{le: equivalence volume forms} and the assumptions, there is a constant $c>1$, depending only on $c_1$ and $n$, such that
  \begin{align*}
    c^{-1} \norm{u}_{H^1(h_\refer)}^2\leq  & \textstyle{\int}_\S [h^{ij}\nabla_i^h u\nabla^h_j u+u^2]\Vol_h\leq c \norm{u}_{H^1(h_\refer)}^2,
  \end{align*}
  where $\nabla^h$ denotes the Levi-Civita connection of $(\S,h)$. This proves (\ref{eq:Hoh Hohref}) for $u \in C^\infty(\S)$. Thus
  (\ref{eq:Hoh Hohref}) holds for $u \in H^1(h)$ since $C^\infty(\S)$ is dense in $H^1(h)$. 
\end{proof}

\begin{lemma}[The Laplace operator] \label{le: Laplace operator}
  Let $(\S, h_\refer)$ be a Riemannian manifold and let $\{E_{i}\}$ be a local orthonormal frame, defined on an open subset
  $U\subseteq \S$. For a second Riemannian metric $h$ on $\S$ and $u\in C^\infty(U)$, the Laplace operator is given by
  \begin{equation}\label{eq:DeltahitoEi}
    \Delta_hu = h^{ij}E_iE_ju + \left( E_i h^{ik} \right) E_ku + \tfrac{1}2 h^{ij}(E_l h_{ij}) h^{lk}E_ku + h^{ij}\lambda_{ki}^{k}E_ju,
  \end{equation}
  where $h_{ij}:=h(E_i,E_j)$; $h^{ij}$ are the components of the inverse of the matrix with components $h_{ij}$; and
  $\lambda_{ij}^{k}$ are defined by $[E_{i},E_{j}]=\lambda_{ij}^{k}E_{k}$. 
\end{lemma}
\begin{remark}
  In the applications of this observation, it is convenient to note that
  \begin{equation}\label{eq:Deltahftt}
    h^{ij}E_iE_ju + \left( E_i h^{ik} \right) E_ku=E_{i}(h^{ij}E_{j}u).
  \end{equation}
\end{remark}
\begin{proof}
Note that $-\Delta_h= - h^{ij} (\n^h)^2_{E_i, E_j}= - h^{ij}E_iE_j + h^{ij} \n^h_{E_i}E_j$.
%\[
%	-\Delta_h
%		= - h^{ij} (\n^h)^2_{E_i, E_j} 
%		= - h^{ij}E_iE_j + h^{ij} \n^h_{E_i}E_j.
%\]
Moreover,
\begin{align*}
	2 h(\n^h_{E_i}E_j, E_l)
		=& E_i h_{jl} + E_j h_{il} - E_l h_{ij} - \lambda_{ji}^mh_{ml} - \lambda_{il}^mh_{mj} - \lambda_{jl}^mh_{mi},
\end{align*}
due to the Koszul formula. Next, $\n^h_{E_i}E_j=h(\n^h_{E_i}E_j, E_l) h^{lk}E_k$. Thus
\begin{align*}
	-\Delta_h
		= & - h^{ij}E_iE_j + \tfrac{1}{2}h^{ij} \left( 2E_i h_{jl} - E_lh_{ij} \right) h^{lk}E_k 
	- \tfrac{1}{2}h^{ij} \left( \lambda_{il}^mh_{mj}
        + \lambda_{jl}^mh_{mi}\right) h^{lk}E_k \\
	%= & - h^{ij}E_iE_j + \tfrac{1}{2}h^{ij} \left( 2 E_i h_{jl} - E_lh_{ij} \right) h^{lk}E_k 
	%- \tfrac{1}{2}h^{ij} \left( \lambda_{ki}^kE_j
        %+ \lambda_{ki}^k E_j \right) \\
        = & - h^{ij}E_iE_j - \left( E_i h^{ik} \right) E_k - \tfrac{1}{2}h^{ij} (E_l h_{ij}) h^{lk}E_k - h^{lk} \lambda_{ml}^m E_k.
\end{align*}
This concludes the proof.
\end{proof}

\subsection{The space $H^{-1}$}

Next we consider the space $H^{-1}(h):= \left( H^1(h) \right)^*$ with norm
\begin{equation}\label{eq:Hmonormdef}
  \norm{ \lambda }_{H^{-1}(h)} := \sup\Big\{\tfrac{\abs{\lambda(v)}}{\norm{v}_{H^1(h)}}\, \Big|\, v \in H^1(h),\  v\neq 0\Big\}.
\end{equation}
There is a natural embedding $\iota_h: L^2(h)\hookrightarrow H^{-1}(h)$ defined by 
\[
\iota_h[u](v) := \textstyle{\int}_\S  v u d\mu_h.
\]
If $X$ is a smooth vector field on $\S$ and $u\in C^{\infty}(\S)$, it is of interest to estimate
\begin{equation}\label{eq:iotaXunormdef}
  \|\iota_h[Xu]\|_{H^{-1}(h)}=\sup\Big\{\tfrac{\abs{\iota_h[Xu](v)}}{\norm{v}_{H^1(h)}}\, |\, v \in H^1(h),\  v\neq 0\Big\}.
\end{equation}
However, assuming $v$ to be $C^{1}$ (cf. (\ref{eq:partial int vf}) below),
\begin{equation}\label{eq:lambda Xu v}
  \iota_h[Xu](v)=\textstyle{\int}_{\S} vXu\Vol_{h}
  =-\textstyle{\int}_{\S}vu(\rodiv_{h}X)\Vol_{h}-\textstyle{\int}_{\S}(Xv)u\Vol_{h}.
\end{equation}
Due to this equality we can define the left hand side via the far right hand side, given $u\in L^2(h)$, $v\in H^1(h)$ and
a $C^1$-vector field $X$. We therefore, for $u\in L^2(h)$ and a $C^1$-vector field $X$ define $X[u]\in H^{-1}(h)$ by
\begin{equation}\label{eq:X brackets u def}
  X[u](v):=-\textstyle{\int}_{\S}vu(\rodiv_{h}X)d\mu_{h}-\textstyle{\int}_{\S}(Xv)ud\mu_{h}.
\end{equation}
This means that 
\begin{equation}\label{eq:lambdaXu}
  \begin{split}
    |X[u](v)| \leq & \|\rodiv_{h}X\|_{C^{0}(h)}\|v\|_{L^{2}(h)}\|u\|_{L^{2}(h)}+\|Xv\|_{L^{2}(h)}\|u\|_{L^{2}(h)}\\
    \leq & [\|\rodiv_{h}X\|_{C^{0}(h)}^{2}+\|X\|_{C^{0}(h)}^{2}]^{1/2}\|v\|_{H^{1}(h)}\|u\|_{L^{2}(h)},
  \end{split}
\end{equation}
where we use the notation introduced in (\ref{eq:Ck g def}). Combining (\ref{eq:lambdaXu}) with (\ref{eq:Hmonormdef}) yields
\begin{equation}\label{eq:XuHmonormest}
  \|X[u]\|_{H^{-1}(h)}\leq [\|\rodiv_{h}X\|_{C^{0}(h)}^{2}+\|X\|_{C^{0}(h)}^{2}]^{1/2}\|u\|_{L^{2}(h)}.
\end{equation}
Combining the above observations yields the conclusion that, given a $C^1$ vector field $X$ and an $L^2$-function $u$, we can think of
$Xu$ as an element of $H^{-1}(h)$.

\subsection{Estimating elliptic second order operators in $H^{-1}$}
Next, it is of interest to note that, given $u\in H^1(h)$ and $L:=-\Delta_h+a$, $Lu$ can be interpreted as an element of $H^{-1}(h)$. Assume, to this end,
first that $u$ and $v$ are smooth functions on $\S$. Then
\begin{equation}\label{eq:-Delta plus a as element of H minus one}
  \textstyle{\int}_\S vLu\Vol_h=\textstyle{\int}_\S [h(\grad_h v,\grad_hu)+avu]\Vol_h
  =\textstyle{\int}_\S [du(\grad_h v)+avu]\Vol_h.
\end{equation}
In particular, it is clear that the right hand side of (\ref{eq:-Delta plus a as element of H minus one}) is meaningful for
$u,v\in H^1(h)$. Given $u\in H^1(h)$, we can therefore define $L[u]\in H^{-1}(h)$ by 
\[
  L[u](v):=\textstyle{\int}_\S [h(\grad_h v,\grad_hu)+avu]d\mu_h. 
\]
Assume that there is a constant $c_2>1$ such that $c_2^{-1}\leq a\leq c_2$ on $\S$. Then
\[
c_2^{-1}\|u\|_{H^1(h)}^2\leq L[u](u)\leq \|L[u]\|_{H^{-1}(h)}\|u\|_{H^1(h)},
\]
so that
\[
\|u\|_{H^1(h)}\leq c_2\|L[u]\|_{H^{-1}(h)}.
\]
Note, however, that $L[u]$ depends on $h$ in two ways. First, $\Delta_h$ is of course the Laplace operator associated with $h$. However, there
is a second dependence on $h$ arising from the choice of the volume form $\Vol_h$ in (\ref{eq:-Delta plus a as element of H minus one}).
The second dependence on $h$ is, in a sense, artificial. One could also consider, for $u,v\in C^\infty(\S)$,
\begin{equation*}
  \begin{split}
    \textstyle{\int}_\S vLu\Vol_{h_{\refer}} = & \textstyle{\int}_\S [h(\grad_h (\psi_h v),\grad_hu)+a\psi_h vu]\Vol_h\\
    = & \textstyle{\int}_\S [\psi_h^{-1}h(\grad_h (\psi_h v),\grad_hu)+avu]\Vol_{h_{\refer}},
  \end{split}
\end{equation*}
where $\psi_h$ is defined by the condition that $\Vol_{h_{\refer}}=\psi_h\Vol_h$. Given $u,v\in H^1(h_{\refer})$, we can thus define
\begin{equation*}
  \begin{split}
    L_{\refer}[u](v) := & \textstyle{\int}_\S [\psi_h^{-1}h(\grad_h (\psi_h v),\grad_hu)+avu]d\mu_{h_{\refer}}\\
    = & \textstyle{\int}_\S [du(\psi_h^{-1}\grad_h (\psi_h v))+avu]d\mu_{h_{\refer}}.
  \end{split}
\end{equation*}
Thus
\begin{equation*}
  \begin{split}
    |L_{\refer}[u](v)| \leq &  \|\nabla^{h_{\refer}}u\|_{L^2(h_{\refer})}\|\psi_h^{-1}\grad_h(\psi_h v)\|_{L^2(h_{\refer})}
    +c_2\|v\|_{L^2(h_{\refer})}\|u\|_{L^2(h_{\refer})}\\
    \leq &  c_1\|\nabla^{h_{\refer}}u\|_{L^2(h_{\refer})}\|\psi_h^{-1}\grad_{h_{\refer}}(\psi_h v)\|_{L^2(h_{\refer})}
    +c_2\|v\|_{L^2(h_{\refer})}\|u\|_{L^2(h_{\refer})}\\
    \leq & C\|h\|_{W^{1,\infty}(h_{\refer})}\|v\|_{H^1(h_{\refer})}\|u\|_{H^1(h_{\refer})},
  \end{split}
\end{equation*}
where $C$ only depends on $c_1$, $c_2$ and $n$. This means that
\begin{equation}\label{eq: H1 invertibility}
  \|L_{\refer}[u]\|_{H^{-1}(h_{\refer})}\leq C\|h\|_{W^{1,\infty}(h_{\refer})}\|u\|_{H^1(h_{\refer})}.
\end{equation}
Next, note that
\begin{equation}\label{eq:Lref uu}
  \begin{split}
    L_{\refer}[u](\psi_h^{-1} u) = & \textstyle{\int}_\S [\psi_h^{-1}h(\grad_h u,\grad_hu)+\psi_h^{-1}au^2]d\mu_{h_{\refer}}.
  \end{split}
\end{equation}
There are thus constants $C$, depending only on $n$, $c_1$ and $c_2$ such that
\begin{equation*}
  \begin{split}
    \|u\|_{H^{1}(h_{\refer})}^2 \leq & CL_{\refer}[u](\psi_h^{-1} u)\leq C\|L_{\refer}[u]\|_{H^{-1}(h_{\refer})}\|\psi_h^{-1} u\|_{H^1(h_{\refer})}\\
    \leq & C\|h\|_{W^{1,\infty}(h_{\refer})}\|L_{\refer}[u]\|_{H^{-1}(h_{\refer})}\|u\|_{H^1(h_{\refer})}.
  \end{split}  
\end{equation*}
Thus
\begin{equation}\label{eq:u Hone bd ito Lref u}
  \|u\|_{H^{1}(h_{\refer})}\leq C\|h\|_{W^{1,\infty}(h_{\refer})}\|L_{\refer}[u]\|_{H^{-1}(h_{\refer})},
\end{equation}
where $C$ is a constant depending only on $n$, $c_1$ and $c_2$. 

\subsection{Elliptic regularity}
Next, we recall elliptic regularity. 
\begin{lemma}\label{lemma:elliptic reg}
  Let $(\S,h)$ be a closed, connected and oriented Riemannian manifold, $a\in C^\infty(\S)$ and $L:=-\Delta_h+a$. Assume $u\in H^1(h)$ to be such
  that $L[u]=\iota_{h}[f]$, where $f\in H^k(h)$ and $k\in\nn{}_0$. Then $u\in H^{k+2}(h)$. Moreover, $Lu=f$ a.e. 
\end{lemma}
\begin{remark}
  In this lemma, we do not assume $a$ to be strictly positive. Moreover, analogous conclusions hold in case $u$ satisfies
  $L[u]+X[u]=\iota_h[f]$ for some smooth vector field $X$.
\end{remark}
\begin{proof}
  Let $(\sfx,U)$ be local coordinates. Let $\zeta\in C^\infty_c(U)$ and consider $L[\zeta u]$. Given $v\in H^1(h)$,
  \begin{equation}\label{eq:L zeta u zeta v diff}
    \begin{split}
      L[\zeta u](v)-L[u](\zeta v) = &
      \textstyle{\int}_\S [d(\zeta u)(\grad_h v)-du(\grad_h (\zeta v))]d\mu_{h}
      %& -\textstyle{\int}_\S [du(\grad_h (\zeta v))+a\zeta vu]d\mu_{h}. 
    \end{split}
  \end{equation}
  Let $f_1:=-du(\grad_h\zeta)$. Then $f_1\in L^2(h)$ and (\ref{eq:L zeta u zeta v diff}) can be written
  \[
  L[\zeta u](v)-L[u](\zeta v) = \iota_{h}[f_1](v)+\textstyle{\int}_{\S}ud\zeta(\grad_h v)d\mu_{h}.
  \]
  In order to rewrite the second term on the right hand side, note that
  \begin{equation*}
    \begin{split}
      ud\zeta(\grad_h v) = & uh(\grad_h\zeta,\grad_h v)=udv(\grad_h\zeta)\\
      = & u(\grad_h\zeta)(v)=(\grad_h\zeta)(uv)-v(\grad_h\zeta)(u).
    \end{split}
  \end{equation*}
  On the other hand, due to Cartan's magic formula (i.e., $\ml_X=d\circ\iota_X+\iota_X\circ d$), Stokes' theorem and
  the fact that $\ml_X\Vol_h=(\rodiv_h X)\Vol_h$, 
  \begin{equation}\label{eq:partial int vf}
    \textstyle{\int}_\S X(\psi)\Vol_h=-\textstyle{\int}_\S\psi (\rodiv_hX) \Vol_h
  \end{equation}
  for all $\psi\in C^\infty(\S)$. In particular, we thus have
  \[
  \textstyle{\int}_\S (\grad_h\zeta)(uv)d\mu_h =-\textstyle{\int}_\S uv\cdot\Delta_h\zeta d\mu_h.
  \]
  Introducing
  \[
  f_2:=-2du(\grad_h\zeta)-u\Delta_h\zeta,
  \]
  it follows that
  \[
  L[\zeta u](v)=L[u](\zeta v)+\iota_h[f_2](v)=\iota_h[f](\zeta v)+\iota_h[f_2](v)=\iota_h[\zeta f+f_2](v).
  \]
  Due to this equality, which holds for all $v\in H^1(h)$, it follows that $(\zeta u)\circ\sfx^{-1}$ is a weak solution
  to an elliptic problem on $\sfx(U)$. Moreover, this weak solution belongs to $H^1(\sfx(U))$. Due to \cite[Theorem~2, p.~332]{evans},
  it follows that $(\zeta u)\circ\sfx^{-1}\in H^2(\sfx(U))$. Thus $\zeta u\in H^2(h)$. Combining this observation with a partition
  of unity, it follows that $u\in H^2(h)$, so that $Lu$ is well defined and equals $f$. If $k\geq 1$, we can iterate this argument
  in order to obtain the desired conclusion. 
\end{proof}

\subsection{Local frames and commutators}\label{ssection:loc fr comm}
Let $(\S,h_{\refer})$ be a closed, connected and oriented Riemannian manifold and let $\{E_i\}$ be a local orthonormal frame with respect to
$h_\refer$, defined on an open subset $U\subseteq\S$ as in the statement of Lemma~\ref{le: Laplace operator}. Fix a $\zeta\in C^\infty_c(U)$ and
a compact set $K\subset U$ such that $\supp \zeta\subset K$. Assume, moreover, that $0\leq\zeta\leq 1$. Then
$E_{\zeta,i}:=\zeta E_i$ is a smooth vector field on $\S$ and $\|E_{\zeta,i}\|_{C^{0}(h_{\refer})}\leq 1$. Moreover, 
\begin{equation}\label{eq:rodiv hrefer Ezetai}
  \rodiv_{h_{\refer}}E_{\zeta,i}=\textstyle{\sum}_{k}h_{\refer}(\nabla^{h_{\refer}}_{E_{k}}E_{\zeta,i},E_{k})
  =\textstyle{\sum}_{k}h_{\refer}([E_{k},E_{\zeta,i}],E_{k})=E_i(\zeta)+\sum_{k}\zeta\lambda^{k}_{ki},
\end{equation}
where the structure coefficients $\lambda_{ij}^{k}$ are defined by $[E_{i},E_{j}]=\lambda_{ij}^{k}E_{k}$. Combining these observations with
(\ref{eq:XuHmonormest}) (in which we replace $h$ by $h_{\refer}$) yields
\begin{equation}\label{eq:EiuHmoest}
  \|E_{\zeta,i}[u]\|_{H^{-1}(h_{\refer})}\leq C\|u\|_{L^{2}(h_{\refer})},
\end{equation}
where $C$ only depends on the frame $\{E_{i}\}$, $n$, $\zeta$ and $K$.

\begin{lemma}[$H^{-1}$-bounds on the commutator] \label{le: commutator Delta X}
  Let $(\S, h_\refer)$ be a closed, connected and oriented Riemannian manifold of dimension $n$. Let $U$, $\{E_{i}\}$, $\zeta$ and $K$ be as
  described prior to the statement of the lemma. Assume $h$ to be a Riemannian metric on $\S$ such that $c_1^{-1} h_\refer\leq h\leq c_1 h_\refer$
  for some constant $c_1 > 1$. Then there is a constant $C$, depending only on $c_1$, $n$, $\{E_i\}$, $\zeta$ and $K$, such that
  \begin{equation}\label{eq:DeltahEicommmod}
    \begin{split}
      \|L_{\refer}[E_{\zeta,i} u]-E_{\zeta,i}[Lu]\|_{H^{-1}}
      \leq & C \big( \norm{h}_{W^{2, \infty}} + \norm{h}_{W^{1, \infty}}^2
      +\|a\|_{W^{1,\infty}}\big)\norm{u}_{H^1}
    \end{split}
  \end{equation}
  for all $u\in H^2$. 
\end{lemma}
\begin{remark}
  In the lemma, it is understood that $L=-\Delta_h+a$, where $a\in C^\infty(\S)$, and that we use the metric $h_{\refer}$ to define all the spaces
  and corresponding norms. Moreover, in the definition of $E_{\zeta,i}[Lu]$, we take it for granted that $h$ is replaced by $h_{\refer}$ on the right hand
  side of (\ref{eq:X brackets u def}).
\end{remark}
\begin{proof}
  Compute, for $u,v\in C^\infty(\S)$, 
  \begin{equation}\label{eq:Lref Ezetai comm}
    \begin{split}
      L_{\refer}[E_{\zeta,i} u](v)-E_{\zeta,i}[Lu](v) = & \textstyle{\int}_\S(LE_{\zeta,i}u-E_{\zeta,i}Lu)vd\mu_{h_{\refer}}.
    \end{split}
  \end{equation}
  On the other hand
  \[
  LE_{\zeta,i}u-E_{\zeta,i}Lu=E_{\zeta,i}\Delta_hu-\Delta_hE_{\zeta,i}u-E_{\zeta,i}(a)u.
  \]
  The contribution to (\ref{eq:Lref Ezetai comm}) arising from the last term on the right hand side can be estimated by
  \[
  C\|a\|_{W^{1,\infty}}\norm{u}_{L^2}\norm{v}_{L^2},
  \]
  where $C$ is numerical; again, all norms are defined using $h_\refer$. Next, we focus on $[E_{\zeta,i},\Delta_h]$. Let us begin by considering the
  commutator with the right hand side of (\ref{eq:Deltahftt}):
  \begin{equation}\label{eq:first terms in Delta h comm}
    \begin{split}
      & E_{\zeta,i}[E_{l}(h^{lj}E_{j}u)]-E_{l}(h^{lj}E_{j}E_{\zeta,i}u)\\
      = & \zeta\lambda_{il}^{m}E_{m}(h^{lj}E_{j}u)+\zeta E_{l}(E_{i}(h^{lj}E_{j}u))-E_l(h^{lj}E_{j}E_{\zeta,i}u)\\
      = & \zeta\lambda_{il}^{m}E_{m}(h^{lj}E_{j}u)+\zeta E_{l}(E_{i}(h^{lj})E_{j}u+h^{lj}[E_{i},E_{j}]u)
      +\zeta E_{l}(h^{lj}E_{j}E_{i}u)-E_l(h^{lj}E_{j}E_{\zeta,i}u).
    \end{split}
  \end{equation}
  The first two terms on the far right hand side can be written as a sum of terms of the form $\varphi E_l(\psi E_mu)$. Here, and in what
  follows, we assume that either $\supp\varphi\subset K$ or $\supp\psi\subset K$. The corresponding contribution to (\ref{eq:Lref Ezetai comm}) is
  \begin{equation*}
    \begin{split}
      \textstyle{\int}_\S \varphi E_l(\psi E_mu)vd\mu_{h_{\refer}}
      = & \textstyle{\int}_\S  E_l(v\varphi\psi E_mu)d\mu_{h_{\refer}}
      -\textstyle{\int}_\S  \psi E_mu\cdot E_l(\varphi v)d\mu_{h_{\refer}}.
    \end{split}
  \end{equation*}
  Recalling (\ref{eq:partial int vf}) (with $h$ replaced by $h_{\refer}$) and (\ref{eq:rodiv hrefer Ezetai}) (with $\zeta=1$), it follows that
  \begin{equation}\label{eq:typexempel comm}
    \big|\textstyle{\int}_\S \varphi E_l(\psi E_mu)vd\mu_{h_{\refer}}\big|
    \leq C\|\varphi\|_{W^{1,\infty}}\|\psi\|_{L^{\infty}}\|u\|_{H^1}\|v\|_{H^1},
  \end{equation}
  where $C$ only depends on $\{E_i\}$ and $K$. The contribution of the first two terms on the far right hand side of
  (\ref{eq:first terms in Delta h comm}) to the left hand side of (\ref{eq:Lref Ezetai comm}) can thus be estimated by 
  \begin{equation}\label{eq:typ right hs}
    C\|h\|_{W^{1,\infty}}\|u\|_{H^1}\|v\|_{H^1},
  \end{equation}
  where $C$ only depends on $c_1$, $n$, $\{E_i\}$, $\zeta$ and $K$; here we used the fact that $|h^{ij}|\leq c_1$ and the fact that
  $\|h\|_{L^\infty}$ is bounded from below by a constant depending only on $c_1$ and $n$. Next, the last two terms on the far right
  hand side of
  (\ref{eq:first terms in Delta h comm}) can be written
  \[
  -E_l(\zeta)h^{lj}E_{j}E_{i}u-E_l(h^{lj}E_{j}(\zeta)E_{i}u).
  \]
  The corresponding contribution to the left hand side of (\ref{eq:Lref Ezetai comm}) can be estimated by appealing to
  (\ref{eq:typexempel comm}). This results in terms of the form (\ref{eq:typ right hs}).

  Next, consider the last term on the right hand side of (\ref{eq:DeltahitoEi}). The corresponding commutator is
  \[
  -h^{lj} \lambda_{kl}^{k} E_jE_{\zeta,i}u+E_{\zeta,i}(h^{lj}\lambda_{kl}^{k}E_ju).
  \]
  These terms are of the same form as ones we have previously handled, and the resulting contribution to the estimate is
  the same. Finally, we need to consider
  \begin{equation}\label{eq:problemtermcomm}
    \begin{split}
      & - \tfrac{1}{2}h^{kj} (E_l h_{kj}) h^{lm}E_mE_{\zeta,i}u+\tfrac{1}{2}E_{\zeta,i}(h^{kj}(E_l h_{kj}) h^{lm}E_mu)\\
      = & \tfrac{1}{2}h^{kj} (E_l h_{kj}) h^{lm}[E_{\zeta,i},E_{m}]u+\tfrac{1}{2}E_{\zeta,i}(h^{kj}(E_l h_{kj}) h^{lm})E_mu.
    \end{split}    
  \end{equation}
  The contribution of both of these terms can be estimated directly. Combining the above observations yields the
  estimate
  \begin{equation*}
    \begin{split}
      & |L_{\refer}[E_{\zeta,i} u](v)-E_{\zeta,i}[Lu](v)|\\
      \leq & C \big( \norm{h}_{W^{2, \infty}} + \norm{h}_{W^{1, \infty}}^2
      +\|a\|_{W^{1,\infty}}\big)\norm{u}_{H^1}
      \norm{v}_{H^1},
    \end{split}
  \end{equation*}
  where $C$ only depends on $c_1$, $n$, $\{E_i\}$, $\zeta$ and $K$. By density and the definition of the $H^{-1}$-norm,
  the desired conclusion follows. 
\end{proof}

\begin{prop} \label{prop: bound on the inverse}
  Let $(\S, h_\refer)$ be a closed, connected and oriented Riemannian manifold of dimension $n$. Assume that $h$ is a Riemannian metric on $\S$
  and that $a$ is a smooth function on $\S$, such that $c_{1}^{-1} h_\refer\leq h\leq c_{1} h_\refer$ and $c_{2}^{-1}\leq a\leq c_{2}$ for some constants
  $c_1,c_2 > 1$. Let $L:=-\Delta_h+a$. Assume $u\in H^1(h)$ and $f\in L^2(h)$ to be such that $L[u]=\iota_h[f]$. Then
  $u\in H^2(h)$ and $Lu=f$ a.e. Next, fix a finite partition of unity subordinate to a finite cover of $\S$ by coordinate neighbourhoods. Then
  there is a constant $C$, depending only on $n$, $c_1$, $c_2$, the cover and the partition of unity such that 
  \begin{equation}\label{eq:basic Htwo est}
    \norm{u}_{H^2}
    \leq C \norm{h}_{W^{1,\infty}}\big(\norm{h}_{W^{1,\infty}}^2 + \norm{h}_{W^{2,\infty}}+ \norm{a}_{W^{1,\infty}} \big) \norm{Lu}_{L^2},
  \end{equation}
  where all the norms are defined using $h_\refer$.
\end{prop}
\begin{remark}\label{remark:basic Htwo est low reg}
  In (\ref{eq:basic Htwo est}), it is sufficient to assume $u$ to be $H^2$, $h$ to be $C^{2}$ and $a$ to be $C^{1}$. In order to see this, it is
  sufficient to take a
  sequence of smooth metrics $h_{j}$ converging to $h$ in $C^{2}$; and a sequence of smooth functions $a_{j}$ converging to $a$ in $C^{1}$. Then
  (\ref{eq:basic Htwo est}) applies with $h=h_{j}$; $a=a_{j}$; and $c_1$ and $c_2$ replaced by $2c_1$ and $2c_2$ respectively. Taking the
  limit of this estimate yields the desired conclusion. 
\end{remark}
\begin{proof}
  The statements concerning the regularity of $u$ follow from the assumptions and Lemma~\ref{lemma:elliptic reg}. In order to prove
  (\ref{eq:basic Htwo est}), we recall (\ref{eq:u Hone bd ito Lref u}). Moreover, from now on, we drop all reference to $h_{\refer}$ since all the
  norms in the present proof are calculated with respect to $h_{\refer}$. Let $U$, $\{E_{i}\}$, $\zeta$, $E_{\zeta,i}$ and $K$ be as described at the
  beginning of Subsection~\ref{ssection:loc fr comm}. Applying (\ref{eq:u Hone bd ito Lref u}) with $u$ replaced by $E_{\zeta,i}u$ yields
  \begin{equation}\label{eq:EiuHoestfs}
    \begin{split}
      \norm{E_{\zeta,i}u}_{H^1} \leq & C \|h\|_{W^{1,\infty}}\norm{E_{\zeta,i}[Lu]-L_{\refer}[E_{\zeta,i}u]}_{H^{-1}}
      +C \|h\|_{W^{1,\infty}}\norm{E_{\zeta,i}[Lu]}_{H^{-1}}.
    \end{split}
  \end{equation}
  Combining this estimate with (\ref{eq:EiuHmoest}) and (\ref{eq:DeltahEicommmod}) yields
  \begin{equation}\label{eq:Ezeta i Lu est}
    \begin{split}
      \norm{E_{\zeta,i}u}_{H^1} \leq & C \|h\|_{W^{1,\infty}}\big( \norm{h}_{W^{2, \infty}} + \norm{h}_{W^{1, \infty}}^2
      +\|a\|_{W^{1,\infty}}\big)\norm{u}_{H^1}\\
      & +C \|h\|_{W^{1,\infty}}\norm{Lu}_{L^2},
    \end{split}
  \end{equation}
  where $C$ only depends on $c_1$, $c_2$, $n$, $\{E_i\}$, $\zeta$ and $K$. On the other hand, due to Lemma~\ref{le: equivalence Sobolev norms},
  (\ref{eq:-Delta plus a as element of H minus one}) and Lemma~\ref{le: equivalence volume forms},
  \begin{equation}\label{eq:u Hone L Ltwo}
    \|u\|_{H^1}^2\leq C\|u\|_{H^1(h)}^2\leq C\|Lu\|_{L^2}\|u\|_{L^2}\leq C\|Lu\|_{L^2}\|u\|_{H^1},
  \end{equation}
  where $C$ only depends on $c_1$, $c_2$ and $n$. Combining (\ref{eq:Ezeta i Lu est}) and (\ref{eq:u Hone L Ltwo}) with a finite partition of unity
  subordinate to a finite cover of $\S$ by coordinate neighbourhoods (on the coordinate  neighbourhoods, we can use the Gram-Schmidt procedure
  to construct an orthonormal frame using the coordinate vector fields) yields the desired conclusion. 
\end{proof}

\begin{lemma}[$L^2$-bounds on higher commutators] \label{le: commutator Delta higher derivatives}
  Let $(\S, h_\refer)$ be a closed, connected and oriented Riemannian manifold of dimension $n$. Let $U$, $\{E_{i}\}$, $\zeta$ and $K$ be as described
  at the beginning of Subsection~\ref{ssection:loc fr comm}. Let $\eta\in C^\infty_c(U)$ be such that $0\leq \eta\leq 1$ and such that $\eta(x)=1$ for
  all $x\in\supp\zeta$. Assume $h$ to be a Riemannian metric on $\S$ and $a$ to be a smooth function on $\S$, such that
  $c_{1}^{-1} h_\refer\leq h\leq c_{1} h_\refer$ and $c_{2}^{-1}\leq a\leq c_{2}$ for some constants $c_1,c_2 > 1$. Let $L:=-\Delta_h+a$. Then there is a
  constant $C$, depending only on $c_1$, $c_2$, $n$, $k$, $\{E_i\}$, $\zeta$ and $\eta$, such that if $|\bfI|\leq k$, then, using the notation
  introduced in Subsection~\ref{ssection:mod moser est},
  \begin{equation}\label{eq:commutHkest}
    \norm{[L, E_{\zeta,\bfI}] u}_{L^2}
    \leq C \left( \norm{h}_{W^{1,\infty}}\norm{u}_{H^{k+1}} +\norm{h}_{W^{1,\infty}}\norm{u}_{W^{1,\infty}}\norm{h}_{H^{k+1}} + \norm{u}_{L^\infty}\norm{a}_{H^k} \right)
  \end{equation}
  for all $u\in H^{k+2}$, where all spaces and norms are defined with respect to $h_\refer$.
\end{lemma}
\begin{remark}
  In some situations, it is useful to keep the following simpler estimate in mind:
  \begin{equation}\label{eq:commutHkest primitive}
    \norm{[L, E_{\zeta,\bfI}] u}_{L^2} \leq C \norm{u}_{H^{k+1}}
  \end{equation}
  for all $u\in H^{k+2}$, where all spaces and norms are defined with respect to $h_\refer$. Moreover, the constant $C$ only depends on $c_1$, $c_2$,
  $n$, $k$, $\{E_i\}$, $\zeta$, $\eta$, the $C^{k+1}$-norm of $h$ and the $C^k$ norm of $a$. 
\end{remark}
\begin{proof}
  Combining (\ref{eq:DeltahitoEi}) and (\ref{eq:Deltahftt}) yields
  \begin{equation}\label{eq:Deltahcomp}
    \Delta_hu = E_{i}(h^{ij}E_{j}u) + \tfrac{1}{2}h^{ij} (E_l h_{ij}) h^{lm}E_mu + h^{ij}\lambda_{mi}^{m}E_ju.
  \end{equation}
  In order to estimate $[-\Delta_h,E_{\zeta,\bfI}]u$, it is convenient to estimate the contribution from each of the terms on the right hand side of
  (\ref{eq:Deltahcomp}). To begin with, we need to estimate
  \begin{equation}\label{eq:fitecomm}
    -E_{i}(h^{ij}E_{j}E_{\zeta,\bfI}u)+E_{\zeta,\bfI}E_{i}(h^{ij}E_{j}u)=[E_{\zeta,\bfI},E_{i}](h^{ij}E_{j}u)+E_{i}(E_{\zeta,\bfI}(h^{ij}E_{j}u)-h^{ij}E_{j}E_{\zeta,\bfI}u).
  \end{equation}
  We leave it to the reader to verify that
  \[
  [E_{\zeta,\bfI},E_{i}]=\textstyle{\sum}_{|\bfJ|\leq |\bfI|}f_{\bfI,i}^{\bfJ}E_{\bfJ}=\textstyle{\sum}_{|\bfJ|\leq |\bfI|}f_{\bfI,i}^{\bfJ}E_{\eta,\bfJ}
  \]
  for some functions $f_{\bfI,i}^{\bfJ}$ vanishing outside of the support of $\zeta$ (a property which justifies the last equality) and depending only on
  the frame and $\zeta$ (including derivatives of $\zeta$). In order to estimate the first term on the right hand side of (\ref{eq:fitecomm}), it is thus
  sufficient to estimate expressions of the form
  \[
    \|E_{\eta,\bfJ_1}h_{k_1 l_1}\cdots E_{\eta,\bfJ_m}h_{k_m l_m}\cdot E_{\eta,\bfJ_0}E_{\eta,j}u\|_{L^2},
  \]
  where $|\bfJ_0|+\dots+|\bfJ_m|\leq |\bfI|$ and we used the fact that $|h^{ij}|$ can be estimated by $c_1$. Appealing to (\ref{eq:moserest mod}), we
  conclude that the first term on the right hand side of (\ref{eq:fitecomm}) can be estimated by
  \[
  C\|u\|_{H^{k+1}}+C\|u\|_{W^{1,\infty}}\|h\|_{H^{k}},
  \]  
  where $C$ only depends on $c_1$, $n$, $k$, $\{E_i\}$ $\zeta$ and $\eta$, where we used the fact that $|h_{ij}|\leq c_1$. It can also be estimated
  by the right hand side of (\ref{eq:commutHkest primitive}). Expanding the second term on the right hand side of (\ref{eq:fitecomm}) yields two types of terms:
  \[
  E_{i}(E_{\zeta,\bfJ}(h^{ij})E_{\zeta,\bfK}E_{j}u), \ \ \ E_{i}(h^{ij}[E_{\zeta,\bfI},E_{j}]u),
  \]
  where $0<|\bfJ|\leq|\bfI|$ and $|\bfJ|+|\bfK|=|\bfI|$ in the first expression. Due to (\ref{eq:moserest mod}) and arguments similar to the above, these
  expressions can, in $L^{2}$, be estimated by
  \begin{equation}\label{eq:Deltahcommestmainterm}
    C\|h\|_{W^{1,\infty}}(\|u\|_{H^{k+1}}+\|u\|_{W^{1,\infty}}\|h\|_{H^{k+1}}),
  \end{equation}
  where $C$ only depends on $c_1$, $n$, $k$, $\{E_i\}$ $\zeta$ and $\eta$. They can also be estimated by the right hand side of
  (\ref{eq:commutHkest primitive}). Next, we need to consider the commutator arising from the second term on the right hand side of
  (\ref{eq:Deltahcomp}). In practice, we need to estimate
  \[
  E_{\zeta,\bfI}(h^{ij}E_{l}(h_{ij})h^{lm}E_{m}u)-h^{ij}E_{l}(h_{ij})h^{lm}E_{m}E_{\zeta,\bfI}u.
  \]
  By arguments similar to the above, this expression can be estimated by (\ref{eq:Deltahcommestmainterm}) in $L^{2}$. It can also be estimated
  by the right hand side of (\ref{eq:commutHkest primitive}). The estimate for the contribution to the commutator from the last term on the right
  hand side of (\ref{eq:Deltahcomp}) is even better. Finally, the contribution from $[E_{\zeta,\bfI},a]$ can be estimated by 
  \[
  C\|u\|_{H^{k}}+C\|u\|_{L^{\infty}}\|a\|_{H^{k}}
  \]
  as well as by the right hand side of (\ref{eq:commutHkest primitive}). The lemma and remark follow, keeping in mind that $\norm{h}_{W^{1,\infty}}$ is
  bounded by below by a constant depending only on $c_{1}$ and $n$. 
\end{proof}

We are finally in a position to prove Theorem~\ref{thm: elliptic estimate}:

\begin{proof}[Proof of Theorem \ref{thm: elliptic estimate}]
  The regularity statement follows immediately from Lemma~\ref{lemma:elliptic reg}. Next, applying Proposition \ref{prop: bound on the inverse} to
  $E_{\zeta,\bfI} u$ yields
  \begin{equation}\label{eq:Hkptwoestfstep}
    \begin{split}
      \norm{E_{\zeta,\bfI} u}_{H^2}
      &\leq C \norm{LE_{\zeta,\bfI} u}_{L^2} \leq C \norm{ E_{\zeta,\bfI} Lu}_{L^2} +  C \norm{ [L, E_{\zeta,\bfI}] u}_{L^2},
    \end{split}
  \end{equation}  
  where $C$ only depends on a cover of $\S$, a partition of unity, $n$, $c_{1}$, $c_{2}$, $\|h\|_{W^{2,\infty}}$ and $\|a\|_{W^{1,\infty}}$, and the
  dependence on the last two quantities is polynomial. By Lemma \ref{le: commutator Delta higher derivatives}, it follows that if $|\bfI|\leq k$, then
  \begin{equation}\label{eq:Hkptwoestfirstste}
    \norm{E_{\zeta,\bfI} u}_{H^{2}} \leq C\left( \norm{Lu}_{H^k} + \norm{u}_{H^{k+1}}
    + \norm{u}_{W^{1, \infty}} \norm{h}_{H^{k+1}}  + \norm{u}_{L^\infty} \norm{a}_{H^k} \right)
  \end{equation}
  for any $k \geq 0$, where $C$ only depends on a cover of $\S$, a partition of unity, $n$, $k$, $c_{1}$, $c_{2}$, $\{E_{i}\}$, $\zeta$, $\eta$,
  $\|h\|_{W^{2,\infty}}$ and $\|a\|_{W^{1,\infty}}$, and the dependence on the last two quantities is polynomial. This means that 
  \begin{equation}\label{eq:indassHkptwoest}
    \norm{u}_{H^{k+2}}\leq C\left( \norm{Lu}_{H^k} +  \norm{u}_{W^{1, \infty}}\norm{h}_{H^{k+1}}
    + \norm{u}_{L^\infty}\norm{a}_{H^k}  + \norm{u}_{H^{k+1}} \right),
  \end{equation}
  where $C$ only depends on $n$, $k$, $c_{1}$, $c_{2}$, a finite cover of $\S$, a finite partition of unity, $\|h\|_{W^{2,\infty}}$ and $\|a\|_{W^{1,\infty}}$,
  the dependence
  on the last two quantities being polynomial. Next, we wish to prove that (\ref{eq:indassHkptwoest}) holds with the $H^{k+1}$-norm of $u$ on the right hand
  side replaced by the $H^1$ norm. If $k=0$, this follows from (\ref{eq:indassHkptwoest}). Assume that an estimate of this form holds for some $k\geq 0$.
  Combining (\ref{eq:indassHkptwoest}), with $k$ replaced by $k+1$, with the inductive hypothesis proves that the inductive hypothesis is satisfied with
  $k$ replaced by $k+1$. Thus the desired estimate holds for all $k\geq 0$. Estimating the
  $H^{1}$-norm of $u$ appearing on the right hand side by appealing to Proposition \ref{prop: bound on the inverse} yields
  \begin{equation}\label{eq:indassHkptwoestnoHo}
    \norm{u}_{H^{k+2}}\leq C\left( \norm{Lu}_{H^k} +  \norm{u}_{W^{1, \infty}}\norm{h}_{H^{k+1}}
    + \norm{u}_{L^\infty}\norm{a}_{H^k} \right),
  \end{equation}
  where $C$ has the same dependence as in the case of (\ref{eq:indassHkptwoest}). The estimate (\ref{eq:mainHkptwoestell}) follows. In order to take the step
  from (\ref{eq:mainHkptwoestell}) to (\ref{eq:Hkptwonormu}), it is sufficient to appeal to Theorem~\ref{thm:globalSchauder}.
  
  In order to prove (\ref{eq:Hkestcrude}), note that the following interpolation estimate holds: Given integers $k$ and $l$ such that
  $0\leq l\leq k$ and $k\geq 1$, there is a constant $C$, depending only on $k$, $l$, $(\S,h_{\refer})$ and a partition of unity, such that
  \begin{equation}\label{eq:psiinterpol}
    \|\psi\|_{H^{l}}\leq C\|\psi\|_{2}^{1-l/k}\|\psi\|_{H^{k}}^{l/k}
  \end{equation}
  for all $\psi\in H^k$ on $\S$. In the case of parallelisable manifolds, a proof of this statement is to be found in \cite[Lemma~120, p.~72]{GPR}. In
  the general case, we can use a partition of unity to divide $\psi$ into a finite number of parts that have compact support in sets that can be considered
  to be subsets of the $n$-torus. Applying (\ref{eq:psiinterpol}) in the case of the $n$-torus then yields (\ref{eq:psiinterpol}) for closed manifolds that
  are not parallelisable. Applying (\ref{eq:psiinterpol}) with $l=k-1$ and $\psi$ replaced by $E_{\zeta,\bfI}\psi$ for $|\bfI|\leq 2$ yields
  \[
    \|E_{\zeta,\bfI}\psi\|_{H^{k-1}}\leq C\|\psi\|_{H^{2}}^{1/k}\|\psi\|_{H^{k+2}}^{1-1/k}.
  \]
  From this estimate, we deduce that
  \[
    \|\psi\|_{H^{k+1}}\leq C\|\psi\|_{H^{2}}^{1/k}\|\psi\|_{H^{k+2}}^{1-1/k}.
  \]
  Combining this estimate with (\ref{eq:indassHkptwoestnoHo}) and Sobolev embedding yields, if $k>n/2$, 
  \begin{equation}\label{eq:u H k plus two aux est}
    \begin{split}
      \norm{u}_{H^{k+2}} %\leq & C\left(\norm{Lu}_{H^k}+\norm{u}_{W^{1, \infty}}\norm{h}_{H^{k+1}}+ \norm{u}_{L^\infty}\norm{a}_{H^k} \right)\\
      \leq & C\norm{Lu}_{H^k}+C(\norm{h}_{H^{k+1}}+\norm{a}_{H^k})\norm{u}_{H^{k+1}}\\
      \leq & C\norm{Lu}_{H^k}+C(\norm{h}_{H^{k+1}}+\norm{a}_{H^k})\|u\|_{H^{2}}^{1/k}\|u\|_{H^{k+2}}^{1-1/k}\\
      \leq & C\norm{Lu}_{H^k}+\tfrac{1}{k}C^{k}(\norm{h}_{H^{k+1}}+\norm{a}_{H^k})^{k}\|u\|_{H^{2}}
      +\tfrac{k-1}{k}\|u\|_{H^{k+2}}.
    \end{split}
  \end{equation}
  The second term on the far right hand side can be estimated by appealing to Proposition~\ref{prop: bound on the inverse}. The third term can be
  taken over to the left hand side. This proves (\ref{eq:Hkestcrude}).

  In order to prove (\ref{eq:u Hk plus two Lu Hk}), note that combining (\ref{eq:Hkptwoestfstep}) with (\ref{eq:commutHkest primitive}) yields
  \begin{equation}\label{eq:indassHkptwoest prim}
    \norm{u}_{H^{k+2}}\leq C( \norm{Lu}_{H^k} +  \|u\|_{H^{k+1}}),
  \end{equation}
  where $C$ only depends on $n$, $k$, $c_{1}$, $c_{2}$, a finite cover of $\S$, a finite partition of unity, the $C^{l_k+1}$-norm of $h$, where
  $l_k:=\max\{k,1\}$, and the $C^{l_k}$-norm of $a$. At this point, an estimate analogous with (\ref{eq:u H k plus two aux est}) combined with
  Proposition~\ref{prop: bound on the inverse} yields (\ref{eq:u Hk plus two Lu Hk}). 
\end{proof}

\subsection{Existence of solutions}

Since existence of solutions to the elliptic equations of interest here is an immediate consequence of the observations made in the present
appendix, and since we did not find a good reference for the statements we need we, next, write down the relevant existence result, even though
it is standard.

\begin{lemma}\label{lemma:existence sol to elliptic PDE}
  Let $(\Sigma,h_{\refer})$ be a closed, connected and oriented Riemannian manifold of dimension $n$. Let $k>n/2+1$ and assume
  $h\in H^{k+1}[T^{(0,2)}T\S;h_{\refer}]$ to be a Riemannian metric and $a\in H^{k}(h_{\refer})$ to be strictly positive. Assume that $f\in L^2(h_{\refer})$. Then
  there is a unique weak solution $u\in H^1(h_{\refer})$ to $L[u]=\iota_h[f]$. Moreover, if $f\in H^k(h_{\refer})$, then $u\in H^{k+2}(h_{\refer})$. 

  Next, let $c_1>1$ and $c_2>1$ be constants and $(\S,h_{\refer})$, $n$ and $k$ be as above. Then there are polynomials
  $\mfD_{1}:\ro_{+}^{2}\rightarrow\ro_{+}$ and $\mfD_{2}:\ro_{+}^{4}\rightarrow\ro_{+}$ with non-negative coefficients; and a continuous function
  $\mfD_{3}:\ro_{+}^{3}\rightarrow\ro_{+}$ which is monotonically increasing in each of its arguments. Here $\mfD_{i}$, $i=1,2,3$, only depend on $c_{1}$,
  $c_{2}$, $n$, $k$ and $(\S,h_{\refer})$. Moreover, if $h$ and $a$ are as above; $c_1^{-1}h_{\refer}\leq h\leq c_1 h_{\refer}$; $c_2^{-1}\leq a\leq c_2$;
  $f\in H^k(h_{\refer})$; and $u\in H^{k+2}(h_{\refer})$ is the solution obtained in the first part of the lemma, then (\ref{eq:mainHkptwoestell}),
  (\ref{eq:Hkestcrude}) and (\ref{eq:Hkptwonormu}) hold. 
\end{lemma}
\begin{remark}
  By the assumptions and Sobolev embedding, we can assume $h$ to be $C^2$ and $a$ to be $C^1$. In particular, $a$ has a strictly positive minimum
  and a finite maximum. 
\end{remark}
\begin{remark}
  Due Lemma~\ref{lemma:elliptic reg}, it follows that if $h$ and $a$ are smooth and $f\in H^l(h_{\refer})$ for some $l\in\nn{}_0$, then
  $u\in H^{l+2}(h_{\refer})$. However, we are here mainly interested in the case that the coefficients are not smooth. 
\end{remark}
\begin{remark}\label{remark:coeff smooth sol smooth}
  If $h$, $a$ and $f$ are all smooth, then an immediate consequence of the lemma is that $u$ is smooth.
\end{remark}
\begin{proof}
  As noted in Subsection~\ref{ssection:sob sp}, $H:=H^1(h)$ is a Hilbert space. Define $B:H\times H\rightarrow \ro$ by
  \[
  B[u,v]:=\textstyle{\int}_\S (h(\grad_hu,\grad_hv)+auv)d\mu_h.
  \]
  By assumption, there is a constant $c_2>1$ such that $c_2^{-1}\leq a\leq c_2$. Thus
  \[
  |B[u,v]|\leq c_2\|u\|_H\|v\|_H,\ \ \
  B[u,u]\geq c_2^{-1}\|u\|_H^2.
  \]
  Next, let $f\in L^2(h)$ and define $F:H\rightarrow \ro$ by  
  \[
  F(v):=\textstyle{\int}_\S fvd\mu_h.
  \]
  Then $F$ is a bounded linear functional. By the Lax-Milgram Theorem, see, e.g., \cite[Theorem~1, p.~315--316]{evans}, it follows that there is a
  unique element $u\in H$ such that $B[u,v]=F(v)$ for all $v\in H$. Due to Lemma~\ref{lemma:elliptic reg}, it follows that if $a$ and $h$ are smooth,
  and $f\in H^l$ for some $l\in\nn{}_0$, then $u\in H^{l+2}$.
  
  Assume now that $h\in H^{k+1}[T^{(0,2)}T\S;h_{\refer}]$ is a Riemannian metric, that $a\in H^{k}(h_{\refer})$ is strictly positive and that $f\in H^k(h_{\refer})$.
  Then there are constants $c_1>1$ and $c_2>1$ such that
  $c_1^{-1}h_{\refer}\leq h\leq c_1 h_{\refer}$ and $c_2^{-1}\leq a\leq c_2$. Next, choose a sequence $h_m$ of smooth Riemannian metrics converging to $h$
  in the $H^{k+1}$-norm and a sequence of smooth strictly positive functions $a_m$ converging to $a$ in the $H^k$-norm. Let $L_m:=-\Delta_{h_m}+a_m$. Then
  there is, by the first part of the proof, functions $u_m\in H^{k+2}(h_{\refer})$ such that $L_mu_m=f$. Letting $b_i=2c_i$, $i=1,2$, we can assume that
  \[
  b_1^{-1}h_{\refer}\leq h_m\leq b_1h_{\refer},\ \ \
  b_2^{-1}\leq a_m\leq b_2
  \]
  for all $m$. Due to Theorem~\ref{thm: elliptic estimate}, there is then a polynomial $\mfD_2:\ro_+^4\rightarrow\ro_+$ with non-negative coefficients,
  depending only on $b_1$, $b_2$, $n$, $k$ and $(\S,h_{\refer})$, such that 
  \[
  \norm{u_m}_{H^{k+2}} \leq \mfD_{2}[\norm{h_m}_{W^{2, \infty}},\norm{a_m}_{W^{1, \infty}},\norm{a_m}_{H^k},\norm{h_m}_{H^{k+1}}] \norm{f}_{H^k}.
  \]
  Since the unit ball in $H^k$ is compact in the weak topology, there is a subsequence of $u_m$, which we still denote $u_{m}$, converging weakly to,
  say, $u$. Moreover, due to the fact that $k>n/2+1$, Sobolev embedding and the fact that $\mfD_2$ is a polynomial, (\ref{eq:Hkestcrude}) holds. 
  Next, due to the Rellich-Kondrakhov theorem, Theorem~\ref{thm:Rellich Kondrakhov}, and Remark~\ref{remark:Rellich Kondrakhov dim one} there is a
  subsequence of $u_m$, which we still denote $u_{m}$, converging in $H^{k+1}$. This means that $u$ is a solution to $Lu=f$. Moreover, combining the
  above observations with Theorem~\ref{thm: elliptic estimate}, we conclude that there is then a polynomial $\mfD_1:\ro_+^2\rightarrow\ro_+$ with
  non-negative coefficients, depending only on $b_1$, $b_2$, $n$, $k$ and $(\S,h_{\refer})$, such that (\ref{eq:mainHkptwoestell}) holds. 
  Next, note that $u\in C^3$ and $f\in C^1$. Moreover, due to Theorem~\ref{thm: elliptic estimate}, there is a continuous function
  $\mfD_3:\ro_+^3\rightarrow\ro_+$, depending only on $b_1$, $b_2$, $n$, $k$ and $(\S,h_{\refer})$, such that (\ref{eq:Hkptwonormu}) holds.
\end{proof}
Next, we prove that there is a complete set of eigenfunctions of the Laplacian. 
\begin{thm}\label{thm:complete set of evs}
  Let $(M,h)$ be a closed, connected and oriented Riemannian manifold of dimension $n$. Then each eigenvalue of $-\Delta_h$ is real and has finite
  multiplicity. Moreover, repeating each eigenvalue according to its multiplicity, the set of eigenvalues is countably infinite and can be ordered
  \[
  0=\lambda_0<\lambda_1\leq \lambda_2\leq\dots.
  \]
  In addition, $\lambda_k\rightarrow\infty$ as $k\rightarrow\infty$. Finally, there is an orthonormal basis $\{u_k\}_{k=0}^\infty$ of $L^2(h)$,
  where $u_0$ is constant; $u_k\in H^1_0(h)$ for $k\geq 1$; and $u_k$ is an eigenfunction of $-\Delta_h$ corresponding to $\lambda_k$.
\end{thm}
\begin{remark}\label{remark:efns smooth}
  Due to Lemma~\ref{lemma:elliptic reg}, the eigenfunctions $u_k$ are smooth.
\end{remark}
\begin{proof}
  Due to Lemma~\ref{lemma:existence sol to elliptic PDE}, we can solve the equation $-\Delta_hu+u=f$ for $f\in L^2(h)$. Then
  \begin{equation}\label{eq:weak sol poisson eq}
    \textstyle{\int}_M[h(\grad_hv,\grad_hu)+uv]d\mu_h=\textstyle{\int}_Mfvd\mu_h
  \end{equation}
  for all $v\in H^1(h)$. In particular, we can define an operator $S:L^2(h)\rightarrow H^1(h)$ taking $f$ to the solution $u$ to
  $-\Delta_hu+u=f$. Due to (\ref{eq:weak sol poisson eq}) with $v=u$ and the Cauchy-Schwartz inequality, it is clear that $S$ is
  continuous. If we consider $S$ to be an operator from $L^2(h)$ to itself, it follows from Theorem~\ref{thm:Rellich Kondrakhov}
  and Remark~\ref{remark:Rellich Kondrakhov dim one} that $S$ is a bounded, linear and compact operator from $L^2(h)$ to itself.
  In order to prove that $S$ is symmetric, let $f,g\in L^2(h)$; $u:=Sf$; and $v:=Sg$. Then 
  \[
  \textstyle{\int}_MSf\cdot gd\mu_h=\textstyle{\int}_Mu\cdot gd\mu_h=\textstyle{\int}_M[h(\grad_hv,\grad_hu)+uv]d\mu_h
  \]
  and
  \[
  \textstyle{\int}_Mf\cdot Sgd\mu_h=\textstyle{\int}_Mf\cdot vd\mu_h=\textstyle{\int}_M[h(\grad_hv,\grad_hu)+uv]d\mu_h.
  \]
  Thus $S$ is symmetric. Finally, $(Sf,f)\geq 0$, where $(\cdot,\cdot)$ denotes the $L^2$-inner product. At this point we can, just as in
  the proof of \cite[Theorem~1, p.~355]{evans}, appeal to \cite[§D.6]{evans} in order to obtain the desired conclusion. More specifically, note
  that $\mH:=L^2(h)$ is an infinite dimensional separable Hilbert space. Due to \cite[Theorem~7, p.~728]{evans}, there is a countable orthonormal
  basis of $\mH$ consisting of eigenvectors of $S$. Let, as in the proof of \cite[Theorem~7, p.~728]{evans}, $\{\eta_k\}$ denote the sequence of distinct
  eigenvalues of $S$, excepting $0$. On the other hand, if $Sf=\mu f$ for some $\mu\in\ro$ and $f\in L^2(h)$, then, if $u:=Sf$,
  \begin{equation}\label{eq:Sf mu estimate}
    \|Sf\|_{L^2(h)}^2\leq \|u\|_{H^1(h)}^2=(Sf,f)=\mu (f,f). 
  \end{equation}
  In particular, $\mu\geq 0$ and if $\mu=0$, then $u=0$, so that $f=0$. Thus $0$ is not an eigenvalue of $S$. If $\mu>0$, then $f=\mu^{-1}Sf$ and,
  since $f\neq 0$
  for an eigenvector, (\ref{eq:Sf mu estimate}) yields the conclusion that $\mu\leq 1$. In other words, if $\mu$ is an eigenvalue of $S$, it has
  to satisfy $0<\mu\leq 1$. Next, let $\mH_k:=N(S-\eta_k I)$ denote the null space
  of $S-\eta_k I$. Then $0<\dim \mH_k<\infty$; see the proof of \cite[Theorem~7, p.~728]{evans}. There are thus infinitely many eigenvalues and we
  can order them $1\geq \eta_1>\eta_2>\cdots>0$. Due to \cite[Theorem~6, p.~727]{evans}, $\eta_k$ is a sequence tending to $0$. On the other hand, 
  $\eta_k>0$ is an eigenvalue of $S$ if and only if $1/\eta_k-1$ is an eigenvalue of $-\Delta_h$ (and $u$ is an eigenvector of $S$ if and only if it
  is an eigenvector of $-\Delta_h$). Note also that $-\Delta_h$ only has one unit eigenfunction with zero eigenvalue. The theorem follows. 
\end{proof}

\section{A spectral definition of Sobolev spaces}\label{appendix:spec def sob}

Next, we give a spectral definition of Sobolev spaces and verify that this definition gives rise to equivalent norms in the case of non-negative
integer Sobolev exponents. 

\subsection{A spectral definition of Sobolev spaces}\label{ssection:specsob}
Let $(\S,h)$ be a closed, connected and oriented Riemannian manifold of dimension $n$. Then $\Delta_{h}$ has eigenvalues $\{-\lambda_{i}^{2}\}$,
where $0=\lambda_{0}<\lambda_{1}\leq\dots$, and a corresponding complete set of orthonormal eigenfunctions $\varphi_{i}$, $i=0,1,\dots$, with respect to
the $L^2(h)$-inner product; see Theorem~\ref{thm:complete set of evs}. Moreover, the $\varphi_i$ are smooth; see Remark~\ref{remark:efns smooth}.

We define a distribution on $\S$ to be a map from $C^{\infty}(\S)$ to $\ro$, say $u$, such that there is a $k\in \nn{}_0$ with the
property that
\[
|u(\psi)|\leq C\|\psi\|_{C^{k}(h)}
\]
for all $\psi\in C^{\infty}(\S)$. We denote the set of distributions on $\S$ by $\mD'(\S)$. We define the set of $\rn{l}$-valued
distributions similarly, and we denote them by $\mD'(\S,\rn{l})$. Given $u\in \mD'(\S,\rn{l})$ and $i\in \nn{}_0$, define
\begin{equation}\label{eq:huidef}
  \hu(i):=u(\varphi_{i}).
\end{equation}
Let $s\in\rn{}$. We say that $u\in H_{(s)}(h;\rn{l})$ if $u\in\mD'(\S,\rn{l})$ and
\[
\textstyle{\sum}_{i=0}^{\infty}\ldr{\lambda_{i}}^{2s}|\hu(i)|^{2}<\infty.
\]
Moreover, we use the notation $H_{(s)}(h):=H_{(s)}(h;\ro)$ and 
\begin{equation}\label{eq:specsobsp}
  \|u\|_{(s)}:=\left(\textstyle{\sum}_{i=0}^{\infty}\ldr{\lambda_{i}}^{2s}|\hu(i)|^{2}\right)^{1/2}.
\end{equation}
Next, note that, since $\{\varphi_i\}$ is complete, if $u,v\in L^{2}(h)$, then
\begin{equation}\label{eq:L two ip ito fc}
  (u,v)_{L^2(h)}:=\textstyle{\int}_{\S}u\cdot vd\mu_{h}=\textstyle{\sum}_{i=0}^{\infty}\hu(i)\cdot\hv(i),
\end{equation}
where the dot denotes the standard Euclidean inner product. This inner product naturally generalises to an inner product on $H_{(s)}(h;\rn{l})$:
\[
(u,v)_{(s)}:=\textstyle{\sum}_{i=0}^{\infty}\ldr{\lambda_{i}}^{2s}\hu(i)\cdot\hv(i). 
\]
Next, note that if $u\in L^{2}(h)$, then we can think of $u$ as an element of $\mD'(\S)$ via
\[
u(\phi):=\textstyle{\int}_{\S}\phi ud\mu_{h}.
\]
Moreover, if $u\in C^\infty(\S)$, $i,k\in\nn{}_0$ and $u_k:=(-\Delta_h+1)^ku$, then, by partial integration,
$\hu_k(i)=\ldr{\lambda_i}^{2k}\hu(i)$. Thus, given $s\in\ro$ and $k\in\nn{}_0$ such that $s\leq 2k$, 
\begin{equation}\label{eq:c infty in Hs}
  \textstyle{\sum}_{i=0}^{\infty}\ldr{\lambda_{i}}^{2s}|\hu(i)|^{2}
  \leq \textstyle{\sum}_{i=0}^{\infty}\ldr{\lambda_{i}}^{4k}|\hu(i)|^{2}
  = \textstyle{\sum}_{i=0}^{\infty}|\hu_k(i)|^{2}=\|u_k\|_{L^2(h)}^2<\infty.
\end{equation}
In particular, $C^\infty(\S)\subset H_{(s)}(h)$. Next, let $s\in\ro$ and $\a_i\in\rn{l}$ for $i\in\nn{}_0$ be such that
\[
  \left(\textstyle{\sum}_{i=0}^{\infty}\ldr{\lambda_{i}}^{2s}|\a_i|^{2}\right)^{1/2}<\infty.
\]
We then wish to prove that there is a $u\in H_{(s)}(h;\rn{l})$ such that $\hu(i)=\a_i$. Define, for $\psi\in C^{\infty}(\S)$,
\[
  u(\psi):=\textstyle{\sum}_i\hat{\psi}(i)\a_i.
\]
Given a $k\in\nn{}_0$ such that $-s\leq 2k$, H\"{o}lder's inequality then yields
\begin{equation*}
  \begin{split}
    |u(\psi)| \leq & \left(\textstyle{\sum}_{i=0}^{\infty}\ldr{\lambda_{i}}^{2s}|\a_i|^{2}\right)^{1/2}
    \left(\textstyle{\sum}_{i=0}^{\infty}\ldr{\lambda_{i}}^{-2s}|\hat{\psi}(i)|^{2}\right)^{1/2}\\
    \leq & \left(\textstyle{\sum}_{i=0}^{\infty}\ldr{\lambda_{i}}^{2s}|\a_i|^{2}\right)^{1/2}\|(-\Delta_h+1)^k\psi\|_{L^2(h)},
\end{split}
\end{equation*}
where we appealed to an estimate analogous to (\ref{eq:c infty in Hs}) in the last step. This means that $u\in \mD'(\S,\rn{l})$.
Moreover, $\hu(i)=\a_i$, so that $u\in H_{(s)}(h;\rn{l})$. Due to the completeness of $\ell^2$, we conclude that $H_{(s)}(h;\rn{l})$
is a Hilbert space. Note also that if $u\in H_{(s)}(h;\rn{l})$ and
\[
  u_k:=\textstyle{\sum}_{i=0}^k\hu(i)\varphi_i
\]
for $k\in\nn{}_0$, then $u_k\in C^\infty(\S,\rn{l})$ and $\|u_k-u\|_{(s)}\rightarrow 0$ as $k\rightarrow\infty$. Thus
$C^\infty(\S,\rn{l})$ is dense in $H_{(s)}(h;\rn{l})$. Next, note that the pairing (\ref{eq:L two ip ito fc}) can be generalised to
a pairing
\[
  \ldr{\cdot,\cdot}:H_{(s)}(h;\rn{l})\times H_{(-s)}(h;\rn{l})\rightarrow\ro,
\]
defined, for $u\in H_{(s)}(h;\rn{l})$ and $v\in H_{(-s)}(h;\rn{l})$, by 
\[
\ldr{u,v}:=\textstyle{\sum}_{i=0}^{\infty}\hu(i)\cdot\hv(i),
\]
where $\cdot$ denotes the standard Euclidean inner product; due to Hölder's inequality
\[
|\ldr{u,v}|\leq \|u\|_{(s)}\|v\|_{(-s)}.
\]

\textbf{Interpolation.} Say that $u\in H_{(s_{0})}(h;\rn{l})\cap H_{(s_{1})}(h;\rn{l})$, where $s_{0}<s_{1}$. Assume, moreover, that $s\in (s_{0},s_{1})$.
Then there is a $\tau\in (0,1)$ such that $s=\tau s_{1}+(1-\tau)s_{0}$. Let $p:=1/\tau$ and $q=1/(1-\tau)$. Then $p,q>1$ and $1/p+1/q=1$. Appealing to
Hölder's inequality then yields
\[
\|u\|_{(s)}^{2}=\textstyle{\sum}_{i=0}^{\infty}\ldr{\lambda_{i}}^{2\tau s_{1}}|\hu(i)|^{2\tau}\ldr{\lambda_{i}}^{2(1-\tau)s_{0}}|\hu(i)|^{2(1-\tau)}
\leq \|u\|_{(s_{1})}^{2\tau}\|u\|_{(s_{0})}^{2(1-\tau)}.
\]
In particular,
\begin{equation}\label{eq:Hsinterpol}
  \|u\|_{(s)}\leq \|u\|_{(s_{1})}^{\tau}\|u\|_{(s_{0})}^{1-\tau}.
\end{equation}

\textbf{The dual space.} Assume that $f$ is an element of the dual of $H_{(s)}(h)$ and let $\|f\|$ denote the norm of $f$ considered as a bounded
linear operator from $H_{(s)}(h)$ to $\ro$. Let, given $N\in\nn{}_0$, 
\[
\a_i:=\ldr{\lambda_{i}}^{-2s}\hf(i), \ \ \
\psi_{N}:=\textstyle{\sum}_{i=0}^{N}\a_i\varphi_{i};
\]
since $C^\infty(\S)\subset H_{(s)}(h)$, we can define $\hf(i)$ by (\ref{eq:huidef}). By definition, $|f(\psi_{N})|\leq \|f\|\cdot \|\psi_{N}\|_{(s)}$.
Moreover, 
\[
f(\psi_{N})=\textstyle{\sum}_{i=0}^{N}\ldr{\lambda_{i}}^{2s}\a_i^2,\ \ \
\|\psi_{N}\|_{(s)}=\big(\sum_{i=0}^{N}\ldr{\lambda_{i}}^{2s}\a_i^{2}\big)^{1/2}.
\]
Combining these observations yields $\|\psi_{N}\|_{(s)}\leq \|f\|$ for all $N\in\nn{}_0$. This means that
\[
\textstyle{\sum}_i\ldr{\lambda_i}^{-2s}|\hf(i)|^2\leq \|f\|^2<\infty,
\]
so that $\chi:=\textstyle{\sum}_{i=0}^{\infty}\hf(i)\varphi_{i}$ defines an element of $H_{(-s)}(h)$ with $\|\chi\|_{(-s)}\leq \|f\|$. Moreover, we can
associate an element $f_{\chi}$ in the dual of $H_{(s)}(h)$ with $\chi$; it is given by
\begin{equation}\label{eq:fchidef}
  f_{\chi}(\phi):=\ldr{\phi,\chi}=\textstyle{\sum}_{i=0}^{\infty}\hat{\chi}(i)\hat{\phi}(i)
  =\sum_{i=0}^{\infty}\hf(i)\hat{\phi}(i)=f(\phi).
\end{equation}
Note also that $\|f_\chi\|\leq \|\chi\|_{(-s)}$. In other words, the dual of $H_{(s)}(h)$ can be identified with $H_{(-s)}(h)$. 

In case $f$ is an element of the dual of $H_{(s)}(h;\rn{k})$, define $f_j$ in the dual of $H_{(s)}(h)$ by $f_j(\phi):=f(\phi e_j)$, where
$e_j$ is the $j$'th element of the standard basis of $\rn{k}$. By the above observations, there is then a $\chi_j\in H_{(-s)}(h)$ such
that $f_j(\phi)=\ldr{\phi,\chi_j}$. This means that
\begin{equation}\label{eq:fchidef md}
  f(\phi)=f\big(\textstyle{\sum}_j\phi_j e_j\big)=\textstyle{\sum}_jf_j(\phi_j)=\textstyle{\sum}_j\ldr{\phi_j,\chi_j}
  =\textstyle{\sum}_i\hat{\chi}(i)\cdot\hat{\phi}(i)=:\ldr{\phi,\chi},
\end{equation}
where $\chi:=\sum_j\chi_je_j\in H_{(-s)}(h;\rn{k})$. 

\subsection{Relating different definitions of Sobolev spaces}\label{subsection:relatingsobolevnorms}
Let $(\S,h)$ be a closed, connected and oriented Riemannian manifold of dimension $n$. Above, we define the norm associated with the Sobolev
space $H_{(s)}(h)$ and in Subsection~\ref{ssection:sob sp}, we introduce the norm $\|\cdot\|_{H^{m}(h)}$. It is of interest to relate the two
when $s=m\in\nn{}_0$. Note, to this end, that for $\psi\in C^\infty(\S)$, 
\[
\ldr{(-\Delta_{h}+1)^{m}\psi,\varphi_{i}}=\ldr{\lambda_{i}}^{2m}\hat{\psi}(i).
\]
This means that
\[
\|\psi\|_{(m)}^{2}=\textstyle{\sum}_{i}\ldr{\lambda_{i}}^{2m}|\hat{\psi}(i)|^2=\ldr{(-\Delta_{h}+1)^{m}\psi,\psi}.
\]
If $m$ is even, we can integrate by parts $m/2$ times in order to obtain $\|\psi\|_{(m)}\leq C\|\psi\|_{H^{m}(h)}$ for a constant $C$ depending
only on $m$ and $(\S,h)$. In case $m=2k+1$, we can integrate by parts $k$ times in order to obtain
\begin{equation}\label{eq:spec norm k odd}
  \begin{split}
    \|\psi\|_{(m)}^{2} = & \textstyle{\int}_\S(-\Delta_{h}+1)^{k+1}\psi\cdot(-\Delta_{h}+1)^{k}\psi\Vol_h\\
    = & \textstyle{\int}_\S(-\Delta_{h}+1)^{k}\psi\cdot(-\Delta_{h}+1)^{k}\psi\Vol_h+\textstyle{\int}_\S|\nabla^h(-\Delta_{h}+1)^{k}\psi|^2\Vol_h.
  \end{split}
\end{equation}
Again, we obtain the desired estimate:
\[
\|\psi\|_{(m)}\leq C\|\psi\|_{H^{m}(h)}
\]
for a constant $C$ depending only on $m$ and $(\S,h)$. Turning to the opposite estimate, if $m=2k$, we can appeal to (\ref{eq:u Hk plus two Lu Hk}) $k$
times. This yields 
\begin{equation}\label{eq:psi spec m geo m}
  \|\psi\|_{H^{m}(h)}\leq C\|(-\Delta_{h}+1)^{k}\psi\|_{2}\leq C\|\psi\|_{(m)}
\end{equation}
for a constant $C$ depending only on $m$ and $(\S,h)$. Assume now that $m=2k+1$. Then, due to (\ref{eq:psi spec m geo m}),
\begin{equation}\label{eq:even part spec bd}
  \|\psi\|_{H^{2k}(h)}\leq C\|\psi\|_{(2k)}\leq C\|\psi\|_{(m)}
\end{equation}
for a constant $C$ depending only on $m$ and $(\S,h)$. Next, let $U$, $\{E_{i}\}$, $\zeta$ and $K$ be as described at the beginning
of Subsection~\ref{ssection:loc fr comm}. Then, due to (\ref{eq:psi spec m geo m}),
\begin{equation*}
  \begin{split}
    \|E_{\zeta,i}\psi\|_{H^{2k}(h)} \leq & C\|E_{\zeta,i}(-\Delta_{h}+1)^{k}\psi\|_{2}+C\|[(-\Delta_{h}+1)^{k},E_{\zeta,i}]\psi\|_{2}\\
    \leq & C\|\nabla^h(-\Delta_{h}+1)^{k}\psi\|_{2}+C\|\psi\|_{H^{2k}(h)}.
  \end{split}
\end{equation*}
Combining this estimate with (\ref{eq:spec norm k odd}) and (\ref{eq:even part spec bd}) yields the conclusion that
$\|\psi\|_{H^{m}(h)}\leq C\|\psi\|_{(m)}$ holds when $m\in\nn{}_0$ is odd. Combining the above estimates with a density argument yields the conclusion
that, given $m\in\nn{}_0$, there is a constant $C_m>1$ depending only on $m$ and $(\S,h)$ such that 
\begin{equation}\label{eq:HmHpmequiv}
  C_{m}^{-1}\|\psi\|_{H^{m}(h)}\leq \|\psi\|_{(m)}\leq C_{m}\|\psi\|_{H^{m}(h)}
\end{equation}
for all $\psi\in H^{m}(h)$.

\section{Smooth dependence on coefficients, elliptic model equation}\label{section:contdeponcoeffellmodeleq}

Let $(\S, h_\refer)$ be a closed, connected and oriented Riemannian manifold of dimension $n$. In what follows, we denote the set of elements in
$H^{k}[T^{(0,m)}T\S;h_\refer]$ that take their values in the set of symmetric covariant $m$-tensor fields by $H^{k}[S^{(0,m)}T\S;h_\refer]$.

\begin{thm} \label{thm: continuous dependence}
  Let $(\S, h_\refer)$ be a closed, connected and oriented Riemannian manifold of dimension $n$. Let $k\in\nn{}$ satisfy $k>n/2+1$; let
  $h_0 \in H^{k + 1}[S^{(0,2)}T\S;h_\refer]$ be a Riemannian metric; let $a_0 \in H^k(h_\refer)$ be a strictly positive function; and let
  $f_0 \in H^k(h_\refer)$. Then there is an open subset
  \[
  \U\subseteq H^{k+1}[S^{(0,2)}T\S;h_\refer] \oplus H^k(h_\refer) \oplus H^k(h_\refer)
  \]
  with $(h_0, a_0, f_0) \in \U$, and a smooth map $\Phi: \U \to H^{k + 2}(h_\refer)$ such that
  \[
  \left( -\Delta_{h}+a \right) \Phi(h, a, f)= f.
  \]
\end{thm}
\begin{remark}
  Since, due to Lemma~\ref{lemma:existence sol to elliptic PDE}, $\Phi$ is well defined on the set of all
  \[
  (h,a,f)\in H^{k+1}[S^{(0,2)}T\S;h_\refer] \oplus H^k(h_\refer) \oplus H^k(h_\refer)
  \]
  such that $h$ is a Riemannian metric and $a$ is strictly positive, we can assume $\U$ to coincide with this set. 
\end{remark}

Theorem \ref{thm: continuous dependence} implies for example the following:

\begin{cor} \label{cor: Ck regularity}
  Let $(\S, h_\refer)$ be a closed, connected and oriented Riemannian manifold of dimension $n$. Let $k$ be an integer such that $k >n/2+1$, let
  $m \in \nn{}_0$, let $t_1<t_2$ be real numbers and let 
  \begin{align*}
    \g:(t_1, t_2)&\to H^{k + 1}[S^{(0,2)}T\S;h_\refer] \oplus H^k(h_\refer) \oplus H^k(h_\refer)
  \end{align*}
  be $C^m$. Assume, moreover, that if $\g(t)=(h_t, a_t, f_t)$, then $h_t$ is positive definite and $a_t$ is strictly positive for all
  $t \in (t_1, t_2)$. Let $\lambda:(t_1, t_2) \to H^{k+2}(\S)$, be defined as follows: for $t\in (t_1,t_2)$, $\lambda(t)$ is the unique solution
  $v$ to $\left( -\Delta_{h_t} + a_t \right)v= f_t$. Then $\lambda$ is $C^m$.
\end{cor}

\begin{proof}[Proof of Corollary \ref{cor: Ck regularity}, assuming Theorem \ref{thm: continuous dependence}]
  That $\lambda$ is well defined follows from Lemma~\ref{lemma:existence sol to elliptic PDE}. Given any $\tau \in (t_1, t_2)$, the map $\Phi$ in
  Theorem~\ref{thm: continuous dependence} is smooth near $(h_\tau, a_\tau, f_\tau)$, so the map
  \begin{align*}
    (\tau-\de, \tau + \de)
    &\to H^{k+2}(\S);\ \ \ t\mapsto \lambda(t) = \Phi(h_t, a_t, f_t),
  \end{align*}
  is of class $C^m$ for some small $\de = \de(\tau) > 0$. Since this is true for all $\tau \in (t_1, t_2)$, the desired statement follows. 
\end{proof}

The main ingredient in the proof of Theorem~\ref{thm: continuous dependence} is the following:

\begin{lemma}[The implicit function theorem in Banach spaces, see \cite{AMR1988}, Theorem~2.5.7, p.~121] \label{le: IFT}
  Let $\BaE$, $\BaF$, $\BaG$ be Banach spaces and let $\U \subset \BaE$, $\V \subset \BaF$ be open subsets.
  Let $\Psi: \U \times \V \to \BaG$ be of class $C^r$ for some $r \geq 1$. 
  For some $x_0 \in \U$, $y_0 \in \V$, assume that
  \begin{equation} \label{eq: second slot derivative}
    D_2\Psi(x_0, y_0): \BaF \to \BaG
  \end{equation}
  is an isomorphism. Then there are neighborhoods $\U_0$ of $x_0$ and $\W_0$ of $\Psi(x_0, y_0)$ and a unique $C^r$ map $\Phi: \U_0 \times \W_0 \to \V$,
  such that for all $(x, w) \in \U_0 \times \W_0$, the equality $\Psi(x, \Phi(x, w))= w$ holds.
\end{lemma}

\begin{proof}[Proof of Theorem \ref{thm: continuous dependence}]
  The proof consists of an application of Lemma \ref{le: IFT}. We begin by choosing
  \begin{align*}
    \BaE := H^{k+1}[S^{(0,2)}T\S;h_\refer] \oplus H^k(h_\refer), \ \ \
    \BaF := H^{k+2}(h_\refer), \ \ \
    \BaG := H^k(h_\refer).
  \end{align*}
  We continue by choosing $\V=\BaF$ and $\U$ to be the set of $(h,a)\in \BaE$ such that $h$ is a Riemannian metric and $a$ is strictly
  positive. Let us now define the map $\Psi: \U \times \V\to \BaG$ by 
  \[
    \Psi((h, a), v):=\left( -\Delta_{h}+a \right) v.
  \]
  We leave it to the reader to verify that $\Psi$ is a smooth map. Next, let $x_0= (h_0, a_0) \in \U$ and $y_0:=v_0 \in \V$, where $v_0$ is
  the unique solution to
  \[
  \left( -\Delta_{h_0} + a_0 \right) v_0= f_0;
  \]
  see Lemma~\ref{lemma:existence sol to elliptic PDE}. Compute
  \[
  D_2 \Psi(x_0, y_0)(v)	= \left( -\Delta_{h_0} + a_0 \right) v.
  \]
  Thus $D_2\Psi(x_0, y_0): \BaF \to \BaG$ is an isomorphism due to Lemma~\ref{lemma:existence sol to elliptic PDE}.  
  The implicit function theorem, Lemma \ref{le: IFT}, gives a smooth map
  \[
  \Phi: \U_0 \times \W_0 \subset \BaE\oplus \BaG \to \BaF,
  \]
  where $\U_0$ is an open neighborhood of $x_0 = (h_0, a_0)$ and $\W_0$ is an open neighborhood of $\Psi(x_0, y_0) = f_0$, such that
  \[
  \Psi((h, a), \Phi((h, a), f)) = f,
  \]
  for all $(h, a) \in \U_0$ and $f \in \W_0$. In other words, $\Phi$ takes $((h, a), f)$ to the unique solution $v$ to 
  \[
  \left( -\Delta_{h}+a \right) v = f.
  \]
  This completes the proof.
\end{proof}

\end{document}